%% file: even-faster-classical-matching.tex
\newif\ifabstract
\abstracttrue
\newif\iffull
\ifabstract \fullfalse \else \fulltrue \fi
\documentclass[11pt]{article}
\usepackage{amsfonts}
\usepackage{amssymb}
\usepackage{amstext}
\usepackage{amsmath}
\usepackage{xspace}
\usepackage{ntheorem}
\usepackage{graphicx}
\usepackage{url}
\usepackage{graphics}
\usepackage{colordvi}
\usepackage{hyperref}
\usepackage{subfigure}
\usepackage{xcolor}
\usepackage{cleveref}

\advance \oddsidemargin by -0.8in
\newcommand{\myparskip}{3pt}
\parskip \myparskip

\newcommand{\revG}{\overleftarrow{G}}
\newcommand{\revH}{\overleftarrow{H}}

\newcommand{\cCMG}{c_{\mathsf{CMG}}}

\newcommand{\otilde}{\widetilde O}
\newcommand{\ohat}{\widehat O}

\newcommand{\LCDS}{\mbox{\sf{Layered Connectivity Data Structure}}\xspace}
\newcommand{\MBM}{\mbox{\sf{Maximum Bipartite Matching}}\xspace}
\newcommand{\MWU}{\mbox{\sf{MWU}}\xspace}

\newcommand{\SSSP}{\mbox{\sf{SSSP}}\xspace}
\newcommand{\stSP}{\mbox{$s$-$t$-\sf{SP}}\xspace}

\newcommand{\tA}{\tilde A}
\newcommand{\tB}{\tilde B}
\newcommand{\APSP}{\mbox{\sf{APSP}}\xspace}

\newcommand{\reconnect}{\ensuremath{\mathsf{ReconnectLayer}}\xspace}

\newcommand{\alg}{\ensuremath{\mathsf{Alg}}\xspace}

\newcommand{\bad}{\mathsf{bad}}

\newcommand{\spec}{\mathsf{spec}}

\newcommand{\val}{\operatorname{val}}
\newcommand{\ceil}[1]{\ensuremath{\left\lceil#1\right\rceil}}
\newcommand{\floor}[1]{\ensuremath{\left\lfloor#1\right\rfloor}}

\newcommand{\event}{{\cal{E}}}

\newcommand{\opt}{\mathsf{OPT}}

\newcommand{\set}[1]{\left\{ #1 \right\}}

\newcommand{\tset}{{\mathcal T}}

\newcommand{\pset}{{\mathcal{P}}}
\newcommand{\qset}{{\mathcal{Q}}}

\newcommand{\aset}{{\mathcal{A}}}
\newcommand{\cset}{{\mathcal{C}}}

\newcommand{\mset}{{\mathcal M}}

\newcommand{\xset}{{\mathcal{X}}}

\newcommand{\yset}{{\mathcal{Y}}}
\newcommand{\rset}{{\mathcal{R}}}

\newcommand{\sset}{{\mathcal{S}}}

\newcommand{\dset}{{\mathcal{D}}}

\newcommand{\tw}{\tilde w}

\newcommand{\be}{\begin{enumerate}}
\newcommand{\ee}{\end{enumerate}}
\newcommand{\bd}{\begin{description}}
\newcommand{\ed}{\end{description}}
\newcommand{\bi}{\begin{itemize}}
\newcommand{\ei}{\end{itemize}}

\newtheorem{theorem}{Theorem}[section]
\newtheorem{lemma}[theorem]{Lemma}
\newtheorem{observation}[theorem]{Observation}
\newtheorem{corollary}[theorem]{Corollary}
\newtheorem{claim}[theorem]{Claim}

\newtheorem{definition}[theorem]{Definition}
\newenvironment{proof}{\par \smallskip{\bf Proof:}}{\hfill\stopproof}
\def\stopproof{\square}
\def\square{\vbox{\hrule height.2pt\hbox{\vrule width.2pt height5pt \kern5pt
\vrule width.2pt} \hrule height.2pt}}

\newenvironment{proofof}[1]{\noindent{\bf Proof of #1.}}%
        {\hfill\stopproof}

\newcommand{\algmwu}{{\sc ALG-MWU}\xspace}

\newenvironment{prog}[1]{
\begin{minipage}{5.8 in}
\begin{center}
{\sc #1}
\end{center}
}
{
\end{minipage}
}

\newcommand{\program}[3]{\begin{figure} \fbox{\vspace{2mm}\begin{prog}{#1} #3 \end{prog}\vspace{2mm}} 
			\caption{#1 \label{#2}} \end{figure}}
\renewcommand{\phi}{\varphi}
\newcommand{\eps}{\epsilon}

\newcommand{\half}{\ensuremath{\frac{1}{2}}}

\newcommand{\poly}{\operatorname{poly}}
\newcommand{\dist}{\mbox{\sf dist}}

\newcommand{\expect}[2][]{\text{\bf E}_{#1}\left [#2\right]}
\newcommand{\prob}[2][]{\text{\bf Pr}_{#1}\left [#2\right]}

\newenvironment{properties}[2][0]
{
\begin{enumerate} \setcounter{enumi}{#1}}{\end{enumerate}}

\newcommand{\mynote}[2][red]{}
\newcommand{\sknote}[2][red]{}
\newcommand{\vol}{\operatorname{Vol}}

\newcommand{\notZ}{\overline{Z}}
\newcommand{\notT}{\overline{T}}
\newcommand{\attime}[1][\tau]{^{(#1)}}

\newcommand{\skipp}{\operatorname{skip}}
\newcommand{\spann}{\operatorname{span}}

\newcommand{\dmax}{\delta_{\mbox{\textup{\footnotesize{max}}}}}

\newcommand{\DLSSSP}{{\sf {DAG-like SSSP}}\xspace}
\newcommand{\ATO}{\ensuremath{\mathcal{ATO}}\xspace}

\newcommand{\DSin}{\mbox{\sf DS}^{\operatorname{in}}}
\newcommand{\DSout}{\mbox{\sf DS}^{\operatorname{out}}}

\newcommand{\OUT}{\operatorname{OUT}}
\newcommand{\INN}{\operatorname{IN}}
\newcommand{\out}{\operatorname{out}}
\newcommand{\inn}{\operatorname{in}}

\newcommand{\shortpath}{\mbox{\sf{short-path}}\xspace}

\newcommand{\pquery}{\mbox{\sf{path-query}}\xspace}

\newcommand{\processcut}{\ensuremath{\mathsf{ProcessCut}}\xspace}
\newcommand{\processcutt}{\ensuremath{\mathsf{ProcessCut}'}\xspace}

\newcommand{\constructexpander}{\ensuremath{\mathsf{ConstructExpander}}\xspace}

\newcommand{\routeandcut}{\ensuremath{\mathsf{RouteAndCut}}\xspace}
\newcommand{\connecttocenters}{\ensuremath{\mathsf{ConnectToCenters}}\xspace}

\newcommand{\maintaincluster}{\ensuremath{\mathsf{MaintainCluster}}\xspace}
\newcommand{\maintainspeccluster}{\ensuremath{\mathsf{MaintainSpecialCluster}}\xspace}
\newcommand{\maintainspecialcluster}{\ensuremath{\mathsf{MaintainSpecialCluster}}\xspace}

\newcommand{\EST}{\mbox{\sf{ES-Tree}}\xspace}

\newcommand{\tc}{\tilde c}
\newcommand{\hn}{\hat n}
\begin{document}

\begin{titlepage}

%\title{An $n^{2+o(1)}$-Time Combinatorial Algorithm for Maximum Bipartite Matching}
\title{Maximum Bipartite Matching in $n^{2+o(1)}$ Time via a Combinatorial Algorithm}
\author{Julia Chuzhoy\thanks{Toyota Technological Institute at Chicago. Email: {\tt cjulia@ttic.edu}. Supported in part by NSF grant CCF-2006464 and NSF HDR TRIPODS award 2216899.}\and Sanjeev Khanna \thanks{University of Pennsylvania, Philadelphia, PA. Email: {\tt sanjeev@cis.upenn.edu}. Supported in part by NSF awards CCF-1934876 and CCF-2008305.}}
	\maketitle
	
\pagenumbering{gobble}
	
	\thispagestyle{empty}
	\begin{abstract}
	
Maximum bipartite matching (MBM) is a fundamental problem in combinatorial optimization with a long and rich history. A classic result of Hopcroft and Karp~\cite{HK73} provides an $O(m \sqrt{n})$-time algorithm for the problem, where $n$ and $m$ are the number of vertices and edges in the input graph, respectively. For dense graphs, an approach based on fast matrix multiplication achieves a running time of $O(n^{2.371})$~\cite{IM81, MS04}. For several decades, these results represented state-of-the-art algorithms, until, in 2013, Madry introduced a powerful new approach for solving MBM using continuous optimization techniques. This line of research, that builds on continuous techniques based on interior-point methods,  led to several spectacular results, culminating in a breakthrough $m^{1+o(1)}$-time algorithm for min-cost flow, that implies an $m^{1+o(1)}$-time algorithm for MBM as well. 

These striking advances naturally raise the question of whether combinatorial algorithms can match the performance of the algorithms that are based on continuous techniques for MBM. 
One reason to explore combinatorial algorithms is that they are often more transparent than their continuous counterparts, and that the tools and techniques developed for such algorithms may be useful in other settings, including, for example, developing faster algorithms for maximum matching in general graphs.  
A recent work of Chuzhoy and Khanna~\cite{CK24} made progress on this question by giving a combinatorial $\tilde{O}(m^{1/3}n^{5/3})$-time algorithm for MBM, thus outperforming both the Hopcroft-Karp algorithm and matrix multiplication based approaches, on sufficiently dense graphs. Still, a large gap remains between the running time of their algorithm and the almost linear-time achievable by algorithms based on continuous techniques. In this work, we take another step towards narrowing this gap, and present a randomized $n^{2+o(1)}$-time combinatorial algorithm for MBM. Thus in dense graphs, our algorithm essentially matches the performance of algorithms that are based on continuous methods.

Similar to the classical algorithms for MBM  and the approach of~\cite{CK24}, our algorithm is based on iterative augmentation of  a current matching using augmenting paths in the corresponding (directed) residual flow network.  Our main contribution is a recursive algorithm that exploits the special structure of the resulting flow problem to recover an $\Omega(1/\log^2 n)$-fraction of the remaining augmentations in $n^{2+o(1)}$ time. 

Finally, we obtain a randomized $n^{2+o(1)}$-time algorithm for maximum vertex-capacitated $s$-$t$ flow in directed graphs when all vertex capacities are identical, using a standard reduction from this problem to MBM. 
\end{abstract}

\newpage

\tableofcontents{}
\end{titlepage}

\pagenumbering{arabic}

\input{intro}

\input{prelims}

\input{high-level-overview}
\input{expander-tools}

\input{SSSP-alg}

\input{connecttocenters}

\input{maintaincluster-main}
\newpage
\appendix
\input{appx-expander-tools}

\input{appx-routeandcut-from-maintaincluster}
\input{appx-proofs-for-connectivity}

%\input{cut-player-alg}
%\input{appx-ball-growing2}
\input{appx-finalizing-cut}

\bibliographystyle{alpha}

\bibliography{faster-classical-matching}

\end{document}

%% file: intro.tex
\section{Introduction}

We consider the classical \MBM problem, where the goal is to compute a maximum-size matching in the given input bipartite graph $G$. Throughout, we denote the number of vertices and the number of edges in $G$ by $n$ and $m$, respectively. \MBM 
is one of the most central and extensively studied problems in computer science and related disciplines, with connections to many other fundamental graph optimization problems. % (see e.g.\cite{Sch03}). 

It is well known that \MBM can be reduced to computing a maximum $s$-$t$ flow in a directed flow network with unit edge capacities. The  Ford-Fulkerson algorithm~\cite{FF56} for maximum $s$-$t$ flow then immediately implies an $O(mn)$-time algorithm for \MBM. The algorithm is conceptually simple, and maintains a matching $M$, starting with $M=\emptyset$. As long as $M$ is not optimal, we can augment it by computing an $s$-$t$ path in the resulting residual flow network. A celebrated work of Hopcroft and Karp~\cite{HK73} provides a significantly more efficient $O(m \sqrt{n})$-time implementation of this idea by iteratively computing a maximal collection of internally disjoint augmenting  $s$-$t$ paths of shortest possible length in the residual flow network. This result remained the fastest known algorithm for several decades, except for the special case of very dense graphs, where fast matrix multiplication techniques were shown to yield an $O(n^{\omega})$-time algorithm~\cite{IM81, MS04}. 
Starting in 2008, a new paradigm emerged, namely, the use of continuous techniques as a method for obtaining fast algorithms for various flow problems, that ultimately revolutionized the field. As a first illustration of this paradigm,
the work of Daitch and Spielman~\cite{DS08}, building on the breakthrough result of Spielman and Teng~\cite{ST04} for efficiently solving Laplacian systems, gave an $\otilde(m^{3/2})$-time algorithm for directed maximum $s$-$t$ flow. 
Later, Madry~\cite{Madry13}  used this paradigm to design an algorithm for directed maximum $s$-$t$ flow with $\tilde{O}(m^{10/7})$ running time, obtaining the first substantial improvement over the algorithm of Hopcroft and Karp for \MBM in sparse graphs. A sequence of remarkable developments~\cite{Madry16, LS19, CMSV17, LS20_stoc, AMV20,BrandLNPSS0W20} recently culminated in a deterministic  $m^{1+o(1)}$-time algorithm for directed maximum $s$-$t$ flow  \cite{ChenKLPGS22,brand2023deterministic}, thereby providing an almost linear-time algorithm for \MBM. In all these recent algorithms, the directed flow problem is cast as a linear program, which is then solved via interior-point methods (IPM). In every iteration of the IPM, one needs to either solve a Laplacian system, or another efficiently solvable problem on undirected graphs, such as min-ratio cycle in~\cite{ChenKLPGS22}. This approach is further combined with dynamic graph data structures to make it even more efficient. 

In view of this recent history, it is natural to ask whether combinatorial techniques can be used to design algorithms for \MBM (and also other flow-like problems), whose performance matches that of algorithms that are based on continuous methods. There are several reasons to focus on combinatorial techniques. First, they tend to be more transparent than their continuous counterparts. Second, it is likely that tools and techniques that are developed in order to design a combinatorial algorithm for as fundamental a problem as \MBM will prove useful in other applications. Lastly, while continuous techniques led to an $m^{1+o(1)}$-time algorithm for  \MBM, the landscape of fast algorithms for the Maximum Matching problem in general graphs did not benefit from these developments. In dense graphs, a fast-matrix multiplication based approach gives $O(n^{2.371})$-time algorithm for Maximum Matching in general graphs~\cite{IM81, MS04}. More interestingly, in sparse to moderately dense graphs, the best known runtime still stands on $\tilde{O}(m \sqrt{n})$~\cite{MicaliV80, Vazirani94, GoldbergK04, Gabow17} and utilizes an augmenting-paths based approach, similar to that used in combinatorial algorithms for \MBM.

In a very recent work, Chuzhoy and Khanna~\cite{CK24} made progress on narrowing the striking gap between the performance of combinatorial and IPM-based approaches for \MBM, by providing a combinatorial deterministic $\tilde{O}(m^{1/3}n^{5/3})$-time algorithm, thus outperforming both the Hopcroft-Karp algorithm, and the matrix multiplication based approaches on sufficiently dense graphs.  Still, a large gap remains between the performance of the best combinatorial algorithms and the almost linear-time achievable by algorithms based on continuous techniques. In particular, on dense graphs, the performance gap incurred by the current best combinatorial algorithm is $\Omega(n^{1/3})$. In this work, we take another step towards narrowing this performance gap, and essentially eliminate it in  dense graphs. 
Our main result is summarized in the following theorem.

\begin{theorem}\label{thm:main}
There is a randomized combinatorial algorithm for the \MBM problem, that, given an $n$-vertex bipartite graph $G$, outputs a maximum matching $M$ in $G$ with probability at least $1 - 1/\poly(n)$. The running time of the algorithm is $O\left(n^2\cdot 2^{O(\sqrt{\log n}\cdot \log \log n)}\right )$.
\end{theorem}

Our algorithm outperforms the Hopcroft-Karp algorithm on graphs with $\omega(n^{1.5})$ edges, and in dense graphs, it essentially matches the performance of algorithms based on continuous techniques. Furthermore, in almost all edge density regimes, this algorithm outperforms the runtime achieved in~\cite{CK24}.

Using a standard reduction from vertex-capacitated flow in directed graphs to \MBM (see Theorem 16.12 in~\cite{Sch03}, for instance), we also obtain a combinatorial algorithm with similar running time for maximum vertex-capacitated flow when all vertex capacities are identical. 

\begin{corollary}
\label{cor:main}
There is a randomized combinatorial algorithm for the directed maximum $s$-$t$ flow problem with uniform vertex capacities, that given an $n$-vertex directed graph $G$, outputs a maximum $s$-$t$ flow with probability at least $1 - 1/\poly(n)$. The running time of the algorithm is $O\left(n^2\cdot 2^{O(\sqrt{\log n}\cdot \log \log n)}\right )$.
\end{corollary}

Similarly to the classical algorithms for \MBM, our approach for proving Theorem~\ref{thm:main} is based on iteratively augmenting a current matching using augmenting paths in the residual flow network. We employ the multiplicative weights update (\MWU) framework, that effectively reduces the underlying flow problem to decremental single-source shortest paths (\SSSP) in directed graphs, a connection first observed by Madry~\cite{madry2010faster} and also used in~\cite{CK24}. As observed in~\cite{CK24}, this reduction results in a {\em special case} of decremental \SSSP that appears significantly easier than general decremental directed \SSSP. Our main contribution is a recursive algorithm that exploits the special structure of the resulting flow problem to recover an $\Omega(1/\log^2 n)$-fraction of the remaining augmentations in $n^{2+o(1)}$ time. We abstract this task as a problem called \routeandcut, where the input is a directed graph $G$, two disjoint sets $A,B$ of its vertices with $|A|\leq |B|$, and two additional parameters $1\leq \eta\leq \Delta$. The goal is to either compute a collection $\pset$ of at least $\Omega(\Delta/\poly\log n)$ paths that connects distinct vertices of $A$ to distinct vertices of $B$ with vertex-congestion at most $\eta$; or to output a cut that approximately certifies infeasibility of the desired routing. Our main result is a randomized algorithm for \routeandcut, whose running time is bounded by $n^{1+o(1)}\cdot (n-|B|)$. It is worth highlighting that when $|B|$ is sufficiently large, this running time may be much smaller than $|E(G)|$. This performance gain for large sets $B$ serves as a crucial building block for our $n^{2+o(1)}$-time algorithm.

As in the work of~\cite{CK24}, the task of efficiently solving \routeandcut  in turn relies on an efficient algorithm for maintaining an expander in a dynamically changing graph, a problem that we refer to as \maintaincluster. One key contribution of our work is the introduction of a {\em parameterized version} of both these problems that allows us to use a bootstrapping approach in the design of our algorithm, where we exploit efficient algorithms for one problem to obtain efficient algorithms for the other problem and vice versa. Another key technical contribution is an efficient algorithm for a new problem that we introduce, called \connecttocenters, that plays a central role in efficiently maintaining short paths from all vertices of a given graph $G$ to a pre-specified collection of ``center'' vertices, even as $G$ undergoes online updates. This problem may be viewed as a representative abstraction of a common paradigm used in many graph algorithms, in which an expander is embedded into the input graph, and all graph vertices are then routed to the vertices of the expander.  As such, our algorithm for this problem may prove useful in other applications. 
Finally, another insight utilized by our algorithm is an explicit recognition of the fact that each iteration of the \MWU algorithm leads to very specific kind of updates in the underlying \SSSP instance, namely, well-behaved increases in the lengths of some edges. While these length increases can easily be simulated as edge deletions, a black-box simulation as an instance of decremental \SSSP gives away some of the inherent algorithmic advantages offered by these special kind of updates that our algorithm exploits. We give a detailed overview of our algorithm and its comparison to the algorithm of~\cite{CK24} in the next subsection.

We conclude by noting that in addition to the conceptual simplicity of a combinatorial augmenting path based approach to solve \MBM, the techniques developed here for speeding up augmentations may also prove useful in obtaining faster algorithms for Maximum Matching in general graphs. We also believe that some of the technical tools that we introduce, such as an efficient algorithm for the \connecttocenters problem that we describe in more detail below, are of independent interest.

\input{intro-techniques}

\subsection{Organization}
We start with preliminaries in \Cref{sec: prelims}. We then provide a high-level overview of our algorithm in \Cref{sec: high level}, where we also define the \routeandcut and the \stSP problems, state our main results for them, and describe the modification of the \MWU method that we use. In \Cref{sec: expander tools} we define the \maintaincluster problem and state our main result for it. We also describe our recursive approach, state the theorems describing the reductions between the \maintaincluster and the \routeandcut problems, and complete the algorithm for the \MBM problem using them. Then in \Cref{sec: SSSP alg} we describe our reduction from the  $(r+1)$-restricted \routeandcut problem to the $r$-restricted \maintaincluster problem. In \Cref{sec: layered connectivity DS} we describe our algorithm for the \connecttocenters problem, and in \Cref{sec: alg for maintaincluster from routeandcut} we complete our reduction from the $r$-restricted \maintaincluster problem to the  $r$-restricted \routeandcut problem.

%% file: intro-techniques.tex
\subsection{Our Techniques}

Our algorithm builds on and extends the techniques of \cite{CK24}, which, in turn, build on the algorithm of  \cite{SCC} for the directed decremental Single-Source Shortest Path (\SSSP) problem. We start with a high-level overview of the algorithm of \cite{CK24}, and then provide the description of our improved algorithm.

It is well known that the \MBM problem in a graph $G$ can be equivalently cast as the problem of computing an integral maximum $s$-$t$ flow in a corresponding directed flow network $G'$ with unit edge capacities. 
We can view any given matching $M$ in $G$ as defining an $s$-$t$ flow $f$ in $G'$ of value $|M|$. We let $H=G'_f$ be the corresponding residual flow network, that we refer to as the \emph{residual flow network corresponding to matching $M$}. 
We note that the residual flow network $H$ has a special structure: namely, each vertex in $H$ has in-degree $1$ or out-degree $1$. Therefore, if $\pset$ is a collection of paths in $H$ causing edge-congestion at most $\eta$, then the paths in $\pset$ cause vertex-congestion at most $\eta$ and vice versa. For all problems that we define below, we assume that their input graph also has this special structure. For convenience, we will focus on edge-congestion, and on edge-based cuts in such graphs. We also note that any directed graph can be converted into a graph with this special structure by replacing every vertex $v$ with a pair $v^+,v^-$ of new vertices, such that all edges that enter $v$ become incident to $v^-$, and all edges leaving $v$ become incident to $v^+$, and inserting the edge $(v^-,v^+)$ into the graph.

The residual network $H$ corresponding to a matching $M$ contains $\Delta=\opt-|M|$ edge-disjoint $s$-$t$ paths, where $\opt$ is the value of the maximum bipartite matching. 
Suppose now that we can design an algorithm that computes a collection $\pset$ of  $\Omega(\Delta/\poly\log n)$ $s$-$t$ paths in $H$, that cause $O(\poly\log n)$ edge-congestion. Using standard methods, we can then efficiently recover a collection $\pset'$ of $\Omega(\Delta/\poly\log n)$ edge-disjoint $s$-$t$ paths in $H$, which can in turn be used in order to augment the current matching $M$, thereby obtaining a new matching $M'$ of cardinality $|M|+\Omega(\Delta/\poly\log n)$. In other words, $\opt-|M'|\leq (\opt-|M|)\cdot (1-1/\poly\log n)$, so the gap between the optimal solution value and the size of the matching the algorithm maintains reduces by at least factor $(1-1/\poly\log n)$. It is then easy to verify that, after $O(\poly\log n)$ such iterations, the algorithm obtains an optimal matching. This is precisely the high-level approach that was used by \cite{CK24}, and we follow it in this work as well.  In order to obtain an algorithm for \MBM, it is now sufficient to design a procedure that, given a residual flow network $H$ corresponding to the current matching $M$, efficiently computes the set $\pset$ of $s$-$t$ paths in $H$ with the above properties.

For technical reasons that will become clear later, we define a slightly more general problem, that we call \routeandcut. In this problem, the input is a directed graph $H$, two disjoint sets $A,B$ of its vertices with $|A|\leq |B|$, and two parameters $1\leq \Delta\leq\min\set{|A|,|B|}$ and $1\leq \eta\leq \Delta$. The goal is to either compute a collection $\qset$ of $\Omega(\Delta/\poly\log n)$ paths, each of which connects a distinct vertex of $A$ to a distinct vertex of $B$, such that the paths in $\qset$ cause congestion $\tilde O(\eta)$; or to compute a cut $(X,Y)$ in $H$ with $|E_H(X,Y)|\ll \Delta/\eta$, with $X$ containing a large fraction of vertices of $A$, and $Y$ containing a large fraction of the vertices of $B$. While \cite{CK24} do not explicitly define and solve this problem, their algorithm can be adapted to solve a special case of \routeandcut, where $\eta\leq O(\poly\log n)$. So for brevity, we will say that the algorithm of \cite{CK24} solves this special case of \routeandcut. Clearly, an algorithm for the \routeandcut problem can be used in order to compute a collection $\pset$ of paths in the residual flow network $H$ corresponding to the current matching $M$, with the desired properties that we described above.

The \routeandcut problem falls into the broader and the extensively studied class of graph routing and flow problems. 
One standard approach for obtaining fast algorithms for such problems, due to \cite{GK98, Fleischer00,AAP93}, is via the Multiplicative Weight Update (\MWU) method. It was further observed by Madry \cite{madry2010faster}  that this approach can be viewed as reducing a given flow problem to a variant of decremental \SSSP or \APSP.  In our case, the reduction is to decremental \SSSP in directed graphs. While strong lower bounds are known for exact algorithms for decremental \SSSP and \APSP (see, e.g. \cite{DorHZ00,RodittyZ11,HenzingerKNS15,abboud2022hardness,DBLP:conf/stoc/AbboudBF23}), we can exploit the special properties of the \SSSP instances that arise from the \routeandcut problem in order to obtain faster algorithms, an approach that was also used by \cite{CK24}.

We note that \cite{AlmostDAG2} provided a  $(1+\eps)$-approximation algorithm for directed decremental \SSSP with total update time  $\otilde(n^2/\eps)$, assuming all edge lengths are poly-bounded. Unfortunately, their algorithm can only withstand an \emph{oblivious adversary}; in other words, the sequence of edge-deletions from the input graph $G$ must be fixed in advance and may not depend on the algorithm's responses to queries. Instances of decremental \SSSP arising from the \MWU framework crucially require algorithms that can withstand an \emph{adaptive adversary}, where the choice of the edge to be deleted in every update may depend on the algorithm's past behavior.
A recent work of \cite{SCC} provided a $(1+\eps)$-approximation algorithm for the directed decremental \SSSP problem with an adaptive adversary, that achieves total update time $O\left (\frac{n^{8/3+o(1)}}{\eps}\right )$ (assuming that all edge lengths are poly-bounded).  While this total update time  is too high for speeding up algorithms for \MBM, their approach was adapted by \cite{CK24} to handle the specific instances of \SSSP that they obtain, leading to faster algorithms for \MBM. Specifically, one of the key observations of \cite{CK24} is that the \SSSP instances that arise from applying the \MWU method to the \MBM problem have the property that all queries are between a fixed pair $(s,t)$ of vertices, and a rather large approximation factor is acceptable. Moreover, by slightly modifying the standard \MWU framework, they ensure that it is sufficient that the algorithm for the \SSSP problem only responds to shortest-path queries as long as the current graph $H$ contains a collection of least $\Omega(\Delta/\poly\log n)$ disjoint and  short $s$-$t$ paths, where $\Delta$ is the target number of augmenting paths that the algorithm aims to produce. We also follow their approach, and reduce the \routeandcut problem, via a slightly modified \MWU method, to a special case of directed decremental \SSSP, that we refer to as decremental  \stSP, that has all of the above properties. We note that \stSP is somewhat more general than the special case of the \SSSP problem that was considered in \cite{CK24}, since they only provide an algorithm for the special case of \routeandcut where $\eta\leq O(\poly\log n)$, while we need an algorithm that works for a wider range of values of parameter $\eta$. For now we focus on the description of their algorithm, and we assume for simplicity that $\eta=1$ in this discussion.

The algorithm of \cite{CK24} for a special case of the decremental \stSP problem follows the high-level approach of \cite{SCC}, that consists of two parts. First, they maintain a partition $\xset$ of graph $H\setminus \set{s,t}$ into a collection of \emph{expander-like graphs}; we abstract the problem of maintaining each such graph, that we call \maintaincluster problem, below. 
Intuitively, the \maintaincluster subroutine is given as input a vertex-induced subgraph $H'$ of $H$, and a distance parameter $d$, with $H'$ undergoing an online sequence of edge deletions. It needs to efficiently support \shortpath queries in $H'$: given a pair $x,y\in V(H')$ of vertices, return an $x$-$y$ path of length at most $d$ in $H'$. However, it may, at any time, produce a cut $(A,B)$ in $H'$ of sparsity at most $O\left(\frac{\poly\log n}{d}\right )$, after which the vertices on one side of the cut are deleted from $H'$, and the algorithm needs to continue with the resulting graph.
The second main ingredient in the algorithm of \cite{CK24} is the Approximate Topological Order (\ATO) framework of \cite{AlmostDAG2} (which is in turn based on the works of \cite{gutenberg2020decremental} and \cite{Bernstein}), combined with the algorithm of \cite{AlmostDAG2} for decremental \SSSP on ``almost'' DAG's. The latter algorithm is applied to the graph $\hat H$, that is obtained from $H$ by contracting every almost-expander $X\in \xset$ into a single vertex. We now discuss each of these components in turn, starting with the \ATO framework.

\paragraph{The \ATO framework.}
The Approximate Topological Order (\ATO) framework of  \cite{AlmostDAG2}\footnote{According to \cite{AlmostDAG2}, while the \ATO technique was first explicitly defined in  \cite{AlmostDAG2}, the work of \cite{gutenberg2020decremental} can be viewed as (implicitly) using this technique.} is a central component in the algorithms of \cite{SCC,CK24}, as well as our algorithm. The framework was first introduced as the means of reducing the decremental \SSSP problem on general directed graphs to decremental \SSSP on ``almost DAG's''. 

An \ATO data structure in a dynamic graph $H$ must maintain a {\em partition} $\xset$ of the vertices of $H$ into subsets. We  refer to the sets $X\in \xset$ as \emph{clusters}, and to $\xset$ as a \emph{clustering}. The only allowed changes to the clustering $\xset$ are \emph{cluster splittings}: given an existing cluster $X\in \xset$ and a subset $X'\subseteq X$ of its vertices, delete the vertices of $X'$ from $X$, and add $X'$ as a new cluster to $\xset$. We assume further that the input graph $H$ contains two special vertices $s$ and $t$, and that clusters $S=\set{s}$ and $T=\set{t}$ always lie in $\xset$. In addition to maintaining the clustering $\xset$, the \ATO must maintain an \emph{ordering} $\sigma$ of its clusters. Assume that $\xset=\set{X_1,\ldots,X_r}$, and that the clusters are indexed according to the ordering $\sigma$. Assume further that a cluster $X_i$ undergoes splitting, with the new cluster $X'_i\subseteq X_i$  inserted into $\xset$. Then the ordering $\sigma$ must evolve in a \emph{consistent} manner, that is, the new ordering must be either $(X_1,\ldots,X_{i-1}, X'_i,X_i\setminus X'_i,X_{i+1},\ldots,X_r)$, or $(X_1,\ldots,X_{i-1},X_i\setminus X'_i, X'_i,X_{i+1},\ldots,X_r)$. Consider now some edge $e=(x,y)$ of $H$, and assume that $x\in X_i$ and $y\in X_j$. If $X_i$ appears before $X_j$ in the ordering $\sigma$, then we say that $e$ is a \emph{left-to-right} edge; if $i=j$, we say that it is a \emph{neutral} edge; and otherwise we say that it is a \emph{right-to-left} edge. If $e$ is a right-to-left edge, then we define its \emph{span}: $\spann(e)=\sum_{i'=j}^i|X_{i'}|$ (we assume here that the sets in $\xset$ are indexed according to the ordering $\sigma$). 
Let $\hat H$ be the contracted graph corresponding to $H$: that is, $\hat H$ is obtained from $H$ by contracting each of the clusters $X\in \xset$ into a vertex $v_X$. For simplicity, we will refer to the vertices $v_S$ and $v_T$ as $s$ and $t$, respectively.
Intuitively, if we could maintain the \ATO without introducing any right-to-left edges, then the corresponding contracted graph $\hat H$ is a DAG, and the ordering $\sigma$ associated with the \ATO naturally defines a topological ordering of the vertices of $\hat H$.
We could then use the algorithm of \cite{AlmostDAG2}  for decremental \SSSP in DAG's, that builds on the work of   \cite{Bernstein,gutenberg2020decremental}, in order to support approximate shortest path queries in $\hat H$ between $s$ and other vertices of $\hat H$, with total update time $\otilde(n^2)$. But in order to be able to support shortest $s$-$t$ path queries in the original graph $H$, we need to ensure that the diameters of the subgraphs $H[X]$ corresponding to the clusters $X\in \xset$ are sufficiently small, and that we can support approximate shortest-path queries between {\em arbitrary pairs} of vertices within each such graph efficiently.

Towards this end, it was observed by \cite{AlmostDAG2} 
that the algorithm for decremental \SSSP in DAG's can be further extended to ``almost DAG's'': suppose $G$ is a directed graph, and let $\rho$ be a fixed ordering of its vertices. Assume that $V(G)=\set{v_1,\ldots,v_n}$, where the vertices are indexed according to the ordering $\rho$. If $e=(v_i,v_j)$ is an edge with $i>j$, then we say that $e$ is a \emph{right-to-left} edge of $G$, and we say that its \emph{width} is $i-j$. It was shown in \cite{AlmostDAG2} that the algorithm for decremental \SSSP on DAG's could be extended to such graphs $G$, provided that the total width of all right-to-left edges is relatively small. Specifically, the running time of their algorithm becomes roughly $\tilde O(n^2+\Gamma\cdot n)$, where $\Gamma$ is the total width of the right-to-left edges.

Assume now that the algorithm for \SSSP on the almost-DAG graph $G$ only needs to respond to approximate shortest-path queries between a specific fixed pair $s,t$ of vertices, and assume, moreover, that it is only needed to support such queries as long as $G$ contains $\Omega(\Delta)$ short edge-disjoint $s$-$t$ paths. It was observed in \cite{CK24} that, in such a case, the running time of the algorithm of \cite{AlmostDAG2} improves to roughly 
$\tilde O(n^2+\Gamma\cdot n/\Delta)$. This observation was one of the key insights that allowed them to obtain a faster running time for the special case of the \stSP problem, and for \MBM.

We now provide additional relevant details of the algorithm of \cite{CK24}. 
Like in \cite{SCC}, the \maintaincluster problem is exploited in order to maintain an \ATO of the input graph $H$. 
Initially, the clustering $\xset$ contains three clusters: $S=\set{s}$, $T=\set{t}$, and $U=V(H)\setminus\set{s,t}$. The algorithm for the \maintaincluster problem is then initialized on graph $H[U]$, with an appropriately chosen distance parameter $d_U$. In general, whenever a new cluster $X$ joins $\xset$,  the algorithm for the \maintaincluster problem is initialized on $H[X]$. Whenever that algorithm produces a sparse cut $(A,B)$ in $X$, we select a subset $Z\in \set{A,B}$ of vertices to be deleted from $X$, update $X$ by deleting these vertices, and add $Z$ as a new cluster to $\xset$, after which the algorithm for the \maintaincluster problem is initialized on $H[Z]$. The key idea is that, since the cuts produced by the algorithm for the \maintaincluster problem are sparse, we can ensure that the total span of all right-to-left edges is sufficiently small. If we then consider the contracted graph $\hat H$, this, in turn, ensures that the total width of all right-to-left edges in $\hat H$ is low. We can now apply the algorithm of \cite{AlmostDAG2} for decremental \SSSP on almost-DAG's to support approximate shortest $s$-$t$ path queries in $\hat H$, while exploiting the fact that such queries only need to be supported as long as $\hat H$ contains a large number of short edge-disjoint $s$-$t$ paths, in order to speed it up. For every cluster $X\in \xset$, we can then use the algorithm for the \maintaincluster problem on $H[X]$, in order to respond to approximate shortest-path queries between pairs of vertices in $X$. Combining these data structures together, we can support approximate shortest $s$-$t$ path queries in the original graph $H$, as long as $H$ contains many short edge-disjoint  $s$-$t$ paths.
This high-level approach allows one to obtain algorithms for decremental \stSP, and for the \routeandcut problem, from algorithms for the \maintaincluster problem, that we now discuss in more detail.

\paragraph{The \maintaincluster problem.}
To recap, in the \maintaincluster problem, the input is a graph $G$ that undergoes an online sequence of edge deletions, and a distance parameter $d$. The goal is to efficiently support short-path queries: given a pair $x,y$ of vertices of $G$, return a path of length at most $d$ connecting them in $G$. The algorithm may, however, at any time, produce a cut $(A,B)$ in $G$ of sparsity at most $O\left (\frac{\poly\log n}{d}\right )$, following which, the vertices of one side of the cut are deleted from $G$. 
The algorithm is used in order to maintain individual clusters of the \ATO. A similar problem was considered by \cite{SCC}, who provide an algorithm with total update time $\ohat(|E(G)|\cdot d^2)$ for it,  where the time to respond to each query is roughly proportional to the number of edges on the path that the algorithm returns. In \cite{CK24} this problem was considered in a more relaxed setting, where the number of queries that the algorithm must support is bounded by a given parameter $\Delta$, and the goal is to minimize the total running time of the algorithm, that is, the sum of the total update time, and the time required to respond to all queries. 
The algorithm of \cite{CK24} for this setting has running time $\otilde(|E(G)|\cdot d+|V(G)|^2)\cdot  \max\set{1,\frac{\Delta\cdot d^2}{|V(G)|}}$, which, for the specific parameters that they employ, becomes $\otilde(|E(G)|\cdot d+|V(G)|^2)$. In order to obtain our improved algorithm for \MBM, we need to generalize this result so that it works for a wider range of parameters, achieving running time $|V(G)|^{2+o(1)}$.

The algorithm of \cite{CK24} follows a rather standard approach. 
First, they use the Cut-Matching game in order to compute a large expander graph $\hat G$, and to embed it into $G$ via short paths that cause low congestion.
The algorithm for the Cut Player is implemented in a rather straightforward manner, since they can afford a running time that is as high as $\Theta(|V(\hat G)|^2)$. The algorithm for the Matching Player essentially needs to solve an instance of the \routeandcut problem. Using the \MWU approach as before, it can be reduced to solving an instance of directed decremental \SSSP. The algorithm of \cite{CK24} then uses the standard Even-Shiloach tree data structure in order to solve the latter problem. In addition to the expander $\hat G$ and its embedding into $G$, the algorithm of \cite{CK24} maintains two additional Even-Shilach trees in $G$. Both trees are rooted in the vertices of $\hat G$, and have depth roughly $d$. One of the trees has all its edges directed away from the root, and the other has all of its edges directed towards the root. In order to respond to a query between a pair $x,y$ of vertices of $G$, the two Even-Shiloach trees are used to compute a short path connecting $x$ to some vertex $x'\in V(\hat G)$, and a short path connecting some vertex $y'\in V(\hat G)$ to $y$. A simple BFS search in the expander $\hat G$ then yields a short path connecting $x'$ to $y'$, which can be turned into an $x'$-$y'$ path in $G$ by exploiting the embedding of $\hat G$ into $G$.

We now describe several sources of inefficiency of the algorithm of \cite{CK24}, and then describe our approach to overcoming them. First, both the algorithms of \cite{CK24} and \cite{SCC} for the \maintaincluster problem are only designed for graphs with unit edge-length. However, both of these works solve (a variant of) the \SSSP problem in graphs with arbitrary edge lengths. To overcome this difficulty, \cite{CK24} use the same approach as \cite{SCC}: namely, they choose a threshold $\tau$, and initially delete all edges whose length is greater than $\tau$ (called \emph{long edges}) from the input graph $H$. The lengths of the remaining edges (called \emph{short edges}) are set to $1$ for the sake of maintaining the \ATO and solving the \maintaincluster problem on the resulting instances. The long edges however are reinserted into the contracted graph $\hat H$, and the actual lengths of the short edges are used in it as well. This approach unfortunately results in a rather large number of right-to-left edges with a large width in $\hat H$, as it may potentially include all long edges. In order to overcome this difficulty, we design an algorithm for the \maintaincluster problem that can handle arbitrary edge lengths, which adds an additional dimension of technical challenges.

The second main source of inefficiency in the algorithm of \cite{CK24} is their use of Even-Shiloach trees in their algorithm for \maintaincluster, both in implementing the Matching Player in the Cut-Matching game, and in order to maintain short paths connecting all vertices of $G$ to the vertices of the expander $\hat G$. It is immediate to see that the problem that the Matching Player needs to solve is essentially an instance of the \routeandcut problem. We also observe that an algorithm for a variant of the \routeandcut problem can be exploited in order to maintain the paths connecting all vertices of $G$ to the vertices of $V(\hat G)$. We abstract this as a new problem, that we call \connecttocenters, and discuss it below. We believe that this problem and our algorithm for solving it are of independent interest. 
We remark that this reduction from the \connecttocenters problem to the \routeandcut problem requires that the algorithm for the \routeandcut problem works for arbitrary congestion parameter $\eta\leq \Delta$, and this is the reason for our more general definition of the \routeandcut problem.

To summarize, as already shown in previous work, in order to obtain an efficient algorithm for the \routeandcut problem, it is enough to obtain an efficient algorithm for the \maintaincluster problem, and we observe that the opposite is also true: an efficient algorithm for the \routeandcut problem implies an efficient algorithm for the \maintaincluster problem. This, however, creates a chicken-and-egg issue, where  in order to  solve one of the two problems efficiently,  we need to design an efficient algorithm for the other. We overcome this barrier by using a recursive approach, that we describe next.

\paragraph{A recursive approach.} We parameterize both the \routeandcut and the \maintaincluster problem using a parameter $r>0$. We say that an instance of the \maintaincluster problem  on an $n$-vertex graph $G$ with a distance parameter $d$ is {\em $r$-restricted}, if $d\leq 2^{r\cdot \sqrt{\log n}}$. Consider now an instance of the \routeandcut problem on an $n$-vertex graph $H$, with two subsets $A,B$ of its vertices, and parameters $\Delta$ and $\eta$. It is not hard to see that, if $\pset$ is any collection of $\Omega(\Delta/\poly\log n)$ paths connecting vertices of $A$ to vertices of $B$, that cause vertex-congestion at most $\eta$, then a large fraction of the paths in $\pset$ have length $\tilde O(n\eta/\Delta)$. We say that an instance of the \routeandcut problem is {\em $r$-restricted} if $\frac{n\eta}{\Delta}\leq 2^{r\cdot \sqrt{\log n}}$. We show a straightforward algorithm for the $1$-restricted \routeandcut problem. Then for all $r\geq 1$, we show that an efficient algorithm for the $r$-restricted \routeandcut problem implies an efficient algorithm for the $r$-restricted \maintaincluster problem. We also show that an efficient algorithm for the $r$-restricted \maintaincluster problem implies an efficient algorithm for the $(r+1)$-restricted \routeandcut problem. Using induction on $r$, we then simultaneously obtain efficient algorithms for the \routeandcut and the \maintaincluster problems for the entire range of values for the parameter $r$.

\paragraph{\connecttocenters Problem.} The \connecttocenters problem is employed as a subroutine in the algorithm for the \maintaincluster problem, but we feel that it is interesting in its own right, as it seems to arise in many different settings. Suppose we are given a dynamic $n$-vertex graph $G$; for now we will assume that $G$ undergoes an online sequence of edge deletions, but in fact our algorithm considers other updates, as described later. In addition to graph $G$, assume that we are given a parameter $d$, and a subset $T\subseteq V(G)$ of vertices that we call \emph{centers}. The goal is to maintain, for every vertex $v\in V(G)$, a path $P(v)$ of length at most $d$, connecting $v$ to some vertex of $T$. As the time progresses, some vertices may be deleted from $T$, but we are guaranteed that $T$ always contains a large enough fraction of the vertices of $G$, e.g. $|T|\geq \Omega(|V(G)|/(d\poly\log n)$. In order to ensure that the deletion of edges from $G$ does not impact too many paths in $\pset^*=\set{P(v)\mid v\in V(G)}$, it is desirable that the paths cause a small edge-congestion (say at most $\tilde O(d)$), and for similar reasons it is desirable that every vertex $x\in T$ serves as an endpoint of relatively few such paths (say at most $\tilde O(d)$). At any time, the algorithm may compute a cut $(A,B)$ of sparsity at most $O\left(\frac{\poly\log n}{d}\right )$, with $|A\cap T|\ll |A|$, after which the vertices of $A$ are deleted from $G$ and the algorithm continues. We note that whenever the by now standard paradigm of embedding an expander into the input graph $G$ and then maintaining short paths connecting all vertices of $G$ to the vertices of the expander is used, one essentially needs to solve a variant of the \connecttocenters problem.  So far this was typically done by employing Even-Shiloach trees, but this data structure becomes inefficient once the depth parameter $d$ is sufficiently large. It is for this reason that we believe that our algorithm for the \connecttocenters problem is of independent interest.

It is immediate to see that the initial collection $\pset^*=\set{P(v)\mid v\in V(G)}$ of paths of length $O(d)$ each, connecting every vertex of $G$ to some vertex of $T$, that cause  edge-congestion $\tilde O(d)$, can be computed by employing an algorithm for the \routeandcut problem. As edges are deleted from $G$, and vertices are deleted from $T$, some of the paths in $\pset^*$ may be destroyed. Whenever a path $P(v)\in \pset^*$ is destroyed (whether because its last endpoint is deleted from $T$, or because one of its edges is deleted from $G$), we say that vertex $v$ becomes \emph{disconnected}. We then need to \emph{reconnect} all disconnected vertices to $T$. This, again, can be done by employing an algorithm for the \routeandcut problem, but doing so directly may be very inefficient. Assume, for example, that we are given an algorithm $\aset$ for the \routeandcut problem, that, on an $n$-vertex graph $G$, has running time $O(n^{2+o(1)})$. Every time a subset of vertices of $G$ becomes disconnected, we would need to employ this algorithm in order to reconnect them to $T$. However, it is possible that only a small number of vertices become disconnected at a time, and spending $\Theta(n^{2+o(1)})$ time to reconnect them each time is prohibitively expensive. A better approach seems to be to consider the set $U$ of vertices that are currently connected, and the set  $U'$ of vertices that are currently disconnected. We could then attempt to route the vertices of $U'$ to the vertices of $U$ by constructing a new collection $\qset=\set{Q(u)\mid u\in U'}$ of paths, where each path $Q(u)$ connects $u$ to some vertex of $U$; and then exploit the existing paths in $\set{P(v)\mid v\in U}$ in order to compute paths connecting every vertex of $U'$ to the vertices of $T$. 
However, we cannot afford to spend $\Theta(n^2)$ time in order to compute the set $\qset$ of paths. On the other hand, 
intuitively, if $|U'|\ll U$, then we may not need to explore the entire graph $G$ in order to compute the set $\qset$ of paths. In order to overcome this difficulty, we require that the algorithm for the \routeandcut problem has running time roughly $n^{1+o(1)}\cdot (n-|B|)$, instead of $n^{2+o(1)}$. In particular, if the graph $G$ is sufficiently dense and $|B|$ is sufficiently large, then this running time may be much smaller than $|E(G)|$. Our reduction from the $(r+1)$-restricted \routeandcut problem to the $r$-restricted \maintaincluster problem follows the high-level approach of \cite{SCC} and \cite{CK24}. However, this additional strict requirement on the efficiency of the algorithm for the \routeandcut problem, in addition to the requirement that the algorithm should work for arbitrary values of the congestion parameter $\eta$, make the reduction more challenging and technical.

Assume now that we are given an algorithm $\aset$ for the \routeandcut problem, that, on an instance $(G,A,B,\Delta,\eta)$ has running time $n^{1+o(1)}\cdot (n-|B|)$, where $n=|V(G)|$. As described above, in order to implement our algorithm for the \connecttocenters problem, whenever we are given a collection $U'$ of vertices that are currently disconnected from $T$, we can now employ Algorithm $\aset$ to construct a collection $\qset$ of paths connecting them to the vertices of $U$ (the set of currently connected vertices), and then compose $\qset$ with the collection $\set{P(v)\mid v\in U}$ of paths to obtain the desired collection $\set{P(v)\mid v\in U'}$ of paths that reconnects the vertices of $U'$ to $T$. Unfortunately, if we follow this approach, and keep appending paths to each other iteration after iteration, we may obtain paths whose lengths are prohibitively large. Instead, we follow a \emph{layered} approach. For a parameter $\lambda=O(\log n)$, we maintain at all times a partition $(U_0,\ldots,U_{\lambda})$ of the vertices of $G$ into \emph{layers}, where $U_0=T$, and $U_{\lambda}$ contains all vertices that are currently disconnected. For all $1\leq i<\lambda$, we also maintain a collection $\qset_i=\set{K(v)\mid v\in U_i}$ of paths, where each path $K(v)$ connects a vertex $v\in U_i$ to a vertex in $\bigcup_{i'=0}^{i-1}U_{i'}$. The paths in each set $\qset_i$ have length at most $\frac d{4\lambda}$, and cause congestion $\tilde O(d)$ in $G$. By composing the paths from different sets $\qset_i$, we can obtain, for every vertex $v\in V(G)$, a path that connects it to some vertex of $T$. At a high level, for all $1\leq i<\lambda$, we reconstruct layer $U_i$ and the set $\qset_i$ of paths from scratch every time that roughly $2^{\lambda-i}$ new vertices become disconnected, and we also ensure that $|U_i|\leq 2^{\lambda-i}$ holds. 
Therefore, when we employ the algorithm for the \routeandcut problem in order to reconnect the vertices of $U_i$, in the resulting instance of $\routeandcut$, $n-|B|\leq 2^{\lambda-i}$ holds, and the running time of the algorithm is roughly $n^{1+o(1)}\cdot 2^{\lambda-i}$. Therefore, as index $i$ becomes smaller, the running time of the \routeandcut algorithm that is used to reconnect the vertices of $U_i$ increases. However, for smaller values of index $i$, we need to reconstruct the set $U_i$ of vertices and the set $\qset_i$ of paths less often. This eventually leads to the desired $n^{2+o(1)}$ running time.

\paragraph{Edge-deletion versus edge-length-increase updates.}
We would like to highlight here what we feel is a somewhat surprising  insight from our work, that may appear minor at first sight, but we believe that it may be useful in other problems. 
Consider the following two settings for dynamic graphs: the first one is the standard decremental setting, where edges are deleted from the input graph $G$ as the time progresses. The second setting is a somewhat more unusual one, where the only updates that are allowed in the input graph $G$ is the doubling of the lengths of its edges; we refer to this type of updates as edge-length-increases. Generally, it is not hard to see that both models are roughly equivalent. Indeed, in order to simulate edge-length-increases in the standard decremental setting, we can simply create a large number of copies of each edge $e$ of various lengths, and then, as the length of $e$ increases, some of these copies are deleted. The reduction in the other direction is also immediate: we can simulate the deletion of an edge $e$ from $G$ by repeatedly doubling its length, until it becomes very high. 

The dynamic graphs that arise from using the \MWU framework typically only undergo edge-length-increase updates, which are then typically implemented as edge-deletions, in order to reduce the problem to the more standard decremental \SSSP, as described above. However, the edge-length-increases that the input graph $G$ undergoes under this implementation of the \MWU method are rather well-behaved: specifically, the lengths of the edges are only increased moderately, and all length increases occur on the edges that participate in the paths that the algorithm returns in response to queries. We observe that the resulting \SSSP problem appears to be easier if we work with edge-length-increase updates directly, instead of the more traditional approach of transforming them into edge-deletion updates.

In order to illustrate this, consider the following simple scenario: we are given a graph $G$, and initially the length of every edge in $G$ is $1$. Assume also that we have computed another graph $X$ (it may be convenient to think of $X$ as an expander), and embedded $X$ into $G$ via paths of length at most $d$, that cause edge-congestion at most $\eta$. If some edge $e$ is deleted from $G$, then every edge $\hat e\in E(X)$, whose embedding path $Q(\hat e)$ uses $e$, must be deleted from $X$ as well. Therefore, the deletion of a single edge from $G$ may lead to the deletion of $\eta$ edges from $X$. Assume now that, instead, the edges of $G$ only undergo increases in their lengths, where the length of each edge may be iteratively doubled, but the total increase in the lengths of all edges is moderate. If the length of an edge $e\in E(G)$ is doubled, then for every edge $\hat e\in E(X)$ whose embedding path $Q(\hat e)$ uses $e$, the length of  $Q(\hat e)$  increases only slightly, and so there is no need to delete $\hat e$ from $E(X)$. We can usually wait until the length of the path $Q(\hat e)$ increases significantly before edge $\hat e$ needs to be deleted from $X$. As mentioned already, in instances arising from the \MWU framework, the total increase in the lengths of all edges in $G$ over the course of the entire algorithm is usually not very large, and in particular most edges whose lengths are doubled are short. This allows us to maintain the expander $X$ and its embedding into $G$ over the course of a much longer sequence of updates to $G$. This is one of the insights that allowed us to obtain a more efficient algorithm for the \maintaincluster problem.

To summarize, our algorithm departs from the algorithm of \cite{CK24} in the following key aspects. First, we use the \MWU method to reduce the \routeandcut problem to \stSP in dynamic graphs that undergo edge-length-increases instead of edge-deletion updates. Second, our algorithm for the \routeandcut problem has running time that significantly decreases when the set $B$ of vertices contains almost all vertices of $G$; in some cases the running time may even be lower than $|E(G)|$. We extend the \maintaincluster problem to handle arbitrary edge lengths, but unlike \cite{CK24} we only allow edge-length-increases, instead of edge-deletion updates. We design an algorithm for the \connecttocenters problem, that we believe is of independent interest, and that can be viewed as reducing the \maintaincluster problem to \routeandcut. Lastly, we use a recursive approach, in which algorithms for \routeandcut rely on algorithms for \maintaincluster and vice versa, by parametrizing both problems with an auxiliary parameter $r$, and then inductively develop algorithms for both problems for the entire range of $r$.

%% file: prelims.tex
\section{Preliminaries}
\label{sec: prelims}

In this paper we work with both directed and undirected graphs.
By default graphs are not allowed to contain parallel edges or self-loops. Graphs with parallel edges are explicitly referred to as \emph{multigraphs}.

\paragraph{Congestion of Paths and Flows.}
Let $G$ be a graph with capacities $c(e)\geq 0$ on edges $e\in E(G)$, and let $\pset$ be a collection of simple paths in $G$. We say that the paths in $\pset$ cause \emph{edge-congestion} $\eta$ if every edge $e\in E(G)$ participates in at most $\eta\cdot c(e)$ paths in $\pset$. When edge capacities are not explicitly given, we assume that they are unit. 
If every edge of $G$ belongs to at most one path in $\pset$, then we say that the paths in $\pset$ are \emph{edge-disjoint}.

Similarly, if we are given a flow value $f(e)\geq 0$ for every edge $e\in E(G)$, we say that flow $f$ causes {\em edge-congestion} $\eta$ if, for every edge $e\in E(G)$, $f(e)\leq \eta\cdot c(e)$ holds. If $f(e)\leq c(e)$ holds for every edge $e\in E(G)$, we may say that $f$ causes \emph{no edge-congestion}.

Assume now that we are given a graph $G$, a pair $s,t$ of its vertices, and a collection $\pset$ of simple paths connecting $s$ to $t$. We say that the paths in $\pset$ cause \emph{vertex-congestion} $\eta$ if every vertex $v\in V(G)\setminus\set{s,t}$ participates in at most $\eta$ paths. If the paths in $\pset$ cause vertex-congestion $1$, then we say that they are \emph{internally vertex-disjoint}.

Throughout the paper, by default, when we talk about congestion caused by a set of paths or by a flow, we mean {\bf edge-congestion}.

\paragraph{Embeddings of Graphs.} Let $H,G$ be two graphs with $V(H)\subseteq V(G)$. An \emph{embedding} of $H$ into $G$ is a collection $\pset=\set{P(e)\mid e\in E(H)}$ of simple paths in $G$, where for every edge $e=(u,v)\in E(H)$, path $P(e)$ connects $u$ to $v$ in $G$. The \emph{congestion} of the embedding is the maximum, over all edges $e'\in E(G)$, of the number of paths in $\pset$ containing $e'$; equivalently, it is the edge-congestion that the set $\pset$ of paths causes in $G$.

\paragraph{Cuts, Sparsity and Expanders for Directed Graphs.}
Let $G$ be a directed graph. A \emph{cut} in $G$ is a an ordered pair $(A,B)$ of subsets of vertices of $G$, such that $A\cap B=\emptyset$, $A\cup B=V(G)$, and $A,B\neq\emptyset$. Note that we do not require that $|A|\leq |B|$ in this definition. The \emph{sparsity} of the cut is:

\[ \Phi(A,B)=\frac{|E(A,B)|}{\min{\set{|A|,|B|}}}.\]

We say that a directed graph $G$ is a $\phi$-expander, for a given value $0<\phi<1$, if the sparsity of every cut in $G$ is at least $\phi$. Notice that, if $G$ is a $\phi$-expander, then for every partition $(A,B)$ of its vertices, both $|E(A,B)|\geq \phi\cdot \min\set{|A|,|B|}$ and $|E(B,A)|\geq \phi\cdot \min{\set{|A|,|B|}}$ hold.

\iffalse
We will sometimes informally say that graph $G$ is an \emph{expander} if $\Phi(G)$ is a constant independent of $|V(G)|$. 
We use the following immediate observation, that was also used in previous works, (see e.g. Observation 2.3 in \cite{detbalanced}). 

\begin{observation}\label{obs: exp plus matching is exp}
	Let $G=(V,E)$ be an $n$-vertex graph that is a $\phi$-expander, and let $G'$ be another graph that is obtained from $G$ by adding to it a new set $V'$ of at most $n$ vertices, and a matching $M$, connecting every vertex of $V'$ to a distinct vertex of $G$. Then $G'$ is a $\phi/2$-expander.
\end{observation}
\fi

We also use the following theorem that provides a fast algorithm for an explicit construction of an expander, that is based on the results of Margulis \cite{Margulis} and Gabber and Galil \cite{GabberG81}. The proof for undirected graphs was shown in \cite{detbalanced}.
\begin{theorem}[Theorem 2.4 in \cite{detbalanced}]
	%[Fast explicit expander construction]
	\label{thm:explicit expander}
	There is a constant $\alpha_0 > 0$ and a deterministic algorithm, that we call \constructexpander, that, given an integer $n>1$, constructs a directed graph $H_n$ with $|V(H_n)|=n$, such that $H_n$ is an $\alpha_0$-expander, and every vertex in $H_n$ has at most $9$ incoming and at most $9$ outgoing edges. The running time of the algorithm is $O(n)$.
\end{theorem}

The above theorem was proved in \cite{detbalanced} for undirected expanders, with maximum vertex degree in the resulting expander bounded by $9$. It is easy to adapt this result to directed expanders: if a graph $G$ is an undirected $\phi$-expander, and a directed graph $G'$ is obtained from $G$ by replacing every edge with a pair of bi-directed edges, then it is easy to verify that $G'$ is a directed $\phi$-expander.

We will also use the following simple observation, that shows that a slight modification of an expander graph still results in a (somewhat weaker) expander. Similar observations were shown in previous work for undirected graphs (see e.g. Observation 2.3 in \cite{detbalanced}).

\begin{observation}\label{obs: expander plus matching}
	Let $H$ be a directed $n$-vertex graph, that is a $\phi$-expander, for some parameter $0<\phi<1$. Let $H'$ be a graph that is obtained from $H$ as follows. We select an arbitrary set $X$ of $n' \leq n$ vertices of $H$, add a set $Y$ of $n'$ new vertices to $H$, and add two matchings of cardinality $n'$ to $H$: a matching $M'\subseteq X\times Y$ and a matching $M''\subseteq Y\times X$. Then the resulting graph $H'$, defined over $(n + n')$ vertices, is a $\phi/3$-expander. 
\end{observation}
\begin{proof}
	Consider any cut $(A,B)$ in $H'$. Assume first that $|A\cap Y|\geq \frac{2|A|}{3}$. Recall that every vertex $y\in Y$ has one edge $e_y=(y,y')\in M''$ with $y'\in X$. 
	Let $Y_1\subseteq A\cap Y$ denote the set of all vertices $y\in A\cap Y$ whose corresponding vertex $y'$ lies in $A$, and let $Y_2$ denote the set of all remaining vertices of $A\cap Y$.
	Since $|A\setminus Y|\leq \frac{|A\cap Y|} 2$, it must be the case that $|Y_2|\geq |Y_1|\geq \frac{|A|}{3}$. But then the edges of $M''$ that are incident to the vertices of $Y_2$ belong to $E_{H'}(A,B)$, and so $|E_{H'}(A,B)|\geq |Y_2|\geq \frac{|A|}{3}$.

	Similarly, if $|B\cap Y|\geq \frac{2|B|}{3}$, then, by using a similar argument with the edges of the matching $M'$ instead of $M''$, we get that $|E_{H'}(A,B)|\geq \frac{|B|}{3}$.
	Therefore, if either $|A\cap Y|\geq \frac{2|A|}{3}$, or $|B\cap Y|\geq \frac{2|B|}{3}$ hold, then the sparsity of the cut $(A,B)$ is at least $\frac{1}{3}$. 
	
	It now remains to consider the case where both $|A\cap Y|< \frac{2|A|}{3}$ and $|B\cap Y|< \frac{2|B|}{3}$ hold. Let $A'=A\setminus Y$ and $B'=B\setminus Y$. Notice that $|A'|\geq |A|/3$ and $|B'|\geq |B|/3$ must hold. Then $(A',B')$ is a cut in graph $H$, and, since $H$ is a $\phi$-expander, $|E_H(A',B')|\geq \phi\cdot \min\set {|A'|,|B'|}\geq \frac{\phi}{3}\cdot \min\set{|A|,|B|}$. Since $E_H(A',B')\subseteq E_{H'}(A,B)$, we get that the sparsity of the cut $(A,B)$ in $H'$ is
	 at least $\frac{\phi}{3}$.
\end{proof}

The following simple observation follows from the standard ball-growing technique, and similar observations were used in many previous works, including in~\cite{CK24}, where a proof is provided for completeness.

\begin{observation}[Observation 2.2 in full version of \cite{CK24}]
\label{obs: ball growing}
	There is a deterministic algorithm, that, given  a directed $n$-vertex graph $G=(V,E)$ with unit edge lengths and maximum vertex degree at most $\dmax$, together with two vertices $x,y\in V$ with $\dist_G(x,y)\geq d$ for some parameter $d\geq 64\dmax\log n$, computes a cut $(A,B)$ of sparsity at most $\phi=\frac{32\dmax\log n}{d}$ in $G$. The running time of the algorithm is $O(\dmax\cdot\min\set{|A|,|B|})$.
\end{observation}

We will also use the following simple observation a number of times. Intuitively, the observation allows us to convert a sequence of small sparse cuts into a single balanced sparse cut; a proof of this observation appears in~\cite{CK24}.

\begin{observation}[Observation 2.3 in full version of \cite{CK24}]
\label{obs: from many sparse to one balanced}
	There is a deterministic algorithm, whose input consists of a directed $n$-vertex graph $G$, parameters $0<\alpha<1$ and $0\leq \phi\leq 1$, and a sequence $\xset=(X_1,X_2,\ldots,X_k)$ of subsets of vertices of $G$, such that the following hold:
	\begin{itemize}
		\item Every vertex of $V(G)$ lies in exactly one set in $\set{X_1,\ldots,X_k}$; 
		\item $\max_i\set{|X_i|}= \alpha n$; and
		
		\item for all $1\leq i<k$, if we denote by $X'_i=\bigcup_{j=i+1}^kX_j$, then either $|E(X_i,X'_i)|\leq \phi\cdot |X_i|$, or $|E(X'_i,X_i)|\leq \phi\cdot |X_i|$.
	\end{itemize}
	The algorithm computes a cut $(A,B)$ in $G$ with $|A|,|B|\geq \min \set{\frac{1-\alpha}{2}\cdot n, \frac n 4}$, and $|E_G(A,B)|\leq \phi\cdot \sum_{i=1}^{k-1}|X_i|$. The running time of the algorithm is $O(|E(G)|)$.  
\end{observation}

\paragraph{Chernoff Bound.}

We use the following standard version of Chernoff Bound (see. e.g., \cite{dubhashi2009concentration}).

\begin{lemma}[Chernoff Bound]
	\label{lem: Chernoff}
	Let $X_1,\ldots,X_n$ be independent randon variables taking values in $\set{0,1}$. Let $X=\sum_{1\le i\le n}X_i$, and let $\mu=\expect{X}$. Then for any $t>2e\mu$,
	\[\Pr\Big[X>t\Big]\le 2^{-t}.\]
	Additionally, for any $0\le \delta \le 1$,
	\[\Pr\Big[X<(1-\delta)\cdot\mu\Big]\le e^{-\frac{\delta^2\cdot\mu}{2}}.\]
\end{lemma}

\subsection{Even-Shiloach Tree for Weighted Directed Graphs}

We use the generalization of the Even-Shiloach data structure of \cite{EvenS,Dinitz} to weighted directed graphs, due to \cite{ES-tree-directed}, who extended similar results of \cite{HenzingerKing} for unweighted directed graphs.

\begin{theorem}[\cite{ES-tree-directed}, see Section 2.1]\label{thm: directed weighted ESTree}
	There is a deterministic algorithm, whose input consists of a directed graph $H$ with integral edge lengths, a source vertex $s\in V(H)$, and a distance parameter $d\geq 1$. The algorithm maintains a shortest-path tree $\tau\subseteq H$ rooted in $s$, up to depth $d$. In other words, for every vertex $v\in V(H)$ with $\dist_H(s,v)\leq d$, vertex $v$ lies in the tree $\tau$, and the length of the shortest $s$-$v$ path 
	in $\tau$ is equal to that in $H$. The data structure also maintains, for all $v\in V(\tau)$, the current value $\dist_H(s,v)$. The total update time of the algorithm is $O(|E(H)|\cdot d)$.
\end{theorem}

We will sometimes refer to the data structure from \Cref{thm: directed weighted ESTree} as $\EST(H,s,d)$.

\input{adjacency-list}

\input{ATO}
\input{almost-DAG}

%% file: adjacency-list.tex
\subsection{The Adjacency-List Representation of Graphs}
\label{subsec: adj-list representation}

In some of the problems that we consider, the algorithm may not need to explore the entire input graph $G$, and in some cases, it is important that the running time of the algorithm is strictly smaller than $|E(G)|$. As an example, consider the following routing problem: we are given as input a directed graph $G$, and two disjoint sets $A$,$B$ of its vertices. The goal is to compute a largest-cardinality collection $\pset$ of paths, where each path $P\in \pset$ connects a vertex of $A$ to a vertex of $B$, the paths in $\pset$ cause low edge-congestion, and every vertex of $A\cup B$ serves as an endpoint of at most one path in $\pset$. We emphasize that vertices of $A\cup B$ may serve as inner vertices on paths in $\pset$, and it is possible that some vertex $v\in A\cup B$ serves both as an endpoint of some path in $\pset$, and as an inner vertex on several other paths. Consider now a scenario where almost all vertices of $G$ lie in $B$. In such a case, one can show that there is an optimal solution $\pset$ to the problem, such that only a small number of vertices of $B$ participate in the paths in $\pset$. In other words, in order to compute such a collection of paths, it is not necessary to explore the entire  graph $G$. In such cases, we will often be ineterested in algorithms with running time roughly $O(n^{1+o(1)}\cdot (n-|B|))$, that  may be significantly smaller than $|E(G)|$, if graph $G$ is sufficiently dense. Therefore, we will need to provide the input graph $G$ to the algorithm in a way that allows it to achieve this low running time, and does not require spending $\Omega(|E(G)|)$ time in order to read the entire input.

In such cases, we will often require that the (directed) input graph $G$ is provided to the algorithm in the adjacency-list model: 
we are given two arrays $\OUT$ and $\INN$ of length $n$ each. For all $1\leq i\leq n$, entry $\OUT[i]$ contains a pointer to a linked-list, that contains all edges in set $\delta^+(v_i)$, that is, all edges that leave the $i$th vertex of $G$. Similarly, $\INN[i]$ contains a pointer to a linked-list of all edges in $\delta^-(v_i)$ -- all edges that enter the $i$th vertex of $G$. If the edges of $G$ have lengths or weights, then these lengths or weights also appear in the corresponding lists $\delta^-(v_i)$ and $\delta^+(v_i)$, where these edges are listed. 

Typically, if $\alg$ is an algorithm for the routing problem mentioned above, then it will only access the entries $\OUT[i],\INN[i]$ of a small subset of the vertices $v_i\in V(G)$, allowing us to achieve a running time that is significantly lower than $\Theta(|E(G)|)$.  Specifically, it is sufficient that the algorithm reads the entries corresponding to the vertices in $V(G)\setminus B$, and of a small subset of vertices in $B$.  This will allow us to obtain running times for such routing problems roughly of the form $O(|V(G)|^{1+o(1)}\cdot (|V(G)|-|B|))$, which may be significantly lower than $\Theta(|E(G)|)$ if set $B$ contains almost all vertices of $G$, and $G$ is sufficiently dense.

We note that providing the graph $G$ as input to such an algorithm $\alg$ may require time $O(|E(G)|+|V(G)|)$, in order to construct the adjacency-list representation of the input graph $G$. However, in a typical application, we will construct the adjacency-list representation of the input graph $G$ once, and then apply Algorithm $\alg$ for solving the routing problem on $G$ multiple times, with different vertex sets $A$ and $B$, possibly with small alterations to the graph $G$ between such executions of $\alg$. Therefore, we will not need to construct the adjacency-list representation of the input graph $G$ from scratch every time Algorithm $\alg$ is called. 

In some cases, given as input a directed graph $G$ in the adjacency-list representation, our routing algorithm will need to work with the graph $\revG$, which is obtained from $G$ by reversing the directions of all its edges. In such cases, we can modify the adjacency-list representation of $G$ to obtain an adjacency-list representation of $\revG$, by simply switching the arrays $\OUT$ and $\INN$.

Finally, given a directed graph $G$ and a vertex $v$ in $G$, we will denote by $\delta^+_G(v)$, the set of edges leaving $v$, and by $\delta^-_G(v)$, the set of edges entering $v$.

%% file: ATO.tex
\subsection{Overview of the \ATO framework}
\label{sec: ATO}

We now provide an overview of the Approximate Topological Order (\ATO) framework that was developed in \cite{AlmostDAG2} and later used in \cite{SCC}. 
A very slight modification of this framework was also used in~\cite{CK24}. We follow here the presentation of~\cite{CK24}, though we further adapt it to allow for presence of a special cluster, that contains most of the vertices of the underlying graph.

Let $G=(V,E)$ be a directed $n$-vertex graph  with two special vertices $s$ and $t$ and non-negative lengths $\ell(e)$ on edges $e\in E$, that undergoes edge-length-increase updates during a time interval $\tset$, that we refer to as the \emph{time horizon}. 
Throughout, we refer to a collection $I\subseteq \set{1,\ldots,n}$ of consecutive integers as an \emph{interval}. The \emph{size} of the interval $I$, denoted by $|I|$, is the number of distinct integers in it.
An \ATO consists of the following components:

\begin{itemize}
\item a partition $\xset$ of the set $V$ of vertices of $G$ into non-empty subsets, such that there is a set $S=\set{s}$ and $T=\set{t}$ in $\xset$; 

\item (optional) a single set $U\in \xset$ that we refer to as a \emph{special cluster}, together with an integer $1\leq \gamma\leq |U|$;

If special cluster $U$ is defined, then all other sets of $\xset$ are called \emph{regular clusters}. If the special cluster is undefined, then all sets of $\xset$ are called regular clusters, and we set $\gamma=1$; and

\item a map $\rho$, that maps every set $X\in \xset$ to an interval $I_X\subseteq \set{1,\ldots,n-\gamma+1}$. The size of the interval $I_X$ must be $|X|$ if $X$ is a regular cluster, and $|X|-\gamma+1$, if it is a special cluster. We require that all intervals in $\set{I_X\mid X\in \xset}$ are mutually disjoint. 
\end{itemize}

For convenience, we will denote an \ATO by $(\xset,\rho, U, \gamma)$, where the special cluster $U$ may be undefined. 
Both the partition $\xset$ and the map $\rho$ evolve over time. We denote by $\xset\attime$ the partition $\xset$ at time $\tau$, and by $\rho\attime$ the map $\rho$ at time $\tau$. For every set $X\in \xset\attime$, we denote by $I^{\attime}_X$ the interval $I_X$ associated with $X$ at time $\tau$. If the special cluster $U$ is defined in the initial \ATO, then we require that $U\in \xset$ holds at all times, though some vertices may be deleted from $U$ as the time progresses, and, additionally, $|U|\geq \gamma$ must hold at all times.

The only type of changes that are allowed to the \ATO over time is the {\bf splitting} of the clusters $X\in \xset$. In order to split a set $X\in \xset$, we select a subset $Y\subseteq X$ of vertices (if $X$ is a special cluster, then $|X\setminus Y|\geq \gamma$ must hold, while otherwise we require that $X\setminus Y\neq \emptyset$). We add $Y$ to $\xset$ as a new cluster, and we delete the vertices of $Y$ from $X$. We also partition interval $I_X$ into two disjoint subintervals: one of these subintervals, of length $|Y|$, becomes the interval $I_Y$ that is associated with the new cluster $Y$. The other subinterval, whose length is $|X|$ if $X$ is a regular cluster, and $|X|-\gamma+1$ if it is a special cluster, becomes the new interval $I_X$ associated with the cluster $X$. Note that the total length of the two new intervals is equal to the length of the original interval $I_X$.

Notice that an \ATO $(\xset,\rho,U,\gamma)$ naturally defines a left-to-right ordering of the intervals in $\set{I_{X}\mid X\in \xset}$: we will say that interval $I_X$ lies to the left of interval $I_{X'}$, and denote $I_X \prec I_{X'}$ iff every integer in interval $I_X$ is smaller than every integer in $I_{X'}$. We will also sometimes denote $X\prec X'$ in such a case. Notice also that, if we are given an ordering $\sigma=(X_1,X_2,\ldots,X_k)$ of the sets in $\xset$, then this ordering naturally defines the map $\rho$, where for all $1\leq i\leq k$, the first integer in $I_{X_i}$ is $\sum_{j=1}^{i-1}|X_j|+1$ if none of the clusters in $\set{X_1,\ldots,X_{i-1}}$ is special, and it is $\sum_{j=1}^{i-1}|X_j|-\gamma+2$ otherwise. For convenience, we will denote the map $\rho$ associated with the ordering $\sigma$ by $\rho(\sigma)$.

Given a vertex $v\in V$, we will sometimes denote by $X^v$ the set of $\xset$ that contains $v$. Consider now some edge $e=(u,v)$ in graph $G$, and assume that $X^u\neq X^v$. We say that $e$ is a \emph{left-to-right} edge if $X^u\prec X^v$, and we say that it is a \emph{right-to-left} edge otherwise.

Lastly, we denote by $H=G_{|\xset}$ the \emph{contracted graph} associated with $G$ and $\xset$ -- the graph that is obtained from $G$ by contracting every set $X\in \xset$ into a vertex $v_X$, that we may call a \emph{supernode}, to distinguish it from other vertices of $G$. We keep parallel edges but discard any self-loops.

\paragraph{Intuition.}
For intuition, we could let $\xset$ be the set of all strongly connected components of the graph $G$. The corresponding contracted graph $H=G_{|\xset}$ is then a DAG, so we can define a topological order of its vertices. This order, in turn, defines an ordering $\sigma$ of the sets in $\xset$, and we can then let the corresponding \ATO be $(\xset,\rho(\sigma), U,\gamma)$, where $U$ is undefined and $\gamma=1$. 
 Notice that in this setting, all edges of $G$ that connect different sets in $\xset$ are left-to-right edges. However, we will not follow this scheme. Instead, we will ensure that every set $X\in \xset$ corresponds to some expander-like graph, and we will allow both left-to-right and right-to-left edges with respect to the \ATO .  It would still be useful to ensure that the corresponding contracted graph $H$ is close to being a DAG, since there are good algorithms for \SSSP on DAG's and DAG-like graphs. Therefore, we will try to limit both the number of right-to-left edges, and their ``span'', that we define below.

\paragraph{Skip and Span of Edges.}
For every edge $e=(u,v)\in E(G)$ we define two values: the \emph{skip} of edge $e$, denoted by $\skipp(e)$, and the \emph{span} of $e$, denoted by $\spann(e)$. Intuitively, the skip value of an edge $(u,v)$ will represent the ``distance'' between the intervals $I_{X^u}$ and $I_{X^v}$. Non-zero span values are only assigned to right-to-left edges. The span value also represents the distance between the corresponding intervals, but this distance is defined differently. We now formally define the span and skip values for each edge, and provide an intuition for them below.

Consider some edge $e=(x,y)\in E(G)$. Let $X,Y\in \xset$ be the sets containing $x$ and $y$, respectively. If $X=Y$, then $\skipp(e)=\spann(e)=0$. 

Assume now that $X\prec Y$, that is, $e$ is a left-to-right edge. In this case, we set $\spann(e)=0$, and we let $\skipp(e)$ be the distance between the last endpoint of $I_X$ and the first endpoint of $Y$. In other words, if the largest integer in $I_X$ is $Z_X$, and the smallest integer in $Y$ is $A_Y$, then $\skipp(e)=A_Y-Z_X$. Note that, as the algorithm progresses, sets $X$ and $Y$ may be partitioned into smaller subsets. It is then possible that $\skipp(e)$ may grow over time, but it may never decrease.

Lastly, assume that $Y\prec X$, so $e$ is a right-to-left edge. In this case, we set $\skipp(e)$ to be the distance between the right endpoint of $I_Y$ and the left endpoint of $I_X$, and we let $\spann(e)$ be the distance between the left endpoint of $I_Y$ and the right endpoint of $I_X$. In other words, if $A_X$ and $Z_X$ are the smallest and the largest integers in $X$ respectively, and similarly $A_Y$ and $Z_Y$ are the smallest and the largest integers in $Y$ respectively, then $\skipp(e)=A_X-Z_Y$ and $\spann(e)=Z_X-A_Y$ (see \Cref{fig: skipspan}). Notice that, as the algorithm progresses, and sets $X$ and $Y$ are partitioned into smaller subsets, the value $\spann(e)$ may only decrease, while the value $\skipp(e)$ may only grow. Moreover, $\skipp(e)\leq \spann(e)$ always holds.

\begin{figure}[h]
	\begin{center}
		\scalebox{0.4}{\includegraphics{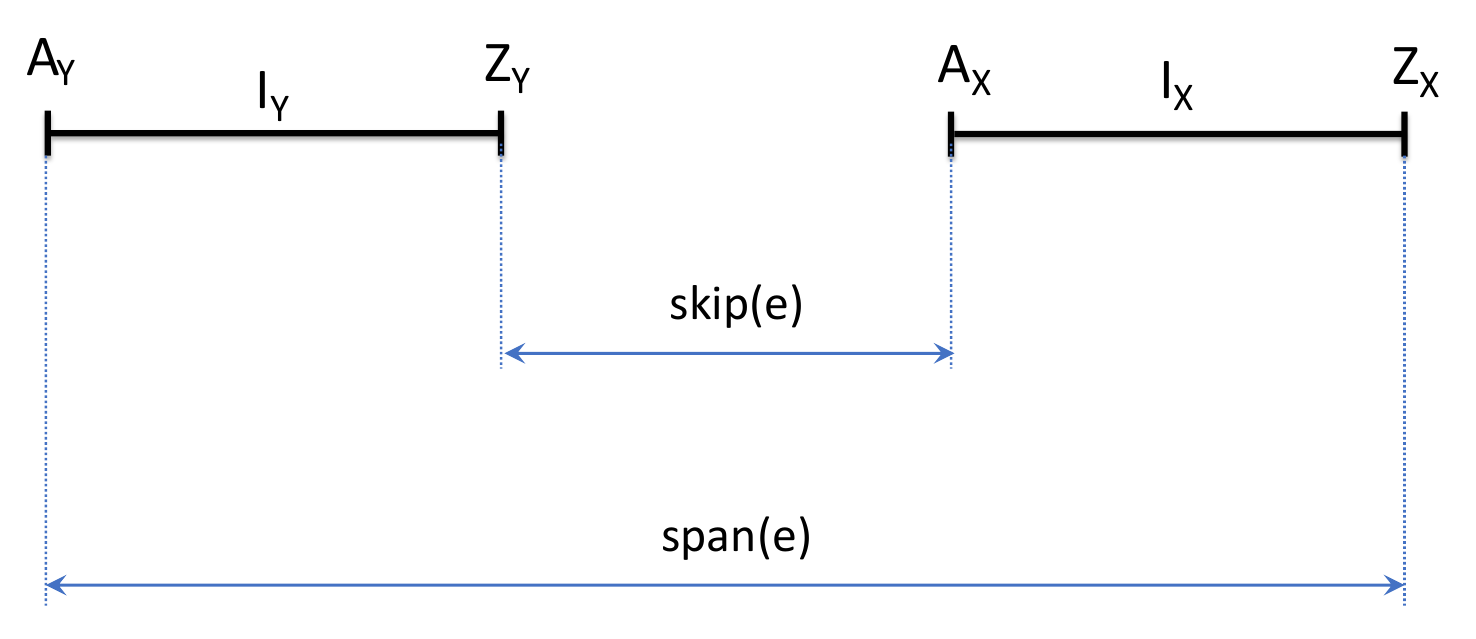}}\caption{Skip and span values of a right-to-left edge. \label{fig: skipspan}}
\end{center}
\end{figure}

For intuition, consider again the setting where $\xset$ represents the  strongly connected components of $G$, and $\rho=\rho(\sigma)$, where $\sigma$ is an ordering of the sets in $\xset$ corresponding to  a topological ordering of the vertices of $H=G_{|\xset}$. Then for any path $P$ in $G$, $\sum_{e\in E(P)}\skipp(e)\leq n$ must hold. This is because the path must traverse the sets of $\xset$ in their left-to-right order. Consider now inserting a single right-to-left edge $e=(x,y)$ into $G$, so  $x\in X$, $y\in Y$, and $Y\prec X$. In this case, path $P$ will generally traverse the sets in $\xset$  from left to right, but it may use the edge $e$ in order to ``skip'' back. It is not hard to see that $\sum_{e'\in E(P)}\skipp(e')\leq n+\skipp(e)+\spann(e)$
now holds, since, by traversing the edge $e$, path $P$ is `set back' by at most $\spann(e)$ units. In other words, path $P$ traverses the sets in $\xset$ in their left-to-right order, except when it uses the edge $e$ to ``skip'' to the left. The length of the interval it ``skips'' over is bounded by $\spann(e)$, so this skip adds at most $\skipp(e)+\spann(e)\leq 2\spann(e)$ to $\sum_{e'\in E(P)}\skipp(e')$.

Using the same intuitive reasoning, if we allow graph $G$ to contain an arbitrary number of right-to-left edges, then for any simple path $P$, $\sum_{e\in E(P)}\skipp(e)\leq n+\sum_{e\in E(G)}2\spann(e)$ holds. However, if we are given any collection $\pset$ of edge-disjoint paths in $G$, then there must be some path $P\in \pset$ with $\sum_{e\in E(P)}\skipp(e)\leq n+\frac{\sum_{e\in E(G)}2\spann(e)}{|\pset|}$, since every edge $e\in E(G)$ may contribute the value $2\spann(e)$ to at most one path in $\pset$. We will use this fact extensively.

\paragraph{From \ATO to \SSSP in DAG-like Graphs.}
Like the algorithm of \cite{SCC} for decremental \SSSP,
our algorithm relies on two main technical tools. The first tool allows us to maintain the partition $\xset$ of $V(G)$ over time. Intuitively, we will ensure that, for every set $X\in \xset$, the corresponding induced subgraph $G[X]$ of $G$ is in some sense expander-like. In particular, we will be able to support short-path queries between vertices of $X$ efficiently. When $G[X]$ is no longer an expander-like graph, we will compute a sparse cut $(A,B)$ in this graph, and we will replace $X$ with $A$ and $B$ in $\xset$ (though formally we will think about this process as splitting the cluster $X$: we add a new cluster, $A$ or $B$, to $\xset$, and then delete the vertices of this new cluster from $X$). Some of the edges connecting the sets $A$ and $B$ may become right-to-left edges with respect to the current \ATO $(\xset,\rho)$ that we maintain. We will control the number of such edges and their $\spann$ values by ensuring that the cut is $\phi$-sparse, for an appropriately chosen parameter $\phi$. 

Given a dynamically evolving \ATO $(\xset,\rho, U,\gamma)$ for graph $G$, we can now construct the corresponding contracted graph $H=G_{|\xset}$. 
For convenience, the supernodes of $H$ representing the sets $S=\set{s}$ and $T=\set{t}$ of $\xset$ are denoted by $s$ and $t$, respectively.
Notice that, as the time progresses, graph $H$ undergoes, in addition to edge-length-increase updates (that we will implement as edge-deletions using standard methods) another type of updates, that we call \emph{vertex-splitting}: whenever some set $X\in \xset$ undergoes a splitting, and a new cluster $Y\subseteq X$ is inserted into $\xset$, we need to insert a supernode corresponding to $Y$ into $H$, and then possibly delete some edges that are incident to $v_X$ from $H$.  We will say that the new supernode  $v_Y$  was \emph{split off} from the supernode $v_X$. Consider now any edge $e=(x,y)\in E(G)$, and assume that $x\in Z$ and $y\in Z'$ for some $Z,Z'\in \xset$ with $Z\neq Z'$. The length of the edge $e$ in graph $H$ remains the same: $\ell(e)$. We will also associate, with every edge $e\in E(H)$, a \emph{weight} $w(e)$, that is set to be the smallest integral power of $2$ that is at least $\skipp(e)$. As observed already, as the algorithm progresses, the value $\skipp(e)$ may grow. In order to avoid inserting edges into $H$ overtime, we will create a number of parallel edges with different $w(e)$ values corresponding to every edge of $G$ whose endpoints lie in different clusters of $\xset$. Specifically, let $e=(x,y)$ be any  edge, with $x\in X$, $y\in Y$, and $X\neq  Y$. For every integer $0\leq i\leq \ceil{\log n}$, such that $2^i\geq \skipp(e)$, we add an edge $e_i=(v_X,v_Y)$ to $H$, with length $\ell(e)$ and $w(e_i)=2^i$. It is easy to verify that this construction ensures that, for every vertex $v\in V(H)\setminus\set{t}$ and integer $0\leq i\leq \ceil{\log n}$, the number of vertices in the set $\set{u\in V(H)\mid \exists e=(v,u)\in E(H)\mbox{ and } w(e)= 2^i}$ is bounded by $2^{i+2}$. Additionally, by carefully controlling the total span of all right-to-left edges in graph $G$, we will ensure that there is always a short $s$-$t$ path $P$ in $H$, such that $\sum_{e\in E(P)}w(e)$ is relatively small. These two properties are crucial in obtaining a fast algorithm for computing short $s$-$t$ paths in $H$, as it undergoes edge-deletion and supernode-splitting updates.

The \ATO framework, combined with an algorithm for maintaining expander-like graphs, gives rise to the problem of repeately computing short $s$-$t$ paths in the contracted graph $H$, as it undergoes the updates that we have outlined above. 
In \cite{CK24} a problem call \DLSSSP was defined, that abstracts the above problem. They also provided an algorithm for it, which is essentially identical to the algorithm of \cite{AlmostDAG2}; a similar algorithm was used in \cite{SCC}.
We now provide the definition of the \DLSSSP problem from \cite{CK24}, and summarize their algorithm for it.
 

%% file: almost-DAG.tex
\subsection{SSSP in DAG-like Graphs}
\label{sec: SSSP in almost DAGS}

The presentation in this section is almost identical to the definition of the \DLSSSP problem from \cite{CK24}.
The input to the \DLSSSP problem is a directed multigraph $G=(V,E)$, with two designated vertices $s$ and $t$, a distance parameter $d>0$, a precision parameter $0<\eps\leq 1/8$, such that $1/\eps$ is an integer, and another parameter $\Gamma>0$. Graph $G$ may have parallel edges, but no self-loops. Every edge $e\in E$ of $G$ is associated with an integral \emph{length} $\ell(e)>0$, and an additional \emph{weight} $w(e)$, that must be an integral power of $2$ between $1$ and $2^{\ceil{\log n}}$. The graph undergoes the following two types of updates:

\begin{itemize}
	\item {\bf Edge deletion:} given an edge $e\in E$, delete $e$ from $G$; and
	
	\item {\bf Vertex splitting:} in this operation, we are given as input a vertex $v\not\in\set{s,t}$ in the current graph, and new vertices $u_1,\ldots,u_k$ that need to be inserted into $G$; we may sometimes say that vertices $u_1,\ldots,u_k$ are \emph{split-off from $v$}. In addition to vertex $v$ and vertices $u_1,\ldots,u_k$, we are given a new collection $E'$ of edges that must be inserted into $G$, and for every edge $e\in E'$, we are given an integral length $\ell(e)>0$ and a weight $w(e)$ that must be an integral power of $2$ between $1$ and $2^{\ceil{\log n}}$. For each such edge $e\in E'$, either both endpoints of $e$ lie in $\set{u_1,\ldots,u_k, v}$; or one endpoint lies in $\set{u_1,\ldots,u_k}$, and the other in $V(G)\setminus\set{v,u_1,\ldots,u_k}$. In the latter case, if $e=(u_i,x)$, then edge $e'=(v,x)$ must currently lie in $G$, and both $w(e')=w(e)$ and $\ell(e')=\ell(e)$ must hold. Similarly, if $e=(x,u_i)$, then edge $e'=(x,v)$ must currently lie in $G$, and both $w(e')=w(e)$ and $\ell(e')=\ell(e)$ must hold.
\end{itemize}

We let $\tset$ be the \emph{time horizon} of the algorithm -- that is, the time interval during which the graph $G$ undergoes all updates. 
We denote by $n$ and by $m$ the total number of vertices and edges, respectively, that were ever present in $G$ during $\tset$. We assume w.l.o.g. that for every vertex $v$ that was ever present in $G$, some edge was incident to $v$ at some time during the update sequence, so $n\leq O(m)$.

For every vertex $v\in V(G)\setminus\set{t}$ and integer $0\leq i\leq \ceil{\log n}$, we define a set:
 $$U_i(v)=\set{u\in V(G)\mid \exists e=(v,u)\in E(G)\mbox{ with } w(e)=2^i}.$$ 
 Note that as the algorithm progresses, each such set $U_i(v)$ may change, with vertices both leaving and joining the set. We require that
the following property must hold at all times $\tau\in \tset$:

\begin{properties}{S}
	\item For every vertex $v\in V(G)\setminus\set{t}$, for every integer $0\leq i\leq \ceil{\log n}$, $|U_i(v)|\leq 2^{i+2}$.  \label{prop: few close neighbors}	
\end{properties}

As graph $G$ undergoes edge-deletion and vertex-splitting updates, at any time the algorithm may be asked a query $\pquery$. It is then required to return a {\bf simple} path $P$ connecting $s$ to $t$, whose length $\sum_{e\in E(P)}\ell(e)$ is bounded by $(1+10\eps)d$, in time $O(|E(P)|)$. 

Lastly, the algorithm may return FAIL at any time and halt. It may only do so if the current graph $G$ does not contain any $s$-$t$ path $P$  with $\sum_{e\in E(P)}\ell(e)\leq d$, and $\sum_{e\in E(P)}w(e)\leq \Gamma$.

The following theorem provides an algorithm for the \DLSSSP problem. Its formal proof, presented in \cite{CK24}, follows almost immediately from the results of \cite{AlmostDAG2}. Specifically, Theorem 5.1 of \cite{AlmostDAG2} provides a similar result with a slightly different setting of parameters, but their arguments can be easily extended to prove the theorem below, as was shown in \cite{CK24}. 

\begin{theorem}[Theorem 5.1 in the full version of \cite{CK24}]\label{thm: almost DAG routing}
	There is a deterministic algorithm for the \DLSSSP problem with total update time $\tilde O\left(\frac{n^2+m+\Gamma\cdot n}{\eps^2}\right )$; here, $m$ and $n$ are the total number of edges and vertices, respectively, that were ever present in $G$.
\end{theorem}

%-------------------------
%-------------------------
%-------------------------
%-------------------------
%-------------------------

%% file: high-level-overview.tex
\section{High-Level Overview of the Algorithm}
\label{sec: high level}

In this section we provide the definitions of well-structured graphs and well-structured cuts, that are used throughout the paper. We also define a new problem, called \routeandcut, which is one of the main building blocks of our algorithm. We then state the main theorem for this section, that provides an efficient algorithm for the \routeandcut problem, and we show that our main result -- the proof of \Cref{thm:main} follows from it. The proof of the theorem that provides an algorithm for the \routeandcut problem is by induction, and is provided in the remainder of this paper. In this section we also provide a simple and somewhat inefficient algorithm for the \routeandcut problem, that will be used as the induction base. We also define an intermediate problem, called $r$-Restricted \stSP problem, and show that an algorithm for this problem implies an algorithm for the \routeandcut problem. We start by defining well-structured graphs and well-structured cuts.

\input{well-structured}

\input{route-and-cut}

\input{completing-main-thm}

\input{high-level-route-and-cut}

\input{advanced-routeandcut}

%% file: well-structured.tex
\subsection{Well-Structured Graphs and Well-Structured Cuts}

Throughout the paper we will work with special kinds of directed graphs, that for convenience we call \emph{well-structured graphs}, and define below.

\begin{definition}[Well-Structured Graphs]\label{def: well structured graphs} Let $G=(L,R,E)$ be a bipartite directed graph. We call the edges of $E_G(L,R)$ \emph{regular} and the edges of $E_G(R,L)$ \emph{special}. We say that $G$ is a \emph{well-structured graph}, if every vertex of $G$ is incident to at most one special edge. If every vertex of $G$ is incident to {\bf exactly} one special edge, then we say that $G$ is a \emph{perfect well-structured graph}.
\end{definition}

If $G$ is a well-structured graph, then we assume that the partition $(L,R)$ of its vertices is given as part of the description of $G$.
We will sometimes consider well-structured graphs with lengths on their edges, and we will sometimes restrict our attention to assignments of lengths to edges that have special properties, that we define below. We call such assignments of lengths to edges \emph{proper}.

\begin{definition}[Proper Assignment of Edge Lengths]
	Let $G=(L,R,E)$ be a well-structured directed graph, and let $\ell(e)\geq 0$ be an assignment of lengths to the edges $e\in E$. We say that $\set{\ell(e)}_{e\in E}$ is a \emph{proper} assignment of edge lengths, if, for every regular edge $e\in E$, $\ell(e)=0$, and for every special edge $e\in E$, $\ell(e)$ is an integral power of $2$ with $\ell(e)\geq 1$. If, for every special edge $e\in E$, $\ell(e)=1$ holds, then we say that $\set{\ell(e)}_{e\in E}$ is a \emph{proper uniform assignment of edge lengths}.
\end{definition}

Recall that a \emph{cut} in a graph $G$ is an ordered pair $(A,B)$ of non-empty subsets of $V(G)$ with $A\cap B=\emptyset$ and $A\cup B=V$.
We will often use the notion of a well-structured cut, that is defined next.

\begin{definition}[Well-Structured Cut]
	Let $G=(L,R,E)$ be a well-structured graph, and let $(A,B)$ be a cut in $G$. We say that $(A,B)$ is a \emph{weakly well-structured cut}, or just a \emph{well-structured cut}, if all edges of $E_G(A,B)$ are special edges. If, additionally, all edges in $E_G(B,A)$ are regular edges, then we say that $(A,B)$ is a \emph{strongly} well-structured cut.
\end{definition}

%% file: route-and-cut.tex
\subsection{The RouteAndCut problem}

In this subsection we define the \routeandcut problem, which is one of the main building blocks of our algorithm. Before we do so, we need to define the notion of \emph{routing}.

\begin{definition}[Routing]
	Let $G=(V,E)$ be a directed graph, and let $A,B$ be two disjoint subsets of its vertices. A \emph{partial routing}, or just a \emph{routing}, from $A$ to $B$ is a collection $\qset$ of paths in $G$, such that every path in $\qset$ connects a vertex of $A$ to a vertex of $B$, and the endpoints of the paths in $\qset$ are all disjoint. The \emph{congestion} of the routing is the edge-congestion that the set $\qset$ of paths causes in graph $G$. 
Vertices of $A\cup B$ may serve both as endpoints and as inner vertices of the paths in $\qset$ simultaneously. A partial routing of cardinality $|A|$ is called a \emph{complete} routing. 
	\end{definition}

We are now ready to define the \routeandcut problem, and its special case, called  $r$-restricted \routeandcut. Intuitively, we use the notion of the $r$-restricted \routeandcut problem in order to discretize the problem instances: our algorithm for \routeandcut  will consider, by induction, $r$-restricted instances of the problem, from smaller to larger values of $r$.

\begin{definition}[\routeandcut problem]\label{def: rounte and cut}
The input to the \routeandcut problem is a well-structured $n$-vertex graph $G=(L,R,E)$, that is given in the adjacency-list representation, 
parameters $N\geq n$, $\Delta\geq 1$, and $1\leq \eta \leq \Delta$, and two disjoint subsets $A,B$ of vertices of $G$, with $|A|,|B|\geq \Delta$. 

The output of the problem is a (possibly partial) routing $\qset$ from $A$ to $B$, whose congestion is bounded by  $4\eta\log N$. Additionally, if $|\qset|<\Delta$, the output must contain a cut  $(X,Y)$ in $G$ with $|E_G(X,Y)|\leq \frac{64\Delta}{ \eta\log^4n}+\frac{256|\qset|}{\eta}$, such that, if 
$A'\subseteq A, B'\subseteq B$ denote the subsets of vertices that do not serve as endpoints of the paths in $\qset$, then $A'\subseteq X$ and $B'\subseteq Y$ hold. The algorithm for the \routeandcut problem is also allowed to return ``FAIL'' without producing any output, but the probability of the algorithm doing so must be bounded by $1/2$.

We say that an instance $(G,A,B,\Delta,\eta,N)$ of the \routeandcut problem is  \emph{$r$-restricted}, for an integer $1\leq r\leq \ceil{\sqrt{\log N}}$, iff $\frac{(n-|B|)\cdot \eta}{\Delta}\leq 2^{r\cdot \sqrt{\log N}}$ holds.
\end{definition}

We provide a brief intuition for the parameters $\Delta$ and $N$ in the definition of the \routeandcut problem, and for the definition of $r$-restricted instances. Consider an arbitrary well-structured $n$-vertex graph $G$, two disjoint sets $A,B$ of its vertices and a parameter $\eta$, and assume that we need to compute a large routing from $A$ to $B$ in $G$, with congestion at most $\eta$. Typically, we will set $\Delta$ to be the cardinality of the desired routing (and if it is not known, we can set $\Delta=\min\set{|A|,|B|}$). Initially, we will also set $N=n$, and consider the resulting instance $(G,A,B,\Delta,\eta,N)$ of the \routeandcut problem. 
Now suppose we compute a routing from $A$ to $B$ in graph $G$ of cardinality $\Delta$ that causes edge-congestion at most $\eta$. Assume w.l.o.g. that, if $b\in B$ is an inner vertex on some path of the routing, then it serves as an endpoint of some other path of the routing.  Then it is not hard to show that most of the paths in the routing must have length at most $d=O\left(\frac{(n-|B|)\cdot \eta}{\Delta}\right )$.
The definition or $r$-restricted instances requires that this parameter $d$ is roughly bounded by $2^{r\sqrt{\log N}}$.
It is easy to see that our starting instance is $r^*$-restricted, for $r^*=\ceil{\sqrt{\log N}}$. In order to solve this problem instance, we will need to solve the problem recursively on smaller subgraphs $G'$ of $G$, but it is important for us to ensure that the resulting instances are $r'$-restricted, for $r'<r^*$. In order to do so, we will let the parameter $N$ in the resulting subinstances correspond to the number of vertices in the original graph $G$ (we may need to slightly adjust it for technical reasons but the adjustments are minor). By appropriately setting the parameters $\Delta$ and $\eta$ for the resulting subinstances, we can then ensure that they are indeed $r'$-restricted, for some $r'<r^*$. Overall, the parameter $N$ can be thought of as roughly the number of vertices in the original instance of the \routeandcut problem that we try to solve, and it is used mostly to define the notion of $r$-restricted instances. The notion of $r$-restricted instances allows us to construct algorithms for the \routeandcut problem inductively, by going from smaller to larger values of $r$. 

The following theorem provides one of our main technical results, namely an efficient algorithm for the \routeandcut problem.

\begin{theorem}\label{thm: main for route and cut}
	There is a randomized algorithm for the \routeandcut problem, that, on an input $(G,A,B,\Delta,\eta, N)$, %either produces a valid output to the problem, or returns ``FAIL''. The probability that the algorithm returns ``FAIL'' is at most $1/2$, and its running time is: 
	has running time at most 	$O\left(n\cdot (n-|B|)\cdot 2^{O(\sqrt{\log N}\cdot \log \log N)}\right )$, where $n=|V(G)|$.
\end{theorem}

The remainder of this paper is largely dedicated to the proof of \Cref{thm: main for route and cut}, but first we show that the proof of \Cref{thm:main} follows from it.

%% file: completing-main-thm.tex
\subsection{Completing the Proof of \Cref{thm:main}}
\label{subsec: completing matching theorem}

Recall that in the \MBM problem, the input is a simple undirected $n$-vertex bipartite graph $G=(L,R,E)$, and the goal is to compute a matching of maximum cardinality in $G$.

We can assume w.l.o.g. that we are given a target integral value $C^*>0$. If $C^*\leq \opt$, then our algorithm must either produce a matching of cardinality $C^*$, or terminate with a ``FAIL'', but we require that the algorithm terminates with a ``FAIL'' with probability at most $1/\poly(n)$ in this case. If $C^*>\opt$, then our algorithm may return an arbitrary matching, or terminate with a ``FAIL''.

We can make this assumption w.l.o.g. since we can use a standard technique of performing a binary search over the value $\opt$ of the optimal solution. 
We use bounds $W<W'$ on the value of the optimal solution, initializing $W=1$ and $W'=n$, and we start with $C^*=\floor{n/2}$.
 Given the current guess $C^*$ on $\opt$, we run the algorithm with this target value $C^*$. If the algorithm successfully computes a matching of cardinality at least $C^*$, then we set $W=C^*$, and
 otherwise, we set $W'=C^*$. In either case, we
  select the next value $C^*=\floor{\frac{W'-W}{2}}$ and continue. Once $W'- W\leq 1$ holds, we run the algorithm once with $C^*=W'$ and once with $C^*=W$, and then output the better of the two solutions.
We say that an algorithm \emph{errs} in an iteration, if it returns ``FAIL'', but $C^*\leq \opt$ holds in that iterations. Notice that, if the algorithm never errs, then after $O(\log n)$ iterations it is guaranteed to terminate with an optimal matching. Since the probability that the algorithm errs in a single iteration is bounded by $1/\poly(n)$, and the number of iterations is bounded by $O(\log n)$, the algorithm is guaranteed to return an optimal matching with probability at least $1-1/\poly(n)$.

From now on, we assume that we are given a target value $C^*$, and that $G$ contains a matching of cardinality $C^*$. Our algorithm must either compute a matching of cardinality $C^*$, or to return ``FAIL'', but the probability for returning ``FAIL'' must be bounded by $1/\poly(n)$. For convenience, we denote $C^*$ by $\opt$. 
 We can assume that $n$ is greater than a large enough constant, since otherwise the problem can be solved in $O(1)$ time.

It is well known that the \MBM problem can be reduced to computing maximum $s$-$t$ flow in a directed flow network with unit edge capacities. In order to do so, we start with the graph $G$, and direct all its edges from the vertices of $L$ towards the vertices of $R$. We then add a source vertex $s$, that connects with an edge to every vertex of $L$, and a destination vertex $t$, to which every vertex of $R$ is connected. All edge capacities are set to $1$. Let $G'$ denote the resulting directed flow network. It is easy to verify that the value of the maximum $s$-$t$ flow in $G'$ is equal to the cardinality of maximum matching in $G$. Moreover, given an integral flow $f$ in $G'$, we can compute a matching $M$ of cardinality $\operatorname{val}(f)$ in $G$, in time $O(n)$: we simply include in $M$ all edges of $G$ that carry $1$ flow unit.

Our algorithm consists of $O(\log^3 n)$ phases. Throughout the algorithm, we maintain a matching $M$ in $G$, starting with $M=\emptyset$, and we denote $\Delta^*=\opt-|M|$. If $M$ is the matching at the beginning of a phase, and $M'$ is the matching obtained at the end of the phase, then we require that $|M'|\geq |M|+\Omega(\Delta^*/\log^2 n)$, where $\Delta^*=\opt-|M|$. We show a combinatorial algorithm, that, given an initial matching $M$, either returns ``FAIL'', or computes such a matching $M'$, in time $O\left (n^{2}\cdot 2^{O(\sqrt{\log n}\log\log n)}\right )$, where the probability that the algorithm returns ``FAIL'' is bounded by $1/\poly(n)$. This will ensure that the number of the phases is bounded by $O(\log^3 n)$, and the total running time of the algorithm is $O\left (n^{2}\cdot 2^{O(\sqrt{\log n}\log\log n)}\right )$.
From now on we focus on the description of a single phase.

\paragraph{Implementation of a Single Phase.}
 We assume that we are given some matching $M$ in the input graph $G$, and we denote $\Delta^*=\opt-|M|$. Our goal is to compute a matching $M'$ of cardinality at least $|M|+\Omega(\Delta^*/\log^2 n)$.
As observed already, matching $M$ defines a flow $f$ of value $|M|$ in the directed flow network $G'$ with unit edge capacities. We denote by $H=G'_f$ the corresponding residual flow network, and we say that $H$ is the \emph{residual flow network of $G$ with respect to matching $M$}. Observe that $H$ is a directed flow network with unit edge capacities, and that the value of the maximum $s$-$t$ flow in $H$ is at least $\opt-|M|=\Delta^*$. 

Next, we will define an instance of the \routeandcut problem. Consider first the graph $H'=H\setminus\set{s,t}$. Notice that, for every edge $e=(u,v)$ of $G$ with $u\in L$ and $v\in R$, if $e\in M$, then edge $(v,u)$ is present in $H'$, and otherwise edge $(u,v)$ is present in $H'$. Therefore, all edges of $H'$ that are directed from vertices of $R$ to vertices of $L$ correspond to the edges of the current matching $M$. Clearly, every vertex of $H'$ may be incident to at most one edge of $E_{H'}(R,L)$, and so graph $H'$ is a well-structured graph.

We let $A\subseteq L$ be the set of vertices $v$, such that edge $(s,v)$ is present in graph $H$, and we let $B\subseteq R$ be the set of vertices $u$, such that edge $(u,t)$ is present in $H$. We also define parameters $\eta=1$, $N=n$, and $\Delta=\Delta^*$. It is easy to verify that $(H',A,B,\Delta,\eta,N)$ is a valid input to the \routeandcut problem, and we can compute an adjacency-list representation of $H'$ in time $O(|E(G)|)$. 
Recall that there is an $s$-$t$ flow $f$ of value $\Delta^*$ in graph $H$, and, from the integrality of maximum flow in integer-capacity networks, we can assume that this flow is integral. This flow naturally defines a routing $\qset^*$ from $A$ to $B$, that causes congestion $\eta=1$, with $|\qset^*|=\Delta^*$. 

We apply the algorithm from \Cref{thm: main for route and cut} to instance $(H',A,B,\Delta,\eta,N)$ of the \routeandcut problem. 
As long as the algorithm returns ``FAIL'', we keep executing it, for up to $\ceil{100\log n}$ iterations. If the algorithm from \Cref{thm: main for route and cut}  returned ``FAIL'' in all $\ceil{100\log n}$ consecutive iterations, we terminate our algorithm and return ``FAIL''. It is easy to verify that this may happen with probability at most $1/n^{100}$. Otherwise, the algorithm from  \Cref{thm: main for route and cut}  must return a routing $\qset$ from $A$ to $B$ with congestion at most $4\log n$. We show that $|\qset|\ge \frac{\Delta^*}{\log n}$ must hold in the next observation.

\begin{observation}\label{obs: arge flow}
	If $\qset$ is the routing returned by the algorithm from \Cref{thm: main for route and cut}, then $|\qset|\geq \frac{\Delta^*}{\log n}$.
\end{observation}
\begin{proof}
	Let 	$A'\subseteq A, B'\subseteq B$ denote the subsets of vertices that do not serve as endpoints of the paths in $\qset$.
	Assume for contradiction that $|\qset|<\frac{\Delta^*}{\log n}$.
	
	Recall that the algorithm from \Cref{thm: main for route and cut} must also return a cut  $(X,Y)$ in $H'$ with $|E_{H'}(X,Y)|\leq \frac{64\Delta^*}{ \log^4n}+256|\qset|\leq \frac{512\Delta^*}{\log n}$, such that $A'\subseteq X$ and $B'\subseteq Y$. Also, recall that we have established that there is a routing $\qset^*$ from $A$ to $B$ in $H'$ of cardinality at least $\Delta^*$, that causes edge-congestion $1$. Since $|\qset|<\frac{\Delta^*}{\log n}$, at least $\frac{\Delta^*}{2}$ of the paths in $\qset^*$ must originate at vertices of $A'$ and terminate at vertices of $B'$. Each such path must contain an edge of $E_{H'}(X,Y)$, so $|E_{H'}(X,Y)|\geq \frac{\Delta^*}{2}$ must hold, a contradiction.
\end{proof}

  Notice that the paths in $\qset$ naturally define a collection $\qset'$ of at least $\frac{\Delta^*}{\log n}$ $s$-$t$ paths in graph $H$ (the  residual flow network with respect to $G$ and the current matching $M$), and they cause congestion at most $4\log n$ (since the endpoints of the paths in $\qset$ are disjoint).
 Next, we show an algorithm that computes a collection $\qset''$ of $\Omega\left(\frac{\Delta^*}{\log^2 n}\right )$ edge-disjoint $s$-$t$ paths in graph $H$. We will then use the paths in $\qset''$ in order to augment the current flow $f$ in graph $G'$, which, in turn, will allow us to compute the new augmented matching $M$.
 
In order to compute the collection $\qset''$ of paths,   
we construct a directed graph $H''\subseteq H$, that consists of all vertices and edges that participate in the paths of $\qset'$.
The capacity of every edge in $H''$ remains unit -- the same as its capacity in $H$. Next, we show that $|E(H'')|\leq O(n\log n)$, and that there is an $s$-$t$ flow of value at least $\frac{\Delta^*}{4\log^2 n}$ in $H''$.

\begin{observation}\label{obs: inner flow}
	 $|E(H'')|\leq O(n\log n)$, and there is an $s$-$t$ flow of value at least $\frac{\Delta^*}{4\log^2 n}$ in $H''$.
\end{observation}
\begin{proof}
	We start by showing that every vertex $v\in V(H'')\setminus\set{s,t}$ participates in at most $8\log n$ paths in $\qset'$.  Indeed, consider any vertex $v\in V(H'')\setminus\set{s,t}$, and let $Q\in \qset'$ be a path in which vertex $v$ participates. Assume first that $v\in L$. Then either edge $(s,v)$ belongs to $H$ and lies on path $Q$, or there is a special edge $(u,v)$ incident to $v$ in $H$, that lies on path $Q$. Since the paths in $\qset'$ cause congestion at most $4\log n$ in $H$, there may be at most $8\log n$ paths in $\qset'$ that contain either of the edges $(s,v)$ or $(u,v)$, and so $v$ participates in at most $8\log n$ paths in $\qset'$. Using a similar argument, every vertex $u\in R$ participates in at most $8\log n$ paths in $\qset'$. Therefore, $|E(H'')|\leq O(n\log n)$.
	
	Consider now a flow in $H''$, where we send $\frac{1}{4\log n}$ flow units along every path in $\qset'$. Since the paths in $\qset'$ cause congestion at most $4\log n$, this is a valid $s$-$t$ flow in $H''$, and the value of this flow is $\frac{|\qset'|}{4\log n}\geq \frac{\Delta^*}{4\log^2 n}$.
\end{proof}

Next, we compute an integral maximum $s$-$t$ flow in $H''$ that obeys the edge capacities in $H''$, using the standard Ford-Fulkerson algorithm. We start by setting the flow $f'(e)$ on every edge $e\in E(H')$ to $0$, and then perform iterations. In every iteration, we use the current flow $f'$ in $H''$ in order to compute a residual flow network $H''_{f'}$. If there is no $s$-$t$ path in $H''_{f'}$, then the algorithm terminates. We are then guaranteed that the value of the current flow $f'$ is optimal, so $\val(f')\geq \frac{\Delta^*}{4\log^2 n}$. Otherwise, we compute an arbitrary $s$-$t$ path $P$ in $H''_{f'}$, and we augment the flow $f'$ along the path $P$. Notice that every iteration of the algorithm can be implemented in time $O(|E(H''_{f'})|)=\otilde(n)$, and the number of iterations is bounded by $n$. Therefore, the total running time of the algorithm that computes the set $\qset''$ of paths in $H''$ is $\otilde(n^2)$. The final integral flow $f'$ can be decomposed, in time $O(|E(H'')|)=\otilde(n)$, into a collection $\qset''$ of at least $\Omega\left(\frac{\Delta^*}{\log^2 n}\right ) $ edge-disjoint paths connecting $s$ to $t$ in $H''$. Since $H''\subseteq H$, we have now obtained the desired collection $\qset''$ of at least $\Omega\left(\frac{\Delta^*}{\log^2 n}\right ) $ edge-disjoint $s$-$t$ paths  in $H$.

We note that we could also directly round the initial fractional $s$-$t$ flow in $H''$, in expected time $\otilde(|E(H'')|)=\otilde(n)$, e.g. by using the algorithm from Theorem 5 in~\cite{LRS13}, that builds on the results of \cite{GKK10}. But since their running time is only bounded in expectation, and since the bottlenecks in the running time of our algorithm lie elsewhere, we instead use the above simple deterministic algorithm.

We can now augment the current flow in $G'$ via the collection $\qset''$ of  augmenting paths, obtaining a new integral flow in graph $G'$ of value $|M|+|\qset''|$, which, in turn, defines a new matching $M'$ with $|M'|\geq |M|+|\qset''|\geq |M|+\Omega\left(\frac{\Delta^*}{\log^2 n}\right )$.

We now bound the running time of a single phase.
Recall that we may execute the  algorithm from \Cref{thm: main for route and cut} at most $O(\log n)$ times per phase. 
The running time of a single execution of the algorithm from \Cref{thm: main for route and cut} is bounded by  $O\left(n^2\cdot 2^{O(\sqrt{\log N}\cdot \log \log N)}\right )\leq O\left(n^2\cdot 2^{O(\sqrt{\log n}\cdot \log \log n)}\right )$.
Additionally, the time required to compute the graph $H''$, and to compute the maximum flow in it is bounded by $\otilde(n^2)$.

Overall, the running time of a single phase is bounded by $O\left(n^2\cdot 2^{O(\sqrt{\log n}\cdot \log \log n)}\right )$. Since the number of phases is bounded by $O(\log^3n)$, the total running time of the algorithm is bounded by $O\left(n^2\cdot 2^{O(\sqrt{\log n}\cdot \log \log n)}\right )$, and the probability that the algorithm ever returns ``FAIL'' is bounded by $1/\poly(n)$.

%% file: high-level-route-and-cut.tex
\subsection{Proof of \Cref{thm: main for route and cut} -- High Level Overview}

In order to prove \Cref{thm: main for route and cut}, we use two large enough constants $c_1\gg c_2$, whose values we set later, and we employ the following theorem.

\begin{theorem}\label{thm: main for i-restricted route and cut}
	For all $r\geq 1$, there is a randomized algorithm for the $r$-restricted \routeandcut problem, that, on an input $(G,A,B,\Delta,\eta,N)$ with $|V(G)|=n$, %either produces a valid output to the problem, or returns ``FAIL''. The probability that the algorithm returns ``FAIL'' is at most $1/2$, and its running time is at most: 
	has running time at most:
	
\[c_1\cdot n\cdot (n-|B|)\cdot  2^{c_2\sqrt{\log N}}\cdot (\log N)^{16c_2(r-1)+8c_2}.\]
\end{theorem}

Note that, if $(G,A,B,\Delta,\eta,N)$ is an instance of the \routeandcut problem, and $|V(G)|=n$, then, from the problem definition, $\eta\leq \Delta$ and $\frac{(n-|B|)\cdot \eta}{\Delta}\leq n$ holds. Therefore, any instance of the problem is also an instance of $r^*$-restricted \routeandcut problem, for $r^*=\ceil{\sqrt{\log n}}$. By using \Cref{thm: main for i-restricted route and cut} with parameter $r^*=\ceil{\sqrt{\log n}}$, we obtain an algorithm for the general \routeandcut problem whose running time is:

\[O\left(n\cdot (n-|B|)\cdot 2^{O(\sqrt{\log N}\cdot \log \log N)}\right ).\]

Therefore, in order to prove \Cref{thm: main for route and cut}, it is enough to prove \Cref{thm: main for i-restricted route and cut}. In the remainder of this paper we focus on the proof of \Cref{thm: main for i-restricted route and cut}. The proof is by induction on $r$, and it relies on a slight modification of the Multiplicative Weight Update (\MWU) framework of  \cite{GK98, Fleischer00}, that essentially reduces the $r$-restricted \routeandcut problem to a special case of the directed \SSSP problem, that we call $r$-Restricted \stSP. Before we describe this approach, we describe an algorithm that, given an instance $(G,A,B,\Delta,\eta,N)$ of the \routeandcut problem, constructs a subgraph $G'\subseteq G$ with $V(G')=V(G)$ and $|E(G')|\leq O((n-|B|)\cdot  n)$, such that a solution to the \routeandcut problem instance  $(G',A,B,\Delta,\eta,N)$ immediately provides a solution to the \routeandcut problem instance $(G,A,B,\Delta,\eta,N)$.

\subsubsection{An Auxiliary Graph $G'$ and its Properties}

Consider an instance $(G=(L,R,E),A,B,\Delta,\eta,N)$ of the $r$-restricted \routeandcut problem, for an integer  $1\leq r\leq \ceil{\sqrt{\log N}}$. Let $G'\subseteq G$ be a graph, that is constructed as follows. We let $V(G')=V(G)$, and we include in $G'$ all special edges of $G$. Additionally, for every vertex $v\in L$, if $|\delta^+_G(v)| \leq 512(n-|B|)$, then we add all regular edges that leave $v$ in $G$ to graph $G'$, and otherwise, we select an arbitrary subset $E'(v)\subseteq \delta^+_G(v)$ of $512(n-|B|)$ such edges, and include these edges in $G'$.

Note that $G'$ is a well-structured $n$-vertex graph, and $(G',A,B,\Delta,\eta,N)$ is a valid instance of the $r$-restricted \routeandcut problem. Moreover, the out-degree of every vertex $v\in L$ is bounded by $512(n-|B|)$, and so $|E(G')|\leq  O((n-|B|)\cdot n)$. Given the adjacency-list representation 
of graph $G$, graph $G'$ can be constructed explicitly in time $ O((n-|B|)\cdot n)$.

Since $G'\subseteq G$, any routing $\qset$ from $A$ to $B$ with congestion at most $4\eta\log N$ in $G'$ defines a routing from $A$ to $B$ with congestion at most $4\eta\log N$ in $G$. Assume now that we are given a routing $\qset$ from $A$ to $B$ in $G'$ with congestion at most $4\eta\log N$, such that $|\qset|<\Delta$. Let $A'\subseteq A, B'\subseteq B$ denote the subsets of vertices that do not serve as endpoints of the paths in $\qset$. Assume, additionally, that we are given  a cut  $(X,Y)$ in $G'$ with $|E_{G'}(X,Y)|\leq \frac{64\Delta}{ \eta\log^4n}+\frac{256|\qset|}{\eta}$, such that $A'\subseteq X$ and $B'\subseteq Y$. Clearly, $(X,Y)$ is also a cut in graph $G$. We claim that $E_G(X,Y)=E_{G'}(X,Y)$. Indeed, since $|\qset|< \Delta$, we get that $|B'|> |B|-\Delta$, and so $|X|\leq n-|B'|< n-|B|+\Delta\leq 2(n-|B|)$ (here, we have used the fact that $\Delta\leq |A|\leq (n-|B|)$). Assume for contradiction that $E_{G}(X,Y)\neq E_{G'}(X,Y)$. Since $G'\subseteq G$, there must be an edge $e=(x,y)\in E_G(X,Y)$, such that $e\not \in E(G')$. This can only happen if $x\in L$, and graph $G'$ contains a subset $E'(x)\subseteq \delta^+_G(x)$ of $512(n-|B|)$ edges. Since $|X|\leq 2(n-|B|)$, at least $510(n-|B|)$ edges of $E'(x)$ must lie in $E_{G'}(X,Y)$. Therefore, $|E_{G'}(X,Y)|\geq 510(n-|B|)$ must hold. However, $|E_{G'}(X,Y)|\leq \frac{64\Delta}{ \eta\log^4n}+\frac{256|\qset|}{\eta}\leq 509\Delta\leq 509(n-|B|)$, a contradiction. From the above discussion, we obtain the following claim.

\begin{claim}\label{claim: construct auxiliary}
	There is a deterministic algorithm, that, given an instance $(G,A,B,\Delta,\eta,N)$ of the $r$-restricted \routeandcut problem (where graph $G$ is given as an adjacency list), constructs another instance $(G',A,B,\Delta,\eta,N)$ of the $r$-restricted \routeandcut problem with $G'\subseteq G$, such that $|E(G')|\leq O(n\cdot (n-|B|))$, and moreover, if $(\qset,(X,Y))$ is a valid solution to instance $(G',A,B,\Delta,\eta,N)$, then it is also a valid solution to instance $(G,A,B,\Delta,\eta,N)$. The running time of the algorithm is $O(n\cdot(n-|B|))$.
\end{claim}

In our algorithms for the \routeandcut problem, we will use \Cref{claim: construct auxiliary} in order to ensure that the number of edges in the input graph is bounded by $O(n\cdot (n-|B|))$.
 We now turn to describe the \MWU-based approach for solving $r$-restricted \routeandcut.

%% file: advanced-routeandcut.tex
\subsection{Solving $r$-Restricted \routeandcut via the Modified MWU Framework and a Reduction to $r$-Restricted \stSP}
\label{subsec: modified MWU}

In this subsection we provide a high-level description of our algorithm for the $r$-restricted \routeandcut problem, including the modified \MWU framework that we use, and a reduction to a problem that we call $r$-restricted \stSP, that can be viewed as a special case of the \SSSP problem. We  show that an algorithm for the $r$-restricted \stSP problem implies an algorithm for the $r$-restricted \routeandcut problem, via the modified \MWU framework. Lastly, we provide a simple though not very efficient algorithm for $r$-restricted \stSP, that will be used as the induction base in the proof of \Cref{thm: main for i-restricted route and cut}.

Let $(G,A,B,\Delta,\eta,N)$ be the input to the $r$-restricted \routeandcut problem instance, where $G=(L,R,E)$ is a well-structured graph with $|V(G)|=n$, that is given as an adjancency list, $A$ and $B$ are disjoint subsets of $V(G)$, and  $N\geq n$, $\Delta\geq 1$, and $1\leq \eta \leq \Delta$ are parameters, with $|A|,|B|\geq \Delta$. Since the instance is $r$-restricted, $\frac{(n-|B|)\cdot \eta}{\Delta}\leq 2^{r\cdot \sqrt{\log N}}$ holds. By using the algorithm from \Cref{claim: construct auxiliary}, we convert this instance into another instance  $(G',A,B,\Delta,\eta,N)$ of $r$-restricted \routeandcut with $|E(G')|\leq O(n\cdot (n-|B|))$, in time $O(n\cdot (n-|B|))$. From now on we focus on solving instance $(G',A,B,\Delta,\eta,N)$, and, for convenience, we denote $G'$ by $G$.

Recall that our goal is to compute a routing $\qset$ from $A$ to $B$, whose congestion is bounded by  $4\eta\log N$. Additionally, if $|\qset|<\Delta$, we need to compute a cut  $(X,Y)$ in $G$ with $|E_G(X,Y)|\leq \frac{64\Delta}{ \eta\log^4n}+\frac{256|\qset|}{\eta}$, such that, if 
$A'\subseteq A, B'\subseteq B$ denote the subsets of vertices that do not serve as endpoints of the paths in $\qset$, then $A'\subseteq X$ and $B'\subseteq Y$ hold.
Recall also that the edges of $E_G(R,L)$ are called special edges, and the remaining edges of $G$ are called regular edges.
From the definition of well-structured graphs, every vertex of $G$ is incident to at most one special edge.
We start with a simple preprocessing step.

\paragraph{A Preprocessing Step.}
As a preprocessing step, we construct a maximal collection $\qset_0$ of disjoint paths, where every path connects a distinct vertex of $A$ to a distinct vertex of $B$, and consists of a single regular edge. We do so by employing a simple greedy algorithm: Start with $\qset_0=\emptyset$. While there is a regular edge $e=(a,b)$ with $a\in A$ and $b\in B$, add a path $Q=(e)$ to $\qset_0 $, and delete $a$ and $b$ from $A$ and $B$, respectively. Clearly, this preprocessing step can be executed in time $O(|E(G)|)\leq O(n\cdot (n-|B|))$. If $|\qset_0|\ge \Delta$, then we terminate the algorithm and return $\qset_0$. We assume from now on that $|\qset_0|<\Delta$.
Let $A_1\subseteq A$, $B_1\subseteq B$ denote the sets of vertices that do not serve as endpoints of the paths in $\qset_0$. Since $|\qset_0|<\Delta$, we get that $|B_1|\geq |B|-\Delta$, and so:

\begin{equation}\label{eq new marked}
n-|B_1|\leq n-|B|+\Delta\leq 2(n-|B|),
\end{equation}

since $n-|B|\geq |A|\geq \Delta$ from the problem definition. The preprocessing step ensures that there is no regular edge in $G$ connecting a vertex of $A_1$ to a vertex of $B_1$, so every path connecting a vertex of $A_1$ to a vertex of $B_1$ must contain at least one special edge.

Next, we describe a modified \MWU framework that will allow us to compute a solution to the $r$-restricted \routeandcut problem instance. %Intuitively, the basic \MWU framework that was used in the proof of  \Cref{lem: main for route and cut basic} reduced the \routeandcut problem to a special case of \SSSP (the problem that was solved by the oracle). However, since the algorithm for solving this instance of \SSSP was not very efficient, the resulting algorithm for the \routeandcut had high running time. 
The modified \MWU framework is designed in such a way that it reduces the $r$-restricted \routeandcut problem to a special case of \SSSP, that we call $r$-restricted \stSP, which seems more tractable than the general decremental \SSSP problem in directed graphs. We will then design an algorithm for the  $r$-restricted \stSP problem, that will allow us to obtain the desired algorithm for $r$-restricted \routeandcut. The algorithm for $r$-restricted \stSP will in turn rely on an algorithm for the $(r-1)$-restricted \routeandcut problem.

\input{modified-MWU2}
\input{restricted-STSP}

\input{induction-base}

%% file: modified-MWU2.tex
\subsubsection{The Modified \MWU Framework}
\label{subsubsec: modified MWU}

We now describe an algorithm that is based on the modified \MWU framework, but ignore the issue of the efficient implementation of the algorithm. We address this issue later, by reducing the problem to the $r$-restricted \stSP problem.

Throughout the algorithm, we use a parameter ${	\Lambda=(n-|B_1|)\cdot \frac{ \eta  \log^5n}{\Delta}}$. While the set $B_1$ of vertices may change over the course of the algorithm, the value of the parameter $\Lambda$ is set at the beginning of the algorithm and remains unchanged throughout the algorithm. 
Our algorithm maintains an assignment $\ell(e)\geq 0$ of \emph{lengths} to the edges of $G$. At the beginning of the algorithm, we assign an initial length $\ell(e)$ to every edge $e\in E(G)$, as follows. If $e$ is a regular edge, we set $\ell(e)=0$, and if it is a special edge, we set $\ell(e)=\frac{1}{\Lambda}$. 

As the algorithm progresses, the lengths of the special edges may grow, but the lengths of the regular edges remain unchanged. 
Whenever we talk about distances between vertices and lengths of paths, it is always with respect to the current lengths $\ell(e)$ of edges $e\in E(G)$. 

\iffalse
We start by defining the following property, that we refert to as a \emph{stopping condition}. If the algorithm ever establishes that this property does not hold, we will terminate it.

\begin{properties}{C}
	\item Let $\pset$ be the largest-cardinality collection of $s$-$t$ paths in the current graph $G$ that obeys edge capacities, such that every path has length at most  $\frac 1 {64}$. Then $|\pset|\geq \frac{\Delta}{\log^{2\hat c}n}$. \label{stopping condition2}
\end{properties} 
\fi

The algorithm gradually constructs a routing $\qset$ from $A_1$ to $B_1$ in $G$, starting with $\qset=\emptyset$. It then performs at most $\Delta$ iterations, and in every iteration, a single path $P$, connecting some vertex $a \in A_1$ to some vertex $b \in B_1$, is added to $\qset$.
We add $P$ to $\qset$, and delete $a$ from $A_1$ and $b$ from $B_1$. Additionally, we may double the lengths of {\em some} special edges on the path $P$. Specifically, our algorithm maintains, for every special edge $e$ of $G$, a counter $n(e)$, which, intuitively, counts the number of paths that were added to $\qset$ that contained $e$, since $\ell(e)$ was last doubled (or since the beginning of the algorithm, if $\ell(e)$ was never doubled yet). At the beginning of the algorithm, we set $n(e)=0$ for every special edge $e$. Whenever a path $P$ is added to $\qset$, we increase the counter $n(e)$ of every special edge $e\in E(P)$. If, for any such edge $e$, the counter $n(e)$ reaches $\eta$, then we double the length of edge $e$, and reset the counter $n(e)$ to be $0$. 

\subsubsection*{The Oracle}

At the heart of our algorithm is an \emph{oracle} -- an algorithm that, in every iteration, either computes a path $P$ of length at most $1$ in the current graph $G$ connecting a vertex of $A_1$ to a vertex of $B_1$, or produces a new assignment $\ell'(e)$ for edges $e\in E(G)$, that we call a \emph{cut-witness}, and define next. Intuitively, the cut-witness is designed in such a way that it can be easily transformed into the desired  cut  $(X,Y)$ in $G$ with $|E_G(X,Y)|\leq \frac{64\Delta}{ \eta\log^4n}+\frac{256|\qset|}{\eta}$, such that, if 
$A'$ and $B'$ denote the current sets $A_1$ and $B_1$ respectively, then $A'\subseteq X$ and $B'\subseteq Y$ hold. Once the oracle produces a cut-witness, the algorithm terminates. We now define the cut-witness formally.

\begin{definition}[Cut-witness]\label{def: cut-witness}
A \emph{cut-witness} is an assignment of lengths $\ell'(e)\geq 0$ to every edge $e\in E(G)$, such that, if  we denote by $A'$ and $B'$ the current sets $A_1$ and $B_1$ of vertices respectively, and by $E^*$ the set of all special edges with both endpoints in $B'$, then:		

\begin{itemize}
%		\item the length $\ell'(e)$ of every regular edge is $0$;
		\item $\sum_{e\in E(G)}\ell'(e)\leq \frac{\Delta}{2\eta\log^4n}+\sum_{e\in E(G)\setminus E^*}\ell(e)$; and
		
		\item the distance in graph $G$, with respect to edge lengths $\ell'(\cdot)$, from $A'$ to $B'$ is at least $\frac{1}{64}$. %, and moreover, the length $\ell'(e)$ of every edge $e$ connecting a pair of vertices in $B'$ is $0$.
	\end{itemize}
\end{definition}

\sknote{I am not sure but I think we did not define distance between 2 sets. If so, I can just add a definition here.}

Next, we define the notion of an \emph{acceptable path} in $G$, and we will require that the oracle, in every iteration, either returns an acceptable path, or returns a cut-witness.

\begin{definition}[Acceptable Path]\label{def: acceptable path}
	A path $P$ in the current graph $G$ is called \emph{acceptable} if $P$ is a simple path, connecting a vertex of $A_1$ to a vertex of $B_1$, the length of $P$ with respect to the current edge lengths $\ell(\cdot)$ is at most $1$, and no inner vertices of $P$ belong to $B_1$.
\end{definition}

We are now ready to define the oracle.

\begin{definition}[The Oracle]\label{def: oracle}
An oracle for the \MWU framework is an algorithm that, in every iteration, either returns an acceptable path $P$, or returns a cut-witness, or returns ``FAIL''. The probability that the oracle ever returns ``FAIL'' must be bounded by $1/2$.
\end{definition}

\subsubsection*{The \MWU-Based Algorithm.}
We describe the modified \MWU-based algorithm, denoted by \algmwu  in Figure \ref{alg: modified mwu}; the description excludes the implementation of the oracle, that will be discussed later.

\program{Alg-MWU}{alg: modified mwu}{
	\begin{itemize}
		\item Let $A_1 \subseteq A$ and $B_1 \subseteq B$ be the sets of vertices that do not serve as endpoints of any path in the set $\qset_0$ constructed in the preprocessing step. 
		\item Initialize the data structures:
		\begin{itemize}
			\item Set $\qset=\emptyset$;
			\item For every edge $e\in E(G)$, if $e$ is a regular edge, set $\ell(e)=0$, otherwise set $\ell(e)=\frac{1}{\Lambda}$.
			\item For every special edge $e$, set $n(e)=0$. 
		\end{itemize}
		\item Perform at most $\Delta-|\qset_0|$ iterations, where in each iteration we apply the oracle.
		\begin{itemize}
				\item If the oracle returned ``FAIL'', then return ``FAIL'' and terminate the algorithm;
				\item If the oracle returned a cut-witness, return the cut-witness and terminate the algorithm;				
			\item Otherwise, the oracle must have returned an acceptable path $P$ connecting a vertex $a\in A_1$ to a vertex $b\in B_1$.
			\begin{itemize}
				\item remove $a$ from $A_1$ and $b$ from $B_1$;
				\item add $P$ to $\qset$;
				\item  for each special edge $e\in E(P)$, increase $n(e)$ by $1$, and, if $n(e)$ reaches $\eta$, double $\ell(e)$ and set $n(e)=0$.
		\end{itemize}
		\end{itemize}
	\end{itemize}
}

Before we discuss the implementation of the oracle, and an efficient implementation of Algorithm \algmwu, we first show that we can use the algorithm in order to solve the $r$-restricted \routeandcut problem. 
Observe first that, if the total running time of the oracle, over the course of all iterations, is bounded by $T$, then the total running time of Algorithm \algmwu is bounded by $O(|T|)+O(|E(G)|)\leq O(|T|)+O(n\cdot (n-|B|))$. Moreover, since the probability that the oracle ever returns ``FAIL'' is at most $1/2$, the probability that algorithm \algmwu terminates with a ``FAIL'' is bounded by $1/2$.

Let $\qset$ be the set of paths obtained when Algorithm \algmwu terminates, and let $\qset'=\qset\cup \qset_0$ be the final set of paths that we obtain. Clearly, every path in $\qset'$ connects a vertex of $A$ to a vertex of $B$. We start by showing that the paths in $\qset'$ cause congestion at most $4\eta\log n$, and that the endpoints of all paths in $\qset'$ are disjoint, in the following simple observation.

\begin{observation}\label{obs: congestion and disjointness}
The endpoints of the paths in $\qset'$ are disjoint, and the congestion caused by the paths in $\qset'$ in $G$ is bounded by $4\eta\log n$.
\end{observation}

\begin{proof}
	It is easy to verify that the endpoints of all paths in $\qset'$ are disjoint. First, the preprocessing step, in every iteration selects a regular edge $e=(a,b)$ with $a\in A$ and $b\in B$, adds the path $Q=(e)$ to $Q_1$, and then deletes $a$ from $A$ and $b$ from $B$. This ensures that vertices $a$ and $b$ cannot serve as endpoints of any paths that are subsequently added to $\qset_0$, or to $\qset$. Additionally, whenever Algorithm \algmwu adds a new path $P$ connecting a vertex $a'\in A_1$ to a vertex $b'\in B_1$ to $\qset$, we delete $a'$ from $A_1$ and $b'$ from $B_1$. Therefore, $a'$ and $b'$ may not serve as endpoints of paths that are added subsequently to $\qset$. Hence the endpoints of all paths in $\qset'$ are disjoint.
	
Next, we show that every {\em special edge} of $G$ belongs to at most $2\eta\log n$ paths of $\qset$. At the beginning of the algorithm, we set $\ell(e)=\frac{1}{\Lambda}=\frac{\Delta}{(n-|B_1|)\cdot \eta \cdot \log^5n}\geq \frac{\Delta}{n\cdot \eta \cdot \log^5n}$. Let $i_1,i_2,\ldots,i_q$ denote the indices of the iterations when $\ell(e)$ was doubled. Notice that, once $\ell(e)$ grows above $1$, it may no longer be used by any paths that are added to $\pset$, so $q\leq \ceil{\log \Lambda}+1 \leq 2\log n$ must hold.  Moreover, for all $1\leq j<q$, the number of paths that were added to $\qset$ after iteration $i_j$ but until iteration $i_{j+1}$ (inclusive), that contained $e$, is bounded by $\eta$. 
Similarly, at  most $\eta$ paths containing $e$ were added to $\qset$ until iteration $i_1$ (inclusive).
Altogether, $e$ may belong to at most  $2\eta\log n$ paths of $\qset$. 
Since paths of $\qset_0$ do not contain special edges, {\em every special edge belongs to at most $2\eta\log n$ paths in $\qset'$}.

Consider now some {\em regular edge} $e=(u,v)\in E(G)$. Recall that no regular edge in $G$ may connect a vertex of $A_1$ to a vertex of $B_1$. Therefore, if a path $P\in \qset$ contains $e$, then it must contain either the unique special edge that is incident to $u$ (if it exists), or the unique special edge that is incident to $v$ (if it exists). Moreover, if a path $Q=(e)$ belongs to $\qset_0$, then every path $Q'\in \qset$ that contains $e$ must contain both the unique special edge incident to $u$ and the unique special edge incident to $v$ (since $u\not\in A_1$ and $v\not \in B_1$ must hold). Therefore, every regular edge $e$ belongs to at most $4\eta\log n$ paths in $\qset'$.
\end{proof}

Let $A'$ and $B'$ denote the sets $A_1$ and $B_1$, respectively, at the end of Algorithm \algmwu. Notice that $A'$ is a set of all vertices $a\in A$ that  do not serve as endpoints of the paths in $\qset'$, and similarly, $B'$ contains all vertices $b\in B$ that  do not serve as endpoints of the paths in $\qset'$.
We now show an algorithm, that, given a cut-witness $\set{\ell'(e)}_{e\in E(G)}$, computes a cut  $(X,Y)$ in $G$ with $|E_G(X,Y)|\leq \frac{64\Delta}{ \eta\log^4n}+\frac{256|\qset|}{\eta}$, such that $A'\subseteq X$ and $B'\subseteq Y$ hold.

\begin{claim}\label{claim: computing the cut}
	There is a deterministic algorithm, that, given  a cut-witness $\set{\ell'(e)}_{e\in E(G)}$ for $G$, obtained at the end of Algorithm \algmwu, computes a cut  $(X,Y)$ in $G$ with $|E_G(X,Y)|\leq \frac{64\Delta}{ \eta\log^4n}+\frac{256|\qset'|}{\eta}$, such that $A'\subseteq X$ and $B'\subseteq Y$ hold. The running time of the algorithm is $O(|E(G)|)\leq O(n\cdot (n-|B|))$.
\end{claim}
\begin{proof}
	Consider the cut-witness $\set{\ell'(e)}_{e\in E(G)}$. We start by bounding $\sum_{e\in E(G)}\ell'(e)$ in the following simple observation.
	
\begin{observation}\label{obs: bound total edge length}
	$\sum_{e\in E(G)}\ell'(e)\leq \frac{\Delta}{\eta\log^4n}+ \frac{2|\qset'|}{\eta}$.
\end{observation}
\begin{proof}
	For every special edge $e\in E(G)$, let $\ell^0(e)$ denote the length of $e$ at the beginning of Algorithm \algmwu, and let $\ell^*(e)$ denote the length of $e$ at the end of the algorithm. 
	We denote by $E'$ the set of all special edges $e\in E(G)$, such that, at the beginning of Algorithm \algmwu, both endpoints of $E'$ lied in $B_1$, but $e\not \in E^*$. For every edge $e\in E'$, at least one endpoint of $e$ was deleted from $B_1$ over the course of the algorithm, and $e$ is the only special edge incident to that endpoint, so $|E'|\leq \Delta$ must hold. Recall that, at the beginning of the algorithm, every special edge was assigned the length $\frac{1}{\Lambda}$. Let $E^{\spec}$ denote the set of all special edges in $G$. For each edge $e\in E^{\spec}$, if $e\not\in E'\cup E^*$, then, at the beginning of Algorithm \algmwu, $e$ was incident to some vertex $x\in V(G)\setminus B_1$, and moreover, $e$ is the only special edge incident to $x$. Therefore, $|E(G)\setminus (E'\cup E^*)|\leq n-|B_1|$, and $|E(G)\setminus E^*|\leq n-B_1+\Delta\le 2(n-|B_1|)$,  since $\Delta\leq |A|\leq n-|B_1|$. Overall, we get that:

\[\sum_{e\in E(G)\setminus E^*}\ell^0(e)\leq \frac{2(n-|B_1|)}{\Lambda}\leq \frac{2\Delta}{ \eta \cdot \log^5n},\]

since $\Lambda=(n-|B_1|)\cdot \frac{\eta  \log^5n}{\Delta}$.

Let $\hat E$ denote the subset of special edge $e\in E^{\spec}\setminus E^*$ with $\ell^*(e)>\ell^0(e)$. Clearly, $\sum_{e\in E(G)\setminus E^*}\ell^*(e)\leq 	\sum_{e\in E(G)\setminus E^*}\ell^0(e)+\sum_{e\in \hat E}\ell^*(e)$.
	
	Let $q$ be the total number of iterations Algorithm \algmwu, so $|\qset|=q$ holds, and denote $\qset=\set{P_1,\ldots,P_q}$, where the paths are indexed in the order in which they were added to $\qset$. For all $1\leq i\le q$, let $C_i$ denote the length of path $P_i$ when it was added to $\qset$. Clearly, $C_i\leq 1$, and $\sum_{i=1}^qC_i\leq q$.

	Consider now some edge $e\in \hat E$, and recall that $\ell(e)$ was doubled at least once during the algorithm. Let $\tau$ be the last time when the length of $e$ was doubled. Then, prior to time $\tau$, there were at least $\eta$ iterations $i$, during which the path $P_i$ that was added to $\qset$ contained $e$, and the length of $e$ during the corresponding iteration was $\frac{\ell^*(e)}{2}$. In other words, edge $e$ contributes at least $\frac{\ell^*(e)\eta}{2}$ to $\sum_{i=1}^qC_i$.
		Therefore, we get that:
	
	\[\sum_{e\in \hat E}\ell^*(e)\leq \frac{2\sum_{i=1}^qC_i}{\eta}\leq \frac{2q}{\eta}\leq \frac{2|\qset|}{\eta}.\]
	
	Altogether, we get that:
	
	\[\sum_{e\in E(G)\setminus E^*}\ell^*(e)\leq \sum_{e\in E(G)\setminus E^*}\ell^0(e)+\sum_{e\in \hat E}\ell^*(e)\leq 
	\frac{2\Delta}{\eta\cdot \log^5n}+\frac{2|\qset|}{\eta}.\]
	
Lastly, from the definition of the cut-witness, we get that:

\[\sum_{e\in E(G)}\ell'(e)\leq \frac{\Delta}{2\eta\log^4n}+\sum_{e\in E(G)\setminus E^*}\ell(e)\leq \frac{\Delta}{\eta\log^4n}+ \frac{2|\qset|}{\eta}\leq \frac{\Delta}{\eta\log^4n}+ \frac{2|\qset'|}{\eta}.\]	
\end{proof}

Consider the graph $G$ with lengths $\ell'(e)$ on its edges $e\in E(G)$. From the definition of a cut-witness, the length of a shortest-path (with respect to edge-lengths $\ell'(\cdot)$) connecting a vertex of $A'$ to a vertex of $B'$ in $G'$ is at least $\frac{1}{64}$.

Consider a thought experiment, where we select a threshold $\rho\in(0,\frac{1}{64})$ uniformly at random, and then let $X$ be the set of all vertices $v\in V(G)$, such that the distance from $A'$ to $v$, with respect to the length $\ell'(\cdot)$ of edges, is at most $\rho$. Let $Y=V(G)\setminus X$. It is easy to verify that, since the distance from $A'$ to $B'$ in $G$ is at least $\frac{1}{64}$, $A'\subseteq X$ and $B'\subseteq Y$ must hold. Moreover, the probability that an edge $e\in E(G)$ lies in $E_{G}(X,Y)$ is bounded by $64\ell'(e)$. Therefore, the expectation of $|E_{G}(X,Y)|$ is bounded by $64\sum_{e\in E(G)}\ell'(e)\leq \frac{64\Delta}{\eta\log^4n}+ \frac{128|\qset'|}{\eta}$.

By performing a standard ball-growing procedure from the set $A'$ of vertices in graph $G$, we can compute, in time $O(|E(G)|)$, a cut $(X,Y)$ in $G$ with $|E_{G}(X,Y)|\leq \frac{64\Delta}{\eta\log^4n}+ \frac{256|\qset'|}{\eta}$, such that $A'\subseteq X$ and $B'\subseteq Y$. 
The total running time of the algorithm is bounded by  $O(|E(G)|)\leq O(n\cdot (n-|B|))$.
\end{proof}

To summarize, we provided an algorithm for the $r$-restricted \routeandcut problem, that assumes the existence of an oracle.

We now discuss an efficient implementation of our algorithm for the \routeandcut problem, given an efficient implementation of the oracle. Recall that the time required for the preprocessing step, and for computing the final cut $(X,Y)$ by the algorithm from \Cref{claim: computing the cut}, is bounded by $O(n\cdot (n-|B|))$. If the running time of the oracle is bounded by $T$, then the running time of Algorithm \algmwu is bounded by $O(T+n\cdot(n-B))$. Therefore, the total running time  of the algorithm for the \routeandcut problem is bounded by $O(T+n\cdot (n-|B|))$. We now focus on implementing the oracle.

\subsubsection*{Implementing the Oracle: Intuition.}

Below we define a problem called $r$-restricted \stSP. The problem is designed to implement the oracle from \Cref{def: oracle}. However, for convenience, we  
multiply all edge lengths by factor $\Lambda$, and then round them up to the next integer. We denote the new edge lengths by $\hat \ell(e)$ for $e\in E(G)$. Notice that, throughout the algorithm, the length $\hat \ell(e)$ of every edge $e\in E(G)$ is either $0$, or it is an integer between $1$ and $\ceil{\Lambda}$. We now define the analogues of an acceptable path and a cut-witness with respect to the new edge lengths, that we refer to as \emph{modified acceptable path} and \emph{modified cut-witness}, respectively.

\begin{definition}[Modified Acceptable Path]\label{def: mod acceptable path}
	A path $P$ in the current graph $G$ is called a \emph{modified acceptable path for $G$ with respect to $A_1$ and $B_1$}, if $P$ is a simple path, connecting a vertex of $A_1$ to a vertex of $B_1$, the length of $P$ with respect to the current edge lengths $\hat \ell(\cdot)$ is at most $\Lambda$, and no inner vertices of $P$ belong to $B_1$.
\end{definition}

\begin{definition}[Modified Cut-witness]\label{def: mod cut-witness}
	A \emph{modified cut-witness} for graph $G$ with respect to $A_1$ and $B_1$ is an assignment of lengths $\ell''(e)\geq 0$ to every edge $e\in E(G)$, such that if  we denote by $A',B'$ the current sets $A_1,B_1$ of vertices, and by $E^*$ the set of all special edges with both endpoints in $B'$, then:		
	
	\begin{itemize}
		\item $\sum_{e\in E(G)}\ell''(e)\leq \frac{\Lambda\cdot \Delta}{\eta\log^4n}+\sum_{e\in E(G)\setminus E^*}\hat \ell(e)$; and
		
		\item the distance in graph $G$, with respect to edge lengths $\ell''(\cdot)$, from $A'$ to $B'$ is at least $\frac{\Lambda}{32}$. %, and moreover, the length $\ell''(e)$ of every edge $e$ connecting a pair of vertices in $B'$ is $0$.
	\end{itemize}
\end{definition}

It is immediate to verify that, if $\set{\ell''(e)}_{e\in E(G)}$ is a modified cut-witness for $G$, then, by setting the length $\ell'(e)=\frac{\ell''(e)}{2\Lambda}$ for every edge $e\in E(G)$, we obtain a valid cut-witness for $G$. Additionally, if $P$ is a valid modified acceptable path for $G$, then it is also an acceptable path for $G$.

%Recall that, given the input adjacency-list representation of graph $G$, we can compute an adjacency-list representaton of $\revG$ in time $O(1)$.
In order to implement the oracle, it is now enough to design an algorithm that, in every iteration, either produces a modified acceptable path in $G$, or produces a modified cut-witness for $G$, or returns ``FAIL'', such that the probability that the algorithm ever returns ``FAIL'' is bounded by $1/2$.
Recall that $n-|B_1|\leq 2(n-|B|)$ from Inequality \ref{eq new marked}, and so, from the definition of $r$-restricted \routeandcut problem,
$\frac{(n-|B_1|)\cdot \eta}{\Delta}\leq  \frac{2(n-|B|)\cdot \eta}{\Delta}\leq 2\cdot 2^{r\cdot \sqrt{\log N}}$ holds.

 We now provide a formal definition of the $r$-restricted \stSP problem, that is designed specifically to provide such an algorithm. For convenience, in the definition of the problem, we will denote the sets $A_1,B_1$ of vertices by $A$ and $B$, respectively, and we will denote the edge lengths $\hat \ell(e)$ by $\ell(e)$.

\iffalse
Consider now the following stopping condition:

\begin{properties}[2]{C}
	\item Let $\qset$ be the largest-cardinality collection of $s$-$t$ paths in the current graph $H$ that obeys edge capacities, such that every path has length at most  $\frac{\Lambda}{64}$. Then $|\qset|\geq \frac{\Delta}{\log^{2\hat c}n}$. \label{stopping condition3}
\end{properties} 

Clearly, if Condition \ref{stopping condition3} is violated in graph $H$, then Condition \ref{stopping condition2} is violated in graph $G$.
Additionally, at any time, if $P$ is an $s$-$t$ path in $H$ of length at most $\Lambda$, then the length of $P$ in $G$ is at most $1$. We are now ready to formally describe the $r$-restricted \stSP problem, that will be applied to graph $H$ in order to implement the oracle.
\fi

%% file: restricted-STSP.tex
\subsubsection{The $r$-Restricted \stSP Problem}
\label{subsec: stSP problem}
The input to the $r$-Restricted \stSP problem is a well-structured $n$-vertex graph $G=(L,R,E)$, parameters $N\geq n$, $\Delta\ge 1$, and $1\leq \eta \leq \Delta$,  and two disjoint subsets $A,B$ of vertices of $G$, such that no regular edge of $G$ connects a vertex of $A$ to a vertex of $B$ in $G$, $\Delta\leq n-|B|$, and  $\frac{(n-|B|)\cdot \eta}{\Delta}\leq 2\cdot 2^{r\cdot \sqrt{\log N}}$ holds.
We also require that $|E|\leq O(n\cdot (n-|B|))$ and that $1\leq r\leq \ceil{\sqrt{\log N}}$.

Initially, every edge $e\in E(G)$ is assigned a length $\ell(e)$ as follows: if $e$ is a regular edge then $\ell(e)=0$, and otherwise $\ell(e)=1$.
We use a parameter $\Lambda=\frac{(n-|B|) \eta \cdot \log^5n}{\Delta}$, and we initialize $\qset=\emptyset$. While the set $B$ of vertices may change over the course of the algorithm, parameter $\Lambda$ remains unchanged.

The execution of the algorithm for the $r$-restricted \stSP problem consists of at most $\Delta$ iterations. In every iteration $i$, the algorithm must do one of the following:

\begin{itemize}
	\item either return FAIL;
	\item or return a modified cut-witness for $G$ with respect to $A$ and $B$ (see \Cref{def: mod cut-witness});
	\item or return a modified acceptable path $P$ for $G$ with respect to $A$ and $B$ (see \Cref{def: mod acceptable path}).
\end{itemize}

In the former two cases, the algorithm terminates. In the last case, let $a\in A$, $b\in B$ denote the endpoints of path $P$. We remove $a$ from $A$ and $b$ from $B$ and add $P$ to $\qset$. Additionally, the lengths of some of the special edges on $P$ will be doubled; we assume that the algorithm receives a list of all such edges. 

For every special edge $e\in E(G)$, the length of $e$ is doubled once $e$ participates in exactly $\eta$ paths that the algorithm returned, since the last time the length of $e$ was doubled, or since the beginning of the algorithm -- whatever happened last.

Lastly, we require that the probability that the algorithm ever returns ``FAIL'' is bounded by $\frac{1}{2}$.
This concludes the definition of the $r$-restricted \stSP problem.

When an algorithm for the problem computes a modified acceptable path $P$, we may say, for convenience, that it responds to a \shortpath query.
We note that the running time of an algorithm for the \stSP problem includes both the time required to maintain its data structures (total update time), and time required to compute and output the modified acceptable paths in every iteration (sometimes referred to as query time). 

The following theorem is immediate from the discussion in Section \ref{subsubsec: modified MWU}, the definition of the $r$-restricted \stSP problem, and Inequality \ref{eq new marked}.

\begin{theorem}\label{thm: oracle to routing} Suppose there is an algorithm for the $r$-restricted \stSP problem, that, given an input $(G,A,B,\Delta,\eta, N)$ with $|V(G)|=n$, has running time at most $n\cdot (n-|B|)\cdot T(n,\Delta,\eta,N)$. Then there is an algorithm for the $r$-restricted \routeandcut problem, that, on input $(G,A,B,\Delta,\eta, N)$ with $|V(G)|=n$, has running time at most $O\left(n\cdot (n-|B|)\cdot T(n,\Delta,\eta,N)\right)+O(n\cdot (n-|B|))$.
\end{theorem}

Lastly, we show a simple but not very efficient algorithm for \stSP, that, combined with \Cref{thm: oracle to routing}, will be used as the induction base in the proof of \Cref{thm: main for i-restricted route and cut}.

%% file: induction-base.tex
\subsubsection{A Simple Algorithm for \stSP}
\label{subsubsec: induction base}

In this subsection we provide the following simple algorithm for the $r$-restricted \stSP problem.

\begin{theorem}\label{thm: simple stSP}
	For all $r\geq 1$, there is a deterministic algorithm for the $r$-restricted \stSP problem, that,  on input $(G,A,B,\Delta,\eta,N)$ with $|V(G)|=n$, has running time:
	$O\left(n\cdot(n-|B|)\cdot \frac{(n-|B|)\cdot \eta}{\Delta}\cdot \log^6 N\right )$.
\end{theorem}
\begin{proof}
We define a dynamic graph $H$, and we will eventually apply the algorithm for maintaining the \EST data structure from \Cref{thm: directed weighted ESTree} to graph $H$, with the distance parameter $d=\Lambda$.

We initialize graph $H$ as follows. We let $V(H)=V(G)\cup \set{s,t}$. For every regular edge $e=(u,v)\in E(G)$, we include edge $e$ with length $\tilde \ell(e)=1$ in $H$. For every special edge $e'=(u',v')\in E(G)$, we include $\ceil{\log \Lambda}+1$ parallel edges $e_1,\ldots,e_{\ceil{\log \Lambda}+1}$, where for all $1\leq i\leq \ceil{\log \Lambda}+1$, the length $\tilde \ell(e_i)=2^{i-1}$. Lastly, for every vertex $a\in A$, we add an edge $(s,a)$ of length $1$, and for every vertex $b\in B$, we add an edge $(b,t)$ of length $1$ to $H$. 
Note that graph $H$ can be initialized in time $O(|E(G)|\cdot \log n)\leq O(n\cdot(n-|B|)\cdot \log n)$.
We then initialize the algorithm for maintaining the \EST data structure from \Cref{thm: directed weighted ESTree} on graph $H$, with the distance threshold $d=\Lambda$.

Our algorithm consists of at most $\Delta$ iterations. In every iteration, we check whether $\dist_H(s,t)\leq \Lambda$ holds, in time $O(1)$, using the \EST data structure. Assume first that this is the case. Then we use the \EST data structure to compute the shortest $s$-$t$ path $P$ in the current graph $H$, of length at most $\Lambda$. From our construction, for every vertex $b\in B$, there is an edge $(b,t)$ of length $1$ in $H$. This ensures that path $P$ may not contain vertices of $B$ as inner vertices. Let $P'$ be the path obtained from $P$ by deleting its first and last vertex. It is then easy to verify that $P'$ is a modified acceptable path, that we return as an output of the current iteration. Let $a\in A$, $b\in B$ be the endpoints of path $P'$. Recall that following the current iteration, $a$ is deleted from $A$ and $b$ is deleted from $B$. Additionally, the lengths of some special edges of $E(P')$ may be doubled.
We modify graph $H$, in order to reflect these changes as follows. We delete the edges $(s,a)$ and $(b,t)$ from $H$. Additionally, for every special edge $e'\in E(P)$ whose length was doubled in $G$, we delete the shortest-length copy of $e'$ from $H$. 

Assume now that, at the beginning of some iteration, we discover that $\dist_H(s,t)>\Lambda$ holds. Then, in the current graph $H$, the length of the shortest path connecting a vertex of $A$ to a vertex of $B$ is at least $\Lambda-2\geq \frac{\Lambda}{2}$. We now define the lengths $\ell'(e)$ of edges $e\in E(G)$ as follows. If $e$ is a regular edge that lies in $H$, then we set $\ell'(e)=0$. If $e=(x,y)$ is a special edge,  then we let $\ell'(e)$ be the length of the shortest parallel edge $(x,y)$ in $H$. 

Let $A'$ and $B'$ denote the sets $A$ and $B$ of vertices at the end of the algorithm, respectively.
Clearly, $\sum_{e\in E(G)}\ell'(e)\leq \sum_{e\in E(G)}\ell(e)$ holds. Next, we claim that, for every path $P$ connecting a vertex $a\in A'$ to a vertex $b\in B'$ in $G$, the length $\sum_{e\in E(P)}\ell'(e)\geq \frac{\Lambda}{32}$. Indeed, recall that we have established that the length of such a path $P$ in graph $H$ is at least $\frac{\Lambda}{2}$. The only difference between the edge lengths $\set{\tilde \ell(e)}_{e\in E(P)}$ and $\set{\ell'(e)}_{e\in E(P)}$ is that the length $\tilde \ell(e)$ of every regular edge $e\in E(H)$ is $1$, and $\ell'(e)=0$. However, since no regular edge may connect a vertex of $A$ to a vertex of $B$ in $G$, regular and special edges must alternate on path $P$, and, since the length of every special edge in $H$ is at least $1$, the total length $\tilde \ell(e)$ of all special edges $e\in E(P)$ is at least $\frac{\sum_{e'\in E(P)}\hat \ell(e')}{4}\geq \frac{\Lambda}{32}$. 

Let $E^*$ be the set of all special edges with both endpoints in $B'$. In order to turn the current assignment $\set{\ell'(e)}_{e\in E(G)}$  of edge lengths to a valid modified cut-witness, we set the length $\ell'(e)$ of every edge $e\in E^*$ to $0$. It is easy to verify that, for every path $P$ connecting a vertex of $A$ to a vertex of $B$ in $G$, its length remains at least $\frac{\Lambda}{32}$, and  $\sum_{e\in E(G)}\ell'(e)\leq \sum_{e\in E(G)\setminus E^*}\ell(e)$ now holds. Therefore, we obtain a valid modified cut-witness for $G$.

It now remains to analyze the running time of the algorithm. We start by analyzing the \emph{total update time} -- the time required to maintain all data structures, which excludes the \emph{query time} -- the time required to compute the modified acceptable paths in every iteration.

Let $m$ denote the total number of edges that are ever present in $H$, and recall that $m\leq O(n\cdot (n-|B|)\cdot \log N)$. Once graph $H$ is initialized, the additional update of the algorithm is bounded by $O(m\cdot \Lambda)$. Therefore, the total update time that is required in order to maintain the data structures (excluding the time required to return the paths in every iteration) is bounded by:

\[  O(n\cdot (n-|B|)\cdot \Lambda\cdot \log N)\leq O\left(n\cdot(n-|B|)\cdot \frac{(n-|B|)\cdot \eta}{\Delta}\cdot \log^6 N\right ).\]

Assume now that the number of iterations in which the algorithm returned a modified acceptable path is $q$, and let $P_1,\ldots,P_q$ denote the paths that the algorithm returned. Then the time that the algorithm spent on returning these paths is bounded by $\sum_{i=1}^qO(|E(P_i)|)$. Recall that no regular edge connects a vertex of $A$ to a vertex of $B$ in $G$, and so regular and special edges alternate on each such path $P_i$. Therefore, for all $1\leq i\leq q$, $|E(P_i)|$ is asymptotically bounded by the number of special edges on $P_i$.

Consider now some special edge $e$ of $G$. When $e$ is first added to $G$, we include at most $O(\log n)$ copies of $e$ in $H$. Every time edge $e$ participates in $\eta$ of the paths that the algorithm returns, one of the copies of $e$ is deleted from $H$. Therefore, copies of a special edge $e$ of $G$ may participate  in at most $O(\eta\log n)$ paths that the algorithm returns, and, overall, $\sum_{i=1}^q|E(P_i)|\leq O(n\eta\log n)$ (since there are at most $n$ special edges in $G$). We conclude that the total running time of the algorithm is bounded by:

\[O\left(n\cdot(n-|B|)\cdot \frac{(n-|B|)\cdot \eta}{\Delta}\cdot \log^6 N\right )+O(n\eta\log N)\leq O\left(n\cdot(n-|B|)\cdot \frac{(n-|B|)\cdot \eta}{\Delta}\cdot \log^6 N\right ),\]

since $\eta\leq \Delta\leq n-|B|$ from the problem definition.
\end{proof}

\iffalse
\paragraph{Remark.}
The proof of \Cref{thm: simple stSP} demostrates why we chose to implement the oracle on graph $\revG$ instead of $G$: our algorithm essentially repeatedy consructs $s$-$t$ paths in graph $H$, as it undergoes edge-deletions (and a limited number of edge insertions). Specifically, we need to be able to insert edges whose both endpoints lie in $B$ into graph $H$. While the role of the vertices $s$ and $t$ in the problem that we are trying to solve is symmertic, the \EST data structure does not treat these two vertices in a symmetric manner, in that it maintains a shortest-path tree that is rooted at $s$. If, for example, we implemented the oracle on graph $G$ instead of $\revG$, then, in graph $H$, we would need to connect $s$ to every vertex $a\in A$, and to connect every vertex $b\in B$ to $t$. Since we maintain an \EST that is rooted at the vertex $s$, it is now not clear how to implement the insertion of edges with both endpoints in $B$ into this data structure efficiently. 
\fi

By combining \Cref{thm: simple stSP} with \Cref{thm: oracle to routing},
and recalling that, in the $r$-restricted \routeandcut problem, $\frac{(n-|B|)\cdot \eta}{\Delta}\leq 2\cdot 2^{r\cdot \sqrt{\log N}}$ holds, we obtain the following immediate corollary.

\begin{corollary}\label{cor: ind base}
	There is a deterministic algorithm for the $r$-restricted \routeandcut problem, that, on input $(G,A,B,\Delta,\eta, N)$ with $|V(G)|=n$, has running time at most $O\left(n\cdot(n-|B|)\cdot 2^{r\cdot \sqrt{\log N}}\cdot \log^6 N\right )$.
\end{corollary}

We will use this corollary in our inductive proof of \Cref{thm: main for i-restricted route and cut}, for the base case when $r=1$, and also for general values of $r$, when the value of the parameter $N$ is sufficiently small.

%% file: expander-tools.tex
\section{Tools for Expander-Like Graphs and the \maintaincluster Problem}
\label{sec: expander tools}

In this section we introduce some tools for expander-like graphs, and define the second main problem that serves as a subroutine in our algorithm for \MBM -- the \maintaincluster problem. We start by definining subdivided graphs, and providing a new definition of cut sparsity, that extends the standard definition to graphs with edge lengths in a natural way.
\input{subdivided-sparsity.tex}

\input{maintaincluster.tex}
\input{maintainspeccluster.tex}

\input{induction-step.tex}

%% file: subdivided-sparsity.tex
\subsection{A Subdivided Graph and a New Definition of Cut Sparsity}

We first recall a standard definition of cut sparsity in directed graphs. Let $G$ be an arbitrary directed graph, and assume that the graph is unweighted, so all edge lengths and capacities are unit. Under the standard definition, the sparsity of the cut $(A,B)$ in $G$ is defined as $\frac{|E_G(A,B)|}{\min\set{|A|,|B|}}$.
	
Suppose now that we are given a well-structured graph $G=(L,R,E)$, with a proper {\bf uniform} assignment of lengths to its edges, so the lengths of all special edges in $G$ are unit. In this case, for a well-structured cut $(A,B)$ in $G$, it would be natural to define its sparsity as $\frac{|E_G(A,B)|}{\min\set{|A|,|B|}}$, to be consistent with the standard definition of sparsity.

Assume now that we are given a well-structured graph $G=(L,R,E)$, with a proper assignment of lengths to its edges, but now the lengths of the special edges may no longer be unit. Several of the expander-based tools that we define below work best when the underlying graph is unweighted; in the context of well-structured graphs, this corresponds to having a proper uniform assignment of edge lengths. A natural way to handle graphs with arbitrary  lengths of special edges in this setting is to suitably \emph{subdivide} its edges, so that the lengths of the special edges become unit, but the lengths of all paths are preserved. We now define a subdivided graph of a well-structured graph, a notion that we will use later, that will also motivate our modified definition of cut sparsity.

\begin{definition}[Subdivided Graph]\label{def: subdivided graph}
	Let $G=(L,R,E)$ be a perfect well-structured graph with a proper assignment of lengths $\ell(e)$ to its edges $e\in E$, and let $G'=(L',R',E')$ be another graph. We say that $G'$ is the \emph{subdivided graph of $G$}, and denote $G'=G^+$, if $G'$ can be obtained from $G$ via the following process. We start with $G'=G$, $L'=L,R'=R$ and $E'=E$, and then consider every special edge $e\in E$ one by one. Let $e=(u,v)$ be any such edge, and denote its length $\ell(e)$ by $q$. We replace edge $e$ with a path $P(e)=(u=x_1,y_1,x_2,y_2,\ldots,x_q,y_q=v)$. Note that $u\in R'$ and $v\in L'$ must hold. We add vertices $x_2,\ldots,x_q$ to $R'$, and vertices $y_1,\ldots,y_{q-1}$ to $L'$. Edges $(x_1,y_1),\ldots,(x_q,y_q)$ become special edges, and edges $(y_1,x_2),\ldots,(y_{q-1},x_q)$ become regular edges  of $G'$. The lengths of the resulting special edges are set to $1$, while the lengths of the resulting regular edges are set to $0$, so the length of the path $P(e)$ is $q=\ell(e)$.
\end{definition}

Note that, if $G$ is a perfect well-structured graph with proper edge lengths, then its subdivided graph $G^+$ is a perfect well-structured graph with proper uniform edge lengths. Therefore, we can use the definition of cut sparsity for well-structured graphs with proper uniform edge lengths from above in graph $G^+$. Every cut $(A',B')$ in graph $G^+$ naturally also defines a cut $(A,B)$ in $G$, where we set $A=V(G)\cap A'$ and $B=V(G)\cap B'$. We will define the sparsity  of cuts in perfect well-structured graphs with proper edge lengths in a natural way, so that, if a well-structured cut $(A',B')$ in $G^+$ is $\phi$-sparse according to the standard definition, then the sparsity of the corresponding cut $(A,B)$ in $G$ is at most $O(\phi)$. Notice that, since graph $G^+$ is obtained from $G$ by subdividing some of its edges, every edge of $G$ may contribute a number of vertices to $A'$ and/or to $B'$. We will assign a \emph{weight} to every vertex of $G$ in order to reflect this contribution. We use the following definition of vertex weights throughout the paper.

\begin{definition}[Vertex Weights]\label{def: vertex weights}
	Let $G=(L,R,E)$ be a perfect well-structured graph with a proper assignment of lengths to its edges. For every vertex $v\in L\cup R$, the \emph{weight} $w(v)$ is defined to be the length of the unique special edge that is incident to $v$. For a subset $S$ of vertices of $G$, we define $w(S)=\sum_{v\in S}w(v)$, and for a subgraph $G'$ of $G$ we define $w(G')=w(V(G'))$.
\end{definition}

We are now ready to define the sparsity of cuts in perfect well-structured graphs.

\begin{definition}[Cut Sparsity]\label{def: cut sparsity}
	Let $G=(L,R,E)$ be a perfect well-structured graph with a  proper assignment of lengths on its edges, and let $(A,B)$ be a well-structured cut in $G$. The \emph{sparsity} of cut $(A,B)$ is:
	
	\[\Phi_G(A,B)=\frac{|E_G(A,B)|}{\min\set{w(A),w(B)}}.\]
\end{definition}

%In the above definition, the cut $(A,B)$ may be either weakly or strongly well-linked.
We will use the following simple observation whose proof can be found in Section \ref{subsec: proof of subdivided cut sparsity} of the Appendix.

\begin{observation}\label{obs: subdivided cut sparsity}
		Let $G=(L,R,E)$ be a perfect well-structured graph with proper lengths $\ell(e)$ on its edges $e\in E$, and let $G^+$ be its subdivided graph. Assume further that we are given a parameter $0<\phi<1$, and that the length of every edge in $G$ is at most $\frac{1}{4\phi}$. Let $(A',B')$ be a well-structured cut in $G^+$ of sparsity at most $\phi$, and let $(A,B)$ be the corresponding cut in $G$, where $A=A'\cap V(G)$ and $B=B'\cap V(G)$. Then $(A,B)$ is a well-structured cut in $G$, $w(A)\geq |A'|/8$, $w(B)\geq |B'|/8$, and $\Phi_G(A,B)\leq 8\phi$.
\end{observation}

The next observation allows us to transform a weakly well-structured cut into a strongly well-structured cut in a perfect well-structured graph, while roughly preserving its sparsity. The proof appears in Section \ref{subsec: proof of weakly to strongly well str} of the Appendix.

\begin{observation}\label{obs: weakly to strongly well str}
	Let $G=(L,R,E)$ be a perfect well-structured graph with a proper assignment of lengths $\ell(e)$ to its edges $e\in E$, and let $0<\phi\leq 1$ be a parameter. Let $W=\sum_{e\in E}\ell(e)$, and assume that, for all $e\in E$, $\ell(e)\leq \frac{W}{2}$. Let $(A,B)$ be a $\frac{\phi}2$-sparse weakly well-structured cut in $G$. Then there is a $\phi$-sparse strongly well-structured cut $(A',B')$ in $G$ with $w(A')\geq w(A)/4$ and $w(B')\geq w(B)/4$. Moreover, there is a deterministic algorithm, that, given $G$, $\phi$, and cut $(A,B)$ in $G$, computes cut $(A',B')$ in $G$ with the above properties in time $O(\min\set{\vol_G(A),\vol_G(B)})$.
\end{observation}

%% file: maintaincluster.tex
\subsection{The \maintaincluster Problem}
We abstract the central technical tool that we use in order to maintain expander-like graphs by the \maintaincluster problem, that we define next. 

As an input to the \maintaincluster problem, we are given a {\bf perfect} well-structured graph $G$ with a proper assignment of lengths $\ell(e)$ to its edges $e\in E(G)$, that subsequently undergoes some updates, though $G$ remains a perfect well-structured graph throughout the algorithm. 
Recall that we have defined vertex weights for perfect well-structured graphs, where the weight of a vertex $v$ is the length of the unique special edge incident to $v$. For every vertex $v\in V(G)$, we denote by $w^0(v)$ its initial weight in the initial input graph $G$, and we denote $W^0(G)=\sum_{v\in V(G)}w^0(v)$. 
We are now ready to define the \maintaincluster problem.

\begin{definition}\label{def: maintaincluster problem}
The input to the \maintaincluster problem consists of a perfect well-structured graph $G=(L,R,E)$ with $|V(G)|=n$, together with a proper assignment of lengths $\ell(e)$ to its edges $e\in E$, and parameters $N\geq W^0(G)$, $1\leq \Delta\leq N$, $1\leq \eta \leq \Delta$, and $d\geq (\log N)^{64}$, such that $\frac{\Delta\cdot d}{\eta}\leq n$, and $N$ is greater than a sufficiently large constant, so that, for example, $N>2^{2^{c_2}}$ holds, where $c_2$ is the constant from \Cref{thm: main for i-restricted route and cut}. 

The algorithm consists of at most $\Delta$ iterations.
At the beginning of every iteration $i$, the algorithm is given two vertices $x,y\in V(G)$, and it must return a simple path $P_i$ in $G$ connecting $x$ to $y$ of length at most $d$. (We may sometimes say that the algorithm is given a \shortpath query between the pair $(x,y)$ of vertices, and it returns the path $P_i$ in response to the query). After that the length of some special edges on path $P_i$ may be doubled. We are guaranteed that the following property holds:

\begin{properties}{P}
	\item If $\tau'>\tau$ are two times at which the length of some edge $e\in E(G)$ is doubled, then the algorithm returned at least $\eta$ paths in response to queries during the time interval $(\tau,\tau']$ that contained $e$.
	\label{maintaincluster-prop: edge congestion}
\end{properties}

The algorithm may, at any time, produce a strongly well-structured cut $(X,Y)$ in $G$ of sparsity $\Phi_G(X,Y)\leq \frac{(\log N)^{64}}{d}$.
If $w(X)\leq w(Y)$, then let $Z=X$ and $Z'=Y$; otherwise, let $Z=Y$ and $Z'=X$. We also let $J'\subseteq Z'$ denote the set of vertices that serve as endpoints of the edges of $E_G(X,Y)$. The vertices of $Z\cup J'$ are then deleted from $G$, and the algorithm continues with the resulting graph $G=G[Z'\setminus J']$. Note that we are guaranteed that $G$ remains a perfect well-structured graph  after this modification. Once $|V(G)|\leq n/2$ holds, the algorithm terminates, even if fewer than $\Delta$ iterations have passed.

The algorithm may also, at any time, return ``FAIL''. The probability that the algorithm ever returns ``FAIL'' must be bounded by $1/N^4$. 

For an integer $1\leq r\leq \ceil{\sqrt{\log N}}$, we say that an instance $(G,\Delta,\eta,d,N)$ of the \maintaincluster problem is $r$-restricted, if $d\leq 2^{r\cdot\sqrt{\log N}}$.
\end{definition}

We prove the following theorem that provides an algorithm for the \maintaincluster problem. Though we do not use the theorem directly, it easily follows from our results and we believe that it is interesting in its own right.

\begin{theorem}\label{thm: alg for maintain cluster}
	There is a randomized algorithm for the \maintaincluster problem, that, on an instance $(G,\Delta,\eta,d,N)$ of the problem, has running time at most $O\left(W^0(G)\cdot (W^0(G)+\Delta)\cdot 2^{O(\sqrt{\log N}\cdot \log\log N)}\right)$.
\end{theorem}

%% file: maintainspeccluster.tex
\subsection{The \maintainspeccluster Problem}

We will sometimes need to use a slight variation of the \maintaincluster problem, that we refer to as the \maintainspeccluster problem. In this problem, as before, we are given  a  perfect well-structured graph $G$ with proper lengths $\ell(e)$ on its edges $e\in E(G)$, that subsequently undergoes some updates, though $G$ remains a perfect well-structured graph throughout the algorithm. 
Graph $G$ is given in the adjacency-list representation, and initially all special edges have length $1$.
For every vertex $v\in V(G)$, we define its initial weight $w^0(v)$ exactly as before, and we denote $W^0(G)=\sum_{v\in V(G)}w^0(v)$ as before. 
In addition to the perfect well-structured graph $G$, the algorithm is given as input a subset $\beta$ of vertices of $G$. Over the course of the algorithm, some vertices may be deleted from the set $\beta$, but we will ensure that, throughout the algorithm, every special edge incident to a vertex in $\beta$ has length $1$. For intuition, the vast majority of the vertices of the input graph $G$ lie in $\beta$, and this property is preserved throughout the entire algorithm. Unlike the \maintaincluster problem, in a \shortpath query, we are now given a single vertex $x\in V(G)$, and we need to return a path $P$ in the current graph $G$ that connects $x$ to any vertex of $\beta$, such that no inner vertices of $P$ lie in $\beta$. The last endpoint of the path $P$ is then deleted from $\beta$, and the lengths of some edges of the path $P$ may be doubled, like in the \maintaincluster problem. Like in the \maintaincluster problem, the algorithm may, at any time, return a sparse strongly well-structured cut $(X,Y)$ in $G$, but now we require that $w(X)\geq 1.8w(\beta \cap X) = 1.8|\beta\cap X|$, that is, much of the weight of the vertices in $X$ must come from vertices that do not lie in $\beta$. The vertices of $X$, and all endpoints of the edges in $E_G(X,Y)$ are then deleted from $G$, and the algorithm continues. 
We are now ready to define the \maintainspeccluster problem formally.

\begin{definition}\label{def: maintainspeccluster problem}
The input to the \maintainspeccluster problem consists of a perfect well-structured graph $G=(L,R,E)$ with a proper assignment of lengths $\ell(e)$ to its edges $e\in E$, that are given in the adjacency-list representation, such that the length of every special edge is $1$. Additionally, we are given a subset $\beta\subseteq V(G)$ of  vertices with $|\beta|\geq \frac{99|V(G)|} {100}$, and parameters $N\geq W^0(G)$, $1\leq \Delta\leq \frac{|\beta|}{(\log N)^{128}}$, $1\leq \eta \leq \Delta$, and $d\geq (\log N)^{64}$, such that $\frac{\Delta\cdot d}{\eta}\leq |V(G)|-|\beta|+\Delta\cdot (\log N)^{64}$, and $N$ is greater than a sufficiently large constant. Over the course of the algorithm, vertices may be deleted from set $\beta$, and the lengths of some special edges may be increased, but the length of every special edge incident to a vertex that currently lies in $\beta$ is always $1$.

The algorithm consists of at most $\Delta$ iterations.
At the beginning of every iteration $i$, the algorithm is given a vertex $x\in V(G)$, and it must return a simple path $P_i$ in $G$ connecting $x$ to any vertex  $y\in \beta$, such that the length of the path is at most $d$, and no inner vertices on the path lie in $\beta$ (if $x\in \beta$, that we require that path $P_i$ only consists of the vertex $x$). After that, the length of some special edges on path $P_i$ may be doubled, and vertex $y$ is deleted from $\beta$. We are guaranteed that Property \ref{maintaincluster-prop: edge congestion} from the definition of the \maintaincluster problem holds.

The algorithm may, at any time, produce a strongly well-structured cut $(X,Y)$ in $G$ of sparsity $\Phi_G(X,Y)\leq \frac{(\log N)^{64}}{d}$, such that $w(X)\geq 1.8|\beta\cap X|$ holds. Let $J'\subseteq Y$ denote the set of vertices that serve as endpoints of the edges of $E_G(X,Y)$. The vertices of $X\cup J'$ are then deleted from $G$, and the algorithm continues with the resulting graph $G=G[Y\setminus J']$. Note that we are guaranteed that $G$ remains a perfect well-structured graph after this modification.

The algorithm may also, at any time, return ``FAIL''. The probability that the algorithm ever returns ``FAIL'' must be bounded by $1/4$. Lastly, if the number of vertices in $\beta$ decreases by a factor $2$, the algorithm terminates.

For an integer $1\leq r\leq \ceil{\sqrt{\log N}}$, we say that an instance $(G,\beta,\Delta,\eta,d,N)$ of the \maintainspeccluster problem is $r$-restricted, if $d\leq 2^{r\cdot\sqrt{\log N}}$.
\end{definition}

We prove the following theorem that provides an algorithm for the \maintainspeccluster problem. %Though we do not use the theorem directly, it easily follows from our results and we believe that it is interesting in its own right.

\begin{theorem}\label{thm: alg for maintainspec cluster}
	There is a randomized algorithm for the \maintainspeccluster problem, that, on an instance $(G,\beta,\Delta,\eta,d,N)$ of the problem, has running time at most: 
	
	\[O\left(W^0(G)\cdot(W^0(G)-|\beta|+\Delta)\cdot 2^{O(\sqrt{\log N}\cdot \log\log N)}\right).\] 
%	
%	where $\beta'$ is the set $\beta$ at the end of the algorithm.
\end{theorem}

In order to prove Theorems \ref{thm: alg for maintain cluster} and \ref{thm: alg for maintainspec cluster}, we use induction. Specifically, we will consider $r$-restricted instances of both problems problem, from smaller to larger values of $r$. The following theorem summarizes our main result for the $r$-restricted \maintaincluster and \maintainspeccluster problems. It uses the large enough constants $c_1\gg c_2$ that also appear in the statment of \Cref{thm: main for i-restricted route and cut}.

\begin{theorem}\label{thm: main for r-restricted maintaincluster}
	For all $r\geq 1$, there is a randomized algorithm for the $r$-restricted \maintaincluster problem, that, on an instance $(G,\Delta,\eta,d,N)$ of the problem, has running time at most: 
	
	\[c_1\cdot (W^0(G))^2\cdot 2^{c_2\sqrt{\log N}}\cdot (\log N)^{16c_2r}+c_1W^0(G)\cdot \Delta\cdot \log^4 N,\]
	
	 and there is a randomized algorithm for the the $r$-restricted \maintainspeccluster problem, that, on an instance $(G,\beta,\Delta,\eta,d,N)$ of the problem, has running time at most: 
	
	\[c_1\cdot W^0(G)\cdot (W^0(G)-|\beta|+\Delta)\cdot 2^{c_2\sqrt{\log N}}\cdot (\log N)^{16c_2r}.\]
%	
%	where $\beta'$ is the set $\beta$ of vertices at the end of the algorithm.
\end{theorem}

Note that the proof of \Cref{thm: alg for maintain cluster} easily follows from \Cref{thm: main for r-restricted maintaincluster}. Indeed, consider any instance  $(G,\Delta,\eta,d,N)$ of the \maintaincluster problem. From the definition of the problem, $d\leq \frac{n\eta}{\Delta}\leq n\leq N$ must hold, and so in particular $d\leq 2^{r\cdot \sqrt{\log N}}$, for $r=\ceil{\sqrt{\log N}}$. Therefore, any instance of the problem is $r$-restricted for $r=\ceil{\sqrt{\log N}}$. By applying \Cref{thm: main for r-restricted maintaincluster} with $r=\ceil{\sqrt{\log N}}$, we obtain an algorithm for the given instance of the \maintaincluster problem, whose running time is bounded by: 

\[\begin{split}
&O\left( W^0(G)\cdot (W^0(G)+\Delta)\cdot 2^{O(\sqrt{\log N})}\cdot (\log N)^{O(\sqrt{\log N})}\right )\\
&\quad\quad\quad\quad\leq O\left(W^0(G)\cdot(W^0(G)+\Delta)\cdot 2^{O(\sqrt{\log N}\cdot \log \log N)}\right ),
\end{split}\] 

as required.

Similarly, the proof of \Cref{thm: alg for maintainspec cluster} easily follows from \Cref{thm: main for r-restricted maintaincluster}. Indeed, consider any instance  $(G,\beta,\Delta,\eta,d,N)$ of the \maintainspeccluster problem, and denote $n=|V(G)|$. As before, from the definition of the problem, $d\leq \frac{n\eta}{\Delta}\leq n\leq N$ must hold, and so in particular $d\leq 2^{r\cdot \sqrt{\log N}}$, for $r=\ceil{\sqrt{\log N}}$. Therefore, any instance of the problem is $r$-restricted for $r=\ceil{\sqrt{\log N}}$. By applying \Cref{thm: main for r-restricted maintaincluster} with $r=\ceil{\sqrt{\log N}}$, we obtain an algorithm for the given instance of the \maintainspeccluster problem, whose running time is bounded by: 

\[\begin{split}
O\left(W^0(G)\cdot (W^0(G)-|\beta|+\Delta)\cdot 2^{O(\sqrt{\log N})}\cdot (\log N)^{O(\sqrt{\log N})}\right )&\\
\quad\quad\quad\quad\quad\quad\quad\quad\quad\quad\quad\quad\quad\quad\quad\quad\leq O\left(W^0(G)\cdot (W^0(G)-|\beta|+\Delta)\cdot 2^{O(\sqrt{\log N}\cdot \log \log N)}\right ).
\end{split}\] 

%where $\beta'$ is the set $\beta$ at the end of the algorithm.

It now remains to complete the proofs of Theorems \ref{thm: main for i-restricted route and cut} and \ref{thm: main for r-restricted maintaincluster}.

%% file: induction-step.tex
\subsection{Completing the Proofs of \Cref{thm: main for i-restricted route and cut} and \Cref{thm: main for r-restricted maintaincluster}}

We prove the two theorems together, by induction on $r$. Specifically, we first show that, given an efficient algorithms for the $r$-restricted \maintaincluster and \maintainspeccluster problems, we can obtain an efficient algorithm for the $(r+1)$-restricted \routeandcut problem. This result is summarized in the following theorem, whose proof is deferred to \Cref{sec: SSSP alg}.

\begin{theorem}\label{thm: from maintaincluster to roundandcut}
	Suppose that, for some parameter $r\geq 1$,
	there exists a randomized algorithm for the $r$-restricted \maintaincluster problem,  that, given as input  an instance
	$(G,\Delta,\eta,d,N)$ of the problem, has running time at most: 
	$c_1\cdot(W^0(G))^2\cdot 2^{c_2\sqrt{\log N}}\cdot (\log N)^{16c_2r}+c_1\cdot W^0(G)\cdot \Delta\cdot \log^4 N$.
	Assume further that there also exists  a randomized algorithm for the $r$-restricted \maintainspeccluster problem,  that, on  an instance
	$(G,\beta,\Delta,\eta,d,N)$ of the problem, has running time at most: 
	$c_1\cdot W^0(G)\cdot (W^0(G)-|\beta|+\Delta)\cdot 2^{c_2\sqrt{\log N}}\cdot (\log N)^{16c_2r}$. %, where $\beta'$ is the set $\beta$ of vertices at the end of the algorithm.
	Then there is a randomized algorithm for the $(r+1)$-restricted \routeandcut problem, that, given as input
	an instance	$(G,A,B,\Delta,\eta,N)$ of the problem with $|V(G)|=n$,
	has running time at most $c_1\cdot n\cdot (n-|B|)\cdot  2^{c_2\sqrt{\log N}}\cdot (\log N)^{16c_2r+8c_2}$. 
\end{theorem}

Next, we show that an efficient algorithm for the $r$-restricted \routeandcut problem implies an efficient algorithm for the $r$-restricted \maintaincluster and \maintainspeccluster problems, in the following theorem, whose proof is deferred to Sections \ref{sec: layered connectivity DS} and \ref{sec: alg for maintaincluster from routeandcut}.

\begin{theorem}\label{thm: from routeandcut to maintaincluster}
	Suppose that, for some parameter $r\geq 1$, there exists a randomized algorithm for the $r$-restricted \routeandcut problem, that, given as an input an instance $(G,A,B,\Delta,\eta,N)$ of the problem with $|V(G)|=n$,
	has running time at most $c_1\cdot n\cdot (n-|B|)\cdot 2^{c_2\sqrt{\log N}}\cdot (\log N)^{16c_2(r-1)+8c_2}$. Then there exists a randomized algorithm for the $r$-restricted \maintaincluster problem, that, on input $(G,\Delta,\eta,d,N)$ has running time at most: $c_1\cdot  (W^0(G))^2\cdot 2^{c_2\sqrt{\log N}}\cdot (\log N)^{16c_2r}+c_1\cdot W^0(G)\cdot \Delta\cdot \log^4 N$, and there exists a randomized algorithm for the $r$-restricted \maintainspeccluster problem, that, on input $(G,\beta,\Delta,\eta,d,N)$ has running time at most: $c_1\cdot W^0(G)\cdot (W^0(G)-|\beta|+\Delta)\cdot 2^{c_2\sqrt{\log N}}\cdot (\log N)^{16c_2r}$. 
\end{theorem}

The majority of the remainder of this paper is dedicated to the proofs of 
\Cref{thm: from maintaincluster to roundandcut} and \Cref{thm: from routeandcut to maintaincluster}. The proofs of \Cref{thm: main for i-restricted route and cut} and \Cref{thm: main for r-restricted maintaincluster} easily follow from these theorems, as we show next.

 Recall that, from \Cref{cor: ind base}, there is an algorithm for the $1$-restricted \routeandcut problem, that, given as input an instance $(G,A,B,\Delta,\eta,N)$ of the problem with $|V(G)|=n$, has running time at most $O\left(n\cdot(n-|B|)\cdot 2^{r\cdot \sqrt{\log N}}\cdot \log^6 N\right )\leq c_1\cdot n\cdot (n-|B|)\cdot 2^{\sqrt{\log N}}\cdot (\log N)^{8c_2}$, since we can assume that $c_1$ is a large enough constant. This proves the induction base for the \routeandcut problem. From \Cref{thm: from routeandcut to maintaincluster} we then conclude that there exists a randomized algorithm for the $1$-restricted \maintaincluster problem, that, on input $(G,\Delta,\eta,d,N)$ has running time at most: 
$c_1\cdot (W^0(G))^2\cdot 2^{c_2\sqrt{\log N}}\cdot (\log N)^{16c_2}+c_1\cdot W^0(G)\cdot \Delta\cdot \log^4 N$, and that there is the $1$-restricted \maintainspeccluster problem, that, on input $(G,\beta,\Delta,\eta,d,N)$ has running time at most: 
$c_1\cdot W^0(G)\cdot (W^0(G)-|\beta|+\Delta)\cdot 2^{c_2\sqrt{\log N}}\cdot (\log N)^{16c_2}$. 
This proves the induction base for the \maintaincluster and the \maintainspeccluster problems.

For the induction step, consider an integer $r\geq 1$, and assume that both \Cref{thm: main for i-restricted route and cut} and \Cref{thm: main for r-restricted maintaincluster} hold for $r$-restricted instances. Then from  
\Cref{thm: from maintaincluster to roundandcut}, there is a randomized algorithm for the $(r+1)$-restricted \routeandcut problem, that, given as input
an instance	$(G,A,B,\Delta,\eta,N)$ of the problem with $|V(G)|=n$,
has running time at most $c_1\cdot n\cdot (n-|B|)\cdot  2^{c_2\sqrt{\log N}}\cdot (\log N)^{16c_2r+8c_2}$, proving \Cref{thm: main for i-restricted route and cut} for $(r+1)$-restricted instances. By applying 
\Cref{thm: from routeandcut to maintaincluster} with integer $(r+1)$, we then conclude that there exists a randomized algorithm for the $(r+1)$-restricted \maintaincluster problem, that, on input $(G,\Delta,\eta,d,N)$ has running time at most: 
$c_1\cdot (W^0(G))^2\cdot 2^{c_2\sqrt{\log N}}\cdot (\log N)^{16c_2(r+1)}+c_1\cdot W^0(G)\cdot \Delta\cdot \log^4 N$,
and that there exists a randomized algorithm for the $(r+1)$-restricted \maintainspeccluster problem, that, on input $(G,\beta,\Delta,\eta,d,N)$ has running time at most: 
$c_1\cdot W^0(G)\cdot (W^0(G)-|\beta|+\Delta)\cdot 2^{c_2\sqrt{\log N}}\cdot (\log N)^{16c_2(r+1)}$.  
This proves
 proving \Cref{thm: main for r-restricted maintaincluster} for $(r+1)$-restricted instances of \maintaincluster and \maintainspeccluster.
This completes the proofs of Theorems \ref{thm: main for i-restricted route and cut} and \ref{thm: main for r-restricted maintaincluster} from Theorems 
\ref{thm: from maintaincluster to roundandcut} and \ref{thm: from routeandcut to maintaincluster}. In the next three sections we prove 
\Cref{thm: from maintaincluster to roundandcut} and \Cref{thm: from routeandcut to maintaincluster}.

%% file: SSSP-alg.tex
\section{From  $r$-Restricted \maintaincluster and \maintainspeccluster to $(r+1)$-Restricted \routeandcut: Proof of \Cref{thm: from maintaincluster to roundandcut}}
\label{sec: SSSP alg}

The proof of \Cref{thm: from maintaincluster to roundandcut} follows from the following theorem:

\begin{theorem}\label{thm: from maintaincluster to stSP}
	Suppose that, for some parameter $r\geq 1$, there exists a randomized algorithm for the $r$-restricted \maintaincluster problem, that, given as input  an instance
	$(G,\Delta,\eta,d,N)$ of the problem, has running time at most: 
	$c_1\cdot  (W^0(G))^2\cdot 2^{c_2\sqrt{\log N}}\cdot (\log N)^{16c_2r}+c_1\cdot W^0\cdot \Delta\cdot \log^4 N$.
	Suppose further that  there exists a randomized algorithm for the $r$-restricted \maintainspeccluster problem, that, given as input  an instance
	$(G,\beta,\Delta,\eta,d,N)$ of the problem, has running time at most: 
	$c_1\cdot W^0(G)\cdot (W^0(G)-|\beta|+\Delta)\cdot 2^{c_2\sqrt{\log N}}\cdot (\log N)^{16c_2r}$. %, where $\beta'$ is  the set $\beta$ of vertices at the end of the algorithm.  
	Then there is a randomized algorithm for the $(r+1)$-restricted \stSP problem, that, given an instance
	$(G,A,B,\Delta,\eta,N)$ of the problem with $|V(G)|=n$, has running running time at most $c_1\cdot n\cdot (n-|B|)\cdot  2^{c_2\sqrt{\log N}}\cdot (\log N)^{16c_2r+7c_2}$. 
\end{theorem}

Before we provide the proof of \Cref{thm: from maintaincluster to stSP}, we show that the proof of \Cref{thm: from maintaincluster to roundandcut} follows from it.

\noindent
{\bf Proof of \Cref{thm: from maintaincluster to roundandcut}:}
Consider a parameter $r\geq 1$, and assume that there exists  a randomized algorithm for the $r$-restricted \maintaincluster problem, that, on input $(G,\Delta,\eta,d,N)$, has running time at most: 
$c_1\cdot  (W^0(G))^2\cdot 2^{c_2\sqrt{\log N}}\cdot (\log N)^{16c_2r}+c_1\cdot W^0\cdot \Delta\cdot \log^4 N$. 
Assume further that there exists  a randomized algorithm for the $r$-restricted \maintainspeccluster problem, that, on input $(G,\beta,\Delta,\eta,d,N)$, has running time at most: 
$c_1\cdot  W^0(G)\cdot (W^0(G)-|\beta|+\Delta)\cdot 2^{c_2\sqrt{\log N}}\cdot (\log N)^{16c_2r}$. %, where $\beta'$ is  the set $\beta$ of vertices at the end of the algorithm. 
 From \Cref{thm: from maintaincluster to stSP}, there exists  a randomized algorithm for the $(r+1)$-restricted \stSP problem, that, on input $(G,A,B,\Delta,\eta,N)$ with $|V(G)|=n$, has running running time at most $c_1\cdot n\cdot (n-|B|)\cdot  2^{c_2\sqrt{\log N}}\cdot (\log N)^{16c_2r+7c_2}$.  From \Cref{thm: oracle to routing}, there exists a randomized algorithm for the $(r+1)$-restricted \routeandcut problem, that, on input $(G,A,B,\Delta,\eta)$ with $|V(G)|=n$, has running time is at most: 

\[\begin{split}
&O\left(c_1\cdot n\cdot (n-|B|)\cdot  2^{c_2\sqrt{\log N}}\cdot (\log N)^{16c_2r+7c_2}\right )+O(n\cdot(n-|B|))\\
&\quad\quad\quad\quad\quad\quad\quad\quad\quad\quad\quad\quad \leq c_1\cdot n\cdot (n-|B|)\cdot  2^{c_2\sqrt{\log N}}\cdot (\log N)^{16c_2r+8c_2},
\end{split}\]

since $c_2$ is a sufficiently large constant. \hfill\stopproof

From now on we focus on the proof of \Cref{thm: from maintaincluster to stSP}. We denote the randomized algorithms for the $r$-restricted \maintaincluster and \maintainspeccluster problems from the theorem statement by $\aset$ and $\aset'$, respectively.

We denote the input to the $(r+1)$-restricted \stSP problem by $(G,A,B,\Delta,\eta,N)$, where $G=(L,R,E)$ is an $n$-vertex  well-structured graph, $A$ and $B$ are two disjoint subsets of vertices of $G$, and $N\geq n$, $\Delta \geq 1$, and $1\leq \eta \leq \Delta$ are parameters. Recall that no regular edge of $G$ connects a vertex of $A$ to a vertex of $B$, $\Delta\leq n-|B|$, $|E|\leq O(n\cdot (n-|B|))$,  and $\frac{(n-|B|)\cdot \eta}{\Delta}\leq 2\cdot 2^{(r+1)\sqrt{\log N}}$ hold.

Recall that, as the algorithm progresses, we maintain lengths $\ell(e)$ for edges $e\in E(G)$, and, additionally, some vertices may leave the sets $A$ and $B$. Since the number of iterations is bounded by $\Delta$, at most $\Delta$ vertices may be deleted from $A$ and from $B$.
For convenience, we denote by $A^0$ and $B^0$ the initial sets $A$ and $B$, respectively, to distinguish them from the sets that are maintained during the algorithm.

Note that we can assume that $N$ is greater than a sufficiently large constant, so that, for example, $N>2^{2^{c_2}}$ holds. Indeed, if this is not the case, then, from \Cref{thm: simple stSP}, we obtain a determinisic algorithm that solves the input instance of \stSP in time:

\[
\begin{split}
O\left(n\cdot(n-|B^0|)\cdot \frac{(n-|B^0|)\cdot \eta}{\Delta}\cdot \log^6 N\right )&\leq O\left(n\cdot(n-|B^0|)\cdot N\cdot \log^6 N\right )
\\&\le c_1\cdot n\cdot (n-|B^0|)\cdot  2^{c_2\sqrt{\log N}}\cdot (\log N)^{16c_2r+7c_2},
\end{split}\]

since $n\leq N\ll c_1$, and $\eta\leq \Delta$ (recall that, by definition, $c_1\gg c_2$ holds). Therefore, we assume from now on that $N$ is sufficiently large.

Additionally, we can assume that $\frac{(n-|B^0|)\cdot \eta}{\Delta}> 2^{4\sqrt{\log N}}$, since otherwise, from \Cref{thm: simple stSP}, we obtain a determinisic algorithm that solves the input instance of \stSP in time:

\[
\begin{split}
O\left(n\cdot(n-|B^0|)\cdot \frac{(n-|B^0|)\cdot \eta}{\Delta}\cdot \log^6 N\right )&\leq O\left(n\cdot(n-|B^0|)\cdot 2^{4\sqrt{\log N}}\cdot \log^6 N\right )
\\&\le c_1\cdot n\cdot (n-|B^0|)\cdot  2^{c_2\sqrt{\log N}}\cdot (\log N)^{16c_2r+7c_2}.
\end{split}\]

Therefore, we assume from now on that:

\begin{equation}\label{eq: bound on d2}
\frac{(n-|B^0|)\cdot \eta}{\Delta}> 2^{4\sqrt{\log N}}
\end{equation}

Throughout the algorithm, we use the parameter  $\Lambda=(n-|B^0|)\cdot \frac{ \eta \cdot \log^5n}{\Delta}$.
We let $H$ be the graph that is obtained from $G$, by adding a source vertex $s$, that connects to every vertex $a\in A$ with a length-$0$ edge, and a destination vertex $t$, to which every vertex $b\in B$ connects with a length-$0$ edge. If an edge $e\in E(G)$ is a special edge for graph $G$, then $e$ remains a special edge for $H$. All other edges of $H$, including edges incident to $s$ and $t$, are called regular edges.
For every edge $e\in E(G)$, its length in $H$ remains the same as in $G$, so for convenience, we denote the length of every edge $e\in E(H)$ by $\ell(e)$.

At the beginning  of the algorithm, for each edge $e\in E(H)$, its length $\ell(e)$ is set as follows: if $e$ is a regular edge, then $\ell(e)=0$, and if $e$ is a special edge, then $\ell(e)=1$. 
The algorithm proceeds as follows. We start by setting $\qset=\emptyset$, and then perform at most $\Delta$ iterations. In every iteration, the algorithm must do one of the following:

\begin{properties}{R}
	\item either return FAIL;\label{outcome: fail}

	\item or return a simple $s$-$t$ path $P$ of length at most $\Lambda$ in the current graph $H$, such that $P$ contains exactly one vertex of $B$ -- the penultimate vertex on the path. \label{outcome:path}
	
		\item or compute an assignment $\ell'(e)\geq 0$ of lengths to all edges $e\in E(H)$, such that, if we denote by $\tilde E$ the set of all special edges of $G$ with both endpoints in $B$, then $\sum_{e\in E(H)}\ell'(e)\leq \frac{\Lambda\cdot \Delta}{\eta \log^4n}+\sum_{e\in E(H)\setminus \tilde E}\ell(e)$; the length $\ell'(e)$ of every edge $e$ that is incident to $s$ or to $t$ is $0$; and the distance from $s$ to $t$ in $H$, with respect to edge lengths $\ell'(\cdot)$, is at least $\frac{\Lambda}{32}$;  \label{outcome: cut}
\end{properties}

Note that, if the algorithm returns assignments of lengths $\ell'(e)$ to edges $e\in E(H)$ as in Outcome \ref{outcome: cut}, then these assignments of edge lengths define a modified cut-witness for $G$ (see \Cref{def: mod cut-witness}). In this case, the algorithm terminates, and  returns the  lengths $\ell'(e)$ of all edges $e\in E(H\setminus\set{s,t})$. Slightly abusing the notation, for the sake of simplicity, we will call an assignment of lengths $\ell'(e)$ to the edges $e\in E(H)$ as  described in Outcome \ref{outcome: cut} a \emph{cut-witness}. 

If the algorithm returns an $s$-$t$ path $P$ as in Outcome \ref{outcome:path}, then $P\setminus\set{s,t}$ is a modified acceptable path for $G$ (see \Cref{def: mod acceptable path}). To simplify the notation, we say that a path $P$ described in  Outcome \ref{outcome:path} is an \emph{acceptable path}. If the algorithm returns an acceptable path $P$, then we let $P'=P\setminus\set{s,t}$, and we let $a\in A$, $b\in B$ be the endpoints of path $P'$. We remove $a$ from $A$ and $b$ from $B$, add $P'$ to $\qset$, and delete the edges $(s,a)$ and $(b,t)$ from $H$. Additionally, the lengths of some of the special edges on $P$ may be doubled; we assume that the algorithm receives the list of all such edges. 

From the problem definition, for every special edge $e\in E(H)$, the length of $e$ is doubled once $e$ participates in exactly $\eta$ paths that the algorithm returned since the last time the length of $e$ was doubled, or since the beginning of the algorithm -- whatever happened last.

Lastly, in the case of Outcome \ref{outcome: fail}, we terminate the algorithm and return ``FAIL'', but we require that the probability that the algorithm ever returns ``FAIL'' is bounded by $\frac{1}{2}$.

\subsection*{Some Useful Bounds}

We denote by $E^{\spec}$ the set of all special edges in graph $G$, and we denote by $E^*\subseteq E^{\spec}$ the set of all special edges $e=(u,v)$, such that $u,v\in B$ holds at the end of the algorithm. Recall that the length of a special edge $e$ may only be increased once it participates in some path that the algorithm returns, and, if a vertex $x\in B$ lies on any path that the algorithm returns, then $x$ is deleted from $B$. Therefore, every edge $e\in E^*$ has length $1$ throughout the algorithm.
We start by bounding $|E^{\spec}\setminus E^*|$ in the following simple observation.

\begin{observation}\label{obs: bound special nonimportant}
	 $|E^{\spec}\setminus E^*|\leq 2(n-|B^0|)$.
\end{observation}
\begin{proof}
We partition the set $E^{\spec}\setminus E^*$ if edges into two subsets: set $E_1$ contains all edges $e$, such that at least one endpoint of $e$ lies in $V(G)\setminus B^0$, and set $E_2$ contains all remaining edges of $E^{\spec}\setminus E^*$. Since every edge of $E_1$ is incident to a vertex of $V(G)\setminus B^0$, and every vertex of $G$ may serve as an endpoint of at most one special edge, $|E_1|\leq n-|B^0|$ must hold. Consider now some edge $e\in E_2$. From the definition, over the course of the algorithm, at least one endpoint of $e$ was deleted from $B$. Since at most $\Delta$ vertices are deleted from $B$ over the course of the algorithm, and since each such vertex may serve as an endpoint of at most one special edge, we get that $|E_2|\leq \Delta\leq n-|B^0|$, from the definition of the \stSP problem. Altogether, we get that 
$|E^{\spec}\setminus E^*|\leq |E_1|+|E_2|\leq 2(n-|B^0|)$.
\end{proof}

For each edge $e\in E^{\spec}$, let $\ell^*(e)$ be the length of $e$ at the end of the algorithm. In the next simple observation we bound $\sum_{e\in E^{\spec}}\ell^*(e)$.

\begin{observation}\label{obs: bound edge lengths}
	$\sum_{e\in E^{\spec}}\ell^*(e)\leq |E^{\spec}|+2(n-|B^0|)\log^5n$, and 	$\sum_{e\in E^{\spec}\setminus E^*}\ell^*(e)\leq 4(n-|B^0|)\log^5n$.
\end{observation} 
\begin{proof}
Recall that, at the beginning of the algorithm, for every edge $e\in E^{\spec}$, we set $\ell(e)=1$. Let $E'\subseteq E^{\spec}$ be the set of all edges $e\in E^{\spec}$ with $\ell^*(e)>1$. Clearly, $\sum_{e\in E^{\spec}}\ell^*(e)\leq |E^{\spec}|+\sum_{e\in E'}\ell^*(e)$.

Let $q$ be the total number of iterations of the algorithm, so that, at the end of the algorithm, $|\qset|=q$ holds, and let $\qset=\set{P_1,\ldots,P_q}$, where the paths are indexed in the order in which they were added to $\qset$. For all $1\leq i\le q$, let $C_i$ denote the length of path $P_i$ when it was added to $\qset$. Clearly, $C_i\leq \Lambda$, and $\sum_{i=1}^qC_i\leq q\cdot \Lambda\leq \Delta\cdot \Lambda$.

Consider now some edge $e\in E'$, and recall that $\ell(e)$ was doubled at least once during the algorithm. Let $\tau$ be the last time when the length of $e$ was doubled. Then, prior to time $\tau$, there were at least $\eta$ iterations $i$, during which the path $P_i$ that was added to $\qset$ contained $e$, and the length of $e$ during the corresponding iteration was $\frac{\ell^*(e)}{2}$. In other words, edge $e$ contributes at least $\frac{\ell^*(e)\eta}{2}$ to $\sum_{i=1}^qC_i$. Altogether, we get that $\sum_{i=1}^qC_i\geq \sum_{e\in E'}\frac{\ell^*(e)\eta}{2}$. Therefore:

\[\sum_{e\in E'}\ell^*(e)\leq \frac{2\sum_{i=1}^qC_i}{\eta}\leq \frac{2\Delta\Lambda}{\eta}\leq 2(n-|B^0|)\cdot \log^5n,\]

since $\Lambda=(n-|B^0|)\cdot \frac{ \eta \cdot \log^5n}{\Delta}$.
Altogether, we get that:

\[\sum_{e\in E^{\spec}}\ell^*(e)\leq |E^{\spec}|+\sum_{e\in E'}\ell^*(e)\leq |E^{\spec}|+2(n-|B^0|)\log^5n.\]

Lastly, from \Cref{obs: bound special nonimportant}, 
	$|E^{\spec}| -|E^*|\leq 2(n-|B^0|)$, and so:

\[\sum_{e\in E^{\spec}\setminus E^*}\ell^*(e)\leq \sum_{e\in E^{\spec}}\ell^*(e)-|E^*|\leq  |E^{\spec}|+2(n-|B^0|)\log^5n-|E^*|\leq 4(n-|B^0|)\log^5n.\]
\end{proof}

\iffalse
	Recall that, at the beginning of the algorithm, the length $\ell(e)$ of every special edge is set to $1$, and so $\sum_{e\in E^{\spec}}\ell(e)\leq n$ holds at the beginning of the algorithm. From Requirement \ref{oracle-prop: small length growth}, the total increase in $\sum_{e\in E^{\spec}}\ell(e)$ over the course of the process described above is bounded by:
	
	\[\frac{\Delta \cdot \Lambda\cdot \log^{\hat c}n}{\eta}\leq n\cdot (\log n)^{\hat c+5},\]
	
	since $\Lambda=\frac{n\eta \cdot \log^5n}{\Delta}$. Altogether, we get that $\sum_{e\in E^{\spec}}\ell^*(e)\leq n+n\cdot (\log n)^{\hat c+5}\leq 2n\cdot (\log n)^{\hat c+5}$.
\end{proof}
\fi

Throughout the algorithm, we use a parameter $N'=\max\set{N,8n\cdot \log^5n}$. Let $\tau$ be any time in the algorithm's execution, and let $H'$ be any subgraph of $H\setminus\set{s,t}$, that is a perfect well-structured graph. Recall that the weight of every vertex $v\in V(H')$ is the length of the unique special edge incident to $v$. From \Cref{obs: bound edge lengths}, and since $|E^{\spec}|\leq n$, the total weight of all vertices of $H'$ must be bounded by $N'$. Notice also that, since $N\geq n$, we get that $N'\leq 8N\log^5n\leq N\log^6N$. It is then easy to see that: 

\begin{equation}\label{eq: bound on log n'}
\log N'\leq \log N+6\log\log N\leq \left(1+\frac{1}{64c_2r}\right)\log N,
\end{equation}

since $r\leq \ceil{\sqrt{\log N}}$ and $N$ is large enough. Additionally, $\sqrt{\log N'}\leq \sqrt{\log N+6\log\log N}\leq \sqrt{\log N}+\sqrt{6\log\log N}$. Therefore:

\begin{equation}
2^{\sqrt{\log N'}}\leq 2^{\sqrt{\log N}}\cdot \log^6 N\leq 2^{2\sqrt{\log N}}.\label{eq: bound on N'}
\end{equation}

We will use these facts later.

\noindent
{\bf Algorithm for the $(r+1)$-restricted \stSP problem:}
The description of our algorithm for the $(r+1)$-restricted \stSP problem consists of two parts. In the first part, we provide an algorithm for maintaining an \ATO for graph $H$. This algorithm in turn exploits algorithm $\aset$ for the $r$-restricted \maintaincluster problem and Algorithm $\aset'$ for the $r$-restricted \maintainspeccluster problem. In the second part, we use the algorithm from \Cref{thm: almost DAG routing} for the \DLSSSP problem on the corresponding contracted graph.

\input{SSSP-part1}

\input{SSSP-runtime-part1}

\input{SSSP-bound-bad-edges}

\input{SSSP-part2}

\input{SSSP-completing-alg}

%% file: SSSP-part1.tex
\subsection{Part 1 - Maintaining the \ATO}

In this part we describe our algorithm for maintaining the \ATO $(\xset,\rho, U, \gamma)$ for graph $H$. Notice that $|V(H)|=|V(G)|+2=n+2$. 
Throughout, we use a parameter ${d= \frac{(n-|B^0|)\cdot\eta}{4\Delta\cdot 2^{\sqrt{\log N}}}}$.
Notice that:

\begin{equation}\label{eq: two bounds on d}
2^{2\sqrt{\log N}}\leq d\leq  2^{r\sqrt{\log N}-1}.
\end{equation}

Here, the first inequality follows from Inequality \ref{eq: bound on d2}.
The second inequality follows since the input instance of \stSP is $(r+1)$-restricted, so $\frac{(n-|B^0|)\cdot\eta}{\Delta}\leq 2\cdot 2^{(r+1)\sqrt{\log N}}$ holds.

Let $U\subseteq V(H)\setminus\set{s,t}$ be the collection of all vertices $v$, such that some special edge of $E^{\spec}$ is incident to $v$, and let $U'=V(H)\setminus U$. Notice that $H[U]$ is a perfect well-structured graph  (see \Cref{def: well structured graphs}).
We let the initial collection $\xset$ of vertex subsets contain the set $U$, and, for every vertex $v\in U'$, a set $X_v=\set{v}$. For convenience, we denote $S=\set{s}$ and $T=\set{t}$. 
We distinguish between two cases. Case 1 happens when $|U|>2(n-|B^0|)\cdot (\log N)^{128}$. In this case, we let $U$ be the special cluster of the \ATO, and we set $\gamma=|U|-256(n-|B^0|)(\log N)^{64}$. 
%(n-|B^0|)\cdot (\log N)^{128}$. 
Otherwise, Case 2 happens, the special cluster of the \ATO is undefined, and we set $\gamma=1$.

Next, we assign, to every resulting set $X\in \xset$, an interval $I_X\subseteq \set{1,\ldots,|V(H)|-\gamma+1}$, so that, for each regular cluster $X\in \xset$, $|I_X|=|X|$, and, if the special cluster is defined, then $|I_U|=|U|-\gamma+1$. In order to do so, it is sufficient to define an ordering $\sigma$ of the sets in $\xset$. We let $S$ appear first in the ordering, followed by the sets in $\set{X_v\mid v\in U'\cap L}$ in an arbitrary order, followed by set $U$, which is followed by the sets  in $\set{X_v\mid v\in U'\cap R}$ in an arbitrary order, followed by $T$. Since vertices of $U'$ are not incident to special edges, it is easy to verify that  no edge of $H$ is a right-to-left edge with respect to this initial \ATO.
Throughout the algorithm, we refer to each set $X\in \xset$ of vertices as a \emph{cluster}. If $|X|=1$, then we say that $X$ is a \emph{singleton} cluster. 
As the algorithm progresses, cluster $U$ may be partitioned into smaller and smaller clusters, with each such cluster assigned an interval that is contained in the initial interval $I_{U}$, but the clusters in $\set{X_v\mid v\in U'}$, together with their corresponding intervals $I_{X_v}$, remain unchanged. For convenience, we denote the initial set $U$ by $U^0$.
We will ensure that, throughout the algorithm, the following invariants hold:

\begin{properties}{I}

	\item If $X\in \xset$ is a non-singleton cluster, then $H[X]$ is a perfect well-structured graph; \label{inv: perfect well structured} 
	
	\item If the special cluster is defined, then $U\in \xset$ holds throughout the algorithm, and moreover, $|U|\geq \gamma$ always holds; and \label{inv: special cluster}
	
	\item  If $e=(x,y)$ is a right-to-left edge with respect to the current \ATO, and $X,Y\in \xset$ are the clusters that contain $x$ and $y$, respectively, so that $Y$ appears before $X$ in the ordering $\sigma$ of clusters of $\xset$ defined by the \ATO, then $e$ must be a special edge, and moreover, $X=\set{x}$ and $Y=\set{y}$ must hold; in other words, $X$ and $Y$ are both singleton clusters. \label{inv: backward edges}
\end{properties}

It is easy to verify that the invariant holds for the initial \ATO.
%Whenever a new cluster $X$ is added to $\xset$, we set $\beta_X=B\cap X$. As the algorithm progresses, vertices may be deleted from set $B$, and vertices may be deleted from $X$. We will always set $\beta_X=B\cap X$, with respect to the current sets $B$ and $X$. Therefore, the set $\beta_X$ is decremental, so vertices may leave it, but they may not join it. 

For each cluster $X$ that is ever added to $\xset$, we denote by $X^0$ the set $X$ when it was just added to $\xset$. If $|X^0|>1$ (that is, $X$ is not a singleton cluster when it is added to $\xset$), then, from Invariant \ref{inv: perfect well structured}, $H[X]$ is a perfect well-structured graph, so weights of vertices of $X$ are well defined. In such a case, we denote by $W^0(X)$ the  total weight of all vertices of $X$ when $X$ is added to $\xset$. Whenever a non-singleton regular cluster $X$ is added to set $\xset$, we define the parameter ${d_X=\min\set{\frac{|X^0|\cdot \eta}{\Delta},d}}$. We may sometimes refer to $d_X$ as $d_{X^0}$.
Next, we partition all clusters that are ever added to $\xset$ into \emph{leaf} and \emph{non-leaf} clusters.

\begin{definition}[Leaf and Non-Leaf Clusters]\label{def: leaf}
	Let $X$ be a cluster that was just added to $\xset$. We say that $X$ is a \emph{leaf cluster}, if $|X^0| \leq \frac{\Delta\cdot (\log N)^{128}}{\eta}$.
If $X$ is not a leaf cluster, then we call it a \emph{non-leaf} cluster.
\end{definition}

Recall that, if $U$ is a special cluster, then $|U^0|>2(n-|B^0|)\cdot (\log N)^{128}$. Since $\Delta\leq n-|B^0|$ from the definition of the \stSP problem, we get that a special cluster may not be a leaf cluster. 

Recall that in general, once a cluster $X$ is added to the collection $\xset$ of clusters of an \ATO, it must remain in $\xset$ throughout the algorithm, though we may delete some vertices from $X$. A new cluster $Y$ may only be added to the set $\xset$ if it is a subset of some cluster $X$ that currently lies in $\xset$. When we describe our algorithm for maintaining the \ATO, it will be convenient for us to sometimes remove a cluster $X$ from the set $\xset$, and then  immediately reinsert it into $\xset$. This is only done for clarity of exposition and convenience of analysis, and we do not need to do so in order to maintain the \ATO. Additionally, we will sometimes replace a cluster $X\in \xset$ with a new collection of clusters, each of which contains a single vertex of $X$. In such a case we could simply identify $X$ with one of the new clusters and implement this update to $\xset$ via a sequence of cluster splitting operations, though for clarity of exposition we will treat all newly inserted clusters as new clusters, and we will think of $X$ as being removed from $\xset$.

Next, we describe an algorithm for maintaining each cluster $X$ that is added to $\xset$. We employ different algorithms for leaf clusters, non-leaf regular clusters, and the special cluster, if it is defined. We now describe each of the three algorithms in turn.

\subsubsection{Leaf Clusters}

Let $X$ be a leaf cluster that was just added to $\xset$. If $X$ is a singleton cluster, then no further processing of this cluster is needed. Assume now that $X$ is a non-singleton cluster.
We immediately remove $X$ from $\xset$, and instead add, for every vertex $v\in X$, a new singleton cluster $\set{v}$ to $\xset$. We assign, to each such new cluster $\set{v}$, an interval $I_{\set{v}}\subseteq I_{X}$ of length $1$, so that all intervals corresponding to the vertices of $X\cap L$ appear before all intervals corresponding to the vertices of $X\cap R$. Notice that the only new right-to-left edges that are introduced by this transformation are the edges of $E_H(X\cap R,X\cap L)$, which are all special edges. The number of such edges is $|X\cap R|=|X\cap L|\leq |X|$. We say that these new right-to-left edges are \emph{bad edges}, and that they are \emph{owned} by the leaf cluster $X$. We denote the set of these edges by $E^{\bad}(X)$. From the above discussion, $|E^{\bad}(X)|\leq |X|$, and each edge $e\in E^{\bad}(X)$ has $\spann(e)\leq |X|$.

\subsubsection{Non-Leaf Regular Clusters}

Whenever a new non-leaf regular cluster $X$ is added to $\xset$, we  initialize Algorithm $\aset$ for the $r$-restricted \maintaincluster problem on graph $H[X^0]$, with parameters $\Delta$ and $\eta$ that remain unchanged, parameter $d$ replaced with $d_X$, and parameter $N$ replaced with $N'$. We denote this instantiation of Algorithm $\aset$ by $\aset(X)$. We now verify that $(H[X^0],\Delta,\eta,d_X,N')$ is indeed a valid input instance for the $r$-restricted \maintaincluster problem, in the following simple but technical observation, whose proof is deferred to Section \ref{subsec: appx: proof valid instance} of Appendix.

\begin{observation}\label{obs: valid instance}
	$(H[X^0],\Delta,\eta,d_X,N')$ is a valid input to the $r$-restricted \maintaincluster problem.
\end{observation}

Recall that Algorithm $\aset(X)$ may, at any time, 
produce a strongly well-structured cut $(Y,Z)$ in $H[X]$ of sparsity $\Phi_{H[X]}(Y,Z)\leq \frac{(\log N')^{64}}{d_X}$. 
Below, we provide an algorithm for processing such a cut.
But before we do so, we provide the remaining details for processing cluster $X$.

Recall that the algorithm for the \stSP problem consists of at most $\Delta$ iterations. At the beginning of every iteration $i$, we will compute, for every non-leaf regular cluster $X\in \xset$, a pair $x,y\in V(X)$ of vertices (where possibly $x=y$). The corresponding algorithm $\aset(X)$ is then required to compute a simple path $P_X$ connecting $x$ to $y$ in graph $H[X]$, whose length is at most $d_X$. We will then process the resulting paths $P_X$ for all non-leaf clusters $X\in \xset$, in order to obtain a short $s$-$t$ path $P_i$ in $H$, as required. For each non-leaf cluster $X\in \xset$, we will ensure that $E(P_i)\cap E(H[X])\subseteq E(P_X)$. Following that, the lengths of some special edges on path $P_i$ may be doubled; we assume that our algorithm is given the list of all such edges. In particular, for each non-leaf regular cluster $X\in \xset$, we are given a subset of the special edges on path $P_X$, whose lengths must be doubled. 
From the definition of the \stSP problem, we are guaranteed that, for every non-leaf regular cluster $X\in \xset$, for every special edge $e\in E(H[X])$, if  $\tau'>\tau$ are two times at which the length of $e$ is doubled, then the number of paths $P_X$ that Algorithm $\aset(X)$ returned  in response to queries during the time interval $(\tau,\tau']$, which contain $e$, is at least $\eta$, so Property \ref{maintaincluster-prop: edge congestion} holds for the instance of the \maintaincluster problem associated with $X$. If Algorithm $\aset(X)$ ever returns ``FAIL'', then we terminate our algorithm for the \stSP problem, and return ``FAIL'' as well.

Lastly, whenever, for a non-leaf regular cluster $X\in \xset$, $|X|$ falls below $\frac{|X^0|}{2}$, we terminate Algorithm $\aset(X)$, and we say that cluster $X$ is \emph{destroyed}. We then think of the resulting set $X$ of vertices as a new cluster that was just added to $\xset$, and treat it as such (so we think that, once cluster $X$ is destroyed, it is deleted from $\xset$, and instead a new cluster, that is identical to $X$, is inserted into $\xset$).
It now remains to describe the algorithm to process sparse cuts produced by Algorithm $\aset(X)$ for non-leaf regular clusters $X\in \xset$.

\subsubsection*{Processing a Cut $(Y,Z)$ in Cluster $X$}

We assume that we are given a non-leaf regular cluster $X\in \xset$, and a strongly well-structured cut $(Y,Z)$ in $H[X]$ of sparsity $\Phi_{H[X]}(Y,Z)\leq \frac{(\log N')^{64}}{d_X}$. 
We denote by $J\subseteq Y$ and $J'\subseteq Z$ the sets of vertices that serve as endpoints of the edges of $E_H(Y,Z)$. Since the edges of $E_H(Y,Z)$ are special edges, it must be the case that $J\subseteq R$ and $J'\subseteq L$. Denote $Y'=Y\setminus J$ and $Z'=Z\setminus J'$.
We start with the following simple observation, whose proof is deferred to Section \ref{subsec: appx: proof of both graphs are perfect well-str} of Appendix.

\begin{observation}\label{obs: both graphs are perfect well-structured}
Each of the graphs $H[Y'],H[Z']$ is a perfect well-structured graph.
\end{observation}

Assume first that $w(Y)\leq w(Z)$.
We update the set $X$ of vertices, by deleting the vertices of $Y\cup J'$ from it; in other words, we set $X=Z'$. We then add $Y'$ as a new cluster to $\xset$, and additionally, for every vertex $v\in J\cup J'$, we add a new singleton cluster $\set{v}$ to $\xset$.  We define the intervals associated with the new clusters, as follows. We assign, to every vertex $v\in J'\cup J$ an interval $I_{\set{v}}\subseteq I_X$ of length $1$, and additionally we define two intervals: $I_{Y'}\subseteq I_X$ of length $|Y'|$ and $I_{Z'}\subseteq I_X$ of length $|Z'|$, so that all resulting intervals $I_{Y'},I_{Z'},\set{I_{\set{v}}}_{v\in J\cup J'}$ are mutually disjoint, and the intervals are ordered as follows: we place the intervals corresponding to the vertices of $J'$ first, followed by interval $I_{Z'}$, then $I_{Y'}$, and lastly the intervals corresponding to the vertices of $J$. Since we have set $X=Z'$, we let $I_X=I_{Z'}$.  From \Cref{obs: both graphs are perfect well-structured}, Invariant \ref{inv: perfect well structured} continues to hold.
We now show that Invariant \ref{inv: backward edges} continues to hold, in the following observation, whose proof is deferred to Section \ref{subsec: appx: proof of obs new right to left} of Appendix. 

\begin{observation}\label{obs: new right to left}
	If $e=(u,v)$ is a right-to-left edge that was added due to the above transformation of \ATO, then $u\in J$ and $v\in J'$ holds, and so $e\in E_H(Y,Z)$, and $e$ is a special edge.
\end{observation}

Since we have defined a separate singleton cluster for every vertex of $J\cup J'$, from \Cref{obs: new right to left}, Invariant  \ref{inv: backward edges} continues to hold. Moreover, the only right-to-left edges that we have introduced during this update are the edges of $E_{H}(Y,Z)$. Notice that the $\spann(e)$ value of each such edge is at most $|X|$. We call all edges of $E_{H}(Y,Z)$ \emph{bad edges}, and we say that they are \emph{owned} by cluster $X$. 

Recall that so far we have assumed that $w(Y)\leq w(Z)$. If $w(Y)>w(Z)$, then the algorithm remains the same, except that at the end we update cluster $X$ by deleting the vertices of $Z\cup J$ from it; or equivalently, we set $X=Y'$. We also set $I_X=I_{Y'}$. This concludes the algorithm for processing the cut $(Y,Z)$ of a non-leaf regular cluster $X\in \xset$.

\subsubsection{The Special Cluster}

Recall that Case 1 happens if
$|U^0|>2(n-|B^0|)\cdot (\log N)^{128}$ holds. In this case, we let $U$ be the special cluster of the \ATO, and we set $\gamma=|U^0|-256(n-|B^0|)(\log N)^{64}$. 
We now show an algorithm for processing the special cluster $U$ if Case 1 happened. We initialize Algorithm $\aset'$ for the $r$-restricted \maintainspeccluster problem on graph $H[U^0]$, with parameters $\Delta$ and $\eta$ that remain unchanged, parameter $d$ replaced by parameter $d_U= \max\set{(\log N')^{64},\left (|U^0\setminus B^0|+\Delta\right )\frac{\eta}{4\Delta\cdot 2^{\sqrt{\log N}}}}$, and parameter $N$ replaced with $N'$. We also initially set $\beta=B^0\cap U^0$, and we ensure that, throughout the algorithm, $\beta=B\cap U$ holds. In other words, whenever a vertex of $\beta$ is deleted from $B$ or from $U$, it is also deleted from $\beta$. 

We denote this instantiation of Algorithm $\aset'$ by $\aset'(U)$. We now verify that $(H[U^0],\beta,\Delta,\eta,d,N')$ is indeed a valid input instance for the $r$-restricted \maintainspeccluster problem, in the following simple observation, whose proof is deferred to Section \ref{subsec: appx: proof valid spec instance} of Appendix.

\begin{observation}\label{obs: valid spec instance}
	$(H[U^0],\beta,\Delta,\eta,d_U,N')$ is a valid input to the $r$-restricted \maintainspeccluster problem.
\end{observation}

Recall that Algorithm $\aset'(U)$ may, at any time, 
produce a strongly well-structured cut $(Y,Z)$ in $H[U]$ of sparsity $\Phi_{H[U]}(Y,Z)\leq \frac{(\log N')^{64}}{d_U}$, such that $w(Y)\geq 1.8|\beta\cap Y|$ holds. 
Below, we provide an algorithm for processing such a cut.
But before we do so, we provide the remaining details for processing cluster $U$.
We will show later that, throughout the algorithm, $|\beta|\geq \frac{|B^0\cap U^0|}{2}$ holds.

Recall that the algorithm for the \stSP problem consists of at most $\Delta$ iterations. At the beginning of every iteration $i$, we may compute a single vertex $x\in U$, that is sent as a query to Algorithm $\aset'(U)$. Algorithm $\aset'(U)$ is then required to compute a simple path $P(x)$ connecting $x$ to some vertex $y\in \beta$ in graph $H[U]$, whose length is at most $d_U$, such that no inner vertex of $P(x)$ lies in $\beta$ (if $x\in \beta$, then we require that path $P(x)$ only consists of the vertex $x$). 
Our algorithm will then compute a short $s$-$t$ path $P_i$ in $H$, as required, so that $E(P_i)\cap E(H[U])\subseteq E(P(x))$. Following that, the lengths of some special edges on path $P_i$ may be doubled; we assume that our algorithm is given the list of all such edges. In particular, we are given a subset of the special edges on path $P(x)$, whose lengths must be doubled. Additionally, vertex $y$ that serves as an endpoint of path $P(x)$ is deleted from set $B$ and from set $\beta$. This ensures that the lengths of the special edges that are incident to vertices of $\beta$ always remain $1$. We note that in some of the iterations the algorithm for the \stSP problem may not ask any queries to algorithm $\aset'(U)$, and in such iterations the lengths of the special edges of $H[U]$ are not doubled. 
From the definition of the \stSP problem, we are guaranteed that, for every special edge $e\in E(H[U])$, if  $\tau'>\tau$ are two times at which the length of $e$ is doubled, then the number of paths $P(x)$ that Algorithm $\aset'(U)$ returned  in response to queries during the time interval $(\tau,\tau']$, which contain $e$, is at least $\eta$, so Property \ref{maintaincluster-prop: edge congestion} holds for the instance of the \maintainspeccluster problem associated with $U$. If Algorithm $\aset'(U)$ ever returns ``FAIL'', then we terminate our algorithm for the \stSP problem, and return ``FAIL'' as well.

Finally, we provide an algorithm for processing the sparse cuts that Algorithm $\aset'(U)$ may produce. Suppose Algorithm $\aset'(U)$ returns  a strongly well-structured cut $(Y,Z)$ in $H[U]$ of sparsity $\Phi_{H[U]}(Y,Z)\leq \frac{(\log N')^{64}}{d_U}$, such that $w(Y)\geq 1.8|\beta\cap Y|$ holds.  
We denote by $J\subseteq Y$ and $J'\subseteq Z$ the sets of vertices that serve as endpoints of the edges of $E_H(Y,Z)$, and we let $Y'=Y\setminus J$ and $Z'=Z\setminus J'$.
 Using the same reasoning as in the proof of \Cref{obs: both graphs are perfect well-structured}, 
each of the graphs $H[Y'],H[Z']$ is a perfect well-structured graph.
We will show below that $|Z'|\geq \gamma$ must hold. We delete the vertices of $Y\cup J'$ from $U$, so that $U=Z'$ now holds, and we add $Y'$ as a new cluster to $\xset$. Additionally, for every vertex $v\in J\cup J'$, we add a new singleton cluster $\set{v}$ to $\xset$.  We define the intervals associated with the new clusters as follows. We assign, to every vertex $v\in J'\cup J$ an interval $I_{\set{v}}\subseteq I_U$ of length $1$, and additionally we define two intervals: $I_{Y'}\subseteq I_U$ of length $|Y'|$ and $I_{Z'}\subseteq I_U$ of length $|Z'|-\gamma+1$, so that all resulting intervals $I_{Y'},I_{Z'},\set{I_{\set{v}}}_{v\in J\cup J'}$ are mutually disjoint, and the intervals are ordered as follows: we place the intervals corresponding to the vertices of $J'$ first, followed by interval $I_{Z'}$, then $I_{Y'}$, and lastly the intervals corresponding to the vertices of $J$. Since we have set $U=Z'$, we let $I_U=I_{Z'}$.  As before, Invariant \ref{inv: perfect well structured} continues to hold. Using the same reasoning as in the proof of \Cref{obs: new right to left}, 	if $e=(u,v)$ is a right-to-left edge that was added due to the above transformation of \ATO, then $u\in J$ and $v\in J'$ holds, and so $e\in E_H(Y,Z)$, and $e$ is a special edge. In particular, Invariant \ref{inv: backward edges} continues to hold.
 Moreover, the only right-to-left edges that we have introduced during this update are the edges of $E_{H}(Y,Z)$. Notice that the $\spann(e)$ value of each such edge is at most $|U^0|-\gamma+1$. We call all edges of $E_{H}(Y,Z)$ \emph{bad edges}, and we say that they are \emph{owned} by cluster $U$. 
We denote by $E^{\bad}(U)$ the set of all bad edges that are owned by cluster $U$ over the course of the algorithm. In the following observation we show that $|U|\geq \gamma$ holds throughout the algorithm, and we bound $|E^{\bad}(U)|$.

\begin{observation}\label{obs: large size of U}
	Over the course of the algorithm, at most $ 64|U^0\setminus B^0|+128\Delta\cdot(\log N)^{64}\leq 256(n-|B^0|)(\log N)^{64}$ vertices are deleted from $U$, and so throughout  the algorithm, $|U|\geq \gamma$ and $|\beta|\geq \frac{|U^0\cap B^0|}{2}$ hold. Moreover, $|E^{\bad}(U)|\leq \frac{\Delta\cdot 2^{6\sqrt{\log N}}}{\eta}$.
\end{observation}
\begin{proof}
	Recall that we denoted by $E^*\subseteq E^{\spec}$ the set of all special edges $e=(u,v)$, such that $u,v\in B$ holds at the end of the algorithm. Let $B^*\subseteq B$ the set of vertices containing all endpoints of the edges in $E^*$, and let $\tilde U=U^0\setminus B^*$. For every vertex $v\in U^0$, we define a weight $\tilde w(v)$ as follows: if, at the end of the algorithm, $v\in U$, then $\tilde w(v)$ is the weight of $v$ at the end of the algorithm; otherwise, $\tilde w(v)$ is the weight of $v$ just before it is deleted from $U$. We use the following observation, whose proof is very similar to the proof of  \Cref{obs: bound edge lengths}, and deferred to Section \ref{subsec: appx: proof of bound on final weights in U} of Appendix.

	\begin{observation}\label{obs: bound final weights in U}
	$\sum_{v\in \tilde U}\tilde w(v)\leq  8|U^0\setminus B^0|+8\Delta\cdot (\log N')^{64}$.
	\end{observation}

	Consider now some cut $(Y,Z)$ of $U$ that Algorithm $\aset'(U)$ produced. Recall that $|E_H(Y,Z)|\leq \frac{(\log N')^{64}}{d_U}\cdot w(Y)$, and that $w(Y)\geq 1.8|\beta\cap Y|$ holds. Since $B^*\cap Y\subseteq \beta\cap Y$, we get that $w(Y)\geq 1.8|B^*\cap Y|$. Since the weight of every vertex in $B^*$ remains $1$ throughout the algorithm, we get that $w(Y\setminus B^*)= w(Y)-|B^*\cap Y|\geq w(Y)-\frac{w(Y)}{1.8}\geq 0.44 w(Y)$. In particular, $|Y\cap B^*|\leq \frac{w(Y)}{1.8}\leq \frac{w(Y\setminus B^*)}{1.8\cdot 0.44}\leq 1.3w(Y\setminus B^*)$.

	Moreover, $|Y|\leq w(Y)=w(Y\setminus B^*)+|Y\cap B^*|\leq 2.3w(Y\setminus B^*)$ holds. Recall that all vertices of $Y$, and also the endpoints of the special edges in $E_H(Y,Z)$ are deleted from cluster $U$ following the cut $(Y,Z)$. Therefore, the total number of vertices deleted from $U$ following the cut $(Y,Z)$ is at most $2|Y|\le 4.6w(Y\setminus B^*)\leq 4.6\sum_{v\in Y\cap \tilde U}\tilde w(v)$. Lastly:

	\[|E_H(Y,Z)|\leq \frac{(\log N')^{64}}{d_U}\cdot w(Y)\le\frac{(\log N')^{64}}{d_U}\cdot 2.3 w(Y\setminus B^*)\leq \frac{2.3\cdot (\log N')^{64}}{d_U}\cdot \sum_{v\in Y\cap \tilde U}\tilde w(v).\]

	Consider now all cuts $(Y_1,Z_1),\ldots,(Y_r,Z_r)$ that Algorithm $\aset'(U)$ ever produced, in the order in which they were produced. For all $1\leq i\leq r$, following the cut $(Y_i,Z_i)$, the vertices of $Y_i$, along with some additional vertices, were deleted from $U$. From the above discussion, the total number of vertices deleted from $U$ following the cut is bounded by $4.6\sum_{v\in Y_i\cap \tilde U}\tilde w(v)$, and $|E_H(Y_i,Z_i)|\leq \frac{2.3\cdot (\log N')^{64}}{d_U}\cdot \sum_{v\in Y\cap \tilde U}\cdot \sum_{v\in Y_i\cap \tilde U}\tilde w(v)$. We conclude that the total number of vertices deleted from $U$ over the course of the algorithm is bounded by:

	\[\begin{split}
	4.6\sum_{i=1}^r\sum_{v\in Y_i\cap \tilde U}\tilde w(v)&\leq 4.6\sum_{v\in \tilde U}\tilde w(v)\\
	&\leq  64|U^0\setminus B^0|+64\Delta\cdot(\log N')^{64}\\
	&\leq 64|U^0\setminus B^0|+128\Delta\cdot(\log N)^{64}\\
	&\leq 256(n-|B^0|)(\log N)^{64},\end{split} \]
	
	from Inequality \ref{eq: bound on log n'}.
	Since, at the beginning of the algorithm, $|U|>2(n-|B^0|)\cdot (\log N)^{128}$ holds, and since
	$\gamma=|U^0|-256(n-|B^0|)(\log N)^{64}$,  we get that $|U|\geq \gamma$ holds throughout the algorithm.
	
Next, we show that, throughout the algorithm, $|\beta|\geq \frac{|U^0\cap B^0|}{2}$ holds. Recall again that $|U^0| >2(n-|B^0|)\cdot (\log N)^{128}$, and that $|U^0\setminus B^0|\leq n-|B^0|$. Therefore, $|U^0\cap B^0|\geq |U^0|-|U^0\setminus B^0|\geq  (n-|B^0|)\cdot (\log N)^{128}$. Set $\beta$ contains all vertices of $U^0\cap B^0$, except for the vertices that were deleted from $B$ (whose number is bounded by $\Delta\leq n-|B^0|$), and the vertices that were deleted from $U$ (whose number is bounded by $256(n-|B^0|)(\log N)^{64}$). Therefore, throughout the algorithm, $|\beta|\geq |U^0\cap B^0|-257(n-|B^0|)(\log N)^{64}$ must hold. Since $|U^0\cap B^0|\geq (n-|B^0|)\cdot (\log N)^{128}$, we get that $|\beta|\geq \frac{|U^0\cap B^0|}{2}$ holds throughout the algorithm.

	Lastly, we get that:
	
	\[
	\begin{split}
	|E^{\bad}(U)|&\leq \sum_{i=1}^r|E_H(Y_i,Z_i)|\\
	&\leq \frac{2.3\cdot (\log N')^{64}}{d_U}\cdot \sum_{i=1}^r\sum_{v\in Y\cap \tilde U}\tilde w(v)\\
&\leq  \frac{10\Delta\cdot 2^{2\sqrt{\log N'}}}{\left (|U^0\setminus B^0|+\Delta \right )\cdot\eta}\cdot \sum_{v\in \tilde U}\tilde w(v)\\
&\leq \frac{10\Delta\cdot 2^{2\sqrt{\log N'}}}{\left (|U^0\setminus B^0|+\Delta \right )\cdot\eta}\cdot \left( 8|U^0\setminus B^0|+8\Delta (\log N')^{64}\right )\\
&\leq  \frac{\Delta\cdot 2^{6\sqrt{\log N}}}{\eta},
	\end{split}\]
	
	since $d_U= \max\set{(\log N')^{64},\left (|U^0\setminus B^0|+\Delta\right )\frac{\eta}{4\Delta\cdot 2^{\sqrt{\log N}}}}\geq \left (|U^0\setminus B^0|+\Delta\right )\frac{\eta}{4\Delta\cdot 2^{\sqrt{\log N}}}$; we have also used Inequality \ref{eq: bound on N'}.
\end{proof}

\subsubsection{Child Clusters, Descendant Clusters, and a Useful Bound}

%Recall that $E^*\subseteq E^{\spec}$ denote the set of all special edges $e=(u,v)$, such that, at the end of the algorithm, $u,v\in B$ holds.

Let $\tilde \xset$ denote the set of all clusters that ever belonged to $\xset$ over the course of the algorithm. We partition $\tilde \xset$ into three subsets: set  $\tilde \xset_0$ contains all singleton clusters, set $\tilde \xset_1$ contains all non-singleton leaf clusters, and set $\tilde \xset_2$ contains all non-leaf clusters. If the special cluster is defined, then we set $\tilde \xset_2'=\tilde \xset_2\setminus\set{U}$, and otherwise we let $\tilde \xset_2'=\tilde \xset_2$.

Let $E^{\bad}$ be the set of all edges of $G$ that ever served as right-to-left edges of the \ATO. For every cluster $X\in \tilde\xset_1\cup \tilde \xset_2$, we denote by $E^{\bad}(X)$ the set of all bad edges that cluster $X$ owns. Recall that, once an edge $e=(x,y)$ is added to the set $E^{\bad}(X)$ of edges, its initial value $\spann(e)$ is bounded by $|X|$. From Invariant \ref{inv: backward edges}, 
once edge $e$ was added to set $E^{\bad}$, if $Y,Y'$ were the clusters of $\xset$ containing $x$ and $y$ respectively, then $Y=\set{x}$ and $Y'=\set{y}$ holds. In other words, when $e$ was added to $E^{\bad}$, $Y$ and $Y'$ were singleton clusters. Therefore, these clusters, and their corresponding intervals $I_Y$ and $I_{Y'}$, remain unchanged for the remainder of the algorithm. In particular, the value $\spann(e)$ also remains unchanged for the remainder of the algorithm.

Recall that, from Invariant \ref{inv: perfect well structured}, for every cluster $X\in \tilde \xset_2'$, graph $H[X]$ remains a perfect well-structured graph for as long as $X$ lies in $\xset$. Recall also that we have defined weights of vertices for such graphs: for every vertex $v\in X$, its weight $w(v)$ is the length of the unique special edge incident to $v$ in $H[X]$. As the algorithm progresses, the lengths of some special edges grow, and so do the weights of the vertices in $H[X]$. 
For every vertex $v\in X^0$, we denote by $w^0_X(v)$ the weight of $v$ when cluster $X$ was first added to set $\xset$, and algorithm $\aset(X)$ was initialized. If $v$ belongs to set $X$ when algorithm $\aset(X)$ is terminated (e.g. because $|X|$ fell below $|X^0|/2$), then we let $w^*_X(v)$ denote the weight of $v$ in graph $H[X]$ when algorithm $\aset(X)$ terminated. Otherwise, we denote by $w^*_X(v)$ the weight of $v$ just before it was deleted from $X$. 
We also denote $W^0(X)=\sum_{v\in X^0}w^0_X(v)$, and $W^*(X)=\sum_{v\in X^0}w^*_X(v)$. 
From \Cref{obs: bound edge lengths}, the total length of all edges in $E^{\spec}$ at the end of the algorithm is bounded by $|E^{\spec}|+2(n-|B^0|)\log^5n\leq 4n\log^5n$, and so, for every cluster $X\in \tilde \xset_2$: 

\begin{equation}\label{eq: final weight of a cluster}
W^0(X)\leq W^*(X)\leq 8n\log^5n. 
\end{equation}

Recall that, at the beginning of the algorithm, $\xset$ may contain at most one non-leaf cluster -- the cluster $U$. A new non-leaf cluster $X'$ may only be inserted into $\xset$ in one of the following two cases: (i) there is some existing non-leaf cluster $X\in \xset$, for which $|X|$ fell below $|X^0|/2$, and $X'=X$; or (ii) there is some existing non-leaf cluster $X\in \xset$, for which Algorithm $\aset(X)$ or $\aset'(X)$ just computed a sparse well-structured cut $(Y,Z)$, and $X'$ is contained in one of the sets $Y$ or $Z$. %Recall that, in the latter case, we are guaranteed that $\sum_{v\in X'}w(v)\leq \frac{\sum_{v\in X}w(v)}2$ holds at the time when $X'$ is added to $\xset$, and moreover, 
In the latter case, the vertices of $X'$ (along, possibly, with some additional vertices) are deleted from $X$. In either one of the above two cases, we say that $X$ is the \emph{parent-cluster} of $X'$, and $X'$ is the \emph{child-cluster} of $X$. We can now naturally define a descendant-ancestor relationship between clusters: Cluster $X'$ is a descendant-cluster of another cluster $X$, iff there is a sequence $X'=X_1,X_2,\ldots,X_r=X$ of $r\geq 1$ clusters, such that, for all $1\leq i<r$, $X_i$ is a child-cluster of $X_{i+1}$. If $X'$ is a descendant-cluster of $X$, then $X$ is the ancestor-cluster of $X'$. It is easy to verify that, for every pair $X_1,X_2\in \tilde \xset_2$ of clusters, either $X^0_1\cap X^0_2=\emptyset$, or $X_1^0\subseteq X^0_2$ (in which case $X_1$ is a descendant cluster of $X_2$), or $X_2^0\subseteq X^0_1$ (in which case $X_2$ is a descendant cluster of $X_1$).
From the above discussion, if $X_1$ is a child-cluster of some cluster $X_2\in \xset'_2$, then either $|X_1^0|<|X_2^0|/2$; or, at the time when $X_1$ was created, $\sum_{v\in X_1}w(v)\leq \sum_{v\in X_2}w(v)/2$ held, following which the vertices of $X_1$ were deleted from $X_2$.

Consider now a pair $X_1,X_2\in \tilde \xset_2'$ of clusters, where $X_2$ is a descendant cluster of $X_1$. We say that $X_2$ is a \emph{close descendant of $X_1$} iff $|X_2^0|>|X_1^0|/2$. In the following claim, we bound the number of close descendants of every cluster in $\tilde X_2'$. The proof is somewhat technical and deferred to Section \ref{subsec: proof of close descendant} of Appendix.

\begin{claim}\label{claim: close descendant}
	Every cluster $X\in \tilde X_2'$ has at most $64\ceil{\log n}$ close descendants.
\end{claim}

For the sake of analysis, we need to bound $\sum_{X\in \tilde \xset_2'}W^*(X)$. This bound will later be used both in order to bound the running time of the algorithm, and in order to bound $\sum_{e\in E^{\bad}}\spann(e)$.

\begin{observation}\label{obs: weights of clusters}
%	For all $0\leq i<\ceil{\log n}$  and $0\leq j\leq 4\ceil{\log n}$, 
%	$\sum_{X\in \xset_{i,j}}W^*(X)\leq 128(n-|B^0|)\cdot (\log N)^{130}$.
	$\sum_{X\in \tilde \xset_2'}W^*(X)\leq (n-|B^0|)\cdot (\log N)^{133}$.
\end{observation}
\begin{proof}
	Recall that for every cluster $X\in \tilde \xset'_2$, $W^0(X)\leq 8n\log^5n$ holds (see Inequality \ref{eq: final weight of a cluster}). 
	For all $0\leq i<\ceil{\log n}$ and $0\leq j\leq 4\ceil{\log n}$, we denote by $\xset_{i,j}$ the set of all clusters $X\in \tilde \xset'_2$ with $2^i\leq |X^0|<2^{i+1}$ and $2^j\leq W^0(X)<2^{j+1}$. 
	
	Fix a pair  $0\leq i<\ceil{\log n}$   and $0\leq j\leq 4\ceil{\log n}$ of indices, and consider the collection $\xset_{i,j}$ of clusters. Notice that, if $X,X'\in \xset_{i,j}$, and $X'$ is a descendant-cluster of $X$, then $X'$ is a close descendant of $X$. From \Cref{claim: close descendant}, for every cluster $X\in \xset_{i,j}$, at most $64\ceil{\log n}$ descendant clusters of $X$ lie in $\xset_{i,j}$. Using a simple greedy algorithm, we can then paratition the clusters in $\xset_{i,j}$ into $z=64\ceil{\log n}+1$ subsets $\yset_1,\ldots,\yset_z$, such that, for all $1\leq z'\leq z$, for every pair $X,X'\in \yset_{z'}$ of clusters, neither of the two clusters is a descendant of the other. In order to do so, we consider the clusters $X\in \xset_{i,j}$ one by one, in the non-decreasing order of $|X^0|$, breaking ties arbitrarily. When cluster $X$ is considered, there must be some index $1\leq z'\leq z$, such that no descendant of $X$ lies in $\yset_{z'}$, and we add $X$ to such a set $\yset_{z'}$.
	For all $1\leq z'\leq z$, we now bound $\sum_{X\in \yset_{z'}}W^*(X)$ separately.
	
	Fix any index $1\leq z'\leq z$, and denote $\yset'=\set{X^0\mid X\in \yset_{z'}}$. Since, for every pair $X,X'$ of clusters in $\yset_{z'}$, neither is the descendant of the other, all sets in $\yset'$ are mutually disjoint. 
	
	Assume first that Case 2 happened, so the special cluster is undefined, and $|U^0|\leq 2(n-|B^0|)\cdot  (\log N)^{128}$ holds. From the construction of the set $U$ of vertices at the beginning of the algorithm, for every special edge $e\in E^{\spec}$, both endpoints of $e$ lie in $U^0$. Therefore, $|E^{\spec}|\leq \frac{|U^0|}{2}\leq (n-|B^0|)\cdot  (\log N)^{128}$. From \Cref{obs: bound edge lengths}, the total length of all edges in $E^{\spec}$ at the end of the algorithm is bounded by $|E^{\spec}|+2(n-|B^0|)\log^5n\leq (n-|B^0|)\cdot  (\log N)^{129}$. Since the clusters in $\yset'$ are mutually disjoint, it is easy to verify that $\sum_{X\in \yset_{z'}}W^*(X)\leq 2(n-|B^0|)\cdot  (\log N)^{129}$.
	
	Next, we consider Case 1, where the special cluster is defined. Let $U^*$ denote the set $U$ of vertices at the end of the algorithm. Notice that $U^*$ is disjoint from all sets of vertices in $\yset'$. For every vertex $v\in U^*$, let $w^*(v)$ be the weight of $v$ at the end of the algorithm, and let $W^*(U)=\sum_{v\in U^*}w^*(v)$. From \Cref{obs: bound edge lengths}, the total length of all edges in $E^{\spec}$ at the end of the algorithm is bounded by $|E^{\spec}|+2(n-|B^0|)\log^5n$, and so: 
	
	\begin{equation}\label{eq: total weight at the end}
	\sum_{X\in \yset_{z'}}W^*(X)+W^*(U)\leq 2|E^{\spec}|+4(n-|B^0|)\log^5n.
	\end{equation}

	Recall that all vertices that serve as endpoints of the edges in $E^{\spec}$ lie in $U^0$. From \Cref{obs: large size of U}, $|U^0\setminus U^*|\leq 256(n-|B^0|)\cdot(\log N)^{64}$. Therefore, $|U^*|\ge 2|E^{\spec}|-256(n-|B^0|)\cdot(\log N)^{64}$, and $W^*(U)\geq |U^*|\geq 2|E^{\spec}|-256(n-|B^0|)\cdot(\log N)^{64}$.
Combining this with Inequality \ref{eq: total weight at the end}, we get that:
	
\[\sum_{X\in \yset_{z'}}W^*(X)\leq 2|E^{\spec}|+4(n-|B^0|)\log^5n-W^*(U) \leq 260(n-|B^0|)\cdot (\log N)^{64}\leq 2(n-|B^0|)\cdot (\log N)^{129}.\]
To conclude, we have shown that, for every index $1\leq z'\leq z$:

\[
\sum_{X\in \yset_{z'}}W^*(X)\leq 2(n-|B^0|)\cdot (\log N)^{129}.
\]

Summing this up over all indices $1\leq z'\leq z$, and recalling that $z=64\ceil{\log n}+1\leq 128\log n$, we get that:

\[
\sum_{X\in \xset_{i,j}}W^*(X)\leq 256(n-|B^0|)\cdot (\log N)^{130}.
\]

Lastly, by summing up the above inequality over all indices $0\leq i<\ceil{\log n}$   and $0\leq j\leq 4\ceil{\log n}$, we get that: 

\[\sum_{X\in \tilde \xset'_2}W^*(X)\leq 2^{12}(n-|B^0|)\cdot (\log N)^{132}\leq (n-|B^0|)\cdot (\log N)^{133}.\]
\end{proof}

We obtain the following corollary of \Cref{obs: weights of clusters}, that allows us to bound the total span of all bad edges owned by the clusters in $\tilde \xset_2'$.

\begin{corollary}\label{cor: span of bad edges type 1}
	$\sum_{X\in \tilde \xset_2'}\sum_{e\in E^{\bad}(X)}\spann(e)\leq \frac{(n-|B^0|)\cdot \Delta\cdot 2^{4\sqrt{\log N}}}{\eta}$.
\end{corollary}
\begin{proof}
	Consider a cluster $X\in \xset_2'$, and let $(Y,Z)$ be a cut that Algorithm $\aset(X)$ returned. Recall that $|E_{H}(Y,Z)|\leq \frac{(\log N')^{64}}{d_X}\cdot \min\set{w(Y),w(Z)}$. If $w(Y)\leq w(Z)$, then let $Z'=Y$, and otherwise let $Z'=Z$. Then $|E_{H}(Y,Z)|\leq \frac{(\log N')^{64}}{d_X}\cdot w(Z')\leq \frac{(\log N')^{64}}{d_X}\cdot\sum_{v\in Z'}w^*_X(v)$ holds. Moreover, all vertices of $Z'$ are deleted from $X$ following the cut $(Y,Z)$, and the edges of $E_H(Y,Z)$ are added to the set $E^{\bad}(X)$ of bad edges that $X$ owns. We assign to every vertex $v\in Z'$ a \emph{charge} of $\frac{(\log N')^{64}}{d_X}\cdot w^*_X(v)$ units, so that the total charge assigned to the vertices of $Z'$ is at least $|E_H(Y,Z)|$. It is then easy to see that, at the end of Algorithm $\aset(X)$, the total number of edges in $E^{\bad}(X)$ is bounded by the total charge to all vertices of $X^0$, which, in turn, is bounded by:
	\[ \frac{(\log N')^{64}}{d_X}\cdot \sum_{v\in X^0}w^*_X(v)=\frac{(\log N')^{64}}{d_X}\cdot W^*(X).\]

	Next, we show that $d_X\geq \frac{|X^0|\cdot \eta}{\Delta\cdot 2^{2\sqrt{\log N}}}$. Indeed, recall that $d= \frac{(n-|B^0|)\cdot\eta}{4\Delta\cdot 2^{\sqrt{\log N}}}$ and so:
	
	\begin{equation}\label{eq: bound on dx}
	d_X=\min\set{\frac{|X^0|\cdot \eta}{\Delta},d}=\min\set{\frac{|X^0|\cdot \eta}{\Delta}, \frac{(n-|B^0|)\cdot\eta}{4\Delta\cdot 2^{\sqrt{\log N}}}}.
	\end{equation}
	
	If Case 2 happened, then $|U^0|\leq 2(n-|B^0|)\cdot (\log N)^{128}$, and, since $X^0\subseteq U^0$, we get that $|X^0|\leq 2(n-|B^0|)\cdot (\log N)^{128}$. If Case 1 happened, then, from \Cref{obs: large size of U}, $|X^0|\leq 256(n-|B^0|)(\log N)^{64}$. Therefore, in any case, $|X^0|\leq 2(n-|B^0|)\cdot (\log N)^{128}$ must hold. Substituting the bound $n-|B^0|\geq \frac{|X^0|}{2(\log N)^{128}}$ into Inequality \ref{eq: bound on dx}, we get that $d_X\geq \frac{|X^0|\cdot \eta}{\Delta\cdot 2^{2\sqrt{\log N}}}$. Altogether, we get that:
	
	\[|E^{\bad}(X)|\leq \frac{(\log N')^{64}}{d_X}\cdot W^*(X)\leq \frac{\Delta\cdot 2^{2\sqrt{\log N}}\cdot  (\log N')^{64}}{|X^0|\cdot \eta} \cdot W^*(X)\leq \frac{\Delta\cdot 2^{3\sqrt{\log N}}}{|X^0|\cdot \eta} \cdot W^*(X).\]
	
	Since for every edge $e\in E^{\bad}(X)$, $\spann(e)\leq |X^0|$, we get that:
	
	\[\sum_{e\in E^{\bad}(X)}\spann(e)\leq |X^0|\cdot |E^{\bad}(X)|\leq \frac{\Delta\cdot 2^{3\sqrt{\log N}}}{\eta} \cdot W^*(X).\]

Finally, from \Cref{obs: weights of clusters}:

\[
\begin{split}
 \sum_{X\in \tilde \xset_2'}\sum_{e\in E^{\bad}(X)}\spann(e)& \leq \frac{\Delta\cdot 2^{3\sqrt{\log N}}}{\eta} \cdot\sum_{X\in \tilde \xset_2'}W^*(X)\\
 &\leq \frac{(n-|B^0|)\cdot \Delta\cdot 2^{3\sqrt{\log N}}\cdot (\log N)^{133}}{\eta}\\
 &\leq \frac{(n-|B^0|)\cdot \Delta\cdot 2^{4\sqrt{\log N}}}{\eta}. 
 \end{split}
 \]
\end{proof}

%% file: SSSP-runtime-part1.tex
\subsubsection{Bounding the Running Time of Part 1 and Probability of Failure}
\label{subsec: bound runtime of part 1}

Recall that, for every cluster $X\in \tilde \xset'_2$, the total running time of Algorithm $\aset(X)$ is bounded by:

\[
\begin{split}
&c_1\cdot (W^0(X))^2\cdot 2^{c_2\sqrt{\log N'}}\cdot (\log N')^{16c_2r}+c_1\cdot W^0(X)\cdot \Delta\cdot \log^4 N'\\
&\quad\quad\quad\quad\quad\quad
\leq c_1\cdot  W^0(X)\cdot (W^0(X)+\Delta)\cdot 2^{c_2\sqrt{\log N}}\cdot (\log N)^{16c_2r+6}\cdot \left(1+\frac{1}{64c_2r}\right )^{16c_2r}\\
&\quad\quad\quad\quad\quad\quad
\leq c_1\cdot  W^0(X)\cdot (W^0(X)+\Delta)\cdot 2^{c_2\sqrt{\log N}}\cdot (\log N)^{16c_2r+7}.
\end{split}
\]

(We have used the fact that, from Inequality \ref{eq: bound on N'}, 
$2^{\sqrt{\log N'}}\leq 2^{\sqrt{\log N}}\cdot \log^6 N$, and from Inequality \ref{eq: bound on log n'}, $\log N'\leq \left(1+\frac{1}{64c_2r}\right )\cdot \log N$. Additionally, we used the assumption that $N$ is sufficiently large.).

\iffalse
Consider now a pair $0\leq i<\ceil{\log n}$   and $0\leq j\leq \ceil{2\log n}$ of indices, and recall that, for every cluster $X\in \xset_{i,j}$, $2^j\leq W^0(X)< 2^{j+1}$. Since, for every cluster $X\in \tilde \xset_2$, $W^0(X)\leq W^*(X)\leq 4n\log^5n$ (see Inequality \ref{eq: final weight of a cluster}), we get that, if $\xset_{i,j}\neq \emptyset$, then $2^j\leq 4n\log^5n$ holds. Additionally, since  $\sum_{X\in \xset_{i,j}}\left(W^*(X)-|\beta'_X|\right)\leq 512(n-|B^0|)\log^6n$ from \Cref{obs: weights of clusters}, we conclude that the total running time of Algorithms $\aset(X)$ for all clusters $X\in \xset_{i,j}$ is bounded by:

\[
\begin{split}
& c_1\cdot  2^{j+1}\cdot\left (\sum_{X\in \xset_{i,j}} \left (W^*(X)-|\beta'_X|\right )\right )\cdot 2^{c_2\sqrt{\log N}}\cdot (\log N)^{16c_2r+7}\\
&\quad\quad\quad\quad\quad\quad
\leq  c_1\cdot n\cdot (n-|B^0|)\cdot 2^{c_2\sqrt{\log N}}\cdot (\log N)^{16c_2r+20}.
\end{split}
\]
\fi

Summing up over all clusters $X\in \tilde \xset_2'$, from \Cref{obs: weights of clusters}, we get that  the total running time of all  Algorithm $\aset(X)$ for clusters $X\in \tilde \xset'_2$ is  bounded by:

\[ 
\begin{split}
&\sum_{X\in \tilde \xset_2'}c_1\cdot  W^0(X)\cdot (W^0(X)+\Delta)\cdot 2^{c_2\sqrt{\log N}}\cdot (\log N)^{16c_2r+7}\\
&\quad\quad\quad\quad\quad\quad\quad\quad\leq \left(\Delta\cdot \sum_{X\in \tilde \xset'_2}W^0(X)+ \left(\sum_{X\in \tilde \xset'_2}W^0(X)\right )^2\right )\cdot c_1\cdot 2^{c_2\sqrt{\log N}}\cdot (\log N)^{16c_2r+7}\\
&\quad\quad\quad\quad\quad\quad\quad\quad\leq c_1\cdot (n-|B^0|)^2 \cdot 2^{c_2\sqrt{\log N}}\cdot (\log N)^{16c_2r+273}\\
&\quad\quad\quad\quad\quad\quad\quad\quad\leq c_1\cdot n\cdot (n-|B^0|)\cdot  2^{c_2\sqrt{\log N}}\cdot (\log N)^{16c_2r+273}.
\end{split}
\]

(We have used the fact that $\Delta\leq n-|B^0|$ from the problem definition).

Assume now that Case 1 happened. In this case, the running time of Algorithm $\aset'(U)$ is bounded by:

\[
\begin{split}
&c_1\cdot  W^0(U)\cdot (W^0(U)-|\beta|+\Delta)\cdot 2^{c_2\sqrt{\log N'}}\cdot (\log N')^{16c_2r}
\\&\quad\quad\quad\quad\quad\quad\quad\quad\leq c_1\cdot  W^0(U)\cdot (W^0(U)-|\beta|+\Delta)\cdot 2^{c_2\sqrt{\log N}}\cdot (\log N)^{16c_2r+6}\cdot \left(1+\frac{1}{64c_2r}\right )^{16c_2r}\\
&\quad\quad\quad\quad\quad\quad\quad\quad
\leq c_1\cdot W^0(U)\cdot (W^0(U)-|\beta|+\Delta) \cdot 2^{c_2\sqrt{\log N}}\cdot (\log N)^{16c_2r+7}\\
&\quad\quad\quad\quad\quad\quad\quad\quad
\leq 2c_1\cdot n\cdot (n-|B^0|) \cdot 2^{c_2\sqrt{\log N}}\cdot (\log N)^{16c_2r+7}.
\end{split}
\]

(As before, we have used the fact that, from Inequality \ref{eq: bound on N'}, 
$2^{\sqrt{\log N'}}\leq 2^{\sqrt{\log N}}\cdot \log^6 N$, and from Inequality \ref{eq: bound on log n'}, $\log N'\leq \left(1+\frac{1}{64c_2r}\right )\cdot \log N$. Additionally, since, at the beginning of the algorithm, the length of every special edge is $1$, $W^0(U)=|U^0|\leq n$. From the problem definition, $\Delta\leq n-|B^0|$, and it is easy to see that $W^0(U)-|\beta|=|U^0|-|\beta|\leq n-|B^0|$ must hold).

Additional time that is required in order to maintain the \ATO is asymptotically bounded by the total number of edges that are ever present in $H$, which, from the problem definition, is bounded by $O(|E(H)|)\leq O(n\cdot (n-|B^0|))$. Overall, since we can assume that $c_2$ is a large enough constant, the running time of Part 1 of the algorithm is bounded by:

\[c_1\cdot n\cdot (n-|B^0|)\cdot 2^{c_2\sqrt{\log N}}\cdot (\log N)^{16c_2r+c_2}.   \]

Notice that, since $N\geq n$ is large enough, the above running time is also bounded by $N^3$. Therefore, the total number of clusters that ever lie in $\xset$ is also bounded by $N^3$. For each regular non-leaf cluster $X\in \tilde \xset$, the probability that Algorithm $\aset(X)$ returns ``FAIL'' is bounded by $\frac{1}{(N')^4}\leq \frac{1}{N^4}$. If Case 1 happens, then the probability that Algorithm $\aset'(U)$ returns ``FAIL'' is bounded by $\frac{1}{4}$. From the union bound, the total probability that our algorithm ever returns ``FAIL'' is bounded by $\frac{1}{2}$, as required.

%% file: SSSP-bound-bad-edges.tex
\subsubsection{Bounding the Total Span of Bad Edges}
\label{subsec: bound span of bad edges}

Recall that, from \Cref{cor: span of bad edges type 1},
	$\sum_{X\in \tilde \xset_2'}\sum_{e\in E^{\bad}(X)}\spann(e)\leq \frac{(n-|B^0|)\cdot \Delta\cdot 2^{4\sqrt{\log N}}}{\eta}$.

If Case 1 happened, then, from \Cref{obs: large size of U}, $|E^{\bad}(U)|\leq \frac{\Delta\cdot 2^{6\sqrt{\log N}}}{\eta}$. It is easy to verify that, for every edge $e\in E^{\bad}(U)$, $\spann(e)\leq |U^0|-\gamma\leq 256(n-|B^0|)(\log N)^{64}$, since $\gamma=|U^0|-256(n-|B^0|)(\log N)^{64}$. Therefore, $\sum_{e\in E^{\bad}(U)}\spann(e)\leq \frac{\Delta\cdot 2^{6\sqrt{\log N}}}{\eta}\cdot 256(n-|B^0|)(\log N)^{64}\leq \frac{(n-|B^0|)\cdot \Delta\cdot 2^{7\sqrt{\log N}}}{\eta}$.

Lastly, consider a leaf cluster $X\in \tilde \xset_1$, and recall that $|X^0|\leq  \frac{\Delta\cdot (\log N)^{128}}{\eta}$.
Recall that $|E^{\bad}(X)|\leq |X^0|$, and so $\sum_{e\in E^{\bad}(X)}\spann(e)\leq |X^0|^2\leq |X^0|\cdot \frac{\Delta\cdot (\log N)^{128}}{\eta}$.

Note that all clusters in $\tilde \xset_1$ are mutually disjoint, and, for each such cluster $X$, $X^0\subseteq U^0$ holds. If Case 2 happened, then $|U^0|\leq 2(n-|B^0|)\cdot (\log N)^{128}$, and so $\sum_{X\in \tilde\xset_1}|X^0|\leq |U^0|\leq 2(n-|B^0|)\cdot (\log N)^{128}$. Otherwise, from \Cref{obs: large size of U}, $\sum_{X\in \tilde\xset_1}|X^0|\leq 256(n-|B^0|)(\log N)^{64}$. Therefore, in either case, $\sum_{X\in \tilde\xset_1}|X^0|\leq 2(n-|B^0|)\cdot (\log N)^{128}$ holds, and so overall:

\[\sum_{X\in \tilde \xset_1}\sum_{e\in E^{\bad}(X)}\spann(e)\leq  \frac{\Delta\cdot (\log N)^{128}}{\eta}\cdot \sum_{X\in \tilde \xset_1}|X^0|\leq   \frac{2(n-|B^0|)\cdot \Delta\cdot (\log N)^{256}}{\eta}.\]

Summing up over the three types of clusters, we get that:

\begin{equation}\label{eq: span of bad edges}
\sum_{e\in E^{\bad}}\spann(e)\leq \frac{4(n-|B^0|)\cdot \Delta\cdot 2^{7\sqrt{\log N}}}{\eta}.
\end{equation}

%% file: SSSP-part2.tex
\subsection{Part 2: Solving \DLSSSP on the Contracted Graph}

The second part of our algorithm, we will define a dynamic graph $\hat H$ that our algorithm will maintain, which in turn will be used to define an instance of the  \DLSSSP problem (see \Cref{sec: SSSP in almost DAGS}). We will then apply the algorithm for \DLSSSP from \Cref{thm: almost DAG routing} to this instance of the problem.
We start by defining the dynamc graph $\hat H$, and by providing an algorithm for maintaining it.

\subsubsection*{Definition of the Dynamic Graph $\hat H$}

Consider any time $\tau$ during the time horizon $\tset$, the graph $H$ and the collection $\xset$ of clusters at time $\tau$. We say that a cluster $X\in \xset$ is \emph{distinguished}, if it contains some vertex that currently lies in $B$.

Graph $\hat H$ at time $\tau$ is defined as follows. We start with graph $H$, and, for every cluster $X\in \xset$ that is not a distinguished cluster, we contract all vertices in $X$ into a single vertex $v_X$ (that we refer to as a \emph{supernode}). Additionally, all vertices that lie in  distinguished clusters, together with the vertex $t$, are contracted into a single supernode, that for convenience we continue to denote by $t$. We delete all loops but keep parallel edges. We also denote the supernode $v_S$ corresponding to cluster $S=\set{s}$ by $s$. Consider now any edge $e$ in the resulting graph. This edge corresponds to some edge in  $H$ whose endpoints lie in different clusters; we do not distinguish between these two edges.
We replace every edge $e=(v_X,v_Y)$ in the resulting graph by a collection $J(e)$ of $O(\log^2n)$ parallel edges $(v_X,v_Y)$, as follows. Let $2^j$ be the smallest power of $2$ that is greater than $\skipp(e)$. If $e$ is a regular edge in $H$, then set $J(e)$ contains, for all $j\leq j'\leq \ceil{\log n}$, a copy $e_{j'}$ of $e$, whose length is $\hat \ell(e_{j'})=1$, and weight $w(e_{j'})=2^{j'}$.
Assume now that $e$ is a special edge of $H$, and  denote its length in $H$ by $\ell(e)=2^i$. For all $i\leq i'\leq \ceil{\log \Lambda}$ and $j\leq j'\leq \ceil{\log n}$, set $J(e)$ contains a copy $e_{i',j'}$ of $e$ with length $\hat \ell(e_{i',j'})=2^{i'}$ and weight $w(e_{i',j'})=2^{j'}$. This completes the definition of graph $\hat H$.

Consider now some vertex $v\in V(\hat H)\setminus\set{t}$, and some integer $0\leq i\leq \ceil{\log n}$. Consider the set of vertices  $U_i(v)=\set{u\in V(\hat H)\mid \exists e=(v,u)\in E(\hat H)\mbox{ with } w(e)=2^i}$. 
Notice that set $U_i(v)$ may only contain the nodes $s,t$, and all nodes $v_Y\in V(\hat H)$, such that some edge $e$ of $H$ connects a vertex of $X$ to a vertex of $Y$, and has value $\skipp(e)\leq 2^{i}$. In the latter case, the distance between intervals $I_X$ and $I_Y$ in the $\ATO$ is at most $2^i$, and so at most $2^i$ other intervals of the $\ATO$ may lie between them. Therefore, it is immediate to verify that $|U_i(v)|\leq 2^{i+2}$ always holds, establishing Property \ref{prop: few close neighbors} from the definition of the \DLSSSP problem.
Let $n'$ be the total number of vertices that are ever present in the dynamic graph $\hat H$. We bound $n'$ in the following simple observation.

\begin{observation}\label{obs: few vertices in contracted}	
$n'\leq 2(n-|B^0|)+2$.
\end{observation}
\begin{proof}
	Let $B^*$ be the set of vertices that lie in $B$ at the end of the algorithm. Clearly, $n'\leq n-|B^*|+2$. On the other hand, since the algorithm consists of $\Delta\leq n-|B^0|$ iterations, and in every iteration at most one vertex is deleted from $B$, we get that $|B^*|\geq |B^0|-\Delta\geq 2|B^0|-n$. Therefore, $n'\leq n-|B^*|+2\leq 2(n-|B^0|)+2$.
\end{proof}

\subsubsection*{Maintaining the Graph $\hat H$}
At the beginning of the algorithm, we construct the initial graph $\hat H$ from graph $H$ as described above. It is easy to verify that this can be done in time $O(|E(H)|\cdot \log^2n)\leq O(n(n-|B^0|)\log^2n)$.

Next, we show an algorithm that maintains the graph $\hat H$, as graph $H$ and the collection $\xset$ of clusters evolve over time. We will ensure that the only updates that $\hat H$ undergoes are edge deletions and vertex splitting, as is required from the input to the \DLSSSP problem.

Recall that the only updates that graph $H$ undergoes are (i)  doubling of the lengths of its edges; and (ii) deletion of edges incident to $s$ and $t$. Additionally, some vertices may be deleted from set $B$.

If an edge $e$ that is incident to $s$ or $t$ is deleted from $H$, then we delete all copies of $e$ from $\hat H$. If the length of some edge $e$ is doubled, and the initial length of edge $e$ was $2^i$, then we delete all copies of $e$ whose length is $2^i$ from graph $\hat H$.

Assume now that some vertex $v$ is deleted from set $B$. Let $X\in \xset$ be the cluster that contains $v$. If $X$ remains a distinguished cluster, then no additional updates to graph $\hat H$ are required. Otherwise, we perform a \emph{vertex splitting} operation, that we apply to supernode $t$, to construct a new supernode $v_X$. The corresponding collection $E'$ of edges is defined as follows. Let $e=(a,b)$ be an edge of $H$, with exactly one endpoint in $X$; assume w.l.o.g. that $a\in X$ and $b\in Y$, where $Y\in \xset$ is some cluster.
Let $2^j$ be the smallest power of $2$ that is greater than $\skipp(e)$. 
Assume first that $e$ is a regular edge in $H$. 
Notice that, if $Y$ is not a distinguished cluster, then, for all $j\leq j'\leq \ceil{\log n}$, 
graph $\hat H$ already contains an edge $(t,v_Y)$, that is a copy of $e$, whose length is $1$, and weight is $2^{j'}$. We add to $E'$ a copy $(v_X,v_Y)$ of 
$e$, whose length is $1$, and weight is $2^{j'}$.  Once  the vertex splitting operation is completed, we will delete from $\hat H$ all copies of $e$ that are incident to $t$. If $Y$ is a distinguished cluster, then, for all $j\leq j'\leq \ceil{\log n}$, we insert a copy $(v_X,t)$ of the edge $e$ into $E'$, whose length is $1$ and weight is $2^{j'}$.
Assume now that $e$ is a special edge of $H$, and  denote its length in $H$ by $\ell(e)=2^i$. Consider a pair $i\leq i'\leq \ceil{\log n}$ and $j\leq j'\leq \ceil{\log n}$ of indices. As before, if $Y$ is not a distinguished cluster, then graph $\hat H$ already contains an edge $(t,v_Y)$, that is a copy of $e$, whose length is $2^{i'}$, and weight is $2^{j'}$. We add to $E'$ a copy $(v_X,v_Y)$ of 
$e$, whose length is $2^{i'}$, and weight is $2^{j'}$.  Once  the vertex splitting operation is completed, we will delete from $\hat H$ all copies of $e$ that are incident to $t$. If $Y$ is a distinguished cluster, then  we insert a copy $(v_X,t)$ of the edge $e$ into $E'$, whose length is $2^{i'}$ and weight is $2^{j'}$. The case where $b\in X$ and $a\not\in X$ is treated similarly.

Lastly, it remains to consider updates to the clusters in $\xset$. All such updates can be implemented by the following two operations: (i) given a cluster $X\in \xset$, and a subset $X'\subseteq X$ of vertices, insert $X'$ into $\xset$ (we call this update \emph{cluster splitting}); and (ii) given 
a cluster $X\in \xset$ and a vertex $x\in X$, delete $x$ from $X$ (we call this update \emph{vertex deletion}).

Consider first a cluster splitting operation, that is applied to cluster $X\in \xset$, through which a new cluster $X'\subseteq X$ is inserted into $\xset$. Assume first that $X$ is not a distinguished cluster. We apply a vertex splitting operation to supernode $v_X$, inserting a new supernode $v_{X'}$ into $\hat H$. The corresponding set $E'$ of edges is defined as follows. 
Let $e=(a,b)$ be an edge of $H$, with exactly one endpoint in $X'$; assume w.l.o.g. that $a\in X'$ and $b\in Y$, where $Y\in \xset$ is some cluster.
Let $2^j$ be the smallest power of $2$ that is greater than $\skipp(e)$. 
Assume first that $e$ is a regular edge in $H$. 
Notice that, if $Y\neq X$, then, for all $j\leq j'\leq \ceil{\log n}$, 
graph $\hat H$ already contains an edge $(v_X,v_Y)$, that is a copy of $e$, whose length is $1$, and weight is $2^{j'}$. We add to $E'$ a copy $(v_{X'},v_Y)$ of 
$e$, whose length is $1$, and weight is $2^{j'}$.  If $Y=X$, then, for all $j\leq j'\leq \ceil{\log n}$, we insert a copy $(v_{X'},v_X)$ of the edge $e$ into $E'$, whose length is $1$ and weight is $2^{j'}$.
Assume now that $e$ is a special edge of $H$, and  denote its length in $H$ by $\ell(e)=2^i$. Consider a pair $i\leq i'\leq \ceil{\log n}$ and $j\leq j'\leq \ceil{\log n}$ of indices. As before, if $Y\neq X$, then graph $\hat H$ already contains an edge $(v_X,v_Y)$, that is a copy of $e$, whose length is $2^{i'}$, and weight is $2^{j'}$. We add to $E'$ a copy $(v_{X'},v_Y)$ of 
$e$, whose length is $2^{i'}$, and weight is $2^{j'}$. If $Y=X$, then  we insert a copy $(v_{X'},v_X)$ of the edge $e$ into $E'$, whose length is $2^{i'}$ and weight is $2^{j'}$. The case where $b\in X'$ and $a\not\in X'$ is treated similarly.

If $X'$ is a distinguished cluster, then so is $X$, and no updates to $\hat H$ are required. If $X$ is a distinguished cluster but $X'$ is not a distinguished cluster, then the update procedure to graph $\hat H$ is almost identical to that when a distinguished cluster becomes a non-distinguished cluster. Specifically, we apply a vertex-splitting operation to $t$, following which a supernode $v_{X'}$ representing the cluster $X'$ is inserted into $\hat H$. The remaining details a straightforward and are omitted here.

Lastly, assume that a vertex $x$ is deleted from some cluster $X\in \xset$. In this case, we consider every edge $e\in E(H)$ that is incident to $x$ one by one. Let $y$ denote the other endpoint of $e$. If $X$ is not a distinguished cluster, then we delete every copy of $e$ that is incident to $v_X$ from $\hat H$. Otherwise, we delete every copy of $e$ that is incident to $t$ from $\hat H$.

To summarize, we view a cluster-splitting operation that is applied to a cluster $X$ as consisting of two steps. In the first step, a new cluster $X'\subseteq X$ is inserted into $\xset$, and in the second step, the vertices of $X'$ are deleted from $X$. Recall that, following such cluster-splitting updates, it is possible that for some edge $e\in E(H)$ whose endpoints lie in different clusters, the value $\skipp(e)$ increases. In such a case, we delete from $\hat H$ all copies of $e$, whose weight is below the new value $\skipp(e)$. This concludes the description of the algorithm for maintaining the graph $\hat H$.

It is easy to verify that the total time that is required in order to compute and maintain graph $\hat H$ is asymptotically bounded by the total time required to maintain graph $H$ and the collection $\xset$ of clusters times $O(\log^6n)$.
Indeed, from \Cref{claim: close descendant}, every cluster in $\tilde \xset'_2$ has $O(\log n)$ close descendants, and so it may have at most $O(\log^2n)$ descendant clusters overall. Therefore, every vertex of $H$ may belong to $O(\log^2n)$ clusters over the course of the algorithm. For each edge $e=(x,y)$, for every pair $X,X'$ of clusters that contain $x$ and $y$ respectively, at any time during the algorithm, $O(\log^2n)$ copies of $e$ connecting the supernodes corresponding to $X$ and $X'$ may be present in $\hat H$. Overall, at most $O(\log^5n)$ copies of an edge of $H$ may ever be present in $\hat H$, and so the total time that is required in order to compute and maintain graph $\hat H$ is asymptotically bounded by the total time required to maintain graph $H$ and the collection $\xset$ of clusters times $O(\log^6n)$.

\subsubsection*{Solving the  \DLSSSP Problem on $\hat H$}

So far we have defined a dynamic graph $\hat H$ with lengths and weights on edges and two special edges $s$ and $t$, that undergoes an online sequence of edge-deletion and vertex-splitting updates, such that Property \ref{prop: few close neighbors} holds.
We consider an instance of the \DLSSSP problem defined on graph $\hat H$
with parameter $\Gamma=64n\cdot 2^{8\sqrt{\log N}}$, parameter $d$ replaced by $\frac{\Lambda}{4}$, and $\eps=\frac 1 {10}$.
We apply the algorithm from \Cref{thm: almost DAG routing} to the resulting instance of the \DLSSSP problem.
Recall that the total update time of this algorithm is bounded by:

\[\tilde O\left((n')^2+|E(\hat H)|+\Gamma\cdot n'\right )\leq \tilde O\left(n\cdot (n-|B^0|)\cdot  2^{8\sqrt{\log N}}\right )\]

Overall, the total time that is required for Part 1 of the algorithm, and also for maintaining graph $\hat H$ and for maintaining the data structures for the \DLSSSP problem from \Cref{thm: almost DAG routing} is bounded by:

\[c_1\cdot n\cdot (n-|B^0|)\cdot 2^{c_2\sqrt{\log N}}\cdot (\log N)^{16c_2r+2c_2}.   \]

\subsubsection*{Computing the Cut-Witness}

Once the algorithm for the \DLSSSP problem terminates, we are guaranteed that graph $\hat H$ contains no $s$-$t$ path $Q$ with $\sum_{e\in E(Q)}\hat \ell(e)\leq \frac{\Lambda}{4}$ and $\sum_{e\in E(Q)}w(e)\leq \Gamma$.

We now define new lengths $\ell'(e)$ for edges $e\in E(H)$, as follows.
For every regular edge $e\in E(H)$, we set $\ell'(e)=0$.
Consider now a special edge $e\in E^{\spec}$. If both endpoints of $e$ lie in $B$ at the end of the algorithm, then we set $\ell(e)=0$. Otherwise, if $e\not\in E^{\bad}$, then we set $\ell'(e)=\ell(e)$. Lastly, if $e\in E^{\bad}$, and not both endpoints of $e$ lie in $B$, we set $\ell'(e)=\ell(e)+\spann(e)\cdot \frac{\eta\cdot\log n}{8\Delta\cdot 2^{7\sqrt{\log N}}}$. 
We now show that $\set{\ell'(e)}_{e'\in E(H)}$ is a valid cut-witness.

\begin{claim}\label{claim: valid cut-witness}
	$\set{\ell'(e)}_{e'\in E(H)}$ is a valid cut-witness.
\end{claim}
\begin{proof}
	Let $\tilde E$ denote all edges of $G$ with both endpoints in $B$ at the end of the algorithm. It is easy to verify that, for each edge $e\in \tilde E$, we have set $\ell'(e)=0$. Additionally, for every edge	$e$ that is incident to $s$ or to $t$, we have set $\ell'(e)=0$. Next, we bound $\sum_{e\in E(H)}\ell'(e)$, as follows:

\[\begin{split} 
\sum_{e\in E(H)}\ell'(e)&\leq  \sum_{e\in E^{\spec}\setminus \tilde E}\ell(e)+\sum_{e\in E^{\bad}\setminus \tilde E} \spann(e)\cdot \frac{\eta\cdot\log n}{8\Delta\cdot 2^{7\sqrt{\log N}}} \\
&\leq \sum_{e\in E^{\spec}\setminus \tilde E}\ell(e)+\frac{\eta\cdot\log n}{8\Delta\cdot 2^{7\sqrt{\log N}}}\cdot \sum_{e\in E^{\bad}}\spann(e)\\
&\leq \sum_{e\in E^{\spec}\setminus \tilde E}\ell(e)+\frac{(n-|B^0|)\log n}{2}\\
&\leq \sum_{e\in E^{\spec}\setminus \tilde E}\ell(e)+\frac{\Lambda\cdot \Delta}{2\eta\log^4n}.
\end{split} 
\]

(we have used Inequality \ref{eq: span of bad edges} and the fact that
$\Lambda=(n-|B^0|)\cdot \frac{ \eta \cdot \log^5n}{\Delta}$).

Lastly, we prove that the distance from $s$ to $t$ in $H$, with respect to edge lengths $\ell'(\cdot)$, is at least $\frac{\Lambda}{32}$. Assume for contradiction that there is an $s$-$t$ path $Q$ in the current graph $H$, such that $\sum_{e\in E(Q)}\ell'(e)< \frac{\Lambda}{32}$. 
Recall that graph $G$ contains no regular edge connecting a vertex of $A$ to a vertex of $B$. Therefore, path $Q$ contains at least one special edge. Moreover, if we discard the first and the last regular edges on $Q$, and denote by $Q'$ the resulting path, then regular and special edges alternate on $Q'$. Since the length of every special edge in $H$ is at least $1$, we get that the number of regular edges on $Q$, that we denote by $m_r$, is bounded by $2\sum_{e\in E(Q)}\ell'(e)+3<\frac{\Lambda}{8}$.

Let $\hat E(Q)\subseteq E(Q)$ be the set of edges that also lie in the contracted graph $\hat H$, and denote $\hat E(Q)=(\hat e_1,\ldots,\hat e_z)$, where the edges are indexed in the order of their appearence on $Q$. 
For all $1\leq i\leq z$, let $\hat e_i'$ be the copy of edge $e_i$ in $\hat H$ that has the shortest length and the smallest weight.
Then the sequence $(\hat e_1',\ldots,\hat e_z')$ of edges defines an $s$-$t$ path in graph $\hat H$, that we denote by $\hat Q$. Notice that $\hat E(Q)$ may not contain an edge with both endpoints in the current set $B$. Therefore, for every special edge $e'_i\in \hat E(Q)$, $\hat \ell(\hat e'_i)=\ell(\hat e_i)\leq \ell'(\hat e_i)$, while for every regular edge $e\in \hat E(Q)$, $\hat\ell(e)=1$.
We conclude that the length of path $\hat Q$ in graph $\hat H$ is bounded by $\sum_{e\in E(Q)}\ell'(e)+m_r\leq \frac{\Lambda}{4}$.
To reach a contradiction, it is now enough to show that $\sum_{\hat e'_i\in \hat E(Q)}w(\hat e'_i)\leq \Gamma$.

Indeed, recall that, for every edge $e\in E^{\bad}$, we have set $\ell'(e)=\ell(e)+\spann(e)\cdot \frac{\eta\cdot\log n}{8\Delta\cdot 2^{7\sqrt{\log N}}}$.
Since $\sum_{e\in E(Q)}\ell'(e)\leq \frac{\Lambda}{32}$, we get that:

\[\sum_{e\in E(Q)\cap E^{\bad}}\spann(e)\cdot \frac{\eta\cdot\log n}{8\Delta\cdot 2^{7\sqrt{\log N}}}\leq \frac{\Lambda}{32}.  \]

Therefore:

\[ \sum_{e\in E(Q)\cap E^{\bad}}\spann(e)\leq \frac{\Lambda\cdot \Delta\cdot 2^{7\sqrt{\log N}}}{4\eta\log n}= \frac{(n-|B^0|) \cdot 2^{7\sqrt{\log N}}\cdot \log^4n}{4},\]

since $\Lambda=(n-|B^0|)\cdot \frac{ \eta \cdot \log^5n}{\Delta}$.
But then: 

\[\sum_{e\in E(Q)}\skipp(e)\leq n+2\sum_{e\in E(Q)\cap E^{\bad}}\spann(e)\leq n+\frac{(n-|B^0|) \cdot 2^{7\sqrt{\log N}}\cdot \log^4n}{2}\leq \frac{\Gamma}{4},\]

since $\Gamma=64n\cdot 2^{8\sqrt{\log N}}$.	
Since, for all $1\leq i<z$, $w(\hat e'_i)\leq 2\skipp(\hat e_i)$, we get that:

\[\sum_{\hat e'_i\in E(\hat Q)}w(\hat e'_i)\leq 2\sum_{e\in E(Q)}\skipp(e)\leq \frac{\Gamma}{2}, \]	
	
a contradiction.		
\end{proof}

%% file: SSSP-completing-alg.tex
\subsection{Completing the Algorithm}
Throughout the algorithm, for every special edge $e=(x,y)$ in graph $H$, we maintain a counter $n(e)$, that counts the number of $s$-$t$ paths containing $e$ that the algorithm reported in response to queries that contain $e$, since $\ell(e)$ was doubled last, or since the beginning of the algorithm -- whatever happened later. 

The algorithm contains at most $\Delta$ iterations, and terminates when either (i)  all iterations are completed; or (ii) the algorithm for the  \DLSSSP problem returns ``FAIL''; or (iii) one of the algorithms for the \maintaincluster problem or for the \maintainspeccluster problem that we use to maintain the clusters of $\xset$ returns ``FAIL''. In the latter case, we terminate our algorithm with a ``FAIL''; we have shown that this may only happen with probability at most $1/2$. In the second case, we return the cut-witness as described above.

We now proceed to describe the $i$th iteration. We use the algorithm for the  \DLSSSP problem from \Cref{thm: almost DAG routing} to compute a simple $s$-$t$ path $P_i$ in the current graph $\hat H$, with $\sum_{e\in E(\hat H)}\hat \ell(e)\leq (1+10\eps)\cdot \frac{\Lambda}{4}\leq \frac{\Lambda}2$, in time $O(|E(P_i)|)$. We denote by $\hat e_1,\hat e_2,\ldots,\hat e_z$ the sequence of the edges on this path. For all $1\leq j\leq z$, let $e_j=(x_j,y_j)$ be the edge corresponding to $\hat e_j$ in graph $H$. Notice that $x_1=s$, and, for all $1\leq j<z$, there is a cluster of $\xset$ that contains both $y_j$ and $x_{j+1}$; we denote this cluster by $X_j$.
Cluster $X_j$ cannot be a distinguished cluster, since path $P_i$ is simple.
 Lastly, there must be a distinguished cluster $X_z\in \xset$, that contains $y_z$. 

Recall that, if Case 1 happened, then $|U^0|\geq 2(n-|B^0|)\cdot (\log N)^{128}$. Moreover, $|U|\geq \gamma$ holds throughout the algorithm, where $\gamma=|U^0|-256(n-|B^0|)(\log N)^{64}\geq (n-|B^0|)\cdot (\log N)^{128}$. At the end of the algorithm, $|V(G)\setminus B|\leq (n-|B^0|)+\Delta\leq 2(n-|B^0|)$ holds. Therefore, if Case 1 happens, then $U$ remains a distinguished cluster throughout the algorithm.

For all $1\leq j<z$, we query Algorithm $\aset(X_j)$ for the \maintaincluster problem with vertices $y_j$ and $x_{j+1}$. Let $Q_j\subseteq H[X_j]$ be the simple path connecting $y_j$ to $x_{j+1}$ that the algorithm returns, whose length is at most $d_{X_j}$. Assume first that either Case 1 did not happen; or Case 1 happened and $X_z\neq U$. In this case, we select an arbitrary vertex $x_{z+1}\in B\cap X_z$, and query Algorithm $\aset(X_j)$ for the \maintaincluster problem with vertices $y_z$ and $x_{z+1}$. Let $Q_z\subseteq H[X_z]$ be the simple path connecting $y_z$ to $x_{z+1}$ that the algorithm returns, whose length is at most $d_{X_z}$. If $Q_z$ contains any vertex of $B$ as an inner vertex, then we truncate it, so it connects $y_z$ to a vertex of $B$, and does not contain any vertex of $B$ as an inner vertex. 
We then append an edge connecting the last vertex of $Q_z$ to $t$ to the path $Q_z$. (For all $1\leq j\leq z$, if $X_j$ is the singleton cluster, then we simply let $Q_j$ contain a single vertex $x_{j}=y_{j+1}$).
If Case 1 happened and $X_z=U$, then we query Algorithm $\aset'(U)$ with vertex $y_z$, to obtain a simple path $Q_z$ connecting $y_z$ to some vertex $b\in B$, such that $Q_z$ has length at most $d_U$ and no inner vertices of $Q_z$ belong to $B$. We then append the edge $(b,t)$ at the end of path $Q_z$.

 By concatenating $e_1,Q_1,e_2\ldots,e_z,Q_z$, we obtain an $s$-$t$ path in the current graph $H$, that we denote by $P'_i$. Since the initial path $P_i$ was simple, all custers in $\xset$ are mutually disjoint, and each of the paths $Q_1,\ldots,Q_z$ is simple, we get that $P'_i$ is a simple path as well. We now bound the length of $P'_i$. Recall that $\sum_{j=1}^z\ell(e_j)\leq \sum_{j=1}^z\hat \ell(\hat e_j)\leq \frac{\Lambda}{2}$.

Recall that, for every cluster $X\in \tilde \xset'_2$, $d_X=\min\set{\frac{|X^0|\cdot \eta}{\Delta},d}$.
Notice that clusters $X_1,\ldots,X_{z-1}$ may not contain vertices of $B$. Since, at the end of the algorithm, $|B|\geq |B^0|-\Delta$, we get that:

\[\sum_{j=1}^{z-1}|X_j|\leq n-|B|\leq n-|B^0|+\Delta\leq 2(n-|B^0|). \]

Moreover, for every cluster $X\in \tilde \xset_2'$, algorithm $\aset(X)$ terminates once $|X|$ falls below $|X^0|/2$, and so:

\[\sum_{j=1}^{z-1}|X^0_j|\le 2\sum_{j=1}^{z-1}|X_j| \leq 4(n-|B^0|). \]

Altogether, we get that:

\[\sum_{j=1}^{z-1} d_{X_j}\leq \sum_{j=1}^{z-1}\frac{|X^0_j|\cdot \eta}{\Delta}\leq \frac{4(n-|B^0|)\cdot \eta}{\Delta}\leq \frac{\Lambda}{4},   \]

since $\Lambda=(n-|B^0|)\cdot \frac{ \eta \cdot \log^5n}{\Delta}$.

Note also that either $d_{X_z}\leq d_U$ (if $X_z$ is the special cluster), or $d_{X_z}\leq d$ (otherwise) must hold.
Recall that $d_U=\max\set{(\log N')^{64},\left (|U^0\setminus B^0|+\Delta\right )\frac{\eta}{4\Delta\cdot 2^{\sqrt{\log N}}}}$. Since $\Lambda=(n-|B^0|)\cdot \frac{ \eta \cdot \log^5n}{\Delta}$, while $|U^0\setminus B^0|+\Delta\leq 2(n-|B^0|)$ and $\frac{(n-|B^0|)\cdot \eta}{\Delta}> 2^{4\sqrt{\log N}}$ from Inequality \ref{eq: bound on d2}, we get that $d_U\leq \frac{\Lambda}{4}$.
Recall also that $d=\frac{(n-|B^0|)\cdot\eta}{4\Delta\cdot 2^{\sqrt{\log N}}}\leq \frac{\Lambda}{4}$.
Therefore, $d_z\leq \frac{\Lambda}{4}$ must hold. Overall, we get that:

\[\sum_{e\in E(P')}\ell(e)\leq \sum_{j=1}^z\ell(e_j)+\sum_{j=1}^{z-1}d_{X_j}\leq  d_{X_z}\leq \Lambda, \]

as required.
We increase the counter $n(e)$ for every edge $e\in E(P')$. If, for any such edge $e$, the counter $n(e)$ reaches $\eta$, then we double the length of $e$, and we reset $n(e)=0$. 

Notice that the time required to respond to the queries posed to algorithms $\aset(X_1),\ldots,\aset(X_{z-1})$, and Algorithm $\aset(X_z)$ (if $X_z$ is a regular cluster) or Algorithm $\aset'(X_z)$ (if it is a special cluster) is already accounted for in the analysis of the running times of these algorithms. The additional time that is required in order to respond to the query is bounded by $O(|E(P_i)|)$.

Overall, the additional time that is required in order to respond to all queries is bounded by $\sum_{i=1}^{\Delta}O(|E(P_i)|)\leq \sum_{i=1}^{\Delta}O(|E(P'_i)|)$.
Recall that, for every special edge $e\in E(H)$, the length of $e$ in $H$ is doubled once it participates in $\eta$ paths that the algorithm returns. Once the length of $e$ reaches $\Lambda$, it may no longer participate in any path that the algorithm returns. Therefore, every special edge may participate in at most $O(\eta\cdot \log \Lambda)\leq O(\eta\log N)$ paths $P'_1,\ldots,P'_{\Delta}$. Since there are at most $n$ special edges in $H$, and since regular and special edges alternate on the paths $P'_1,\ldots,P'_{\Delta}$, we get that 
$\sum_{i=1}^{\Delta}|E(P'_i)|\leq O(n\eta\log N)\leq O(n\Delta\log N)\leq O(n\cdot (n-|B^0|)\cdot \log N)$, since $\eta\leq \Delta$ and $\Delta\leq n-|B^0|$ from the problem definition.
The total additional time that is required in order to respond to all queries is then bounded by $ O(n\cdot (n-|B^0|)\cdot \log N)$, and the total running time of the algorithm is bounded by:

\[O\left(c_1\cdot n\cdot (n-|B^0|)\cdot 2^{c_2\sqrt{\log N}}\cdot (\log N)^{16c_2r+2c_2}\right )\leq c_1\cdot n\cdot (n-|B^0|)\cdot  2^{c_2\sqrt{\log N}}\cdot (\log N)^{16c_2r+7c_2}.   \]

%% file: connecttocenters.tex
\section{Layered Connectivity Data Structure and an Algorithm for the  \maintainspeccluster problem}
\label{sec: layered connectivity DS}

The goal of this section is to prove the second part of \Cref{thm: from routeandcut to maintaincluster}, by showing that an algorithm for the $r$-restricted \routeandcut problem implies an algorithm for the $r$-restricted \maintainspeccluster problem with the required running time. In order to prove this theorem, we define a new data structure, that we call a \LCDS, and provide an algorithm for maintaining this data structure in a graph that undergoes online increases in the lengths of its edges. The same data structure, and the same algorithm for maintaining it, will also be used in \Cref{sec: alg for maintaincluster from routeandcut} in order to prove the first part of \Cref{thm: from routeandcut to maintaincluster}, namely that  an algorithm for the $r$-restricted \routeandcut problem implies an algorithm for the $r$-restricted \maintaincluster problem with the required running time. Because of this, we define both the \LCDS, and the algorithm for maintaining this data structure slightly more generally than what is needed for the \maintainspeccluster problem. But first we start with intuition.

Recall the definition of the \maintainspeccluster problem. At a high level, the input is a perfect well-structured graph $G=(L,R,E)$  with a proper assignment of lengths $\ell(e)$ on its edges $e\in E$ that are given in the adjacency list representation, together with a subset $\beta\subseteq V(G)$ of vertices with $|\beta|\geq \frac{99|V(G)|} {100}$, and parameters $N\geq W^0(G)$, $1\leq \Delta\leq \frac{|\beta|}{(\log N)^{128}}$, $1\leq \eta \leq \Delta$, and $d\geq (\log N)^{64}$, such that $\frac{\Delta\cdot d}{\eta}\leq |V(G)|-|\beta|+\Delta\cdot (\log N)^{64}$, and every special edge has length $1$. Over the course of the algorithm, vertices may be deleted from set $\beta$, and the lengths of some edges may be increased, but the length of every special edge incident to a vertex that currently lies in $\beta$ remains $1$.

The algorithm consists of at most $\Delta$ iterations.
At the beginning of every iteration $i$, the algorithm is given a vertex $x\in V(G)$, and it must return a simple path $P_i$ in $G$ connecting $x$ to any vertex  $y\in \beta$, such that the length of the path is at most $d$, and no inner vertices on the path lie in $\beta$. After that, the length of some special edges on path $P_i$ may be doubled, and vertex $y$ is deleted from $\beta$. Recall that, for any edge $e\in E(G)$, if $\ell(e)$ is doubled at times $\tau'$ and $\tau$ with $\tau'<\tau$, then at least $\eta$ paths that the algorithm returns during the time interval $[\tau',\tau)$ must contain $e$.
The algorithm may, at any time, produce a strongly well-structured cut $(X,Y)$ in $G$ of sparsity $\Phi_G(X,Y)\leq \frac{(\log N)^{64}}{d}$, such that $w(X)\geq 1.8|\beta\cap X|$ holds. Let $J'\subseteq Y$ denote the set of vertices that serve as endpoints of the edges of $E_G(X,Y)$. The vertices of $X\cup J'$ are then deleted from $G$, and the algorithm continues with the resulting graph $G=G[Y\setminus J']$.

A natural way to solve this problem is to simply maintain a collection 
$\qset=\set{K(v)\mid v\in V(G)}$ of paths in  $G$, where for each vertex $v\in V(G)$, path $K(v)$ connects $v$ to some vertex of $\beta$, such that the length of every path in $\qset$ is at most $d$, and every vertex of $\beta$ serves as endpoint of a small number of such paths. 
Then, whenever a query vertex $x$ arrives, we could simply return the path $K(x)$ in response to the query, possibly truncating it to ensure that no inner vertices on the path belong to $\beta$.
As the time progresses, some of the paths in $\qset$ may however be destroyed. For example, as the lengths of some special edges are doubled, the lengths of some paths in $\qset$ may become too large, so we can no longer use them in order to respond to queries. Additionally, some vertices may be deleted from $\beta$, and so paths connecting vertices of $G$ to such vertices can no longer be used in order to respond to queries. Lastly, when our algorithm returns a  sparse cut $(A,B)$, and some of the vertices are deleted from $G$ following this, it is possible that, for some of the paths $K(v)\in \qset$, some inner vertices on $K(v)$ are deleted from $G$, while the first endpoint of the path remain in $G$. Whenever, for some vertex $v\in V(G)$, the corresponding path $K(v)\in \qset$ is destroyed, we say that vertex $v$ became \emph{disconnected}. We need to ensure that, before we respond to each query, every vertex of $G$ has a valid path $K(v)\in \qset$. Therefore, whenever a vertex $v$ becomes disconnected, we need to compute a new path $K(v)$ connecting it to some vertex that currently lies in $\beta$. The difficulty is that, even if the number of vertices that are currently disconnected is quite small, if we attempt to connect them to the vertices of $\beta$ directly, the amount of time that we need to invest to compute such a routing may be potentially as large as $\Omega(|V(G)|\cdot (|V(G)|-|\beta|)$. We note that we can ensure that the total number of vertices that ever become disconnected over the course of the algorithm (counting every time a vertex becomes disconnected separately, even if the same vertex is disconnected several times), is relatively small. However, there may be a large number of iterations, in each of which a small number of vertices become disconnected. Therefore, computing a new collection of paths connecting these disconnected vertices to $\beta$ from scratch in every iteration is prohibitively expensive in terms of the running time. A different approach that could address this issue is the following: let $U$ be the set of vertices that are currently disconnected, and let $U'=V(G)\setminus U$. Instead of connecting the vertices of $U$ to the vertices of $\beta$ directly, we can instead try to connect them to the vertices of $U'$, and then exploit the paths in $\set{K(v)\mid v\in U'}$ in order to route the vertices of $U$ to the vertices of $\beta$. In doing so we will only need to explore a relatively small subgraph of $G$, and we can use the algorithm for the $r$-restricted \routeandcut problem, whose existence is assumed in the statement of \Cref{thm: from routeandcut to maintaincluster}, in order to do so in time roughly $O(W(G)^{1+o(1)}\cdot |U|)$. The difficulty with this approach is that, as we compose the routing paths with each other, iteration after iteration, we may obtain routing paths that are very long. To conclude, we need to balance between the need to ensure that all paths in $\qset$ are sufficiently short, with the need to ensure that the total running time invested in reconnecting the disconnected vertices to $\beta$ is not too high. We use a natural \emph{layering} approach, by defining a \LCDS, which we describe below. 

So far we provided intuition for the \LCDS in the context of obtaining an algorithm for the $r$-restricted \maintainspeccluster problem. We will also exploit this data structure in our algorithm for the $r$-restricted \maintaincluster problem. Intuitively, instead of the set $\beta$ of vertices that is given as part of input to the \maintainspeccluster problem, in order to solve the \maintaincluster problem, we will select a random set $\beta$ of such vertices of an appropriate cardinality, and then embed an expander that is defined over the vertices of $\beta$ into the input graph $G$. The \LCDS data structure will then be used to maintain paths connecting every vertex of $G$ to some vertex of $\beta$. We will also maintain a similar data structure in graph $\revG$, that is obtained by reversing the direction of every edge in $G$, obtaining a collection of paths that, for every vertex $v\in V(G)$, connects some vertex of $\beta$ to $v$. 
We now formally define the \LCDS data structure, followed by an algorithm for maintaining such a data structure in dynamic graphs that arise from the \maintaincluster and \maintainspeccluster problems.
Lastly, we design an algorithm for the \maintainspeccluster problem that relies on this data structure.

\input{layered-connectivity-DS}

\input{reconnect-layer}

\input{connecttocenters-problem}

\input{connecttocenters-alg}

\input{alg-maintainspecial}

%% file: layered-connectivity-DS.tex
\subsection{The Layered Connectivity Data Structure}
\label{subsubsec: layered connectivity DS}

We assume that we are given a perfect well-structured graph $G=(L,R,E)$ with $|V(G)|=n$, and a proper assignment $\ell(e)$ of lengths to the edges $e\in E$, that are given in the adjacency-list representation. Recall that for every vertex $v\in V(G)$, its weight $w(v)$ is the length of the unique special edge that is incident to $v$ in $G$. We denote $W=\sum_{v\in V(G)}w(v)$, and we let $H$ be the subdivided graph corresponding to $G$ (see \Cref{def: subdivided graph}). We assume further that we are given a parameter $N\geq \max\set{n,W/32}$ that is greater than a large enough constant, together with additional parameters $d'\geq 1$, $\rho\geq \log^2N$, and $1\leq r\leq \ceil{\sqrt{\log N}}$, such that $d'\leq\min\set{\frac{2^{r\sqrt{\log N}}}{\log^8N},\frac{n}{8}}$. Lastly, we are given a collection $\beta\subseteq V(H)$ of vertices, that we refer to as \emph{centers}, with $|\beta|\geq \frac{W}{d'\cdot \log^2N}$.

Let $\lambda=\ceil{\log(64N)}+1$.
A \LCDS for graph $G$ consists of the following ingredients:

\begin{itemize}
	\item The first ingredient is a partition $(U_0,U_1,\ldots,U_{\lambda})$ of $V(H)$ into disjoint subsets, such that $U_0=\beta$. For all $0\leq i\leq \lambda$, we refer to the vertices in set $U_i$ as \emph{layer-$i$ vertices}, and we sometimes refer to the vertices of $U_{\lambda}$ as \emph{disconnected vertices}. We also denote by $U^{\geq i}=\bigcup_{i'=i}^{\lambda}U_{i'}$ and by $U^{\leq i}=\bigcup_{i'=0}^{i}U_{i'}$. Similarly, we denote by  $U^{> i}=\bigcup_{i'=i+1}^{\lambda}U_{i'}$ and by $U^{< i}=\bigcup_{i'=0}^{i-1}U_{i'}$.
	
	\item The second ingredient is collections $\qset_1,\ldots,\qset_{\lambda-1}$ of paths in graph $H$. For all $1\leq i\leq \lambda-1$, $\qset_i=\set{K(v)\mid v\in U_i}$, where, for every vertex $v\in U_i$, path $K(v)$ connects vertex $v$ to some vertex $v'\in U^{<i}$. We say that vertex $v'$ is the \emph{parent-vertex} of $v$, and vertex $v$ is a \emph{child-vertex} of $v'$. We also say that $v$ is a \emph{descendant} of $v'$. The descendant relation is transitive: if $v$ is a descendant of $v'$, and $v'$ is a descendant of $v''$, then $v$ is a descendant of $v''$.  Additionally, every vertex is a descendant of itself. If vertex $x$ is a descendant of vertex $y$, then we say that $y$ is an \emph{ancestor} of $x$. Note that ancestor-descendant relation is only determined by the endpoints of the paths in $\bigcup_{i=1}^{\lambda-1}\qset_i$, and it is independent of the inner vertices on these paths.
	For every vertex $x\in \beta$ and index $0\leq i<\lambda$, we denote by $\dset_i(x)$ the set of all vertices $v\in U_i$ that are descendants of $x$, and we denote by $\dset(x)=\bigcup_{i=0}^{\lambda-1}\dset_i(x)$ the set of all descendant-vertices of $x$.
	
	Lastly, we require that the following properties hold for a \LCDS:
	
	\begin{properties}{L}
		\item For all $1\leq i< \lambda$, the paths in $\qset_i$ have length at most $4d'\log^{10}N$ each, and they cause edge-congestion at most $4d'\log^{9}N$ in $H$;  \label{prop: routing paths}
		
		\item For every vertex $x\in \beta$, for every layer $1\leq i\leq \lambda-1$,  $|\dset_i(x)|\leq \rho$; and \label{prop: descendants}
		\item For every layer $1\leq i\leq \lambda-1$,  every vertex $v\in U^{<i}\setminus U^0$, may serve as an endpoint of at most one path in $\qset_i$. \label{prop: descendants one to one}
	\end{properties}
\end{itemize}

The \LCDS data structure must also maintain, for every vertex $x\in \beta$ and layer $1\leq i< \lambda$, the collection $\dset_i(x)\subseteq U_i$ of layer-$i$ descendants of $x$. Additionally, for every  layer $1\leq i\leq \lambda-1$ and  vertex $v\in U^{<i}$, if $v$ serves as an endpoint of a path $K(u)\in \qset_i$, then we maintain a pointer from $v$ to $u$.

In order to construct and maintain a \LCDS, as graph $G$ undergoes updates, we will employ a procedure that we call \reconnect, and describe next.

%% file: reconnect-layer.tex
\subsection{Procedure \reconnect}
\label{subsubsection proc reconnect layer}
As before, we assume that we are given a perfect well-structured graph $G=(L,R,E)$ with $|V(G)|=n$, and a proper assignment $\ell(e)$ of lengths to the edges $e\in E$, that are given in the adjacency-list representation, together with a parameter $N\geq \max\set{n,W/32}$ that is greater than a large enough constant, and parameters $d'\geq 1$, $\rho\geq \log^2N$, and $1\leq r\leq \ceil{\sqrt{\log N}}$, where $W=\sum_{v\in V(G)}w(v)$, and $d'\leq\min\set{\frac{2^{r\sqrt{\log N}}}{\log^8N},\frac{n}{8}}$.
We additionally assume that $\frac{W\cdot \log^3N}{|\beta|}\leq \frac{\rho}{16}$, and that the length of every edge in $G$ is at most $d'$.
Lastly, we assume that we are given a subdivided graph $H$ corresponding to $G$, in the adjacency-list representation, and a subset $\beta\subseteq V(H)$ of its vertices, with $|\beta|\geq \frac{W}{d'\cdot \log^2N}$.

A valid input to Procedure \reconnect consists of an index $1\leq i<\lambda$ of a layer, and a valid \LCDS for $G$, such that, for all $i\leq  i'<\lambda$, $U_{i'}=\emptyset$, and $|U_{\lambda}|\leq 2^{\lambda-i}$.
We denote $k=| U_{\lambda}|$.
A valid output of Procedure \reconnect is one of the following:

\begin{itemize}
	\item either a valid \LCDS, in which, for all $i<i'<\lambda$, $U_{i'}=\emptyset$ and $|U_{\lambda}|\leq \frac{k}{4\lambda}$ holds; additionally, layers $U_1,\ldots,U_{i-1}$, together with the corresponding collections of paths $\qset_1,\ldots,\qset_{i-1}$ must remain the same as in the input \LCDS; or

	\item a strongly well-structured $\frac{1}{4d'}$-sparse cut $(A,B)$ in $G$ with $w(A),w(B)\geq \frac{k}{2^{14}\log N}$ and $w(A)\geq 1.8|\beta\cap A|$; or
	
		\item A weakly well-structured cut $(A',B')$ in $H$ with $|A'|\geq \frac{511|V(H)|}{512}$,  $|A'\cap \beta|<\frac{2|\beta|}{\log^3N}$, and $|E_H(A',B')|\leq  \frac{|V(H)|}{1024d'\cdot \log^6N}$.
\end{itemize}

Next, we show an algorithm, that, given a valid input to Procedure \reconnect,  produces a valid output for the procedure. The algorithm relies on Algorithm $\aset$ for the $r$-restricted \routeandcut problem, whose existence was assumed in the statement of \Cref{thm: from routeandcut to maintaincluster}.

\begin{claim}\label{claim: proc reconnect}
Suppose that, for some parameter $r\geq 1$, there exists a randomized algorithm for the $r$-restricted \routeandcut problem, that, given as an input an instance $(G,A,B,\Delta,\eta,N)$ of the problem with $|V(G)|=n$,
has running time at most $c_1\cdot n\cdot (n-|B|)\cdot 2^{c_2\sqrt{\log N}}\cdot (\log N)^{16c_2(r-1)+8c_2}$.
Then there is a randomized algorithm, that,  given a valid input to Procedure \reconnect, where graph $G$ and its corresponding subdivided graph $H$ are given in the adjacency-list representation, with the same parameter $r$, either produces a valid output for the procedure, or returns ``FAIL''. The running time of the algorithm is at most
$ O\left(c_1\cdot W\cdot k\cdot 2^{c_2\sqrt{\log N}}\cdot (\log N)^{16c_2(r-1)+8c_2+14}\right )$, and the probability that it returns ``FAIL'' is at most $\frac{1}{N^{10}}$. 
\end{claim} 

The remainder of this subsection is dedicated to the proof of \Cref{claim: proc reconnect}. We denote by $\aset$ the algorithm for the $r$-restricted \routeandcut problem, whose existence is assumed in the statement of the claim.

 The most natural way to design an algorithm for Procedure \reconnect is to simply cast it as a special case of the $r$-restricted \routeandcut problem, in the underlying graph $H$, with the set $A=U_{\lambda}$ of source vertices, and the set $B=\bigcup_{i'=0}^{i-1}U_{i'}$ of destination vertices. For technical reasons, we must allow every vertex $x\in \beta $ to serve as an endpoint of at least $\Theta\left(\frac{W\cdot \poly\log N}{|\beta|}\right )$ of the resulting routing paths, while each of the remaining vertices of $B$ may serve as an endpoint of at most one path. To accommodate this requirement, we slightly modify graph $H$. The modification essentially replaces every vertex $x\in \beta$ with $\Theta\left(\frac{W\cdot \poly\log N}{|\beta|}\right )$ distinct vertices, while ensuring that the resulting graph remains a uniform well-structured graph. The main difficulty with this approach, however, is that, for every vertex $x\in \beta$, and every layer $1\leq i'<i$, there may be up to $\rho$ vertices of $U_{i'}$ that are descendants of $x$. It is possible that we will connect a distinct vertex of $U_{\lambda}$ to each such descendant vertex of $x$ via the routing computed by the algorithm for the \routeandcut problem. The vertices of $U_{\lambda}$ from which these routing paths originate are then added to layer $U_i$, and, as the result, the number of descendants of $x$ in $U_i$ may become too high. In order to overcome this difficulty, we will iteratively solve ${4\lambda}$ instances of the \routeandcut problem, with the goal of constructing, for every vertex $v\in U_{\lambda}$, a collection $\rset(v)$ of at most ${4\lambda}$  paths that originate at $v$.
 Let $U'\subseteq U_{\lambda}$ be the subset of vertices $v$ with $|\rset(v)|\geq 2\lambda$, and let $\rset=\bigcup_{v\in U'}\rset(v)$. Our goal is to ensure that set $U'$ contains almost all vertices of $U_{\lambda}$, the paths in $\rset$ cause low edge-congestion, and they terminate at the vertices of $\bigcup_{i'=0}^{i-1}U_{i'}$. We also need to ensure that every vertex of $\beta$ serves as an endpoint to at most  $\Theta\left(\frac{W\cdot \poly\log N}{|\beta|}\right )$ paths in $\rset$, while every vertex of $\bigcup_{i'=1}^{i-1}U_{i'}$ serves as an endpoint of at most one such path. 
 Intuitively, by sending $\frac{1}{2\lambda}$ flow units on every path in $\rset$, we obtain a fractional flow where every vertex in $U'$ sends at least $1$ flow unit to the vertices of $\bigcup_{i'=0}^{i-1}U_{i'}$, and, for every vertex $x\in \beta$, the descendants of $x$, including $x$ itself, receive at most $\rho$ flow units. We can then compute an integral flow, which defines, for every vertex $v\in U'$, a path $K(v)$ connecting $v$ to some vertex in $\bigcup_{i'=0}^{i-1}U_{i'}$, such that, for every vertex $x\in \beta$, at most $\rho$ such paths terminate at the descendant vertices of $x$. %Since the paths in the resulting set $\set{K(v)\mid v\in U'}$ cause a relatively low edge-congestion, we will show that most of these paths are sufficiently short. The vertices $v\in U'$ for which $K(v)$ is sufficiently short are then added to set $U_i$, and their corresponding paths $K(v)$ are added to $\qset_i$.

We now provide a formal description of the procedure.
Recall that we have denoted $k=|U_{\lambda}|$. We assume first that $k\geq 2^{20}\cdot d'\cdot \log^6N$. We show an algoritm for the special case where $k<2^{20}\cdot d'\cdot \log^6N$ later.
  We start by describing the modified graph $H'$, to which the algorithm for the \routeandcut problem will be applied iteratively. 
  Throughout, we use a parameter $M=\frac{W\cdot \log^3N}{|\beta|}$.  Since  $\frac{W\cdot \log^3N}{|\beta|}\leq \frac{\rho}{16}$ from the problem definition, we get that $M\leq \frac{\rho}{16}$.

 \subsubsection{Modified Graph $H'$.}
 Recall that we are given as input a perfect well-structured graph $G$, and its corresponding subdivided graph $H$. Both graphs are given in the adjacency-list representation. We first define the graph $H'$, and then describe an algorithm to compute its adjacency-list representation efficiently.
 
Graph $H'$ is defined as follows.
We start with $H'=H$, and the bipartition $(L'',R'')$ of its vertices, that is identical to the bipartition $(L',R')$ of $V(H)$. We then inspect every vertex $x\in \beta$ one by one. If $x\in L''$, then we add the collection $\set{x_1,\ldots,x_{M-1}}$ of new vertices to graph $H'$, that we call \emph{copies} of $x$. The vertices in $\set{x_1,\ldots,x_{M-1}}$ are also added to $R''$, and, for all $1\leq j\leq M-1$, we add a regular edge $(x,x_j)$ to graph $H'$. If $x\in R''$, then we also add a collection  $\set{x_1,\ldots,x_{M-1}}$ of new vertices to $H'$, called copies of $x$, but now we add $x_1$ to $L''$, and $x_2,\ldots,x_{M'-1}$ to $R''$. We add a special edge $(x,x_1)$, and, for all $2\leq j\leq M-1$, we add a regular edge $(x_1,x_j)$ to $H'$. This concludes the definition of graph $H'$.

Note that $H'$ is a uniform well-structured graph. Since $M=\frac{W\cdot \log^3N}{|\beta|}$, and since $|V(H)|=W$, we get that $|V(H')|\leq O(W\log^3N)$, and moreover, $|E(H')\setminus E(H)|\leq O(W\log^3N)$. It is easy to modify the adjacency-list representation of $H$ to obtain the adjacency-list representation of $H'$, by inserting all vertices of $V(H')\setminus V(H)$ and all edges $E(H')\setminus E(H)$ into the former, in time $O(W\log^3N)$.
%, and, since $|V(X)|\leq \frac{8\hat W}{d'}$, we get that $|V(H')|\leq 2^{11}\cdot \hat W+|V(H)|\leq 2^{12}\cdot \hat W$. Recall that graph $H$ is given as input to Procedure \reconnect as an adjacency-list. We can augment this adjacency-list to create an adjacency-list representation of graph $H'$ in time $O(\hat W)$.  
 
\subsubsection{Parameters for the \routeandcut Problem Instances}
Our algorithm will perform $4\lambda$ iterations. In every iteration, we  will apply Algorithm $\aset$ for the $r$-restricted \routeandcut problem to graph $H'$, with suitably chosen sets $A'$ and $B'$ of vertices of $H'$. We now define the parameters that will be used in each such instance of the $\routeandcut$ problem, and verify that this setting of the parameter indeed defines a valid instance of the problem.

We start by letting $N'=64N\log^3N$. 
It is easy to see that: 

\begin{equation}\label{eq: bound on log n' new}
\log N'\leq \log N+6\log\log N\leq \left(1+\frac{1}{64c_2r}\right)\log N,
\end{equation}

since $r\leq \ceil{\sqrt{\log N}}$ and $N$ is large enough. Additionally, $\sqrt{\log N'}\leq \sqrt{\log N+6\log\log N}\leq \sqrt{\log N}+\sqrt{6\log\log N}$. Therefore:

\begin{equation}
2^{\sqrt{\log N'}}\leq 2^{\sqrt{\log N}}\cdot \log^6 N\leq 2^{2\sqrt{\log N}}.\label{eq: bound on N' new}
\end{equation}

We will use these facts later.

We also use parameters  $\Delta'=k$, and ${ \eta'=2^{20}\cdot d'\cdot \log^6N}$. In all instances of the \routeandcut problem that we construct, we will ensure that $U_{\lambda}\subseteq A'$, while set $B'$ contains all vertices of $U_0$ and their copies, and additionally, $|V(H')\setminus B'|\leq 4\lambda k$ holds. We show that all such resulting instances are valid instances of the $r$-restricted \routeandcut problem in the following simple claim, whose proof is deferred to Section \ref{subsec: proof of valid routeandcut} of Appendix.

\begin{claim}\label{claim: routeandcut instance}
Let $A',B'$ be any pair of disjoint subsets of $V(H')$ with $U_{\lambda}\subseteq A'$, and with all vertices of $U_0$ and their copies lying in $B'$, such that $|V(H')\setminus B'|\leq 4\lambda k$. Then  $(H',A',B',\Delta', \eta',N')$ is a valid instance of the $r$-restricted \routeandcut problem.
\end{claim}

We are now ready to describe our algorithm, that consists of at most $4\lambda$ iterations. The algorithm starts with a collection $\rset=\emptyset$ of paths, and in every iteration, at most $k$ paths are added to $\rset$. We now describe a single iteration.

\subsubsection{Description of an Iteration}
We describe the $j$th iteration, for $1\leq j\leq 4\lambda$.
We construct a set $B_j$ of vertices of $H'$ as follows. We add to $B_j$ all vertices of $\beta$ and their copies. Additionally, we add to $B_j$ all vertices of $U^{<i}\setminus U_0$, except for those that serve as endpoints of the paths that currently lie in $\rset$. Notice that set $V(H')\setminus B_j$ contains all vertices of $U_{\lambda}$, and, additionally, every vertex $v\in U^{<i}\setminus U_0$ that serves as an endpoint of some path that currently belongs to $\rset$. Since at most $k$ paths are added to $\rset$  in every iteration, and since the number of iterations is bounded by $4\lambda$, we get that $|V(H')\setminus B_j|\leq 4\lambda\cdot k$.
We also set $A_j=U_{\lambda}$, and we consider the resulting instance $(H',A_j,B_j,\Delta', \eta',N')$ of the \routeandcut problem. From \Cref{claim: routeandcut instance}, this is a valid $r$-restricted instance of the problem, and therefore we can apply Algorithm $\aset$ to this instance. 
Recall that Algorithm $\aset$ may return ``FAIL'', with probability at most $1/2$. In this case, we simply repeat Algorithm $\aset$, until either it does not return ``FAIL'', or until it returns ``FAIL'' in at least $\log^3N$ consecutive iterations. In the latter case, we terminate our algorithm and return ``FAIL'', and we say that iteration $j$ of the algorithm terminated with a `` FAIL''. Notice that the probability that iteration $j$ terminates with a ``FAIL'' is at most $2^{-\log^3N}$. We now assume that iteration $j$ did not terminate with a ``FAIL'', and we consider the output of the last execution of Algorithm $\aset$. Let $\rset'_j$ be the routing that the algorithm returned. We consider two cases, depending on whether $|\rset'_j|\geq \left(1-\frac{1}{64\lambda}\right )k$ or not.

\subsubsection*{Case 1: $|\rset'_j|\geq  \left(1-\frac{1}{64\lambda}\right )k$}

Recall that the paths in $\rset'_j$ cause congestion at most $4 \eta'\cdot \log N'$.
We can assume w.l.o.g. that, for every path $R\in \rset'_j$, if $x\in B_j$ is an inner vertex on path $R$, then $x$ also serves as an endpoint of another path in $\rset'_j$ (as otherwise we can truncate path $R$, so that it terminates at $x$). 

We say that a path $Q\in \rset'_j$ is \emph{long} if $Q$ contains at least $d'\cdot \log^{10}N$ special edges, and we say that it is \emph{short} otherwise. We use the next simple claim in order to show that $\rset'_j$ contains many short paths.

\begin{claim}\label{claim: few long paths}
	The number of long paths in $\rset'_j$ is at most $\frac{k}{64\lambda}$.
\end{claim}
\begin{proof}
Let $E'$ be the set of all special edges of $H'$ that participate in the paths in $\rset'_j$. We claim that $|E'|\leq 6\lambda k$. Indeed, we can partition $E'$ into two subsets: set $E'_1$ containing all special edges $e$, such that at least one endpoint of $e$ lies in $V(H')\setminus B_j$, and set $E'_2$ containing all remaining edges. Since $|V(H')\setminus B_j|\leq 4\lambda k$, we get that $|E'_1|\leq 4\lambda k$. From the above discussion, for every edge $e\in E'_2$, at least one endpoint of $e$ must serve as an endpoint of a path in $\rset'_j$, and so $|E'_2|\leq |\rset'_j|\leq k$. Therefore, overall, $|E'|\leq |E'_1|+|E'_2|\leq 6\lambda k$. Since the paths in $\rset'_j$ cause congestion at most $4\eta'\log N'$, we get that:
$\sum_{Q\in \rset'_j}|E(Q)\cap E'|\leq |E'|\cdot 4 \eta'\cdot \log N'\leq 24\eta'\lambda k\log N'\leq 512 k \eta'\log^2 N$, since $\lambda=\ceil{\log(64N)}+1\leq 8\log N$ and $\log N'\leq 2\log N$.

Assume for contradiction that there are more than  $\frac{k}{64\lambda}$ long paths in $\rset'_j$. Then:

\[\sum_{Q\in \rset'_j}|E(Q)\cap E'|\geq \frac{k}{64\lambda}\cdot d'\cdot \log^{10}N\geq \frac{k}{2^{9}}\cdot d'\cdot \log^9N> 512k\eta'\log^2 N \]

since $\lambda=\ceil{\log(64N)}+1\leq 8\log N$, $\eta'=2^{20}\cdot d'\cdot \log^6N$, and $N$ is sufficiently large. This is a contradiction since we have shown that $\sum_{Q\in \rset'_j}|E(Q)\cap E'|\leq  512 k \eta'\log^2 N$.
\end{proof}

We let $\rset_j\subseteq \rset'_j$ be the set of all short paths, so  $|\rset_j|\geq \left(1-\frac{1}{32\lambda}\right )k$. Since the paths in $\rset_j$ are short, and since regular and special edges alterante on each path, the length 
of each such path is at most $4d'\log^{10}N$. The paths in $\rset_j$ cause congestion at most $4 \eta'\cdot \log N'\leq d'\log^8N$, since 
$\log N'\leq 2\log N$ from Inequality \ref{eq: bound on log n' new}, $\eta'=2^{20}\cdot d'\cdot \log^6N$, and we have assumed that $N$ is large enough. 
We add the paths of $\rset_j$ to set $\rset$, and continue to the next iteration.

\subsubsection*{Case 2: $|\rset'_j|<\left(1-\frac{1}{64\lambda}\right )k$}

We now consider the second case, where $|\rset'_j|<\left(1-\frac{1}{64\lambda}\right )k$. Let $A'_j\subseteq A_j$ and $B'_j\subseteq B_j$ be the sets of vertices that do not serve as endpoints of the paths in $\rset'_j$, so $|A'_j|,|B'_j|\geq \frac{k}{64\lambda}$. Consider the cut $(\tilde Y,\tilde Y')$ in graph $H'$ that the algorithm for the \routeandcut problem returns, and recall that 
$A'_j\subseteq \tilde Y$, $B'_j\subseteq \tilde Y'$, and additionally:

\[|E_{H'}(\tilde Y,\tilde Y')|\leq \frac{64\Delta'}{ \eta'\log^4(|V(H')|)}+\frac{256|\rset'_j|}{\eta'}\leq \frac{512k}{ \eta'}\leq \frac{k}{1024d'\log^6N},\]

since $\eta'=2^{20}\cdot d'\cdot \log^6N$.

Since $A'_j\subseteq \tilde Y$, we get that $|\tilde Y\cap U_{\lambda}|\geq \frac{k}{64\lambda}$. Next, we show that $|\tilde Y|\leq O(k\log N)$, 
$|A'_j|\geq |\tilde Y\cap \beta|\cdot \log N$,
and that $\tilde Y'$ must contain at least $|\beta|\cdot\left(1-\frac{2}{\log^3N}\right )$ vertices of $\beta$, in the following simple observation, whose proof is deferred to Section \ref{subsec: proof of lots of terminals on one side} of Appendix.

\begin{observation}\label{obs: lots of terminals on one side}
$|\tilde Y|\leq O(k\log N)$, 	$|\tilde Y'\cap \beta|\geq \max\set{|\beta|\cdot\left(1-\frac{2}{\log^3N}\right ),|\beta|-\frac{2k}{\log^3N}}$, and $|A'_j|\geq |\tilde Y\cap \beta|\cdot \log N$.
\end{observation}

If Case 2 happens, then we transform the cut $(\tilde Y,\tilde Y')$ in $H'$ into either a weakly well-structured cut $(A',B')$ in $H$ with $|A'|\geq \frac{511|V(H)|}{512}$,  $|A'\cap \beta|<\frac{2|\beta|}{\log^3N}$, and $|E_H(A',B')|\leq  \frac{W}{1024d'\cdot \log^6N}$; or into a strongly well-structured $\frac{1}{4d'}$-sparse cut $(A,B)$ in $G$ with $w(A),w(B)\geq \frac{k}{2^{14}\log N}$, and $w(A)\geq 1.8|\beta\cap A|$. We do so in three steps. In the first step, we obtain a weakly well-structured cut $(X_1,X'_1)$ in $H$. If this cut has the  required properties,
namely that $|X_1|\geq \frac{511|V(H)|}{512}$,  $|X_1\cap \beta|<\frac{2|\beta|}{\log^3N}$, and $|E_H(X_1,X_1')|\leq  \frac{W}{1024d'\cdot \log^6N}$,
 then we return this cut and terminate the algorithm. Otherwise, in the second step, we compute a weakly well-structured cut $(X_2,X'_2)$ in $G$. Finally, in the third step, we compute the required strongly well-structured cut $(A,B)$ in $G$.

\paragraph{Step 1: a Weakly Well-Structured Cut in $H$.}

We start by defining a cut $(X_1,X'_1)$ in $H$ with  $X_1=\tilde Y\cap V(H)$ and $X'_1=\tilde Y'\cap V(H)$. Clearly, $|E_H(X_1,X'_1)|\leq |E_{H'}(\tilde Y,\tilde Y')|\leq \frac{k}{1024d'\log^6N}$. Moreover, $|X_1|\geq |A'_j| \geq \frac{k}{64\lambda}$. 
Furthermore, from \Cref{obs: lots of terminals on one side}, we get that $|X_1|\leq |\tilde Y|\leq O(k\log N)$; $|X_1|\ge |A'_j| 
\geq |\tilde Y\cap \beta|\cdot \log N=|X_1\cap \beta|\cdot \log N$; and $|X_1'\cap \beta|\geq |\tilde Y'\cap \beta|\geq \max\set{|\beta|\cdot\left(1-\frac{2}{\log^3N}\right ),|\beta|-\frac{2k}{\log^3N}}$.

Next, we turn the cut $(X_1,X_1')$ into a weakly well-structured one. In order to do so,  we consider every vertex $v\in X_1$ one by one. When vertex $v\in X_1$ is considered, we check whether there is a regular edge $(v,u)$ that is incident to $v$ and belongs to $E_H(X_1,X_1')$. If so, then $v\in L'$ must hold. We then we move $v$ from $X_1$ to $X'_1$. Notice that this move may introduce at most one new edge to the cut $E_H(X_1,X_1')$ -- the special edge that is incident to $v$, while at least one edge that is incident to $v$ is removed from $E_H(X_1,X_1')$. 
Therefore, throughout this transformation, $|E_H(X_1,X_1')|$ cannot grow, and $|X_1|$ decreases by at most $|E_H(X_1,X_1')|\leq \frac{k}{d'\cdot \log^6N}$. Since, at the beginning of Step 1, $|X_1|\geq \frac{k}{64\lambda}\geq \frac{k}{256\log N}$ held, at the end of this step, $|X_1|\geq \frac{k}{512\log N}$ holds, and 
moreover $|X_1|\geq \frac{|X_1\cap \beta|\cdot \log N}{2}$. Since we did not remove any vertices from $X'_1$, it is still the case that 
$|X_1'\cap \beta|\geq \max\set{|\beta|\cdot\left(1-\frac{2}{\log^3N}\right ),|\beta|-\frac{2k}{\log^3N}}$.
It is easy to verify that
every edge in $E_H(X_1,X_1')$ is a special edge, and so $(X_1,X_1')$ is a weakly well-structured in $H$.
Notice that the time required for this step is bounded by $O(\vol_H(X_1))\leq O(|X_1|\cdot W)\leq O(k\cdot W\cdot \log N)$.

We now consider two cases. The first case happens if $|X_1'|<\frac{k}{512}$. Then
we are guaranteed that: 

\[|E_H(X_1,X_1')|\leq  \frac{k}{1024d'\cdot \log^6N}\leq \frac{|V(H)|}{1024d'\cdot \log^6N}.\] %\leq \frac{W}{d'\cdot \log^6N}.\]

%\leq \frac{|T|}{32\log^6N}.\]

Additionally, we get that $|X_1'|\leq \frac{k}{512}\leq \frac{|V(H)|}{512}$, and
$|X_1'\cap \beta|\geq |\tilde Y'\cap \beta| \geq |\beta|\cdot\left(1-\frac{2}{\log^3N}\right )$. Therefore, $|X_1|\geq \frac{511|V(H)|}{512}$, and $|X_1\cap \beta|<\frac{2|\beta|}{\log^3N}$. We return cut $(X_1,X_1')$ in $H$ and terminate the algorithm.

From now on, we assume that $|X_1'|\geq \frac{k}{512}$. We denote $\phi=\frac{1}{d'}$.
Since $|E_H(X_1,X_1')|\leq \frac{k}{d' \log^6N}$, while $|X_1|\geq \frac{k}{512\log N}$, the sparsity of the cut $(X_1,X_1')$ in $H$ is at most $\frac{1}{d'\log^4N}=\frac{\phi}{\log^4N}$. Moreover, $(X_1,X_1')$ is a weakly well-structured cut in $H$ and $|X_1\cap \beta|<\frac{2k}{\log^3N}$.

\paragraph{Step 2: a Weakly Well-Structured Cut in $G$.}
Next, we turn the cut $(X_1,X_1')$ in $H$ into a 
cut $(X_2,X_2')$ in $G$, in a natural way: we let $X_2=X_1\cap V(G)$, and $X_2'=X_2\cap V(G)$.
Recall that the length of every edge in $G$ is at most $d'\leq \frac{1}{\phi}$. From \Cref{obs: subdivided cut sparsity}, $(X_2,X_2')$ is a weakly well-structured cut in $G$ of sparsity at most $\frac{8\phi}{\log^4N}$, with $w(X_2)\geq \frac{|X_1|}{8}\geq \frac{k}{2^{12}\log N}$ and $w(X_2')\geq \frac{|X_1'|}8\geq \frac{k}{2^{12}\log N}$. The time required to compute the cut $(X_2,X_2')$ in $G$ given the cut $(X_1,X_1')$ in $H$ is bounded by $O(|V(H)|)\leq O( W)$. Clearly, $|X_2\cap \beta|\leq |X_1\cap \beta|<\frac{2k}{\log^3N}$.

\paragraph{Step 3: a Strongly Well-Structured Cut in $G$.}
Recall that for every special edge $e\in E(G)$, $\ell(e)\leq d'\leq \frac{n}{8}$ (from the definition of the input to Procedure \reconnect), while $\sum_{e\in E(G)}\ell(e)\geq  \frac{n}{2}$. 
 We can now apply the algorithm from \Cref{obs: weakly to strongly well str} to cut $(X_2,X_2')$ compute a strongly well-structured cut $(A,B)$ in $G$ with $w(A)\geq \frac{w(X_2)}{4}\geq \frac{k}{2^{14}\log N}$ and $w(B)\geq \frac{w(X_2')}{4}\geq \frac{k}{2^{14}\log N}$, whose sparsity is at most $\frac{16\phi}{\log^4N}\leq \frac{1}{4d'}$. 
 
 We now show that $w(A)\geq 1.8|\beta\cap A|$. 
 Recall that the algorithm from \Cref{obs: weakly to strongly well str} starts with the cut $(A,B)=(X_2,X_2')$, and considers every special edge $e=(u,v)\in E_G(B,A)$ in turn. For each such edge, it either moves $u$ from $B$ to $A$, or it moves $v$ from $A$ to $B$.
 We partition the set $A\cap \beta$ of vertices into three subsets: set $\beta_1$ contains all vertices $v\in A\cap \beta$ that lied in $X_2$; recall that $|\beta_1|\leq |X_2\cap \beta|\leq \frac{2k}{\log^3N}$. 
 We now define the remaining two sets $\beta_2$ and $\beta_3$. 
 Each vertex $v\in (\beta\cap A)\setminus \beta_1$ was added to set $A$ by the algorithm from 
  \Cref{obs: weakly to strongly well str}. Therefore, for each such vertex $v$, there is a special edge $(v,u)\in E_G(X_2',X_2)$. If $u\in \beta$ holds (in which case $u\in \beta_1$ holds), we add $v$ to $\beta_2$, and otherwise we add it to $\beta_3$. It is easy to verify that $|\beta_2|\leq |\beta_1|\leq \frac{2k}{\log^3N}$. Consider now a vertex $v\in \beta_3$, and the corresponding special edge $e=(v,u)$. Since $w(v),w(u)\geq 1$, it is easy to see that $|\beta_3|\leq w(A)/2$. Overall, we get that:
  
  \[|\beta\cap A|=|\beta_1|+|\beta_2|+|\beta_3|\leq \frac{4k}{\log^3N}+\frac{w(A)}2\leq  0.51w(A). \]
  
  Therefore, $w(A)\geq 1.8|\beta\cap A|$, as required.

 The running time of the algorithm from \Cref{obs: weakly to strongly well str}  is bounded by $O(\vol_G(X_2))\leq O(\vol_H(X_1))\leq O(|X_1|\cdot |V(H)|)\leq O(k\cdot W\cdot \log N)$.
 %Recall that the algorithm from \Cref{obs: weakly to strongly well str} starts with the cut $(A,B)=(X_2,x_2')$, and then considers every special edge  
We return the cut $(A,B)$ and terminate the algorithm.

Overall, the total time we have spent on iteration $j$ when Case 2 happens, excluding the time required to execute Algorithm $\aset$, is bounded by $O(k\cdot  W\cdot \log N)$. 
This concludes the description and the analysis of a single iteration.

\subsubsection{Constructing a New \LCDS.}

 If any iteration terminated with a ``FAIL'', then the algorithm terminates with a ``FAIL''. Since the number of iterations is bounded by $4\lambda\leq 32\log N$, and the probability that a single iteration terminates with a ``FAIL'' is bounded by $2^{-\log^3 N}$, we get that the overall probability that our algorithm implementing procedure \reconnect terminates with a ``FAIL'' is bounded by $\frac{1}{N^{10}}$. We assume from now on that no iteration terminated with a `` FAIL''.
If, in any iteration, Case 2 occurred, then the algorithm terminated with a cut.

We assume from now on that neither of the $4\lambda$ iterations terminated with a ``FAIL'', and that, in each iteration, Case 1 occurred. Therefore, for all $1\leq j\leq 4\lambda$, we have computed a collection
$\rset_j$ of paths, with  $|\rset_j|\geq \left(1-\frac{1}{32\lambda}\right )k$, such that the length of every path in $\rset_j$ is at most $4d'\log^{10}N$. 
Let $\rset=\bigcup_{j=1}^{4\lambda}\rset_j$. Since, for all $1\leq j\leq 4\lambda$, the paths in $\rset_j$ cause congestion at most $d'\log^8N$, and $\lambda=\ceil{\log(64N)}+1\leq 4\log N$, we get that the paths in $\rset$ cause congestion at most $4d'\log^9N$.
For every vertex $v\in U_{\lambda}$, let $\rset(v)\subseteq \rset$ be the set of paths that originate at $v$. Clearly, $|\rset(v)|\leq 4\lambda$ holds.

We say that a vertex $v\in U_{\lambda}$ is \emph{good} if $|\rset(v)|\geq 2\lambda$, and we say that it is bad otherwise. We claim that at most $\frac{k}{4\lambda}$ vertices of $U_{\lambda}$ may be bad. Indeed, let $\Pi$ be the collection of all pairs $(v,j)$, where $v\in U_{\lambda}$ and $1\leq j\leq 4\lambda$, such that no path of $\rset_j$ originates at $v$. Since, for all $1\leq j\leq 4\lambda$, $|\rset_j|\geq \left(1-\frac{1}{32\lambda}\right )k$, we get that $|\Pi|\leq \frac{k}{32\lambda}\cdot 4\lambda=\frac{k}{8}$. However, if more than $\frac{k}{4\lambda}$ vertices of $U_{\lambda}$ are bad, then $|\Pi|>\frac{k}{4\lambda}\cdot 2\lambda\geq \frac{k}{2}$, a contradiction.
We let $U'\subseteq U_{\lambda}$ be the set of all good vertices, so $|U'|\geq \left(1-\frac{1}{4\lambda}\right )k$.

Let $\tilde \rset\subseteq \rset$ be a set of path that contains, for every vertex $v\in U'$, exactly $2\lambda$ paths that originate at $v$. From our construction, the paths in $\tilde \rset$ have length at most $4d'\log^{10}N$ each, and they cause congestion at most $4d'\log^9N$. Additionally, every vertex $x\in \beta$, the number of paths in $\rset$ that terminate at $x$ is bounded by: 

\[(4\lambda)\cdot M\leq \frac{4\lambda\cdot  W\log^3N}{|\beta|}\leq \frac{\lambda\rho}{4}, \]

since $M=\frac{W\cdot \log^3N}{|\beta|}$ and $\frac{W\cdot \log^3N}{|\beta|}\leq \frac{\rho}{16}$ from the problem definition.
Our algorithm also ensures that vertex $y\in U^{<i}\setminus \beta$ may serve as an endpoint of at most one such path. Consider now some vertex $x\in \beta$, and let $\dset^{<i}(x)$ be the set of all descendant vertices of $x$ that belong to layers $U_1,\ldots,U_{i-1}$, including vertex $x$ itself. From Property \ref{prop: descendants} of \LCDS, $|\dset^{<i}(x)|\leq i\cdot \rho$. 
Let $\tilde \rset^x\subseteq \tilde \rset$ denote the collection of all paths $Q\in \tilde \rset$, such that $Q$ terminates either at one of the vertices in $\dset^{<i}(x)$. Then $|\tilde \rset^x|\leq \frac{\lambda\rho}{4}+\lambda\cdot \rho\leq 1.5\lambda\cdot\rho$.

Notice that for every path $Q\in \tilde \rset$, if $v\in U^{<i}$ is the last endpoint of $Q$, then we can compute a vertex $x\in \beta$ with $v\in \dset^{<i}(x)$ in time $O(1)$, using the current \LCDS. We compute the collection $\beta'\subseteq \beta$ of vertices, where a vertex $x\in \beta'$ if and only if some path in $\tilde R$ terminates at a vertex of $\dset^{<i}(x)$. This set of vertices can be computed in time $O(|\tilde R|)\leq O(\lambda k)\leq O(k \log N)$.

Let $V'\subseteq U^{<i}$ be the set of vertices that serve as endpoints of the paths in $\tilde \rset$. From our construction, every vertex $v\in V'\setminus \beta$ serves as an endpoint of exactly one such path.
Next, we construct an auxiliary directed flow network $\tilde H$. We start by letting $V(\tilde H)=U'\cup V'\cup \beta'$. For every path $Q\in \tilde R$, if $v$ is the first endpoint of $Q$, and $v'$ is the last endpoint of $Q$ then we add a directed edge $e=(v,x)$ to the flow network, whose capacity is $1$. We say that edge $e$ \emph{corresponds to the path $Q$}.
For every vertex $u\in V'$, if $x\in \beta'$ is the ancestor of $v$, then we add an edge $(u,x)$ of capacity $1$ to the graph.
 Additionally, we add a source vertex $s$, and connect it to every vertex in $U'$ with an edge of capacity $1$, and a destination vertex $t$. Every vertex $x\in \beta'$ connects to $t$ with an edge of capacity $\rho$. This concludes the description of the directed flow network $\tilde H$. We claim that the value of the maximum $s$-$t$ flow in this network is $|U'|$. Indeed, since the total capacity of all edges incident to $s$ is $|U'|$, the value of the maximum $s$-$t$ flow cannot be greater than $|U'|$. Consider now the following fractional flow in $\tilde H$: we send one flow unit on every edge incident to $s$, and, for every path $Q\in \tilde \rset$, we send $\frac{1}{2\lambda}$ flow units on its corresponding edge $(v,u)\in E(\tilde H)$, and the same amount  of flow on edge $(u,x)$, where $x\in \beta'$ is the ancestor of $u$. Since, for every vertex $v\in U'$, set $\tilde \rset$ contains exactly $2\lambda$ paths originating at $v$, the total amount of flow leaving $v$ is exactly $1$. 
 Since every vertex $u\in V'$ is an endpoint of at most one path in $\tilde \rset$, if $x\in \beta'$ is the ancestor of $u$, then edge $(u,x)$ carries at most $1$ flow unit.
 
 Consider now some vertex $x\in \beta'$, and recall that $|\tilde \rset^x|\leq 1.5\lambda \cdot \rho$. Since, for every path $Q\in \tilde \rset^x$, the corresponding edge carries $\frac{1}{2\lambda}$ flow units, the total amount of flow entering $x$ is at most $\frac{1.5\lambda \cdot \rho}{2\lambda}\leq \rho$. We set the flow on the edge $(x,t)$ to be equal to the total amount of flow entering $x$. Overall, we obtain a valid $s$-$t$ flow in $\tilde H$, whose value is $|U'|$.

Next, we  employ the standard Ford-Fulkerson algorithm in order to compute an integral flow of value $|U'|$ in $\tilde H$. We start by setting the flow on every edge of $\tilde H$ to $0$, and then perform at most $|U'|\leq k$ iterations. In every iteration, we compute a residual flow network, along with an augmenting path, and then modify the current flow accordingly. Clearly, every iteration can be executed in time $O(|E(\tilde H)|)\leq O(k\log N)$, while the number of iterations is bounded by $k$. Altogether, the running time of the algorithm that computes the flow network $\tilde H$, and computes a maximum integral $s$-$t$ flow in it is bounded by $O(k^2\log N)\leq O(k\cdot W\log N)$.

We are now ready to construct a new \LCDS. The sets $U_0,\ldots, U_{i-1}$ of vertices, and the sets $\qset_1,\ldots,\qset_{i-1}$ of paths remain unchanged.
For all $i<i'<\lambda$, we set $U_{i'}=\emptyset$.
 We also set $U_i=U'$ and $U_{\lambda}=V(H)\setminus\left(\bigcup_{i'=0}^iU_{i'}\right )$. Since $|U'|\geq \left(1-\frac{1}{4\lambda}\right )k$, we get that $|U_{\lambda}|\leq \frac{k}{4\lambda}$.
For every vertex $v\in U'$, we consider the unique edge $e$ of $\tilde H$ that leaves the vertex $v$ and carries one flow unit, and we let $Q\in \tilde \rset$ be the path corresponding to $e$. We then set $K(v)=Q$, and we let $\qset_i=\set{K(v)\mid v\in U_i}$. Since, for every vertex $x\in \beta'$, the capacity of the edge $(x,t)$ in $\tilde H$ is $\rho$, this guarantees that the number of descendant-vertices of $x$ in $U_i$ is bounded by $\rho$, as required. 
For every vertex $u\in V'$, if $x\in \beta'$ is its ancestor, then the capacity of the edge $(u,x)$ is $1$, so $u$ may serve as an endpoint in at most $1$ path in $\qset_i$.
Additionally, our construction guarantees that the paths in $\qset_i$ have length at most $4d'\log^9N$ each, and they cause edge-congestion at most $4d'\log^9N$ in $H$, as required. Therefore, we obtain a valid \LCDS.

\subsubsection{Analysis of the Running Time}
Recall that the time required to construct the adjacency-list representation of graph $H'$ from the adjacency-list of $H$ is bounded by $O( W\log^3N)$. The algorithm consists of $4\lambda=O(\log N)$ iterations. In every iteration, we apply Algorithm $\aset$ for the $r$-restricted \routeandcut problem to instance $(H',A_j,B_j,\Delta', \eta',N')$ at most $O(\log^3N)$ times. 
Recall that $|V(H')|\leq O(W\cdot \log^3N)$, and $|V(H')\setminus B_j|\leq 4\lambda k\leq O(k\log N)$. Therefore,
the running time of each execution of Algorithm $\aset$ is bounded by:

\[\begin{split} 
&c_1\cdot |V(H')|\cdot (|V(H')|-|B_j|)\cdot 2^{c_2\sqrt{\log N'}}\cdot (\log N')^{16c_2(r-1)+8c_2}\\
&\quad\quad\quad\quad\quad\quad\quad\quad \leq O\left(c_1\cdot W\cdot k\cdot  2^{c_2\sqrt{\log N'}}\cdot (\log N')^{16c_2(r-1)+8c_2+4}\right )\\
&\quad\quad\quad\quad\quad\quad\quad\quad \leq O\left(c_1\cdot W\cdot k\cdot 2^{c_2\sqrt{\log N}}\cdot (\log N)^{16c_2(r-1)+8c_2+10}\cdot \left(1+\frac{1}{64c_2r}\right )^{16c_2(r-1)+8c_2+4}\right )\\
&\quad\quad\quad\quad\quad\quad\quad\quad \leq O\left(c_1\cdot W\cdot k\cdot 2^{c_2\sqrt{\log N}}\cdot (\log N)^{16c_2(r-1)+8c_2+10}\right )
\end{split}
\]

(We have used the fact that, from Inequality \ref{eq: bound on N' new},
$2^{\sqrt{\log N'}}\leq 2^{\sqrt{\log N}}\cdot \log^6 N$, and from Inequality \ref{eq: bound on log n' new}, $\log N'\leq \left(1+\frac{1}{64c_2r}\right )\cdot \log N$. 

The time spent on executing Algorithm $\aset$ in all $4\lambda=O(\log N)$ iterations is then bounded by:

\[  O\left(c_1\cdot W\cdot k\cdot 2^{c_2\sqrt{\log N}}\cdot (\log N)^{16c_2(r-1)+8c_2+14}\right ).\]

If the algorithm terminates with a cut, then we may need to invest additional $O(k\cdot  W\cdot \log N)$ time in order to convert the cut returned by Algorithm $\aset$ into a desired well structured cut in $H$ or in $G$. Lastly, if the algorithm terminates with a new \LCDS, then the additional time that is required in order to compute the flow network $\tilde H$, the maximum integral $s$-$t$ flow in it, and then compute the corresponding set $\qset_i$ of paths, is bounded by $O(k\cdot W\cdot \log N)$. Overall, we get that the running time of the algorithm is bounded by:

\[  O\left(c_1\cdot W\cdot k\cdot 2^{c_2\sqrt{\log N}}\cdot (\log N)^{16c_2(r-1)+8c_2+14}\right ).\]

\subsubsection{Special Case: $k<2^{20}\cdot d'\cdot \log^6N$.}

It now remains to consider a special case where $k<2^{20}\cdot d'\cdot \log^6N$.
In this case, we employ the following simple algorithm. We start with $U'=\emptyset$, $\rset=\emptyset$, and $T=V(H)\setminus U_{\lambda}$. 
In every iteration, we will add a single path to $\rset$, connecting some vertex in $U_{\lambda}$ to a vertex of $T$.
For every vertex $x\in \beta$, we maintain a counter $n_x$, that counts the number of paths in $\rset$ that terminate at the vertices of $\bigcup_{i'<i}\dset_{i'}(x)$. Once this number reaches $\rho$, we remove vertex $x$, and all vertices in $\bigcup_{i'<i}\dset_{i'}(x)$ from $T$. We also delete from $T$ any vertex that serves as an endpoint of a path in $\rset$, except for the vertices that lie in $\beta$. We initialize $n_x=0$ for all $x\in \beta$. We will ensure that $|\rset|\leq k<2^{20}\cdot d'\cdot \log^6N$. This ensures that the total number of vertices $x\in \beta$ for which $n_x=\rho$ holds at the end of the algorithm is bounded by $\frac{k}{\rho}$. From Requirement \ref{prop: descendants} of the \LCDS, for each such vertex $x\in \beta$, $|\dset(x)|\leq \lambda\cdot \rho$, and so the total number of vertices that are ever deleted from $T$ over the course of the algorithm is bounded by $\frac{k}{\rho}\cdot \lambda\cdot \rho+k\leq 2\lambda\cdot k\leq 8k\log N$. Therefore, throughout the algorithm, $|V(H)\setminus T|\leq k+8k\log N\leq 9k\log N\leq 2^{24}\cdot d'\cdot \log^7N$ holds.

We consider every vertex $v\in U_{\lambda}$ one by one. When vertex $v$ is considered, we either compute a path $K(v)$ connecting $v$ to a vertex of $T$, such that no inner vertices on $K(v)$ lie in $T$;  or we compute a set $S_v\subseteq V(H)\setminus T$ of vertices with no outgoing edges, that contains $v$. In the former case, the length of $K(v)$ is at guaranteed to be at most $|V(H)\setminus T|\leq 2^{14}\cdot d'\cdot \log^7N\leq 4d'\log^{10}N$. We add vertex $v$ to set $U'$, add path $K(v)$ to set $\rset$, and, if $u\in T$ is the other endpoint of path $K(v)$, and $x\in \beta$ is the vertex for which $u\in \dset(x)$, we increase the counter $n_x$, and, if $u\not\in \beta$, delete $u$ from $T$. If $n_x$ reaches $\rho$, then we delete from $T$ all vertices of $\bigcup_{i'<i}\dset_{i'}(x)$. We then continue to the next iteration. 
In the latter case, if we obtain a set $S_v\subseteq V(H)\setminus T$ of vertices with no outgoing edges with $v\in S_v$, we just continue to the next iteration. 

If, at the end of this procedure, $|U'|\geq \left(1-\frac{1}{4\lambda}\right )k$, then we obtain a new \LCDS by keeping the layers $U_0,\ldots,U_{i-1}$ unchanged, together with the corresponding collections $\qset_1,\ldots,\qset_{i-1}$ of paths, setting $U_i=U'$ and $\qset_i=\rset=\set{K(v)\mid v\in U'}$. 
Note that the paths in $\qset_i$ cause congestion at most $k\leq 2^{20}\cdot d'\cdot \log^6N\leq 4d'\log^9N$.
We also let $U_{\lambda}$ contain all vertices that originally lied in $U_{\lambda}$ and were not added to $U'$, and, for all $i<i'<\lambda$, we set $U_{i'}=\emptyset$.  
From our construction, for every vertex $x\in \beta$, $|\dset_i(x)|\leq \rho$ holds, as required, and every vertex $v\in U^{<i}$ serves as endpoint of at most one path in $\qset_i$.
It is immediate to verify that we obtain a valid \LCDS, with $|U_{\lambda}|\leq \frac{k}{4\lambda}$.

Assume now that, at the end of the algorithm, $|U'|<\left(1-\frac{1}{4\lambda}\right )k$ holds. 
Consider the cut $(\tilde Y,\tilde Y')$ in $H$, where $\tilde Y=\bigcup_{v\in U_{\lambda}\setminus U'}S_v$ and $\tilde Y'=V(H)\setminus \tilde Y$. Clearly, $E_H(\tilde Y,\tilde Y')=\emptyset$. Moreover, from our construction, $|\tilde Y|\geq |U_{\lambda}\setminus U'|\geq \frac{k}{4\lambda}\geq\frac{k}{512\log N}$. 
At the same time, $\tilde Y\subseteq V(H)\setminus T$, so $|\tilde Y|\leq |V(H)\setminus T|\leq  2^{24}\cdot d'\cdot \log^7N$.
Since the total number of vertices $x\in \beta$ with $n_x=\rho$ at the end of the algorithm is bounded by $\frac{k}{\rho}\leq \frac{W}{\rho}$, while, from the problem definition, $\frac{W\cdot \log^3N}{|\beta|}\leq \frac{\rho}{16}$, we get that 
$W\leq\frac{|\beta|\cdot \rho}{16\log^3N}$, and $|\tilde Y'\cap \beta|\geq |\beta|-\frac{W}{\rho}\geq |\beta|\left(1-\frac{2}{\log^3N}\right )$.

Since $E_H(\tilde Y,\tilde Y')=\emptyset$, cut $(\tilde Y,\tilde Y')$ is weakly well-structured in $H$, and it is $\frac{1}{d'\log^4N}$-sparse. If $|\tilde Y'|<\frac{k}{512}$, 
then we return cut $(\tilde Y,\tilde Y')$ in $H$, that is guaranteed to have all required properties. Otherwise, we transform cut $(\tilde Y,\tilde Y')$ into a a strongly well-structured $\frac{1}{4d'}$-sparse cut $(A,B)$ in $G$ with $w(A),w(B)\geq \frac{k}{2^{14}\log N}$ and $w(A)\geq 1.8|\beta\cap A|$, exactly as before, and output this cut as the outcome of the algorithm.

We now provide an algorithm to efficiently implement the above procedure. Throughout the algorithm, we denote $Z=V(H)\setminus T$. Each vertex $v\in Z$ is classified as either \emph{marked} or \emph{unmarked}; a vertex $v\in Z$ is marked only if we have correctly established that no path connects $v$ to a vertex of $T$. Notice that vertices may leave $T$, but they cannot join $T$. Therefore, if, at some time $\tau$ during the algorithm, no path conencts $v$ to vertices of $T$, then this remains so for the remainder of the algortihm. For every vertex $v\in Z$ that is unmarked, we maintain a heap $h(v)$ that contains all unmarked vertices $u\in Z$ with edge $(v,u)\in E(H)$, and also all vertices $u'\in T$ with edge $(v,u')\in E(H)$.

At the beginning of the algorithm, $Z=U_{\lambda}$, and all vertices in $Z$ are unmarked. Notice that the time that is required in order to construct the initial heaps $h(v)$ for vertices $v\in Z$ is bounded by $O(|U_{\lambda}|\cdot |V(H)|)\leq O(k\cdot W)$. Next, we perform iterations. In every iteration, we select an arbitrary vertex $v\in Z$, and perform a DFS, starting from $v$, in the graph that is obtained from $H$ by deleting the marked vertices, until a vertex of $T$ is reached. In other words, when the DFS search arrives at some vertex $z$, it only explores edges $(z,u)$ with $u\in h(z)$. Whenever the DFS search backtracks from a vertex $a$, we know that no path connects $a$ to a vertex of $T$, so we mark $a$. We then examine all edges $(b,a)\in E(H)$ with $b\in Z$, and delete $a$ from the corresponding heap $h(b)$. The time that the DFS spent on exploring vertex $a$, on marking it, and on deleting it from the heaps of vertices $b\in Z$ with $(b,a)\in E(H)$ is \emph{charged} to vertex $a$ itself. Excluding the time charged to the newly marked vertices, the DFS search from a single vertex $v$ takes time $O(|Z|)=O(k\log N)$. If, as the result of the search, we arrive at a vertex of $T$, then we obtain the desired path $K(v)$. Otherwise, vertex $v$ becomes marked as well, and we process it like all other marked vertices. Notice that the number of such DFS searches that our algorithm performs is bounded by $|U_{\lambda}|\leq k$, and so the total time spent on the DFS searches, excluding the time charged to the marked vertices, is bounded by $O(k^2\log N)\leq O(k\cdot W\cdot \log N)$.

Recall that, throughout the algorithm, $|Z|=|V(H)\setminus T|\leq O(k\log N)$ holds. Every vertex in $Z$ may be marked at most once. When a vertex $a\in Z$ is marked, we need to consider all vertices $b\in Z$ with $(b,a)\in E(H)$, and delete $a$ from the heap $h(b)$. Since $|Z|\leq O(k\log N)$, we spend $O(k\log N)$ to process each newly marked vertex, and $O(k^2\log^2N)\leq O(k\cdot W\log^2N)$ time on all marked vertices throughout the algorithm.

%Lastly, whenever a vertex $v\in T$ is deleted from $T$, we need to add it to the heaps $h(u)$ of all vertices $u\in Z$ with $(u,v)\in E(H)$. Since $|Z|\leq O(k\log^2N)$, and since the total number of vertices deleted from $T$ over the course of the algorithm is at most $O(k\log^2N)$, the algorithm spends $O(k^2\log^4N)\leq O(k\cdot W\cdot \log^4N)$ time to process all vertices that are ever deleted from $T$. 

Once every vertex of $U_{\lambda}$ is either added to $U'$ or is marked, the algorithm terminates. The set $\tilde Y$ of vertices can then be constructed by simply including all marked vertices of $Z$ in it. From the above discussion, the total running time of the algorithm for this special case is $O(k\cdot W\cdot \log^2N)$.
This completes the proof of \Cref{claim: proc reconnect}.

%% file: connecttocenters-problem.tex
\subsection{\connecttocenters Problem}

Suppose we are given a perfect well-structured graph $G$ with a proper assignment $\ell(e)$ of lengths to its edges, together with a parameter $d'\geq 1$, and a subset $\beta\subseteq V(G)$ of vertices of $G$.
Assume further that $G$ undergoes an online sequence of updates of the following two kinds: (i) double the length of a special edge $e\in E(G)$; and (ii) delete a vertex $v\in \beta$ from the set $\beta$.
 Intuitively, the purpose of the \connecttocenters problem is to maintain the \LCDS in $G$  with $U_{\lambda}=\emptyset$, except that we slightly relax Requirement \ref{prop: routing paths}, and allow the lengths of the paths in $\bigcup_{i=1}^{\lambda-1}\qset_i$ to be bounded by $16d'\log^{10}N$, instead of $4d'\log^{10}N$. 
 The parameter $\rho$ corresponding to the data structure is $\rho=\frac{2^{12}\cdot W^0\cdot \log^3N}{|\beta|}$, where $W^0$ is the initial total weight of all vertices of $G$.
 The algorithm is allowed, at any time, to return a strongly well-structured $\frac{1}{4d'}$-sparse cut $(A,B)$ in $G$ with $w(A)\geq 1.8|\beta\cap A|$. Following that, if $w(A)\leq w(B)$, then the vertices of $A$, and the endpoints of the edges in $E_G(A,B)$ are deleted from $G$, while otherwise the vertices of $B$, and the endpoints of the edges in $E_G(A,B)$ are deleted from $G$. In either case, the algorithm must continue with the resulting graph. The algorithm is also allowed, at any time, to return a weakly well-structured cut $(A',B')$ in graph $H$ -- the subdivided graph corresponding to $G$ -- with $|A'|\geq \frac{511|V(H)|}{512}$,  $|A'\cap \beta|<\frac{2|\beta|}{\log^3N}$, and $|E_H(A',B')|\leq  \frac{|V(H)|}{1024 d'\cdot \log^6N}$, after which the algorithm terminates. It is also allowed, at any time, to terminate with a ``FAIL'', but we require that the probability that this ever happens is bounded by $1/N^7$. We model this algorithm as consisting of a number of phases, where every phase consists of a number of iterations. In every iteration, the algorithm must either (i) produce a valid \LCDS for the current graph $G$ with $U_{\lambda}=\emptyset$ (with the relaxed Requirement \ref{prop: routing paths}), after which the phase terminates; or (ii) produce  a strongly well-structured $\frac{1}{4d'}$-sparse cut $(A,B)$ in $G$, after which $G$ is updated as described above, and some vertices may be deleted from $\beta$; or (iii) produce  a weakly well-structured cut $(A',B')$ in graph $H$, with $|A'|\geq \frac{511|V(H)|}{512}$,  $|A'\cap \beta|<\frac{2|\beta|}{\log^3N}$, and $|E_H(A',B')|\leq  \frac{|V(H)|}{1024d'\cdot \log^6N}$, after which the algorithm terminates. At the end of each phase, the algorithm is given a collection of special edges of $G$ whose lengths are doubled, and a collection of vertices that are deleted from $\beta$, and then a new phase begins. The algorithm only needs to continue as long as $|\beta|$ remains sufficiently large (it most contain at least $1/4$ of the vertices from the original set $\beta$), and as long as the total increase in the lengths of the edges is not too large. We now define the total increase of edge lengths formally, followed by a formal defintion of the \connecttocenters problem.

 \begin{definition}[Total increase in edge lengths]
 	Let $G$ be a perfect well-structured graph with a proper assignment $\ell(e)$ of lengths to its edges, that undergoes an online sequence of updates, consisting of the following operations: (i) double the length of a special edge $e\in E(G)$; and (ii) delete a vertex from $V(G)$, together with its incident edges. Consider any time $\tau$ during the update sequence, and denote by $E^0$ the set of special edges in the initial graph $G$. For each edge $e\in E^0$, let $\ell^0(e)$ be the length of $e$ at the beginning of the algorithm, and let $\ell^{\tau}(e)$ be the length of $e$ at time $\tau$ if $e\in E(G)$ at time $\tau$; or  the length of $e$ just before it was deleted from $\tau$ otherwise. Then the total increase in edge lengths at time $\tau$ is: $\sum_{e\in E^0}\left(\ell^{\tau}(e)-\ell^0(e)\right )$.
 \end{definition}

We are now ready to define the \connecttocenters problem.

\begin{definition}[\connecttocenters problem]
The input to the \connecttocenters problem consists of a perfect well-structured graph $G=(L,R,E)$ with $|V(G)|=n$, and a proper assignment $\ell(e)$ of lengths to the edges $e\in E$, that are given in the adjacency list representation, together with a subdivided graph $H$ corresponding to $G$ in the adjacency-list representation. Let $W^0$ denote the weight of all vertices in this initial graph $G$. Additionally, we are given  a parameter $N\geq \max\set{n,W^0/8}$ that is greater than a large enough constant, and parameters $d'\geq 1$,  and $1\leq r\leq \ceil{\sqrt{\log N}}$, such that $d'\leq\min\set{\frac{2^{r\sqrt{\log N}}}{\log^8N},\frac{n}{8}}$. Lastly, we are given a set $\beta\subseteq V(H)$ of vertices, with $|\beta|\geq \frac{W^0}{256d'}$. We denote by $\beta^0$ the initial set $\beta$.

The algorithm consists of a number of phases, and every phase is partitioned into iterations. In every iteration, the algorithm is required to return one of the following:

\begin{properties}{O}
	\item either a valid \LCDS in the current graph $G$ with the parameter  $\rho=\frac{2^{12}\cdot W^0\cdot \log^3N}{|\beta^0|}$, in which Requirement \ref{prop: routing paths} is relaxed to allow the lengths of the paths in $\bigcup_{i=1}^{\lambda-1}\qset_i$ to be bounded by $16d'\log^{10}N$ and $U_{\lambda}=\emptyset$; \label{outcome: LCDS}
	
	\item or a strongly well-structured $\frac{1}{d'}$-sparse cut $(A,B)$ in $G$ with $w(A)\geq 1.8|\beta\cap A|$;\label{outcome: sparse cut}
	
	\item or  a weakly well-structured cut $(A',B')$ in graph $H$, with $|A'|\geq \frac{511|V(H)|}{512}$,  $|A'\cap \beta|<\frac{2|\beta|}{\log^3N}$, and $|E_H(A',B')|\leq  \frac{|V(H)|}{1024d'\cdot \log^6N}$. \label{outcome: finalizing cut}
\end{properties}

In the case of Outcome \ref{outcome: finalizing cut}, the algorithm terminates. In the case of Outcome \ref{outcome: sparse cut}, the endpoints of the edges in $E_G(A,B)$ are deleted from $G$. Additionally, if  $w(A)\leq w(B)$, then the vertices of $A$ are deleted from $G$, and otherwise the vertices of $B$ are deleted from $G$. All the vertices of $\beta$ that are deleted from $G$ are also deleted from $\beta$, along with some additional vertices that may be deleted from $\beta$ (given by the adversary). In case of Outcome \ref{outcome: LCDS}, the current phase terminates, and the algorithm is given a collection of special edges of $G$, whose lengths are doubled, together with a list of vertices that are deleted from $\beta$.

Once the total increase in the edge lengths becomes greater than $8n$, or at least $3/4$ of vertices in the initial set $\beta$ are deleted from $\beta$, the algorithm terminates. The algorithm may also return ``FAIL'' and terminate at any time, but we require that the probability that this ever happens is bounded by $\frac{1}{N^7}$.
\end{definition}

Next, we provide an algorithm for the \connecttocenters problem, that relies on the algorithm for the $r$-restricted \routeandcut problem, whose existence is assumed in \Cref{thm: from routeandcut to maintaincluster}. We then show that the algorithm for the $r$-restricted \maintainspecialcluster problem with the properties required in \Cref{thm: from routeandcut to maintaincluster} follows immediately from the algorithm for the \connecttocenters problem. This algorithm will also be used as one of the main ingredients in our algorithm for the \maintaincluster problem, that will complete the proof of \Cref{thm: from routeandcut to maintaincluster}.

%% file: connecttocenters-alg.tex
\subsection{An Algorithm for the \connecttocenters Problem}
\label{subsec: connecttocenters-alg}

In this subsection, we provide an algorithm for the \connecttocenters problem, that is summarized in the following theorem.

\begin{theorem}\label{thm: connecttocenters}
	Suppose that, for some parameter $r\geq 1$, there exists a randomized algorithm for the $r$-restricted \routeandcut problem, that, given as an input an instance $(G,A,B,\Delta,\eta,N)$ of the problem with $|V(G)|=n$,
	has running time at most $c_1\cdot n\cdot (n-|B|)\cdot 2^{c_2\sqrt{\log N}}\cdot (\log N)^{16c_2(r-1)+8c_2}$.
	Then there is a randomized algorithm for the \connecttocenters problem, for the same value of parameter $r$, whose running time is at most:

\[	O\left(c_1\cdot W^0\cdot ( W^0-|\beta^0|+|\beta'|
	\cdot \rho +\Lambda)\cdot  2^{c_2\sqrt{\log N}}\cdot (\log N)^{16c_2(r-1)+8c_2+26}\right ),\]
	
	 where $\Lambda$ is the total increase in the edge lengths over the course of the algorithm, $\beta'$ is the set of all vertices deleted from $\beta$ over the course of the algorithm (including the vertices that are deleted from $H$ following updates to graph $G$), and $W^0$ is the total initial weight of the vertices in the input graph $G$. 
\end{theorem} 

The remainder of this subsection is dedicated to the proof of \Cref{thm: connecttocenters}. The algorithm is quite straightforward: in every iteration we simply apply the algorithm for the \reconnect problem from \Cref{claim: proc reconnect} to the current graph $G$, and the current \LCDS. The only subtlety is that the algorithm requires that the length of every edge in $G$ is at least $d'$. If this is not the case, and $e=(u,v)$ is a special edge of length greater than $d'$, then we can define a strongly well-structured cut $(A,B)$ in $G$ with $A=\set{u}$ and $B=V(G)\setminus \set{u}$, whose sparsity is at least $1/d'$ (since $w(u),w(v)\geq d'$), and return this cut. We also need to verify that, whenever we apply the algorithm from \Cref{claim: proc reconnect} to the current graph $G$, we indeed obtain a valid instance of the \reconnect problem. The main difficulty is in the analysis of the running time of this algorithm, that we discuss below.

\subsubsection{The Data Structures}

Recall that we are given as input a perfect well-structured graph $G=(L,R,E)$ with $|V(G)|=n$, and a proper assignment $\ell(e)$ of lengths to the edges $e\in E$, that are given in the adjacency list representation, together with a subdivided graph $H=(L',R',E')$ corresponding to $G$ in the adjacency-list representation. 
We denote by $W^0$ the weight of all vertices in this initial graph $G$. Recall that we are also given  a parameter $N\geq \max\set{n,W^0/8}$ that is greater than a large enough constant, and parameters $d'\geq 1$,  and $1\leq r\leq \ceil{\sqrt{\log N}}$, such that $d'\leq\min\set{\frac{2^{r\sqrt{\log N}}}{\log^8N},\frac{n}{8}}$.
We also denote by $\beta^0$ the initial set $\beta\subseteq V(H)$ of vertices that is given as part of the problem input; recall that $|\beta^0|\geq \frac{W^0}{256d'}$.
Throughout, we use the parameter $\rho=\frac{2^{12}\cdot W^0\cdot \log^3N}{|\beta^0|}$.

Recall that graph $H$ is obtained from $G$ by replacing every special edge $e=(u,v)$ with a new path, that we denote by $P^*(e)$, and call \emph{the path representing $e$ in $H$}. The number of special edges on $P^*(e)$ is exactly $\ell(e)$. Therefore, $|V(H)|=\sum_{v\in V(G)}w(v)$, and initially, $|V(H)|=W^0$. %For every special edge $e'$ on path $P^*(e)$, we say that $e$ is the \emph{parent-edge} of $e'$, and that $e'$ is the \emph{child-edge} of $e$.

Assume that at some time during the algorithm's execution, the length $\ell(e)$ of a special edge $e\in E(G)$ is doubled. Then for every special edge $e'=(x,y)$ on path $P^*(e)$, we replace the edge by a path $(x,x',y',y)$, where $x',y'$ are new vertices, edges $(x,x')$ and $(y',y)$ are special, and edge $(x,y')$ is regular; we call this process the \emph{splitting} of edge $e'$. Notice that, since $(x,y)$ is a special edge of $H$, $x\in R'$ and $y\in L'$ must hold. We add $x'$ to $L'$ and $y'$ to $R'$. It is easy to verify that $H$ remains the subdivided graph of $G$ after this update. Since the total increase in the lengths of the edges of $G$ over the course of the algorithm is bounded by $\Lambda$, the total time required to maintain graph $H$ in the adjacency-list representation, once the representation of the initial graph $H$ is given, is bounded by $O(\Lambda)$.

Our algorithm maintains a valid \LCDS for the current graph $G$, with the parameter $\rho$ defined above, and with Requirement \ref{prop: routing paths} relaxed to allow the lengths of the paths in $\bigcup_{i=1}^{\lambda-1}\qset_i$ to be bounded by $16d'\log^{10}N$. However, whenever a new path is added to set $\bigcup_{i=1}^{\lambda-1}\qset_i$, we require that its initial length is bounded by $4d'\log^{10}N$.

At the beginning of the algorithm, we set $U_0=\beta$, $U_1=U_2=\cdots=U_{\lambda-1}=\emptyset$, and $U_{\lambda}=V(H)\setminus \beta$. For all $1\leq i< \lambda$, we also set $\qset_i=\emptyset$. It is easy to verity that this is a valid \LCDS. Throughout, we denote $\qset=\bigcup_{i=1}^{\lambda-1}\qset_i$.
We also maintain, for every vertex $x\in \beta$ and layer $1\leq i<\lambda$, the set $\dset_i(x)$ of the descendants of $x$ at layer $i$, and, for every vertex $v\in \bigcup_{i=1}^{\lambda-1}\dset_i(x)$, a pointer from $v$ to $x$. At the beginning of the algorithm, for all 
$x\in \beta$ and $1\leq i<\lambda$, we set $\dset_i(x)=\emptyset$. Additionally, for every layer $1\leq i<\lambda$, for every vertex $v\in U^{<\lambda}$, if $v$ serves as an endpoint of a path $K(u)\in \qset_i$, then we maintain a pointer from $v$ to $u$.

For every edge $e\in E(H)$, we will maintain a list $S(e)$ of all vertices 
$v\in \bigcup_{i=1}^{\lambda-1}U_i$ with $e\in K(v)$. The task of maintaining these lists is quite straightforward. We will ignore it when we describe the remainder of the algorithm below, but we will show how to extend the algorithm in order to maintain these lists efficiently later.

For every vertex $v\in \bigcup_{i=1}^{\lambda-1}U_i$, we will maintain the length $\hat \ell(v)$ of the path $K(v)\in \qset$. 
When the path $K(v)$ is first added to $\qset$, we initialize its length $\hat \ell(v)$. Then whenever a special edge on path $K(v)$ undergoes splitting, we increase $\hat \ell(v)$ by $1$. 
The time that is required in order to maintain the values $\hat \ell(v)$ for vertices $v\in \bigcup_{i=1}^{\lambda-1}U_i$ is clearly subsumed by the time required to compute and to maintain the collection $\qset$ of paths. In our description of the remainder of the algorithm below, we will ignore the task of maintaining these values.

Before we provide the algorithm, we need one more definition, and a simple bound. We denote by $G^0$ the initial graph $G$. For every vertex $v\in V(G^0)$, we denote by $w^0(v)$ the weight of $v$ at the beginning of the algorithm. If $v$ lies in $G$ at the end of the algorithm, then we let $w^*(v)$ be the weight of $v$ at the end of the algorithm, and otherwise we let $w^*(v)$ be the weight of $v$ just before it is deleted from $G$. We also denote $W^*=\sum_{v\in V(G^0)}w^*(v)$. Since the total increase in the lengths of edges over the course of the algorithm is at most $8n$, we get that $W^*\leq W^0+16n$. Since $N\geq \max\set{n,W^0/8}$, we get that $W^*\leq 32N$.

\subsubsection{The Algorithm}

We now describe the algorithm for executing a single phase.
We denote by $M$ the multiset of all vertices that ever lied in $U_{\lambda}$ over the course of the algorithm. If a vertex is added to and removed from $U_{\lambda}$  multiple times during the algorithm, we add multiple copies of that vertex to $M$. We will bound the running time of the algorithm in terms of $|M|$, and then later we will bound $|M|$. Whenever our algorithm computes a strongly well-structured $\frac{1}{d'}$-sparse cut $(A,B)$ in $G$ with $w(A)\geq 1.8|\beta\cap A|$, we call Procedure \processcut that
we describe below. Following this procedure, the vertices that serve as endpoints of the edges in $E_G(A,B)$, along with the vertices of $A$ (if $w(A)\leq w(B)$) or the vertices of $B$ (otherwise) are deleted from $G$.
We partition the execution of a single phase into three stages.

\paragraph{First stage.}
The first stage continues as long as the current graph $G$ contains an edge of length greater than $d'$. Let $e=(u,v)$ be any such edge, and recall that it must be a special edge.  Consider the cut $(A,B)$ in $G$, where $A=\set{u}$, and $B=V(G)\setminus\set{u}$. It is immediate to verify that this cut is strongly well-structured. Moreover, since $w(u)=w(v)=\ell(e)$, while $|E_G(A,B)|=1$, the sparsity of the cut $(A,B)$ is at most $\frac{1}{d'}$, and $w(A)\geq 1.8|\beta\cap A|$ (as $|A|=1$). We then output the cut $(A,B)$, and call Procedure \processcut that is described below.
This completes the description of the algorithm for the first stage of a phase.  At the end of the first stage, graph $G$ does not contain any edge whose length is greater than $d'$.

\paragraph{Second stage.}
The second stage is executed as long as $U_{\lambda}\neq\emptyset$.
We let $1\leq i<\lambda$ be the smallest index for which $\sum_{j=i+1}^{\lambda}|U_j|\geq 2^{\lambda-i-1}$ holds.  Notice that, since $|U_{\lambda}|\ge 1$, if we choose $i=\lambda-1$, then $\sum_{j=i+1}^{\lambda}|U_j|\geq 2^{\lambda-i-1}=1$ holds, so such an index $i$ must exist. 
We need the following claim.

\begin{claim}\label{claim: small i}
 $\sum_{j=i}^{\lambda}|U_j|\leq  2^{\lambda-i}$.
\end{claim}
\begin{proof} 
Assume otherwise, that is,  $\sum_{j=i}^{\lambda}|U_j|>  2^{\lambda-i}$. Since we did not choose index $i-1$ instead of $i$, it must be the case that $i=1$. Recall however that $\lambda=\ceil{\log(64N)}+1$, and so $2^{\lambda-1}\geq 64N$. Therefore, $\sum_{j=1}^{\lambda}|U_j|\geq 64N$.
 However, $|V(H)|\leq W^*< 64N$ must hold, a contradiction. 
\end{proof}

 For every layer $i\leq j\leq \lambda-1$ and vertex $v\in U_j$, we delete $v$ from $U_j$ and add it to $U_{\lambda}$. We also delete path $K(v)$ from set $\qset_j$, and delete $v$ from the set $\dset_j(x)$, where $x$ is its ancestor-vertex in $\beta$. 
 We denote by $k=|U_{\lambda}|$, so $2^{\lambda-i-1}\leq k \leq  2^{\lambda-i}$.
 Notice that the running time that we have spent so far in this iteration is $O(2^{\lambda-i})+O(\lambda\cdot \log n)\leq O(k)+O(\log^2N)$. 
 
 Next, we apply the algorithm for Procedure \reconnect from \Cref{claim: proc reconnect} to the current graphs $G$, $H$, and the index $i$. For brevity of notation, we denote this procedure by $\Pi$. We verify that we obtain a valid input to Procedure \reconnect. 
 
 First, the current weight of all vertices in $G$ is bounded by $W^*\leq 32N$, so $N\geq \max\set{n,W/32}$ holds. The inequalities $1\leq r\leq \ceil{\sqrt{\log N}}$ and $1\leq d'\leq\min\set{\frac{2^{r\sqrt{\log N}}}{\log^8N},\frac{n}{8}}$ hold from the definition of the \connecttocenters problem. The length of every edge in $G$ is at most $d'$, and, for all $i\leq j<\lambda$, $U_j=\emptyset$ holds, while $|U_{\lambda}|\leq 2^{\lambda-i}$ from \Cref{claim: small i}. 
 Recall that we have set $\rho=\frac{2^{12}\cdot W^0\cdot \log^3N}{|\beta^0|}$. 
 Since $|\beta^0|\leq W^0$, we get that $\rho\geq \log^2N$. Let $W=\sum_{v\in V(G)}w(v)$.
 From the problem definition, $W\leq W^*\leq 32W^0$, and $|\beta|\geq \frac{|\beta^0|}4$. Therefore:
 
 \[\frac{W\cdot \log^3N}{|\beta|}\leq  \frac{128W^0\cdot \log^3N}{|\beta^0|}  \leq \frac{\rho}{16},\]
 
 as required.

 It remains to verify that $|\beta|\geq \frac{W}{d'\log^2N}$ holds. From the definition of the \connecttocenters problem, $|\beta^0|\geq \frac{W^0}{256d'}$, and, throughout the algorithm, $|\beta|\geq\frac{|\beta^0|}{4}\geq \frac{W^0}{1024d'}$ holds. Since the total weight $W$ of the vertices of $G$ is bounded by $W^*\le 32W^0$, and $N$ is sufficiently large, we get that $|\beta|\geq \frac{W^0}{1024d'}\geq \frac{W}{2^{15}d'}\geq \frac{W}{d'\log^2N}$ holds. Therefore, we obtain a valid input for Procedure \reconnect.
 Recall that the running time of procedure $\Pi$ is bounded by:

\[\begin{split} 
&O\left(c_1\cdot W\cdot k\cdot 2^{c_2\sqrt{\log N}}\cdot (\log N)^{16c_2(r-1)+8c_2+14}\right )\\
&\quad\quad\quad\quad\quad\quad\quad\quad\leq O\left(c_1\cdot W^0\cdot k\cdot 2^{c_2\sqrt{\log N}}\cdot (\log N)^{16c_2(r-1)+8c_2+14}\right ),
\end{split}\]

 and the probability that it returns ``FAIL'' is at most $\frac{1}{N^{10}}$. If the procedure returns ``FAIL'', then we terminate our algorithm, and return ``FAIL'' as well. Otherwise, we consider three cases.
 
 In the first case, the procedure returns a weakly well-structured cut $(A',B')$ in $H$ with $|A'|\geq \frac{511|V(H)|}{512}$,  $|A'\cap \beta|<\frac{2|\beta|}{\log^3N}$, and $|E_H(A',B')|\leq  \frac{|V(H)|}{1024d'\cdot \log^6N}$. In this case, we return the cut $(A',B')$, and the algorithm terminates. We say that Procedure $\Pi$ produced a type-1 outcome in this case.
 
 In the second case, the procedure returns a valid \LCDS, in which, for all $i<i'<\lambda$, $U_{i'}=\emptyset$ and $|U_{\lambda}|\leq \frac{k}{4\lambda}$ holds; additionally, layers $U_1,\ldots,U_{i-1}$, together with the corresponding collections of paths $\qset_1,\ldots,\qset_{i-1}$ remain the same as in the input \LCDS. In this case we say that Procedure $\Pi$ produced a type-2 outcome. Notice that, at the beginning of the procedure, $|U_{\lambda}|=k$ held, while at the end of the procedure, $|U_{\lambda}|\leq \frac{k}{4\lambda}$. Therefore, at least $k\left(1-\frac{1}{4\lambda}\right )$ vertices were removed from $U_{\lambda}$ by this procedure. For each such vertex $v$, we \emph{charge} the vertex $O\left(c_1\cdot W^0\cdot 2^{c_2\sqrt{\log N}}\cdot (\log N)^{16c_2(r-1)+8c_2+14}\right )$ units for the execution of the procedure, so that the total charge to all vertices that were removed from $U_{\lambda}$ by the procedure is at least the total running time of the procedure.
 
 Lastly, we consider the third case, where Procedure $\Pi$ returned a strongly well-structured $\frac{1}{4d'}$-sparse cut $(A,B)$ in $G$ with $w(A),w(B)\geq \frac{k}{2^{14}\log N}$ and $w(A)\geq 1.8|\beta\cap A|$. In this case, we say that the outcome of Procedure $\Pi$ is of type $3$.
 We then output the cut $(A,B)$, and call Procedure \processcut that is described below.
 Recall that, following this cut, all vertices that serve as endpoints of the edges in $E_G(A,B)$, and also either all vertices of $A$ or all vertices of $B$, are deleted from $G$.  Let $Z$ denote the set of all vertices that are deleted from $G$. 
 Notice $w(Z)\geq \min\set{w(A),w(B)}\geq \frac{k}{2^{14}\log N}$.  For every vertex $v\in Z$, we charge the vertex $O\left(c_1\cdot W^0\cdot w(v)\cdot  2^{c_2\sqrt{\log N}}\cdot (\log N)^{16c_2(r-1)+8c_2+15}\right )$ units for the execution of the iteration, so that the total charge to all vertices of $Z$ is at least as high as the running time of the iteration.

 \paragraph{Third Stage.} The second stage terminates once we obtain a valid \LCDS (with relaxed Requirement \ref{prop: routing paths}), with $U_{\lambda}=\emptyset$. Recall that, following this, the algorithm is given a set $E'\subseteq E(G)$ of special edges whose lengths need to be doubled, and a set $\hat \beta\subseteq \beta$ of vertices, that need to be deleted from $\beta$. 
For every special edge $e\in E'$, we consider each of the special edges $e'\in P^*(e)$ one by one. We start by performing a splitting update of the edge $e'$, that introduces two new special edges $e_1,e_2$ in graph $H$, that replace edge $e'$, and two new vertices, that we add to set $U_{\lambda}$. 
We initialize $S(e_1')=S(e_2')=\emptyset$. For every vertex $v\in S(e)$, we update the corresponding path $K(v)$ by splitting the edge $e'$, and we update the length $\hat \ell(v)$ of the path by increasing it by (additive) $2$. If $\hat \ell(v)$ becomes greater than $16d'\log^{10}N$, then let $U_j$ be the layer to which $v$ belongs. We delete $v$ from $U_j$ and add it to $U_{\lambda}$. We also delete the path $K(v)$ from $\qset_j$, and we say that path $K(v)$ is \emph{destroyed}. We also delete $v$ from set $\dset_j(x)$, where $x\in \beta$ is the ancestor of $v$. Let $\Lambda'$ denote the total initial length of the edges in $E'$, so the lengths of all edges increased by additive $\Lambda'$ in this iteration, and the number of edge-splitting operations performed in graph $H$ is bounded by $\Lambda'$. Since, from Requirement  \ref{prop: routing paths}, for all $1\leq i<\lambda$, the paths in $\qset_i$ cause edge-congestion at most $4d'\log^9N$, and since $\lambda\leq O(\log N)$, we get that, for every edge $e'\in E(H)$, $|S(e)|\leq O(d'\log^{10}N)$. Therefore, the total running time of the algorithm for Stage 3 so far is $O(\Lambda'\cdot d'\cdot \log^{10}N)$. %Additionally, the total number of vertices added to $U_{\lambda}$ so far during Stage 3 is bounded by $2\Lambda'$.

Next, we process the vertices of $\hat \beta$, that are deleted from set $\beta$. 
We consider the vertices $x\in \hat \beta$ one by one. For each such vertex $x$, for every layer $1\leq j<\lambda$, and for every descendant $v\in \dset_i(x)$ of $x$ at layer $j$, we delete $v$ from $U_j$ and add it to $U_{\lambda}$. We also delete the path $K(v)$ from $\qset_j$ (we may sometimes say that the path is \emph{destroyed}). 
This completes the description of the algorithm for the third stage of a phase. %We denote by $n_3$ the total number of vertices that were added to $U_{\lambda}$ during this stage, and by $\hat n_3$ the total number of vertices that were deleted from $\beta$ during this stage. 
Recall that, from Property \ref{prop: descendants} of the \LCDS, for every vertex $x\in \beta$, for every layer $1\leq j\leq \lambda-1$, $|\dset_j(x)|\leq  \rho$.  Therefore, 
the total number of vertices that are added to $U_{\lambda}$  due to the deletion of their ancestors from $\beta$ is bounded by $O(|\hat \beta|\cdot \rho\cdot\lambda)\leq O( |\hat \beta|\cdot \rho\cdot \log N)$.
%$n_3\leq \hat n_3\cdot \lambda\cdot \rho+2\Lambda'\leq 8\hat n_3\cdot \rho\cdot \log N+2\Lambda'$.
The running time of Stage 3 is bounded by $O(\Lambda'\cdot d'\cdot \log^{10}N)+O(|\hat \beta|\cdot\rho\cdot \log N)\leq O(\Lambda'\cdot W^0\cdot \log^{10}N)+O(|\hat \beta|\cdot\rho\cdot \log N)$, since  $d'\leq n\leq W^0$ from the problem definition.
This completes the description of the algorithm for a single phase.

Recall that $\Lambda$ bounds the total increase in the lengths of edges of the course of the algorithm.
Let $\hat W$ denote the total weight of all vertices that were ever deleted from $G$ over the course of the algorithm, and let $\beta'$ be the set of all vertices that were deleted from $\beta$ over the course of the algorithm. Observe that the vertices of $\beta^0\setminus\beta'$ remain in $G$ throughout the algorithm, and so:

\begin{equation}\label{eq: deleted weight}
\hat W\leq W^*-|\beta^0\setminus \beta'|\leq W^0+2\Lambda-|\beta^0|+|\beta'|.
\end{equation}

From the above discussion, the total running time of algorithms for stages 1, 2 and 3 (excluding the running times of the executions of Procedure $\Pi$ and Procedure \processcut), over the course of the entire algorithm, is bounded by:

\[O(|\beta'|\cdot\rho\cdot \log N)+O(\Lambda\cdot d'\cdot \log^{10}N). \]  

Since, from the problem definition, $d'\leq n\leq W^0$, 
this is bounded by:

\[O(| \beta'|\cdot\rho\cdot \log N)+O(\Lambda\cdot W^0\cdot \log^{10}N). \]

Next, we bound the total running time of all executions of Procedure \reconnect over the course of the algorithm. Observe first that type-1 outcome of the procedure may only occur once over the course of the algorithm. The running time of Procedure \reconnect in case of type-1 outcome is bounded by:

\[O\left(c_1\cdot W^0\cdot k\cdot 2^{c_2\sqrt{\log N}}\cdot (\log N)^{16c_2(r-1)+8c_2+14}\right ),\]

where $k\leq |V(H)\setminus \beta|$.
Recall that $\beta'$ is the set of all vertices that were deleted from $\beta$ over the course of the algorithm, so $|\beta|\geq |\beta^0|-|\beta'|$. Recall also that $|V(H)|\leq W^*\leq W^0+2\Lambda$. Therefore, $k\leq |V(H)\setminus \beta|
\leq W^*-|\beta^0|+|\beta'|\leq W^0+2\Lambda-|\beta^0|+|\beta'|$. We then get that the total running time of Procedure \reconnect when type-1 outcome occurs is bounded by:

\[O\left(c_1\cdot W^0\cdot (W^0-|\beta^0|+\Lambda+|\beta'|)\cdot 2^{c_2\sqrt{\log N}}\cdot (\log N)^{16c_2(r-1)+8c_2+14}\right ).\]

Recall that we have charged the running times of the executions of Procedure $\Pi$ in which type-2 or type-3 outcomes occured to vertices that were removed from $U_{\lambda}$ by the procedure (in case of type-2 outcome), or to the weights of the vertices that were deleted from $G$ following the procedure. Recall also that $M$ is the multiset of all vertices that ever lied in $U_{\lambda}$. 
The total running time of the remaining executions of Procedure $\Pi$ is bounded by:

\[
\begin{split}
&O\left(c_1\cdot W^0\cdot (\hat W+|M|)\cdot  2^{c_2\sqrt{\log N}}\cdot (\log N)^{16c_2(r-1)+8c_2+15}\right )\\
&\quad\quad\quad\quad\quad\quad\quad\quad\quad \leq O\left(c_1\cdot W^0\cdot ( W^0+|\beta'|-|\beta^0|+\Lambda +|M|)\cdot  2^{c_2\sqrt{\log N}}\cdot (\log N)^{16c_2(r-1)+8c_2+15}\right ),
\end{split}\]

from Inequality \ref{eq: deleted weight}.
The total running time of the algorithm, excluding Procedure \processcut, is therefore bounded by:

\[ O\left(c_1\cdot W^0\cdot ( W^0-|\beta^0|+|\beta'|\cdot \rho +\Lambda+|M|)\cdot  2^{c_2\sqrt{\log N}}\cdot (\log N)^{16c_2(r-1)+8c_2+15}\right ). \]

Next, we describe Procedure \processcut and analyze its running time, and then bound $|M|$.

\subsubsection{Procedure \processcut}

Procedure \processcut is called whenever our algorithm produces a strongly well-structured $\frac{1}{d'}$-sparse cut $(A,B)$ in $G$ with $w(A)\geq 1.8|\beta\cap A|$. Let $Z$ be the set of vertices of $G$ that contains all endpoints of the edges in $E_G(A,B)$. Additionally, if $w(A)\leq w(B)$, we add the vertices of $A$ to $Z$, and we add the vertices of $B$ to $Z$ otherwise. Recall that, following this cut, the vertices of $Z$ are deleted from $G$.  Additionally, some vertices may be deleted from $\beta$. We denote the latter set of vertices by $\hat \beta$, and we assume that $Z\cap \beta\subseteq \hat \beta$. 

We start by updating the graph $H$, as follows: we delete the vertices of $Z$ from $H$, and, for every special edge $e\in E(G)$ that is incident to a vertex of $Z$, we delete all vertices and edges on the corresponding path $P^*(e)$ from $H$. 
For every vertex $v$ that is deleted from $H$, if $U_j$ is the layer of the \LCDS to which $v$ belongs, we delete $v$ from $U_j$ and from set $\dset_j(x)$, where $x\in \beta$ is the ancestor of $v$, and we also delete path $K(v)$ from $\qset_j$.
Note that the update of graph $H$ takes time $O(w(Z))$. We denote by $Z'$ the set of vertices that were deleted from $H$, so $|Z'|\leq O(w(Z))$.

Next, we consider the vertices $x\in \hat \beta$ one by one. For each such vertex $x$, for every layer $1\leq j<\lambda$, and for every descendant $v\in \dset_j(x)$ of $x$ at layer $j$, we delete $v$ from $U_j$ and add it to $U_{\lambda}$. We also delete the path $K(v)$ from $\qset_j$ (we may sometimes say that the path is \emph{destroyed}).  Additionally, if $x$ remains in $H$, we add it to $U_{\lambda}$ as well.
Let $n'$ be the number of vertices that were added to $U_{\lambda}$ so far. Recall that, from Property \ref{prop: descendants} of the \LCDS, for every vertex $x\in \hat \beta$, for every layer $1\leq j\leq \lambda-1$, $|\dset_j(x)|\leq  \rho$.  Therefore, $n'\leq |\hat \beta|\cdot \lambda\cdot \rho\leq 8|\hat \beta|\cdot \rho\cdot \log N$. The running time of the algorithm so far is bounded by $O(w(Z)+|\hat \beta|\cdot \rho\cdot \log N)$.

Lastly, it is possible that, for some vertices $v\in \bigcup_{j=1}^{\lambda-1}U_j$, some vertex of $K(v)$ was deleted from $H$ by the procedure, while vertex $v$ itself was not deleted. Note that, if the ancestor $x\in \beta$ of $v$ was
deleted from $\beta$, then we have already removed $v$ from its corresponding layer.
Let $V'\subseteq \bigcup_{j=1}^{\lambda-1}U_j$ denote the set of all vertices $v\in \bigcup_{j=1}^{\lambda-1}U_j$, such that $v\not\in Z'$, the ancestor $x\in \beta$ of $v$ is not in $\hat \beta$, and some vertex on path $K(v)$ is in $Z'$.
We further partition set $V'$ into two subsets: set $V'_1$ contains all vertices $v\in V'$, such that the last vertex $u$ on path $K(v)$ lies in $Z'$; while set $V'_2$ contains all remaining vertices.

From Requirement \ref{prop: descendants one to one} of \LCDS, $|V'_1|\leq \lambda\cdot Z'\leq O(w(Z)\cdot \log N)$. Since the \LCDS data structure maintains, for every vertex $v\in V'$ a pointer from the last endpoint $u$ of $K(v)$ to vertex $v$ itself, set $V'_1$ can be computed in time $O(|V'_1|)\leq O(w(Z)\cdot \log N)$.

In order to compute the set $V'_2$ of vertices, consider the collection $E_G(A,B)$ of edges. Recall that, for every edge $e\in E_G(A,B)$, we have defined a path $P^*(e)$ that represents edge $e$ in graph $H'$. Let $E'\subseteq E(H)$ be the set of edges that contains, for each edge $e\in E_G(A,B)$, the first and the last edge on path $P^*(e)$. We claim that, for every vertex $v\in V'_2$, 
path $K(v)$ must contain an edge of $E'$.

\begin{observation}\label{obs: problematic vertex contains edges}
	For every vertex $v\in V'_2$, 
	path $K(v)$ must contain an edge of $E'$.
\end{observation}
\begin{proof}
	Assume first that $w(A)\leq w(B)$.
Consider the graph $H\setminus E'$, and let $(a,b)$ be any edge in this graph with $a\in Z'$ and $b\not\in Z'$. From the definition of the set $Z'$ of vertices and the set $E'$ of edges, it must be the case that $a$ is the last endpoint of some path $P^*(e)$ for an edge $e\in E_G(A,B)$; equivalently, $a\in V(G)$ is the tail of one of the edges in $E_G(A,B)$. But then $a\in L'$, and it only has a single incoming edge in $H$, which must belong to $E'$. 
Consider now some vertex $v\in V_2'$, and its corresponding path $K(v)$. Assume for contradiction that $K(v)$ does not contain an edge of $E'$. Since $K(v)$  connects two vertices of $V(H)\setminus Z'$, and contains at least one vertex of $Z'$ as an inner vertex, it must contain an edge $(a,b)$ with $a\in Z'$ and $b\not\in Z'$. But the only edge entering $a$ lies in $E'$, a contradiction.

%Therefore, if $v\in V'_2$, then path $K(v)$ may not contain $a$. We conclude that for every vertex $v\in V'_2$, path $K(v)$ may not contain any edge $(a,b)$ in $H\setminus E'$ with $a\in Z'$ and $b\not\in Z'$. But since each such path connects two vertices of $V(H)\setminus Z'$, and contains at least one vertex of $Z'$ as an inner vertex, we conclude that it must contain an edge of $E'$.

If $w(A)>w(B)$, then the proof is the same, except that we consider an edge $(a,b)$ with $a\not\in Z'$ and $b\in Z'$. In this case, $b$ must be a head of an edge of $E_G(A,B)$, and so it has no outgoing edges in $H\setminus E'$.
\end{proof}

In order to compute the set $V'_2$ of vertices, we consider the edges $e\in E_G(A,B)$ one by one. When edge $e$ is considered, we let $e'$ and $e''$ be the first and the last edges on the corresponding path $P^*(e)$. We then add to $V'_2$ all vertices that lie in set $S(e')\cup S(e'')$. Notice that, from Requirement \ref{prop: routing paths} of the \LCDS, the congestion caused by the set $\qset$ of paths is bounded by $O(d'\lambda\log^9N)\leq O(d'\log^{10}N)$, and so for every edge $e^*\in E(H)$, $|S(e^*)|\leq O(d'\log^{10}N)$. Altogether, we get that:
$|V_2'|\leq O\left(|E_G(A,B)|\cdot d'\cdot \log^{10}N\right )$. 
Since cut $(A,B)$ is $\frac 1 {d'}$-sparse, we get that $|E_G(A,B)|\leq \frac{\min\set{w(A),w(B)}}{d'}\leq \frac{w(Z)}{d'}$. Overall, we get that $|V_2'|\leq O(w(Z)\cdot \log^{10}N)$, and set $V'_2$ can be computed in time $O(w(Z)\cdot \log^{10}N)$. From our discussion, $|V'|\leq O(w(Z)\cdot \log^{10}N)$, and it can be computed in time $O(w(Z)\cdot \log^{10}N)$.

For every vertex $v\in V'$, let $1\leq j\le \lambda$ be the index of the layer to which $v$ belongs. If $j<\lambda$, then we delete $v$ from $U_j$ and add it to $U_{\lambda}$, and we also delete the path $K(v)$ from $\qset_j$. Additionally, if $x\in \beta$ is the ancestor of $v$, then we delete $v$ from set $\dset_j(x)$.

From our discussion, the total number of vertices that were added to $U_{\lambda}$ over the course of the procedure is bounded by $ O(w(Z)\cdot \log^{10}N)+O(|\hat \beta|\cdot \rho\cdot \log N)$, and the running time of the procedure is bounded by $ O(w(Z)\cdot \log^{10}N)+O(|\hat \beta|\cdot \rho\cdot \log N)$.

Recall that we have denoted by $\hat W$ the total weight of all vertices that were ever deleted from $G$ over the course of the algorithm, and by $\beta'$ the set of all vertices that remain in $\beta$ at the end of the algorithm. Recall also that, from Inequality \ref{eq: deleted weight}, 
$\hat W\leq  W^0+2\Lambda-|\beta^0|+|\beta'|$.
Therefore, the total number of vertices that were ever added to set $U_{\lambda}$ by Procedure \processcut is bounded by:

\[O(\hat W\cdot \log^{10}N)+O(|\beta'|\cdot \rho\cdot \log N)\leq 
O\left((W^0-|\beta^0|+|\beta'|\cdot \rho+\Lambda)\cdot \log^{10}N\right ),
 \]
 
 and the total running time of all calls to Procedure \processcut over the course of the algorithm is bounded by 

\[O(\hat W\cdot \log^{10}N)+O(|\beta'|\cdot \rho\cdot \log N)\leq 
O\left((W^0-|\beta^0|+|\beta'|\cdot \rho+\Lambda)\cdot \log^{10}N\right ).
\]

The total running time of the algorithm for the \connecttocenters problem is then bounded by:

\begin{equation}\label{eq: run time so far}
O\left(c_1\cdot W^0\cdot ( W^0-|\beta^0|+|\beta'|\cdot\rho +\Lambda+|M|)\cdot  2^{c_2\sqrt{\log N}}\cdot (\log N)^{16c_2(r-1)+8c_2+15}\right ). 
\end{equation}

\subsubsection{Bounding $|M|$ and the Running Time of the Algorithm}

Recall that $M$ is the multiset containing all vertices that were ever added to 
$U_{\lambda}$ over the course of the phase. We partition the set $M$ into six subsets, $M_1,\ldots,M_6$, depending on the process through which the vertices were added to $U_{\lambda}$, and we then bound the cardinality of each set separately.

\paragraph{Set $M_1$.}
Recall that at the beginning of the algorithm, every vertex of $V(H)\setminus \beta^0$ is added to $U_{\lambda}$. Set $M_1$ contains all these vertices, and, additionally, whenever a new vertex is inserted into graph $H$ due to the splitting of a special edge, this new vertex, which is then added to $U_{\lambda}$, is also added to $M_1$. 
Lastly, if Procedure \reconnect ever terminates with a type-1 outcome, then all vertices that were added to $U_{\lambda}$ in the same iteration, just prior to that execution of Procedure \reconnect, are also added to set $M_1$.
Since the vertices of $\beta^0\setminus \beta'$ always lie in $U_0$ and are never added to $U_{\lambda}$, it is easy to see that $|M_1|$ is asymptotically bounded by the total number of vertices ever present in $H$, excluding the vertices of $|\beta^0\setminus \beta'|$. Therefore, $|M_1|\leq O(W^*-|\beta^0\setminus \beta'|)\leq O( W^0+2\Lambda-|\beta^0|+|\beta'|)$.

\paragraph{Set $M_2$.}
The second set $M_2$ contains all vertices $v$ that were added to $U_{\lambda}$ when the length of the current path $K(v)\in \qset$ exceeded $16d'\log^{10}N$. 
Recall that, when path $K(v)$ was added to $\qset$, its length was bounded by $4d'\log^{10}N$. Therefore, between the time path $K(v)$ was added to $\qset$, and the time it was destroyed, its edges underwent at least $4d'\log^{10}N$ splitting updates.
Recall that the total increase in the lengths of all edges in $G$ is denoted by $\Lambda$, so the total number of edge splitting updates in graph $H$ is bounded by $\Lambda$. Since every edge of $H$ may belong to at most $4d'\lambda\log^9N\leq 32d'\log^{10}N$ paths in $\qset$, we get that the total number of edge splitting updates that the paths in $\qset$ undergo over the course of the algorithm is bounded by $O(\Lambda\cdot d'\log^{10}N)$. Since each destroyed path must undergo at least $4d'\log^{10}N$ edge-splitting updates, we get that $|M_2|\leq O(\Lambda)$.

\paragraph{Set $M_3$.}
The third set, $M_3$, contains all vertices $v$ that were added to $U_{\lambda}$ because the ancestor-vertex of $v$ was deleted from $\beta$ during Phase 3 of the algorithm. Note that, from Property \ref{prop: descendants} of the \newline \LCDS, every vertex $x\in \beta$ has at most $\lambda\cdot\rho$ descendants at any time. Therefore, $|M_3|\leq O(|\beta'|\cdot \rho\cdot \lambda)\leq O(|\beta'|\cdot \rho\cdot \log N)$.

\paragraph{Set $M_4$.}
The fourth set, $M_4$, contains all vertices that were added to $U_{\lambda}$ by Procedure $\processcut$. As we showed already, $|M_4|\leq 
O\left((W^0-|\beta^0|+|\beta'|\cdot \rho+\Lambda)\cdot \log^{10}N\right )$.

\paragraph{Set $M_5$.}
The fifth set, $M_5$, contains all vertices that were added to $U_{\lambda}$  during the iterations of Stage 2, just prior to the call to Procedure \reconnect, that resulted in a type-3 outocme of the procedure.
Recall that in each such iteration, the number of vertices that were added to $U_{\lambda}$ was bounded by $k$, while the total weight of the vertices that are subsequently deleted from $G$ is at least $\Omega\left(\frac{k}{\log N}\right )$. Therefore, $|M_5|\leq O(\hat W\cdot \log N)\leq O\left ((W^0+2\Lambda-|\beta^0|+|\beta'|)\cdot \log N\right )$ from Inequality \ref{eq: deleted weight},

So far we have defined five subsets $M_1,\ldots,M_5$ of the set $M$, and we have shown that:

\begin{equation}\label{eq: bound on Ms}
\sum_{j=1}^5|M_j|\leq O\left((W^0-|\beta^0|+|\beta'|\cdot \rho+\Lambda)\cdot \log^{10}N\right ).
\end{equation}

\paragraph{Set $M_6$.}
We let the sixth set, $M_6$, contain all remaining vertices of $M$. Each such vertex $v$ must have been added to $U_{\lambda}$  during an iteration of the second stage, just prior to a call to Procedure \reconnect, that resulted in a type-2 outcome.
Consider any such iteration of Phase 2, and let $1\leq i<\lambda$ be the layer that was chosen by the iteration. We then say that this iteration of Phase 2 is a \emph{type-$i$ iteration}. For brevity, in the remainder of the proof, we use ``type-$i$ iteration'' only to refer to type-$i$ iterations of Phase 2 in which Procedure \reconnect resulted in a type-2 outcome.
Recall that, at the beginning of such an iteration, $\sum_{j=i+1}^{\lambda}|U_j|\geq 2^{\lambda-i-1}$ holds, and the vertices of $\bigcup_{j=i}^{\lambda-1}U_j$ are then moved from their current layers to $U_{\lambda}$. Let $k=\sum_{j=i}^{\lambda}|U_j|$, so  $2^{\lambda-i-1}\leq k\leq 2^{\lambda-i}$. Then at the beginning of the iteration, at most $k$ vertices were added to $U_{\lambda}$, while at the end of the iteration, $|U_{\lambda}|\leq \frac{k}{4\lambda}$ held. Moreover, layers $U_1,\ldots,U_{i-1}$ remain unchanged over the course of the iteration.

From the description of our algorithm, for any pair $1\leq i',i''\leq \lambda$ of indices, a vertex may only move directly from layer $U_{i'}$ to layer $U_{i''}$ if either $i'=\lambda$ or $i''=\lambda$.
At the beginning of a type-$i$ iteration, $|U^{>i}|\geq 2^{\lambda-i-1}$ holds, while at the end of the iteraiton, $|U^{>i}|\leq \frac{k}{4\lambda}\leq \frac{2^{\lambda-i}}{4\lambda}$.
Therefore, between every pair of type-$i$ iterations, at least $2^{\lambda-i-1}-\frac{2^{\lambda-i}}{4\lambda}\geq 2^{\lambda-i-2}$ vertices must have moved from $\bigcup_{i'=1}^{i}U_{i'}$ to $U_{\lambda}$. 
For all $1\leq i<\lambda$, we denote by $M^i_6\subseteq M_6$ the set of vertices that were added to $M_6$ during type-$i$ iterations.
We start by bounding the number of type-$i$ iterations, for $1\leq i<\lambda$.

\begin{claim}
	There is a large enough constant $c^*$, such that for all $1\leq i<\lambda$, the total number of type-$i$ iterations is bounded by $\frac{c^*\cdot (W^0-|\beta^0|+|\beta'|\cdot \rho+\Lambda)\cdot \log^{10}N}{2^{\lambda-i}}$.
\end{claim}

\begin{proof}
	The proof is by induction. The base case is when $i=1$. 
	%Recall that at the end of each type-$i$ iteration, $U_2=\cdots=U_{\lambda-1}=\emptyset$ hold, and $|U_{\lambda}|\leq \frac{2^{\lambda-1}}{4\lambda}$. 
	From the above discussion, between every pair of type-1 iterations, at least $2^{\lambda-3}$ new vertices must have been moved from $U_1$ to $U_{\lambda}$. All such vertices must belong to $M_1\cup\cdots\cup M_5$. Since $\sum_{j=1}^5|M_j|\leq O\left((W^0-|\beta^0|+|\beta'|\cdot \rho+\Lambda)\cdot \log^{10}N\right )$, and $c^*$ is a sufficiently large constant, we get that the total number of  type-$1$ iterations is bounded by: 

\[O\left (\frac{(W^0-|\beta^0|+|\beta'|\cdot \rho+\Lambda)\cdot \log^{10}N}{2^{\lambda}}\right )\leq \frac{c^*\cdot (W^0-|\beta^0|+|\beta'|\cdot \rho+\Lambda)\cdot \log^{10}N}{2^{\lambda-1}}.\]
	
	Consider an index $1< i<\lambda$, and assume that the claim holds for all indices $1\leq i'<i$. We define a potential function $\Phi_i=\sum_{j=i+1}^{\lambda}|U_{j}|$. Recall that, at the beginning of a type-$i$ iteration,  $\Phi_i=\sum_{j=i+1}^{\lambda}|U_j|\geq 2^{\lambda-i-1}$ holds, while at the end of the iteration, $|U_{\lambda}|\leq  \frac{2^{\lambda-i}}{4\lambda}$, and for all $i<i'<\lambda$, $U_{i'}=\emptyset$, so $\Phi_i\leq \frac{2^{\lambda-i}}{4\lambda}$. Therefore, in every type-$i$ iteration, $\Phi_i$ decreases by at least $2^{\lambda-i-1}\cdot \left(1-\frac{1}{2\lambda}\right )\geq 0.48\cdot 2^{\lambda-i}$. The potential $\Phi_i$ is increased whenever vertices are added to sets $M_1,\ldots,M_5$, and the total increase due to the addition of such vertices to $M_1,\ldots,M_5$ is bounded by $\sum_{j=1}^5|M_j|$. Additionally, for all $1\leq i'<i$, each type-$i'$ iteration may increase the potential $\Phi_i$ by at most $\frac{2^{\lambda-i'}}{4\lambda}$. Since, from the induction hypothesis, the number of   type-$i'$ iterations is bounded by  $\frac{c^*\cdot (W^0-|\beta^0|+|\beta'|\cdot \rho+\Lambda)\cdot \log^{10}N}{2^{\lambda-i'}}$, 
	the total increase in the potential $\Phi_i$ due to type-$i'$ iterations is bounded by: 
	
	\[ \frac{c^*\cdot (W^0-|\beta^0|+|\beta'|\cdot \rho+\Lambda)\cdot \log^{10}N}{2^{\lambda-i'}}\cdot \frac{2^{\lambda-i'}}{4\lambda}\leq   \frac{c^*\cdot (W^0-|\beta^0|+|\beta'|\cdot \rho+\Lambda)\cdot \log^{10}}{4\lambda}.\]

 Therefore, the total increase in the potential $\Phi_i$ over the course of the entire algorithm is bounded by:
	
	\[
	\begin{split}
	\sum_{j=1}^5|M_j|+&(i-1)\cdot \frac{c^*\cdot (W^0-|\beta^0|+|\beta'|\cdot \rho+\Lambda)\cdot \log^{10}N}{4\lambda}\\
	&\leq 0.26 c^*\cdot (W^0-|\beta^0|+|\beta'|\cdot \rho+\Lambda)\cdot \log^{10}N.
	\end{split}
	\]

	Since every successful type-$i$ iteration decreases the potential $\Phi_i$ by at least $0.48\cdot 2^{\lambda-i}$, we get that the total number of type-$i$ iterations is bounded by:
	
	\[\frac{0.26 c^*\cdot (W^0-|\beta^0|+|\beta'|\cdot \rho+\Lambda)\cdot \log^{10}N}{0.48\cdot 2^{\lambda-i}}\leq \frac{c^*\cdot (W^0-|\beta^0|+|\beta'|\cdot \rho+\Lambda)\cdot \log^{10}N}{2^{\lambda-i}}. \]
\end{proof}

Since, for all $1\leq i<\lambda$, every type-$i$ iteration adds at most $O(2^{\lambda-i})$ vertices to $U_{\lambda}$, we get that $|M_6^i|\leq O\left((W^0-|\beta^0|+|\beta'|\cdot \rho+\Lambda)\cdot \log^{10}N\right )$,  and altogether: 

\[|M_6|=\sum_{i=1}^{\lambda-1}|M_6^i|\leq O\left((W^0-|\beta^0|+|\beta'|\cdot \rho+\Lambda)\cdot \lambda\cdot \log^{10}N \right ).\]
 
Combining this with Inequality \ref{eq: bound on Ms}, and substituting $\lambda\leq O(\log N)$, we get that:

\[|M|\leq O\left((W^0-|\beta^0|+|\beta'|\cdot \rho+\Lambda)\cdot \log^{11}N\right ).\]

Lastly, substituting this bound on $|M|$ into the bound from Equation \ref{eq: run time so far}, we get that the total running time of the algorithm so far is bounded by:

\[
\begin{split}
&O\left(c_1\cdot W^0\cdot ( W^0-|\beta^0|+|\beta'|\cdot \rho +\Lambda+|M|)\cdot  2^{c_2\sqrt{\log N}}\cdot (\log N)^{16c_2(r-1)+8c_2+15}\right )\\
&\quad \leq O\left(c_1\cdot W^0\cdot ( W^0-|\beta^0|+|\beta'|
\cdot \rho +\Lambda)\cdot  2^{c_2\sqrt{\log N}}\cdot (\log N)^{16c_2(r-1)+8c_2+26}\right )
\end{split}
\]

\subsubsection{Maintaining the Additional Data Structures}

It remains to show an algorithm that maintains, for every edge $e\in E(H)$, the set $S(e)$ of vertices, whose corresponding path $K(v)$ contains $e$.  Additionally, for every vertex $x\in \beta$, we need to maintain, for all $1\leq i<\lambda$, the set $\dset_i(x)$ of the descendants of $x$ at layer $i$, and, for every vertex $v\in \bigcup_{i=1}^{\lambda-1}U_i$, a pointer to its unique ancestor in $\beta$. Lastly, for every  layer $1\leq i\leq \lambda-1$ and  vertex $v\in \bigcup_{j=1}^{i-1}U_j$, if $v$ serves as an endpoint of a path $K(u)\in \qset_i$, then we need to maintain a pointer from $v$ to $u$.
We now show how to extend our algorithm, in a straightforward manner, to maintain these additional data structures.

At the beginning of the algorithm, for every edge $e\in E(H)$, we set $S(e)=\emptyset$, and for every vertex $x\in \beta$ and index $1\leq i<\lambda$, we set $\dset_i(x)=\emptyset$. Whenever a new vertex $v$ is added to any layer $U_i$, for $1\leq i<\lambda$, we inspect its corresponding path $K(v)$, that is added to $\qset_i$. For every edge $e\in E(K(v))$, we add vertex $v$ to the list $S(e)$. Let $y$ denote the last endpoint of path $K(v)$, and let $x\in \beta$ be the unique ancestor of $y$. We then add $v$ to $\dset_i(x)$, and we add a pointer from $v$ to $x$, and another pointer from $y$ to $v$.

Whenever any vertex $u$ is deleted from a set $U_i$, for any $1\leq i<\lambda$, we again inspect the corresponding path $K(v)$, and, for every edge $e\in E(K(v))$, we delete $v$ from set $S(v)$. If $y$ is the other endpoint of path $K(v)$, then we delete the pointer from $y$ to $v$. Additionally, if $x\in\beta$ is the unique ancestor of $v$ at level $0$, then we delete $v$ from $\dset_i(x)$. Note that the time required to execute all these updates is asymptotically bounded by the time required to compute all paths $K(v)$ that are ever added to $\qset$. 

Therefore, the total update time of the algorithm remains bounded by:

\[ O\left(c_1\cdot W^0\cdot ( W^0-|\beta^0|+|\beta'|
\cdot \rho +\Lambda)\cdot  2^{c_2\sqrt{\log N}}\cdot (\log N)^{16c_2(r-1)+8c_2+26}\right ). \]

Recall that the algorithm from Procedure \reconnect may terminate with a ``FAIL'', and if this happens, our algorithm terminates with a ``FAIL'' as well. However, the number of calls to Procedure \reconnect is bounded by $|M|\leq \tilde O((W^0)^2)\leq N^3$, and the probability that a single call terminates with a ``FAIL'' is bounded by $\frac{1}{N^{10}}$. Therefore, overall, the probability that our algorithm terminates with a ``FAIL'' is bounded by $\frac{1}{N^7}$.

%% file: alg-maintainspecial.tex
\subsection{An Algorithm for the \maintainspecialcluster Problem}
\label{subsec: alg for maintainspecial}

In this subsection we prove the second part of  \Cref{thm: from routeandcut to maintaincluster}, by showing that an algorithm for the $r$-restricted \routeandcut problem implies an algorithm for the $r$-restricted \maintainspeccluster problem with the required running time. We summarize this result in the following theorem.

\begin{theorem}\label{thm: from routeandcut to maintainspeccluster}
	Suppose that, for some parameter $r\geq 1$, there exists a randomized algorithm for the $r$-restricted \routeandcut problem, that, given as an input an instance $(G,A,B,\Delta,\eta,N)$ of the problem with $|V(G)|=n$,
	has running time at most $c_1\cdot n\cdot (n-|B|)\cdot 2^{c_2\sqrt{\log N}}\cdot (\log N)^{16c_2(r-1)+8c_2}$. Then  there exists a randomized algorithm for the $r$-restricted \maintainspeccluster problem, that, on input $(G,\beta,\Delta,\eta,d,N)$ has running time at most: $c_1\cdot W^0(G)\cdot (W^0(G)-|\beta|+\Delta)\cdot 2^{c_2\sqrt{\log N}}\cdot (\log N)^{16c_2r}$. 
\end{theorem}

The remainder of this subsection is dedicated to the proof of \Cref{thm: from routeandcut to maintainspeccluster}. 
Recall that we are given as input a perfect well-structured graph $G=(L,R,E)$ with a proper assignment of lengths $\ell(e)$ to its edges $e\in E$, that are given in the adjacency-list representation, such that the length of every special edge is $1$.
We are also given a subset $\beta\subseteq V(G)$ of  vertices with $|\beta|\geq \frac{99|V(G)|}{100}$, and parameters $N\geq W^0(G)$, $1\leq \Delta\leq \frac{|\beta|}{(\log N)^{128}}$, $1\leq \eta \leq \Delta$, and $d\geq (\log N)^{64}$, such that $\frac{\Delta\cdot d}{\eta}\leq |V(G)|-|\beta|+\Delta\cdot (\log N)^{64}$. Lastly, we are given an integer $1\leq r\leq \ceil{\sqrt{\log N}}$, and we are guaranteed that $d\leq 2^{r\cdot\sqrt{\log N}}$.

Over the course of the algorithm, vertices may be deleted from set $\beta$, and the lengths of some edges may be increased, but the length of every special edge incident to a vertex that currently lies in $\beta$ is always $1$. We denote by $\beta^0$ the initial set $\beta$ of vertices. Once $|\beta|$ falls below $\frac{|\beta^0|}{2}$, the algorithm terminates.

Recall that the algorithm consists of at most $\Delta$ iterations.
At the beginning of every iteration $i$, the algorithm is given a vertex $x\in V(G)$, and it must return a simple path $P_i$ in $G$ connecting $x$ to any vertex  $y\in \beta$, such that the length of the path is at most $d$, and no inner vertices on the path lie in $\beta$ (if $x\in \beta$, that we require that path $P_i$ only consists of the vertex $x$). After that, the length of some special edges on path $P_i$ may be doubled, and vertex $y$ is deleted from $\beta$. In this case, we say that vertex $y$ was deleted from $\beta$ \emph{directly}. If $\tau'>\tau$ are two times at which the length of some edge $e\in E(G)$ is doubled, then we are guaranteed that the algorithm returned at least $\eta$ paths in response to queries during the time interval $(\tau,\tau']$ that contained $e$.

	The algorithm may, at any time, produce a strongly well-structured cut $(X,Y)$ in $G$ of sparsity $\Phi_G(X,Y)\leq \frac{(\log N)^{64}}{d}$, such that $w(X)\geq 1.8|\beta\cap X|$ holds. Let $J'\subseteq Y$ denote the set of vertices that serve as endpoints of the edges of $E_G(X,Y)$. The vertices of $X\cup J'$ are then deleted from $G$, and the algorithm continues with the resulting graph $G=G[Y\setminus J']$. Let $Z$ be the set of vertices that are deleted from $G$ following this cut. Notice that the vertices of $Z\cap \beta$ are deleted from graph $G$, and hence from set $\beta$ as well. We say that the vertices of $Z\cap \beta$ are deleted from $\beta$ \emph{indirectly}.
	
Let $\Lambda$ denote the total increase in the lengths of edges over the course of the algorithm, and denote $W^0=W^0(G)=n$. We also denote the initial graph $G$ by $G^0$, and the set of all vertices that were deleted from $\beta$ over the course of the algorithm, either directly or indirectly, by $\beta'$. We need the following claim and its corollary, whose proofs are deferred to Section \ref{subsec: proof of bound on lambda} and \ref{subsec: proof of small weight} of Appendix, respectively.

\begin{claim}\label{claim: bound lambda and deleted from beta}
	$\Lambda\leq 4n-4|\beta^0|+5\Delta\cdot (\log N)^{64}\leq 4n$ and $|\beta'|\leq 45n-45|\beta^0|+56\Delta  \cdot (\log N)^{64}$.
\end{claim}

\begin{corollary}\label{cor: small weight}
	Throughout the algorithm, $\sum_{v\in V(G)\setminus \beta}w(v)\leq 0.2|\beta|$.
\end{corollary}

Note that the subdivided graph $H$ of $G$ is identical to graph $G$, so we can use the adjacency list representation of $G$ also as the adjacency list representation of $H$. We define a parameter $d'=\frac{d}{(\log N)^{64}}$. Recall that, from the definition of the \maintainspecialcluster problem, we are guaranteed that $d\leq 2^{r\cdot\sqrt{\log N}}$, and also that $d \leq \frac{\eta}{\Delta}\left(n-|\beta^0|+\Delta\cdot (\log N)^{64}\right )\leq n$, since $\eta\leq \Delta$, and $ \Delta\leq \frac{|\beta^0|}{(\log N)^{128}}$ from the problem definition.
Therefore, $d'\leq\min\set{\frac{2^{r\sqrt{\log N}}}{\log^8N},\frac{n}{8}}$. Since $|\beta^0|\geq \frac{99n}{100}$ from the problem definition, and $d'\geq 1$ (as $d\geq (\log N)^{64}$), we get that $|\beta^0|\geq \frac{W^0}{256d'}$.

We apply the algorithm for the \connecttocenters problem from \Cref{thm: connecttocenters} to graph $G$, set $\beta$ of vertices, and parameters $N,d'$, and $r$. From our discussion, this is a valid input to the problem. Recall that the algorithm uses a parameter $\rho=\frac{2^{12}\cdot W^0\cdot \log^3N}{|\beta^0|}$. Since, from the definition of the \maintainspecialcluster problem, $|\beta^0|\geq \frac{99n}{100}$, while $W^0=n$, we get that $\rho=O(\log^3N)$. Our algorithm consists of at most $\Delta$ phases, that correspond to the phases of the \connecttocenters problem. Recall that the purpose of each such phase is to construct a valid \LCDS with $U_{\lambda}=\emptyset$. Our algorithm will only respond to queries at the end of each phase, once such a data structure is obtained. 

We now consider a single phase of the algorithm for the \connecttocenters problem, and one of its iterations. Assume first that the algorithm produces 
 a strongly well-structured $\frac{1}{d'}$-sparse cut $(A,B)$ in $G$ with $w(A)\geq 1.8|\beta\cap A|$. Since $d'=\frac{d}{(\log N)^{64}}$, the sparsity of the cut is at most $\frac{(\log N)^{64}}{d}$. We claim that $w(A)< w(B)$ must hold in the following simple obsevation, whose proof is deferred to Section \ref{subsec: proof of smaller side} of appendix.

\begin{observation}\label{obs: smaller side}
	 $w(A)< w(B)$ must hold.
\end{observation}

We delete from $G$ the vertices of $A$, and the endpoints of the edges of $E_G(A,B)$ that lie in $B$, and continue with the resulting graph.

Next, we show that the algorithm to the \connecttocenters
may never produce a weakly well-structured cut $(A',B')$ in graph $H$, with $|A'|\geq \frac{511|V(H)|}{512}$,  $|A'\cap \beta|<\frac{2|\beta|}{\log^3N}$, and $|E_H(A',B')|\leq  \frac{|V(H)|}{d'\cdot \log^6N}$.
Indeed, assume for contradiction that the algorithm produces such a cut. Since every special  edge incident to a vertex of $\beta$ has length $1$, we get that $|V(H)\setminus \beta|=\sum_{v\in V(G)\setminus \beta}w(v)|\leq 0.2|\beta|$ from \Cref{cor: small weight}.
But then:

\[|A'|\leq |V(H)\setminus \beta|+\frac{2|\beta|}{\log^3N}\leq 0.21|\beta|. \]

Since $|V(H)|\geq |\beta|$, it is then impossible that $|A'|\geq \frac{511|V(H)|}{512}$, a contradiction.

Therefore, if the algorithm for the  \connecttocenters problem does not produce a a strongly well-structured $\frac{1}{d'}$-sparse cut $(A,B)$ in $G$ with $w(A)\geq 1.8|\beta\cap A|$, it must produce a valid \newline \LCDS, with $U_{\lambda}=\emptyset$. In this case, our algorithm is ready to respond to a query. Let $x$ be a query vertex that it receives. If $x\in \beta$, then we return a path consisting of the vertex $x$ itself. Otherwise, we construct a path $P$ connecting $x$ to some vertex $y\in \beta$, as follows. Since $U_{\lambda}=\emptyset$, there must be an integer $1\leq i<\lambda$ with $x\in U_i$. We consider the path $K(x)\in \qset_i$, that connects $x$ to some vertex $x_1\in U_{i_1}$, for some $0\leq i_1<i$. If $x'\in \beta$, then we terminate this process, and otherwise we consider the path $K(x_1)$, which is concatenated to path $K(x)$. After at most $\lambda$ such iterations, we must obtain a path $P$, that connects $x$ to some vertex $y\in \beta$. Since the length of every path in set $\bigcup_{j=1}^{\lambda-1}\qset_j$ is at most $4d'\log^{10}N$, the final length of the resulting path is bounded by $4\lambda d'\log^{10}N\leq 32d'\log^{11}N<d$, since $d'=\frac{d}{(\log N)^{64}}$, and $N$ is sufficiently large. In time $O(d)$, we can turn this path into a simple path, and truncate it so that it still connects $x$ to some vertex of $\beta$, but the inner vertices on the path do not belong to $\beta$. The time required to respond to a single query is bounded by $O(d)$, and the total time required to respond to all queries is bounded by $O(d\cdot \Delta)\leq O\left (n\cdot (n-|\beta^0|+\Delta)\cdot  (\log N)^{64}\right )$,
since, from the problem definition, $\Delta\cdot d\leq \eta\cdot \left(|V(G)|-|\beta|+\Delta\cdot (\log N)^{64}\right )$, and $\eta\leq \Delta\leq |\beta|\leq n$.

Recall that, in our setting, $\rho=O(\log^3N)$, and so the running time of the algorithm for the \connecttocenters problem from 
\Cref{thm: connecttocenters} is bounded by:

\[	O\left(c_1\cdot W^0\cdot ( W^0-|\beta^0|+|\beta'| +\Lambda)\cdot  2^{c_2\sqrt{\log N}}\cdot (\log N)^{16c_2(r-1)+8c_2+29}\right ).\]

Since, from \Cref{claim: bound lambda and deleted from beta},
$\Lambda\leq 4n-4|\beta^0|+5\Delta\cdot (\log N)^{64}$ and $|\beta'|\leq 45n-45|\beta^0|+56\Delta  \cdot (\log N)^{64}$, while $n=W^0$, we get that the total running time of the algorithm is bounded by:

\[\begin{split}	
&O\left(c_1\cdot W^0\cdot ( W^0-|\beta^0|+\Delta)\cdot  2^{c_2\sqrt{\log N}}\cdot (\log N)^{16c_2(r-1)+8c_2+100}\right )\\
&\quad\quad\quad\quad \leq c_1\cdot W^0\cdot (W^0-|\beta|+\Delta)\cdot 2^{c_2\sqrt{\log N}}\cdot (\log N)^{16c_2r}.
\end{split}
\]

Finally, recall that the algorithm for the \connecttocenters problem may terminate with a ``FAIL'', with probability at most $\frac{1}{N^7}$. If this happens, then our algorithm terminates with a ``FAIL'' as well, but the probability for this happening is bounded by $\frac{1}{N^7}$.

%% file: maintaincluster-main.tex
\section{From the \maintaincluster Problem to \routeandcut Problem -- Comleting the Proof of \Cref{thm: from routeandcut to maintaincluster}}
\label{sec: alg for maintaincluster from routeandcut}
	
The goal of this section is to complete the proof of 	 \Cref{thm: from routeandcut to maintaincluster}, by providing an algorithm for the \maintaincluster problem, that is summarized in the following theorem.

\begin{theorem}\label{thm: from routeandcut to maintaincluster part 2}
	Suppose that, for some parameter $r\geq 1$, there exists a randomized algorithm for the $r$-restricted \routeandcut problem, that, given as an input an instance $(G,A,B,\Delta,\eta,N)$ of the problem with $|V(G)|=n$,
	has running time at most $c_1\cdot n\cdot (n-|B|)\cdot 2^{c_2\sqrt{\log N}}\cdot (\log N)^{16c_2(r-1)+8c_2}$. Then there exists a randomized algorithm for the $r$-restricted \maintaincluster problem, that, on input $(G,\Delta,\eta,d,N)$ has running time at most: $c_1\cdot  (W^0(G))^2\cdot 2^{c_2\sqrt{\log N}}\cdot (\log N)^{16c_2r}+c_1\cdot W^0(G)\cdot\Delta\cdot\log^4N$.
\end{theorem}

We assume that there exists a randomized algorithm $\aset$  for the $r$-restricted \routeandcut problem, for some integer $r\geq 1$, that, on an input an instance $(G,A,B,\Delta,\eta,N)$ of the problem with $|V(G)|=n$,
has running time at most $c_1\cdot n\cdot (n-|B|)\cdot 2^{c_2\sqrt{\log N}}\cdot (\log N)^{16c_2(r-1)+8c_2}$.

Recall that the input to the $r$-restricted \maintaincluster problem consists of a perfect well-structured $n$-vertex graph $G=(L,R,E)$ with a proper assignment $\ell(e)$ of lengths to its edges.  Recall that we have defined vertex weights for perfect well-structured graphs, where the weight $w(v)$ of a vertex $v$ is the length of the unique special edge incident to $v$. For every vertex $v\in V(G)$, we denote by $w^0(v)$ its initial weight, and we denote $W^0(G)=\sum_{v\in V(G)}w^0(v)$. Additionally, we are given a parameter $N\geq W^0(G)$, that is greater than a large enough constant, with $r\leq \ceil{\sqrt{\log N}}$,
and parameters $1\leq \Delta\leq N$, $1\leq \eta \leq \Delta$, and $(\log N)^{64} \leq d\leq 2^{r\cdot \sqrt{\log N}} $, such that $\frac{\Delta\cdot d}{\eta}\leq n$. In particular:

\begin{equation}\label{eq: bound on d} 
d\leq \frac{n\eta}{\Delta}\leq n.
\end{equation}

Since $d\geq (\log N)^{64}$, we get that $(\log N)^{64}\leq n\leq N$, and so:

\begin{equation}\label{eq: comparing logs}
\log\log N\leq \log n\leq \log N.
\end{equation}

Since we have assumed that $N$ is greater than a large enough constant, and in particular that $N>2^{2^{c_2}}$, we can also assume that $n$ is greater than a large enough constant.
We also get  following inequality that will be useful for us later:

\begin{equation}\label{eq: more logs}
\log (16N)=4+\log N\leq \left(1+\frac{1}{32c_2r}\right )\log N.
\end{equation}

Additionally, $\sqrt{\log(16N)}\leq \sqrt{4+\log N}\leq 2+\sqrt{\log N}$, and so:

\begin{equation}\label{eq: even more logs}
2^{c_2\sqrt{\log(16N)}}\leq 2^{2c_2}\cdot 2^{c_2\sqrt{\log N}}\leq  2^{c_2\sqrt{\log N}}\cdot \log N,
\end{equation}

since $N$ is sufficiently large.

The algorithm consists of at most $\Delta$ iterations.
At the beginning of every iteration $i$, the algorithm is given two vertices $x,y\in V(G)$, and it must return a simple path $P_i$ in $G$ connecting $x$ to $y$, whose length is at most $d$. (We may sometimes say that the algorithm is given a \shortpath query between the pair $(x,y)$ of vertices, and it returns the path $P_i$ in response to the query). 
The lengths of some of the special edges on path $P_i$ are then doubled. 
If $e$ is a special edge of $G$, and $\tau<\tau'$ are consecutive times during the algorithm's execution when the length of $e$ was doubled, then we are guaranteed that the number of the paths $P_i$ that the algorithm returned in response to the queries during time interval $(\tau,\tau']$ with $e\in E(P_i)$ is at least $\eta$.

%Throughout, we use a parameter $\phi=\frac{(\log N)^{64}}{d}$.
The algorithm may, at any time, produce a strongly well-structured cut $(X,Y)$ in $G$ of sparsity $\Phi_G(X,Y)\leq \frac{(\log N)^{64}}{d}$. 
If $w(X)\leq w(Y)$, then let $Z=X$ and $Z'=Y$; otherwise, let $Z=Y$ and $Z'=X$. We also let $J'\subseteq Z'$ denote the set of vertices that serve as endpoints of the edges of $E_G(X,Y)$. The vertices of $Z\cup J'$ are then deleted from $G$, and the algorithm continues with the resulting graph $G=G[Z'\setminus J']$. This modification ensures that $G$ remains a perfect well-structured graph. Once $|V(G)|\leq n/2$ holds, the algorithm terminates, even if fewer than $\Delta$ iterations have passed.
The algorithm may also, at any time, return ``FAIL'', but the probability of this ever happening must be bounded by $1/N^4$.

We can assume w.l.o.g. that, in the initial graph $G$, the length of every edge is at most $d$. Indeed, if $(x,y)$ is an edge of $G$ of length greater than $d$, then cut $(A,B)$ with $A=\set{x}$ and $B=V(G)\setminus\set{x}$ is a strongly well-structured cut of sparsity at most $\frac{1}{d}$, since $w(x),w(y)=d$, and $|E_G(A,B)|=1$. As long as special edges of length greater than $d$ lie in $G$, we can therefore produce $\frac{1}{d}$-sparse cuts that will result in the removal of such edges from $G$. The time required to perform this preprocessing step is $O(|E(G)|)\leq O(W^0(G))$. We will also ensure that, throughout the algorithm, every edge of $G$  has length at most $2d$. Therefore, if we denote by $W\attime$ the total weight of all edges at some time $\tau$ during the algorithm's execution, then $W\attime\leq 2n\cdot d\leq 2n^2$ must hold (we have used Inequality \ref{eq: bound on d}). Since  the algorithm terminates once the number of vertices falls below $n/2$,  the total weight $W$ of all vertices of $G$ satisfies $\frac{n}{2}\leq W\leq 2n^2$, at all times, and therefore, at all times the following inequality holds:

\begin{equation}\label{eq comparing other logs}
\log n-1\leq \log W\leq 4\log n.
 \end{equation}

Throughout, we denote by $G^0$ the initial graph $G$, and by $E^{\spec}$ be the set of all special edges of $G^0$.
For each such edge $e\in E^{\spec}$, let $\ell^0(e)$ be the length of $e$ at the beginning of the algorithm, and let $\ell^*(e)$ be defined as follows. If $e$ lies in graph $G$ at the end of the algorithm, then $\ell^*(e)$ is the length of edge $e$ at the end of the algorithm. Otherwise, $\ell^*(e)$ is the length of the edge $e$ just before it was deleted from $G$. We will use the following simple observation to bound $ \sum_{e\in E^{\spec}}\ell^*(e)$.
The proof is deferred to Section \ref{subsec: proof of final length of edges} of Appendix.

\begin{observation}\label{obs: final length of edges}
	\[ \sum_{e\in E^{\spec}}\ell^*(e)\leq 2\sum_{e\in E^{\spec}}\ell^0(e)+\frac{2d\Delta}{\eta} \leq  2\sum_{e\in E^{\spec}}\ell^0(e)+2n.  \]
\end{observation}

Throughout the algorithm, instead of working with graph $G$, it will be convenient for us to work with its corresponding subdivided graph $G^+$ (see \Cref{def: subdivided graph}). We denote the two sides of the bipartition of the vertices of $G^+$ by $L^+$ and $R^+$, respectively. Recall that graph $G^+$ is obtained from $G$ by replacing every special edge $e=(u,v)$ with a new path, that we denote by $P^*(e)$, and call \emph{the path represeinting $e$ in $G^+$}. The number of special edges on path $P^*(e)$ is exactly $\ell(e)$. For every special edge $e'$ on path $P^*(e)$, we say that $e$ is the \emph{parent-edge} of $e'$, and that $e'$ is the \emph{child-edge} of $e$.
	
Assume that at some time during the algorithm's execution, the length $\ell(e)$ of an edge $e\in E(G)$ is doubled. Then for every special edge $e'=(x,y)$ on path $P^*(e)$, we replace the edge by a path $(x,x',y',y)$, where $x'$ and $y'$ are new vertices, edges $(x,x')$ and $(y',y)$ are special, and edge $(x',y')$ is regular; we call this process the \emph{splitting} of edge $e'$. Notice that, since $(x,y)$ is a special edge of $G^+$, $x\in R^+$ and $y\in L^+$ must hold. We add $x'$ to $L^+$ and $y'$ to $R^+$. It is easy to verify that $G^+$ remains the subdivided graph of $G$ after this update.
	
In order to simplify the notation, we will denote $G^+$ by $H$ from now on, and we will also denote the bipartition $(L^+,R^+)$ by $(L',R')$. Recall that $H$ is a uniform perfect well-structured graph. The initial number of vertices in $H$ is $W^0$. As the algorithm progresses and edges of $H$ are subdivided, we will denote by $W$ the current number of vertices in $H$, which is equal to $\sum_{v\in V(G)}w(v)$. 
	
From now on we will mostly work with graph $H$. In every iteration, we are given a pair $(x,y)$ of vertices of $H$, and we need to return a simple path $P$ connecting $x$ to $y$ in $H$, of length at most $d$. Such a path can be easily trasformed into a path in graph $G$ of length at most $d$ connecting $x$ to $y$. After that, some special edges on path $P$ are split. 
From the definition of the \maintaincluster problem, we are guaranteed that the following property holds throughout the algorithm.

\begin{properties}{Q'}\label{prop: splitting of edges}
	\item  Consider any special edge $e$  of $H$ that did not lie in the original graph $H=(G^0)^+$. Assume that $e$ was added to $H$ at some time $\tau$ during the algorithm's execution, and was later split at some time $\tau'>\tau$ (so after time $\tau'$, edge $e$ no longer lies in $H$). Then the number of paths $P$ that the algorithm returned in response to queries during the time interval $(\tau,\tau']$ with $e\in E(P)$ is at least $\eta$.
\end{properties}

Notice that $W^0=2\sum_{e\in E^{\spec}}\ell^0(e)$, and
the total number of vertices that ever belonged to $H$ over the course of the algorithm is bounded by $\sum_{e\in E^{\spec}}2\ell^*(e)\le 4\sum_{e\in E^{\spec}}\ell^0(e)+4n\leq 2W^0+4n\leq 6W^0$, 
from \Cref{obs: final length of edges}. In particular, $|V(H)|\leq 6W^0$ holds at all times.

Throughout, we use a parameter $d'$, that is the smallest integral power of $2$ that is greater than $\frac{d}{(\log N)^{28}}$, and we set $\phi=\frac{4}{d'}$. Intuitively, if, at any time during the algorithm's execution, the length of any special edge $e=(x,y)$ in $G$ becomes at least $d'/2$, then we can immediately produce a cut $(A,B)$ with $A=\set{x}$ and $B=V(G)\setminus \set{x}$. It is easy to verify that $(A,B)$ is a strongly well-structured cut in $G$. Moreover, since $w(x),w(y)=\ell(e)$, the sparsity of the cut is at most $ \phi$, and so our algorithm could output this cut. In reality our algorithm is somewhat more complex, but it uses a similar method to ensure that, at all times, the lengths of all edges in $G$ are bounded by $d'$. Additionally, as long as the algorithm does not terminate, $\sum_{e\in E(G)}\ell(e)\geq \frac{|V(G)|}{2}\geq \frac{n}{4}>d'$ must hold.

Let $W$ be the total weight of all vertices of $G$ at any time during the algorithm's execution.
Then $W\geq |V(G)|\geq \frac{n}{2}\geq \frac{d}{2}$ (from Inequality \ref{eq: bound on d}). Therefore, throughout the algorithm:

\begin{equation}\label{eq: bound on d'}
d'\leq \frac{2d}{\log^{28} N}\leq \frac{4 W}{\log^{28}N}
\end{equation}

holds.
As long as the lengths of all edges in $G$ are bounded by $d'/2$,  our algorithm may produce, at any time, a weakly well-structured cut $(A,B)$ in the current graph $H$, of sparsity at most $\frac{\phi}{4}$. Let $A'=A\cap V(G)$ and $B'=B\cap V(G)$. From \Cref{obs: subdivided cut sparsity}, $(A,B)$ is a weakly well-structured cut in $G$ of sparsity at most $2\phi$. Furthermore, using the algorithm from \Cref{obs: weakly to strongly well str}, we can transform cut $(A,B)$ into a strongly well-structured cut $(A^*,B^*)$ in $G$, of sparsity at most $4\phi$. 
 We then return cut $(A^*,B^*)$, following which graph $G$ is updated with the deletion of the vertices on the side of the cut whose current weight is smaller, and the endpoints of the edges in $E_G(A^*,B^*)$. We update graph $H$ by deleting from it every vertex that was deleted from $G$, and, additionally, for every edge $e$ that was deleted from $G$, its corresponding path $P^*(e)$ is deleted from $H$. As before, the algorithm is allowed to halt and return ``FAIL'' at any time, but we are required that this happens with probability at most $1/N^4$.

\subsection{Partition into Phases and Probability of Failure}
\label{subsec: phases}

Our algorithm is partitioned into phases. A phase terminates if one of the following happened: (i) the algorithm returned ``FAIL'', in which case we say that the phase is of \emph{type 1}; or (ii) at least $\frac{\hat W}{\log^{30}N}$ vertices have been deleted from graph $H$ during the current phase, where $\hat W$ is the number of vertices in $H$ at the beginning of the phase, in which case we say that the phase is of \emph{type 2}; or (iii) at least $\Delta'=\frac{\Delta}{\log^{30}N}$ queries have been responded to; in which case we say that the phase is of \emph{type 3}.
We will ensure that the probability that the algorithm returns ``FAIL'' in a single phase is bounded by $\frac{1}{2}$.
 Since the total number of vertices that ever belong to graph $H$ is bounded by $6W^0\leq O(N)$, we get that the total number of type-2 and type-3 phases is bounded by $O(\log^{30}N\cdot \log N)\leq \log^{32}N$. 
If $\ceil{16\log N}$ consecutive type-1 phases occur, then we terminate the algorithm and return ``FAIL''. Notice that, for each type-2 phase and type-3 phase, the probability that it is followed by $\ceil{16\log N}$ consecutive type-1 phases is bounded by $\left(\frac 1 2 \right )^{\ceil{16\log N}}\leq \frac{1}{N^{16}}$. Therefore, overall, the probability that the algorithm ever returns ``FAIL'' is bounded by $\frac{(\log N)^{32}}{N^{16}}<\frac{1}{N^4}$, as required. The total number of phases in the algorithm is then bounded by $O(\log^{32}N\cdot \log N)\leq O(\log^{33}N)$. 
\iffalse
We will ensure that the total running time of a single phase is bounded by:

\[O\left(\hat W^2\cdot 2^{c_2\sqrt{\log N}}\cdot (\log N)^{16c_2r-c_2}\right ),\]

where $\hat W$ is the number of vertices in $H$ at the beginning of the phase. Since $\hat W\leq 12W^0$, and since $c_2$  is a sufficiently large constant, 
this will ensure that the total running time of the algorithm is bounded by $c_1\cdot  (W^0(G))^2\cdot 2^{c_2\sqrt{\log N}}\cdot (\log N)^{16c_2r}$, as required.
\fi

At a very high level, our algorithm follows a somewhat standard framework, though the specific details are quite different. As a first step, we select a random subset $T$ of vertices of $H$, that we call \emph{terminals}, where every vertex of $H$ is selected to be added to $T$ independently with probability $\frac 1 {d'}$. Next, we will attempt to embed an expander $X$, whose vertex set is $T$, into graph $H$, so that every embedding path has length $\otilde(d')$, and the paths cause edge-congestion $O(\poly\log n)$. We show that, with high probability, if we are unable to compute such an embedding, then we can compute a cut $(A,B)$ in $H$, that can be transformed into a $\frac{(\log N)^{64}}{d}$-sparse strongly well-structured cut $(A',B')$ in $G$ with $w(A'),w(B')\geq \frac{\hat W}{\log^{30}N}$. We then output this cut, and update graphs $G$ and $H$ with the corresponding deletions of vertices and edges, which must lead to the deletion of at least $\frac{\hat W}{\log^{30}N}$ vertices from $H$. Therefore, with high probability, if we cannot successfully compute an expander $X$, and embed it into $H$, the current phase will terminate. Assume now that we have successfully computed an expander $X$ over the set $T$ of roughly $\frac {\hat W}{d'}$ vertices, and embedded it into $H$. 
Then we will employ the \LCDS from \Cref{sec: layered connectivity DS}, with the set $V(X)$ of vertices serving as the set $\beta$ of centers, in order to maintain a collection $\qset^{\out}$ of short paths, connecting every vertex of $V(H)$ to some vertex of $V(X)$. We will also employ the same data structure on graph $\revG$, that is obtained from $G$ by reversing the direction of each of its edges, to maintain a collection $\qset^{\inn}$ of short paths, that, for every vertex $v\in V(H)$, contains a path that connects some vertex of $V(X)$ to $v$.

We now formally describe the algorithm for executing a single phase. 
We assume w.l.o.g. that, at the beginning of the phase, graph $G$ does not contain any edge of length greater than $d'$. Otherwise, as long as such an edge $e=(x,y)$ lies in $G$, then cut $(A,B)$ with $A=\set{x}$ and $B=V(G)\setminus\set{x}$ is a strongly well-structured cut of sparsity $\frac{1}{d'}\leq \phi$, since $w(x)=d'$, and $|E_G(A,B)|=1$. As long as special edges of length greater than $d'$ lie in $G$, we can therefore produce $\phi$-sparse cuts that will result in the removal of such edges from $G$. The time required to perform this preprocessing step, including updating the graph $H$, is $O(|V(H)|)\leq O(W^0)$.
Therefore, we assume that, at the beginning of the phase, $G$ does not contain edges of length greater than $d'$. We denote by $\hat W$ the number of vertices in $H$ at the beginning of the phase.

We start with a preprocessing step, in which we sample the vertices of $H$ to be added to the set $T$. Then we describe our two main data structures. The first data structure maintains a large expander that is defined over the set $T$ of vertices, while the second data structure is \LCDS in graphs $G$ and $\revG$.

\input{maintaincluster-subsampling}

\input{maintaincluster-finalizing}

\input{maintaincluster-expander}

\input{connect-to-terminals}

\input{cleanup}

%% file: maintaincluster-subsampling.tex
\subsection{Preprocessing Step: Subsampling}
\label{subsec: subsample}

Recall that the number of  vertices in $H$ at the beginning of the phase, denoted by $\hat W$, is bounded by $12W^0\leq 12N$.

In this step, we construct a set $T$ of vertices of $H$, that we call \emph{terminals}, as follows: every vertex $v\in V(H)$ is added to $T$ with probability $\frac 1{d'}$ independently. This is the only randomized part in our algorithm for the \MBM problem. Next, we will define four ``bad'' events, and we will show that each of them may only happen with low probability. Eventually, if our algorithm returns ``FAIL'' and halts, we will guarantee that at least one of these bad events must have happened. This will ensure that the probability that the phase terminates with the ``FAIL'' is at most $\half$, as required.

\subsubsection{Bad Event $\event_1$}

Observe first that the expected cardinality of $T$ is $\frac{\hat W}{d'}$. We say that bad event $\event_1$ happens if $|T|>\frac{8\hat W}{d'}$ or $|T|<\frac{\hat W}{2d'}$. From \Cref{lem: Chernoff}, $\prob{|T|>\frac{8\hat W}{d'}}\leq 2^{-8\hat W/d'}\leq \frac{1}{512}$, and $\prob{|T|<\frac{\hat W}{2d'}}\leq e^{-\hat W/(8d')}\leq \frac{1}{512}$, since $d'\leq  \frac{4 \hat W}{\log^{28}N}$ from Inequality \ref{eq: bound on d'}. 
Overall, we get that $\prob{\event_1}\leq \frac{1}{256}$, and, if $\event_1$ does not happen, then $\frac{\hat W}{2d'}\leq |T|\leq \frac{8\hat W}{d'}$ holds.

Before we define the remaining two bad events, we need to inspect the well-structured cuts of $H$ more closely.

\subsubsection{Well-Structured Cuts and Their Structure}

Recall that a cut $(A,B)$ of $H$ is weakly well-structured iff all edges of $E_H(A,B)$ are special. Let $E'$ be any subset of special edges of $H$. We denote by $\cset(E')$ the collection of all weakly well-structured cuts $(A,B)$ in $H$ with $E_H(A,B)=E'$. Next, we investigate the structure of the cuts in $\cset(E')$.

Let $J(E')$ denote the set of all vertices $v\in V(H)$, such that some edge of $E'$ leaves $v$, and let $J'(E')$ denote the set of all vertices $u\in V(H)$, such that some edge of $E'$ enters $u$. Next, we denote by $S^*(E')$ the set of all vertices $v\in H$, such that there is a path in $H\setminus E'$ connecting a vertex of $J(E')$ to $v$. Notice that, for every cut $(A,B)\in \cset(E')$, it must be the case that $S^*(E')\subseteq A$. Indeed, assume otherwise, and let $v\in S^*(E')$ be any vertex that lies in $B$. Since $J(E')\subseteq A$ must hold, there is a path $P$ connecting a vertex of $A$ to $v$, such that $P$ is disjoint from the edges of $E'$. But then some edge of $P$ must lie in $E_H(A,B)$, contradicting the fact that $E_H(A,B)=E'$.

Similarly, we denote by $S^{**}(E')$ the set of all vertices $v\in H$, such that there is a path in $H\setminus E'$ connecting $v$ to a vertex of $J'(E')$. From the same arguments as above, for every cut $(A,B)\in \cset(E')$,  $S^{**}(E')\subseteq B$ must hold.

Let $\xset'$ be the set of all strongly connected components of $H\setminus \left (S^*(E')\cup J(E')\cup S^{**}(E')\cup J'(E')\right )$, where for each $X\in \xset'$, we view $X$ as a set of vertices corresponding to the connected component. Denote $\xset(E')=\xset'\cup\set{S^*(E')\cup J(E'),S^{**}(E')\cup J'(E')}$. Clearly, $\xset(E')$ is a partition of $V(H)$. Consider again any cut $(A,B)\in \cset(E')$. As observed already, $S^*(E')\cup J(E')\subseteq A$ and $S^{**}(E')\cup J'(E')\subseteq B$ must hold. Moreover, it is easy to see that, since $E_H(A,B)=E'$, for all $X\in \xset'$, either $X\subseteq A$, or $X\subseteq B$ must hold. From the definition of the sets $S^*(E'),S^{**}(E')$ of vertices, graph $H\setminus E'$ contains no path connecting a vertex of $S^*(E')\cup J(E')$ to a vertex lying outside this set. Similarly, no path in $H\setminus E'$ may connect a vertex of $V(H)\setminus \left(S^{**}(E') \cup J'(E')\right )$ to a vertex of $S^{**}(E')\cup J'(E')$. It is now easy to verify that, for every pair $X,X'\in \xset(E')$ of sets, it cannot be the case that $H\setminus E'$ contains both a path $P$ connecting a vertex of $X$ to a vertex of $X'$, and a path $P'$ connecting a vertex of $X'$ to a vertex of $X$. Indeed, if such paths did exist, then neither of these paths could contain vertices of $S^*(E')\cup J(E')\cup S^{**}(E')\cup J'(E')$, and so the vertices of $X\cup X'$ should have belonged to a single strongly connected component of $H\setminus \left(S^*(E')\cup J(E')\cup S^{**}(E')\cup J'(E')\right)$.

We can now define an arbitrary topological ordering $(X_1,X_2,\ldots,X_q)$ of the sets in $\xset(E')$, such that $X_1=S^{**}(E')\cup J(E')$, $X_q=S^{*}(E')\cup J'(E')$, and for all $1\leq i<j\leq q$, there is no edge of $E(H)\setminus E'$ that connects a vertex of $X_j$ to a vertex of $X_i$. Notice that, for all $1\leq i\leq q$, we can define a cut $(A_i,B_i)$ with $A_i=X_{i+1}\cup\cdots\cup X_q$ and $B_i=X_1\cup\cdots\cup X_i$. From our discussion so far, the only edges that connect vertices of $A_i$ to vertices of $B_i$ are the edges of $E'$, and so for each such $1\leq i< r$, cut $(A_i,B_i)$ lies in $\cset(E')$. We call all cuts in $\set{(A_i,B_i)\mid 1\leq i\leq q}$ \emph{canonical cuts} corresponding to $E'$, and denote the set of all such cuts by $\cset'(E')$. Notice that $\cset'(E')\subseteq \cset(E')$, and it is possible that $\cset(E')\setminus\cset'(E')\neq \emptyset$.
Note that the definition of the set $\cset'(E')$ of cuts depends on the topological ordering of the sets in $\xset(E')$, and there may be several such topological orderings. We assume that we are given some tie-breaking rule that selects a single topological ordering of the sets in $\xset(E')$, and we define $\cset'(E')$ with respect to that topological ordering. We also assume that the above unique topological ordering can be computed in time $O(n)$, given the sets in $\xset(E')$.
We are now ready to define the second bad event.

\subsubsection{Bad Event $\event_2$}

 We start with intuition. Recall that our goal is to embed an expander over the set $T$ of vertices into $H$, with polylogarithmic congestion. Intuitively, if we are not able to do so, we will construct a cut $(A,B)$ in $H$, such that $|T\cap A|,|T\cap B|\geq \Omega(|T|/\poly\log N)$, and $|E_H(A,B)|$ is sufficiently low. If this happens, we would like to be able to compute a sparse and a roughly balanced cut in graph $H$. If we denote by $E'=E_H(A,B)$, then ideally we would like to ensure that some canonical cut corresponding to $E'$ has this property. The second bad event happens if this is not the case for some such cut $(A,B)$. Next, we formally define the second bad event, and show that it can only happen with small probability.

We say that bad event $\event_2$ happens if there is a weakly well-structured cut $(A,B)$ in $H$, with $|A\cap T|,|B\cap T|\geq \frac{|T|}{\log^4N}$, and $|E_H(A,B)|\leq \frac{|T|}{\log^6N}$, such that, if we denote by $E'=E_H(A,B)$, then for every canonical cut $(A',B')\in \cset'(E')$, either $|A'|\leq \frac{\hat W}{\log^{5}N}$, or $|B'|\leq \frac{\hat W}{\log^{5}N}$ holds.

\begin{claim}\label{claim: second bad event}
	\[\prob{\event_2}\leq \frac{1}{64}.\]
\end{claim}

\begin{proof}
%	Throughout the proof, we assume that Event $\event_1$ did not happen, so $\frac{\hat W}{2d'}\leq |T|\leq \frac{8\hat W}{d'}$ holds. 
Recall that, if event $\event_1$ does not happen, then $\frac{\hat W}{2d'}\leq |T|\leq \frac{8\hat W}{d'}$ holds, and so, if $E'$ is a set of special edges with $|E'|\leq \frac{|T|}{\log^6N}$, then $|E'|\leq \frac{8\hat W}{d'\log^6N}$.
	
	For every subset $E'$ of at most $\frac{8\hat W}{d'\log^6N}$ special edges, we define a bad event $\event(E')$ to be the event that there is some cut $(A,B)$ in $H$ with $E_H(A,B)=E'$, such that $|A\cap T|,|B\cap T|\geq \frac{|T|}{\log^4N}$, and for every canonical cut $(A',B')\in \cset'(E')$, either $|A'|\leq \frac{\hat W}{\log^{5}N}$, or $|B'|\leq \frac{\hat W}{\log^{5}N}$ holds.
	We bound $\prob{\event(E')}$ for each such set $E'$ of edges separately, and then use the union bound in order to bound $\prob{\event_2}$.
	
	Fix any subset $E'$ of at most $\frac{8\hat W}{d'\log^6N}$ special edges. Let $X\in \xset(E')$ be the set of vertices with maximum cardinality. We claim that Event $\event(E')$ may only happen if $|X|\geq \hat W-\frac{2\hat W}{\log^{5}N}$, in the following simple observation.
	
	\begin{observation}\label{bad event large component}
		If  $\event(E')$ happened, then $|X|\geq \hat W-\frac{2\hat W}{\log^{5}N}$ must hold.
	\end{observation}
\begin{proof}
	Assume otherwise, that is, Event $\event(E')$ happened, but  $|X|< \hat W-\frac{2\hat W}{\log^{5}N}$. 
	 Denote $\xset(E')=(X_1,X_2,\ldots,X_q)$, where the sets are listed according to their topological order, and assume that $X=X_i$. 
	 
	 Assume first that $|X_i|\geq \frac{\hat W}{4}$. Then either $|\bigcup_{i'<i}X_{i'}|\geq \frac{\hat W}{\log^{5}N}$, or  $|\bigcup_{i'>i}X_{i'}|\geq \frac{\hat W}{\log^{5}N}$. In the former case, cut $(A_{i},B_{i})$ is a canonical cut with $|A_{i}|,|B_{i}|\ge \frac{\hat W}{\log^{5}N}$, while in the latter case, cut  $(A_{i-1},B_{i-1})$ is a canonical cut with $|A_{i-1}|,|B_{i-1}|\ge \frac{\hat W}{\log^{5}N}$, a contradiction (if $i=1$ then we are guaranteed that $(A_i,B_i)$ is the cut with the desired properties, and if $i=q$, then we are guaranteed that $(A_{q-1},B_{q-1})$ is such a cut).
	 
	 Assume next that $|X_i|<\frac{\hat W}{4}$. In this case, we let $1<j<q$ be the smallest index for which $\sum_{j'<j}|X_{j'}|\geq \frac{\hat W}{2}$. It is easy to verify that $\sum_{j'<j}|X_{j'}|\leq \frac{3\hat W}{4}$ must hold, and so cut $(A_{j},B_{j})\in \cset'(E')$ satisfies that both $|A_{j}|,|B_{j}|\geq \frac{\hat W}{\log^{5}N}$, a contradiction.
	\end{proof}
	 
	 We assume from now on that  $|X|\geq \hat W-\frac{2\hat W}{\log^{5}N}$ holds for the set $E'$ of edges. Let $Y=V(H)\setminus X$, so $|Y|\leq \frac{2\hat W}{\log^{5}N}$. Consider now any cut $(A,B)\in \cset(E')$, and assume that $|A\cap T|,|B\cap T|\geq \frac{|T|}{\log^4N}$. Notice that all vertices of $X$ must lie on one side of the cut $(A,B)$. If $X\subseteq A$, then $B\subseteq Y$, and, since $|T\cap B|\geq \frac{|T|}{\log^4N}$, we get that $|T\cap Y|\ge \frac{|T|}{\log^4N}$. Similarly, if $X\subseteq B$, then $A\subseteq Y$, and, since 
	$|T\cap A|\geq \frac{|T|}{\log^4N}$, we get that $|T\cap Y|\ge \frac{|T|}{\log^4N}$. In either case, we conclude that Event $\event(E')$ may only happen if $|X|\geq \hat W-\frac{2\hat W}{\log^{5}N}$ and $|T\cap Y|\geq \frac{|T|}{\log^4N}$. 
	
	We start by bounding the probability that $|T\cap X|<\frac{\hat W}{4d'}$. Since we have assumed that $|X|\geq \hat W-\frac{2\hat W}{\log^{5}N}\geq \frac{\hat W}{2}$, and since every vertex is selected into $T$ independently with probability $1/d'$, we get that $\expect{|T\cap X|}\geq \frac{\hat W}{2d'}$. From \Cref{lem: Chernoff}, $\prob{|T\cap X|<\frac{\hat W}{4d'}}\leq e^{-\hat W/(16d')}$.
	
	Assume now that $|T\cap X|\geq \frac{\hat W}{4d'}$, so in particular $|T|\geq \frac{\hat W}{4d'}$. In this case, if $|T\cap Y|\geq \frac{|T|}{\log^4N}$, then $|T\cap Y|\geq \frac{\hat W}{4d'\log^4N}$ must hold.
	Since $\expect{|T\cap Y|}= \frac{|Y|}{d'}\leq \frac{2\hat W}{d'\log^5N}$, from \Cref{lem: Chernoff}, $\prob{|T\cap Y|\geq \frac{\hat W}{4d'\log^4N}}\leq 2^{-\hat W/(4d'\log^4N)}\leq e^{-\hat W/(16d')}$.
	
Altogether, we get that:

\[\prob{\event(E')}\leq \prob{|T\cap X|<\frac{\hat W}{4d'}}+\prob{|T\cap Y|\geq \frac{|T|}{\log^4N}\mid |T\cap X|\geq \frac{\hat W}{4d'}}\leq 2\cdot e^{-\hat W/(16d')}. 
\]

	Observe that $H$ has $\hat W$ special edges, and so the number of sets $E'$ containing at most $\frac{8\hat W}{d'\log^6N}$ special edges is bounded by:
	
	\[ (2\hat W)^{8\hat W/(d'\log^6N)}\leq (24N)^{8\hat W/(d'\log^6N)}\leq 2^{32\hat W/(d'\log^5N)} .\]
	
	(We have used the fact that $\hat W\leq 12\hat W^0\leq 12N$). Using the union bound over all sets $E'$ of at most  $\frac{8\hat W}{d'\log^6N}$ special edges, we get that the probability that event $\event(E')$ happens for any such set of edges is bounded by:
	
	\[ 2^{\frac{32\hat W}{d'\log^5N}} \cdot 2\cdot 2^{-\frac{\hat W}{16d'\log^4N}}\leq 2^{-\frac{\hat W}{32d'\log^4N}}\leq \frac 1 {512}, \]
	
	since $d'\leq  \frac{8 \hat W}{\log^{28}N}$ from Inequality \ref{eq: bound on d'}.

Finally, if Event $\event_2$ happened, then either $\event_1$ happened, or, for some set $E'$ of at most  $\frac{8\hat W}{d'\log^6N}$ special edges, event $\event(E')$ has happened. Therefore, overall, $\prob{\event_2}\leq \frac{1}{256}+\frac{1}{512}\leq \frac{1}{64}$.
\end{proof}

\subsubsection{Bad Event $\event_3$}

We start again with intuition. Once we construct the expander $X$ and embed it into $H$, we will need to route the remaining vertices of $H$ to the vertices of $T$ via paths that cause congestion roughly $\otilde(d')$. If we are unable to do so, we will compute a sparse cut that separates some vertices of $H$ from most terminals in $T$. Denote this cut by $(A,B)$. The bad scenairo that we would like to avoid is when one side of the cut (say $A$) contains almost all vertices of $H$, while the other side of the cut contains few vertices of $H$, but it contains most terminals. If this happens, then the vertices of $B$ will be subsequently deleted from $H$, thereby destroying the expander that we have constructed. As we need to invest a large amount of time into constructing the expander $X$ and embedding it into $H$, we would like to avoid such a situation. Let $E'=E_G(A,B)$. If some canonical cut $(A',B')\in \cset'(E')$ is roughly balanced with respect to the vertices of $H$, then we can return cut $(A',B')$ instead of cut $(A,B)$, and ensure that a sufficiently large number of vertices is deleted from $H$, and the current phase can be terminated. However, if $\cset'(E')$ contains no such cut, then it is not clear how to bypass the problem mentioned above. The third bad event is designed to capture exactly this scenario. We now formally define it and show that it is unlikely to happen.

We say that the bad event $\event_3$ happens if there is 
a weakly well-structured cut $(A,B)$ in $H$, with $|E_H(A,B)|\leq \frac{|T|}{\log^6N}$, such that one of the sides of the cut contains at least $\frac{\hat W} 2$ vertices of $H$ and at most $\frac{|T|}{32}$ vertices of $T$, and additionally,  if we denote by $E'=E_H(A,B)$, then for every canonical cut $(A',B')\in \cset'(E')$, either $|A'|\leq \frac{\hat W}{\log^{5}N}$, or $|B'|\leq \frac{\hat W}{\log^{5}N}$ holds.

\begin{claim}\label{claim: third bad event}
	\[\prob{\event_3}\leq \frac{1}{64}.\]
\end{claim}

\begin{proof}
	As before, if event $\event_1$ does not happen, and $E'$ is a set of special edges with $|E'|\leq \frac{|T|}{\log^6N}$, then $|E'|\leq \frac{8\hat W}{d'\log^6N}$ must hold.
For every subset $E'$ of at most $\frac{8\hat W}{d'\log^6N}$ special edges, we define a bad event $\event'(E')$ to be the event that there is some weakly well-structured cut $(A,B)$ in $H$ with $E_H(A,B)=E'$, such that one side of the cut cut contains at least $\hat W/2$ vertices of $H$ and at most $\frac{|T|}{32}$ vertices of $T$, and additionally,
 for every canonical cut $(A',B')\in \cset'(E')$, either $|A'|\leq \frac{\hat W}{\log^{5}N}$, or $|B'|\leq \frac{\hat W}{\log^{5}N}$ holds.
	We bound $\prob{\event'(E')}$ for each such set $E'$ of edges separately, and then use the union bound in order to bound $\prob{\event_3}$.
	
	Fix any subset $E'$ of at most $\frac{8\hat W}{d'\log^6N}$ special edges. Let $X\in \xset(E')$ be the set of vertices with maximum cardinality. Using the same arguments as in the proof of \Cref{claim: second bad event}, Event $\event'(E')$ may only happen if $|X|\geq \hat W-\frac{2\hat W}{\log^{5}N}$. We assume from now on that  $|X|\geq \hat W-\frac{2\hat W}{\log^{5}N}$ holds.
	Let $Y=V(H)\setminus X$, so that $|Y|\leq \frac{2\hat W}{\log^5N}$.

Consider now some weakly well-structured cut $(A,B)$ in $H$ with $E_H(A,B)=E'$, and assume that one of the sets $Z\in \set{A,B}$ contains at least $\hat W/2$ vertices of $H$ and at most $\frac{|T|}{32}$ vertices of $T$. Since set $X$ must be contained in one of the sides of the cut, it must be the case that $X\subseteq Z$, and $|T\cap X|\leq \frac{|T|}{32}=\frac{|T\cap X|+|T\cap Y|}{32}$. Therefore, $|T\cap X|\leq \frac{|T\cap Y|}{31}$ must hold.

Using the same reasoning as in the proof of  \Cref{claim: second bad event}, $\prob{|T\cap X|<\frac{\hat W}{4d'}}\leq e^{-\hat W/(16d')}$. 

Assume now that $|T\cap X|\geq \frac{\hat W}{4d'}$. In this case, if $|T\cap X|\leq \frac{|T|}{32}$ holds, then $|T\cap Y|\geq 31|T\cap X|\geq \frac{31\hat W}{4d'}$ must hold.
Notice that, since $|Y|\leq \frac{2\hat W}{\log^5N}$, $\expect{|T\cap Y|}\leq \frac{2\hat W}{d'\log^5N}$. From \Cref{lem: Chernoff}, $\prob{|T\cap Y|\geq \frac{31\hat W}{4d'}}\leq 2^{-31\hat W/(4d')}\leq e^{-\hat W/(16d')}$.

Altogether, we get that:

\[\prob{\event'(E')}\leq \prob{|T\cap X|<\frac{\hat W}{4d'}}+\prob{|T\cap X|\leq \frac{|T|}{32}\mid |T\cap X|\geq \frac{\hat W}{4d'}}\leq 2\cdot e^{-\hat W/(16d')}. 
\]

As in the proof of \Cref{claim: second bad event}, 	the number of sets $E'$ containing at most $\frac{8\hat W}{d'\log^6N}$ special edges is bounded by $2^{32\hat W/(d'\log^5N)}$.

Using the union bound over all sets $E'$ of at most  $\frac{8\hat W}{d'\log^6N}$ special edges, we get that the probability that event $\event'(E')$ happens for any such set of edges is bounded by:

\[ 2^{\frac{32\hat W}{d'\log^5N}} \cdot 2\cdot 2^{-\frac{\hat W}{16d'\log^4N}}\leq 2^{-\frac{\hat W}{32d'\log^4N}}\leq \frac 1 {512}, \]

since $d'\leq  \frac{8 \hat W}{\log^{28}N}$ from Inequality \ref{eq: bound on d'}.

Finally, if Event $\event_3$ happened, then either event $\event_1$ happened, or, for some set $E'$ of at most  $\frac{8\hat W}{d'\log^6N}$ special edges, event $\event'(E')$ has happened. Therefore, overall, $\prob{\event_3}\leq \frac{1}{256}+\frac{1}{512}\leq \frac{1}{64}$.
\end{proof}

\subsubsection{Bad Event $\event_4$}

Recall that for every edge $e\in E(G)$, there is a corresponding path $P^*(e)$, containing $2\ell(e)$ vertices, that represents $e$ in graph $H$. Recall that we have ensured that, at the beginning of the phase, the length of every edge $e\in E(G)$ is at most $d'$. The fourth bad event, $\event_4$, happens if there is some edge $e\in E(G)$, such that at least $16\log N$ vertices of $P^*(e)$ are in $T$.

\begin{claim}
	$\prob{\event_4}\leq \frac{1}{64}$.
\end{claim}
\begin{proof}
	For every special edge $e\in E(G)$, let $\event''(e)$ be the bad event that at least $16\log N$ vertices of $P^*(e)$ are added to $T$. Since the total number of vertices on $P^*(e)$ is $2\ell(e)\leq 2d'$, and every edge is chosen to be added to $T$ with probability $\frac{1}{d'}$, the expectation of $|P^*(e)\cap T|$ is at most $2$. From \Cref{lem: Chernoff}, the probability of $\event''(e)$ is then bounded by $2^{-16\log N}\leq N^{-16}$.
	
	By using the union bound over all special edges $e\in E(G)$, and since the number of such edges is bounded by $n\leq N$, we get that:
	
	\[\prob{\event_4}\leq N\cdot N^{-16}\leq \frac{1}{64}.\]
\end{proof}

Lastly, we let $\event$ be the bad event that either of the events $\event_1,\event_2$, $\event_3$, or $\event_4$ happened. Clearly:

\[\prob{\event}\leq \sum_{i=1}^4\prob{\event_i}\leq \frac 1 {16}.\]

Note that we can check, in time $O(N)$, whether events $\event_1$ and $\event_4$ hold. If either of these events holds, we terminate the algorithm and return ``FAIL''. Therefore, we assume from now on that neither of the events $\event_1,\event_4$ happened.

As the phase progresses, vertices may be deleted from graphs $G$ and $H$. We view the set $T$ of terminal vertices as static, that is, it does not change even as some vertices of $T$ are deleted from $H$.

%% file: maintaincluster-finalizing.tex
\subsection{Finalizing Cuts}
\label{subsec: finalizing cuts}

In this subsection we define a special kinds of cuts in graph $H$, that we call \emph{finalizing} cuts. If, at any time during the execution of a phase, our algorithm computes such a cut in the current graph $H$, we will be able to terminate the current phase, by either correctly establishing that at least one of the events $\event_2,\event_3$ must have happened, or by providing a strongly well-structured $\phi$-sparse cut $(A,B)$ in the current graph $G$, such that $w(A),w(B)\geq \frac{\hat W}{\log^{30}N}$. In the former case, we will terminate the current phase with a ``FAIL'', while in the latter case, once cut $(A,B)$ is processed, and the corresponding vertices and edges are deleted from $H$, we are guaranteed that at least $\frac{\hat W}{\log^{30}N}$ vertices have been deleted from $H$ since the beginning of the phase, and so the current phase terminates. We are now ready to define finalizing cuts in $H$.

\begin{definition}[Finalizing Cut]\label{def: finalizing cut}
	Let $(A,B)$ be a cut in graph $H$ at any time during the execution of a phase. We say that the cut $(A,B)$ is a \emph{finalizing} cut, if the following hold:
	
	\begin{itemize}
		\item all edges in $E_H(A,B)$ are special edges; 
		\item $|E_H(A,B)|\leq \frac{|T|}{8\log^6N}$; and
		\item one of the sides of the cut $Z\in \set{A,B}$ contains at least $|V(H)|-\frac{\hat W}{32}$ vertices of $H$, and at most $|T|\cdot \left(1-\frac{7}{\log^4N}\right )$ vertices of $T$.
	\end{itemize}
\end{definition}

Next, we show that, if our algorithm ever produces a finalizing cut during a phase, then the phase can be terminated, in the following claim.

\begin{claim}\label{claim: finalizing cut}
	Suppose, at any time during the execution of a phase, there is a finalizing cut $(\tilde A,\tilde B)$ in the current graph $H$. Assume further that neither of the events $\event_2,\event_3$ has happened, and that every edge in the current graph $G$ has length at most $d'/2$. Then there is a strongly well-structured $\phi$-sparse cut $(A^*,B^*)$ in the current graph $G$, such that $w(A^*),w(B^*)\geq \frac{\hat W}{\log^{30}N}$. Moreover, there is a deterministic algorithm, that, given the graph $G$ at the beginning of the current phase (that we denote by $G_0$), the sequence of the strongly well-structured $\phi$-sparse cuts $(A_1,B_1),\ldots(A_q,B_q)$ in $G$ that the algorithm has produced during the current phase so far,  the current graphs $G$ and $H$, and the finalizing cut $(\tilde A,\tilde B)$ in the current graph $H$, either correctly estabishes that at least one of the events $\event_2,\event_3$ has happened,  or computes a strongly well-structured $\phi$-sparse cut $(A^*,B^*)$ in the current graph $G$, such that $w(A^*),w(B^*)\geq \frac{\hat W}{\log^{30}N}$. The running time of the algorithm is $O(\hat W^2)$.
\end{claim}

The remainder of this subsection is dedicated to the proof of \Cref{claim: finalizing cut}. 
We denote by $G_0$ the graph $G$ at the beginning of the current phase, and by $H_0$ its corresponding subdivided graph $H$. The main idea of the proof is to exploit the cuts $(A_1,B_1),\ldots(A_q,B_q)$ in graph $G$, together with the finalizing cut $(\tilde A,\tilde B)$ in graph $H$, in order to compute another well-structured cut $(\tilde S,\tilde S')$ in $H_0$, such that $|E_H(\tilde S,\tilde S')|\leq \frac{|T|}{\log^6N}$, and one side of the cut contains all but at most $\frac{\hat W}{16}$ vertices of $H_0$, but only a relatively small fraction of the vertices of $T$. Let $E''=E_H(\tilde S,\tilde S')$. By inspecting all canonical cuts in $\cset'(E'')$, we will then either correctly establish that  at least one of the events $\event_2,\event_3$ happened, or we will compute a cut $(A^*,B^*)$ in $G$ with the required properties.

In the following technical claim, whose proof is deferred to Section \ref{sec: proof of intermediate finalizing cut} of Appendix, we provide an algorithm for computing the cut $(\tilde S,\tilde S')$ in $H_0$ with the desired properties.

\begin{claim}\label{claim: compute intermediate cut}
There is a deterministic algorithm, that, given the graph $G_0$, the sequence of the strongly well-structured $\phi$-sparse cuts $(A_1,B_1),\ldots(A_q,B_q)$ in $G$ that the algorithm has produced during the current phase so far, the current graphs $G$ and $H$, such that all edge lengths in $G$ are bounded by $d'/2$, and the finalizing cut $(\tilde A,\tilde B)$ in the current graph $H$, computes a weakly well-structured cut $(\tilde S,\tilde S')$ in $H_0$, such that $|E_{H_0}(\tilde S,\tilde S')|\leq \frac{|T|}{\log^6N}$, and one side $\tilde Z\in \set{\tilde S,\tilde S'}$ of the cut contains all but at most $\hat W/8$ vertices of $H_0$, and at most $\left(1-\frac{2}{\log^4N}\right )\cdot |T|$ vertices of $T$. The 
running time of the algorithm is $O(\hat W^2)$.
\end{claim}

We are now ready to complete the proof of \Cref{claim: finalizing cut}. Let $E''=E_{H_0}(\tilde S,\tilde S')$. We construct, in time $O(|E(H_0)|)\leq O(\hat W^2)$, the collection $\xset(E'')$ of subsets of $H_0$ associated with the set $E''$ of edges, and then examine all canonical cuts $(A,B)\in \cset'(E'')$. 
For each such cut, we check whether $|A|\leq \frac{\hat W}{\log^{5}N}$, or $|B|\leq \frac{\hat W}{\log^{5}N}$ hold.
From the definition of canonical cuts, we can check this, for every cut $(A,B)\in \cset'(E'')$, in time $O(|V(H)|)\leq O(\hat W)$. We first claim that if, for each such canonical cut $(A,B)$, either $|A|\leq \frac{\hat W}{\log^{5}N}$ or $|B|\leq \frac{\hat W}{\log^{5}N}$ holds, then at least one of the events $\event_2,\event_3$ must have happened.

\begin{observation}\label{obs: when bad events happened}
	If, for every canonical cut $(A,B)\in \cset'(E'')$, either $|A|\leq \frac{\hat W}{\log^{5}N}$ or $|B|\leq \frac{\hat W}{\log^{5}N}$ holds, then at least one of the events $\event_2,\event_3$ must have happened.
\end{observation}
\begin{proof}
Recall that	one side $\tilde Z\in \set{\tilde S,\tilde S'}$ of the cut $(\tilde S,\tilde S')$ in $H$ contains all but at most $\hat W/8$ vertices of $H_0$, and at most $\left(1-\frac{2}{\log^4N}\right )\cdot |T|$ vertices of $T$. Denote the other side of the cut by $\tilde Z'$. We now consider two cases. The first case happens if $|\tilde Z\cap T|\leq  \frac{|T|}{32}$. Then $\tilde Z$ contains at least $\frac{\hat W} 2$ vertices of $H$ and at most $\frac{|T|}{32}$ vertices of $T$. Since, for every canonical cut $(A,B)\in \cset'(E'')$, either $|A|\leq \frac{\hat W}{\log^{5}N}$ or $|B|\leq \frac{\hat W}{\log^{5}N}$ holds, we conclude that Event $\event_3$ must have happened.

Otherwise, $|\tilde Z\cap T|\geq \frac{|T|}{32}$, and, since $|\tilde Z\cap T|\leq \left(1-\frac{2}{\log^4N}\right )\cdot |T|$, we get that $|\tilde Z\cap T|,|\tilde Z'\cap T|\geq \frac{|T|}{\log^4N}$. Since, for every canonical cut $(A,B)\in \cset'(E'')$, either $|A|\leq \frac{\hat W}{\log^{5}N}$ or $|B|\leq \frac{\hat W}{\log^{5}N}$ holds, we conclude that Event $\event_2$ must have happened.
\end{proof}

If, for every canonical cut $(A,B)\in \cset'(E'')$, either $|A|\leq \frac{\hat W}{\log^{5}N}$ or $|B|\leq \frac{\hat W}{\log^{5}N}$ hold, then we terminate the algorithm, and report that at least one of the events $\event_2,\event_3$ has happened. We assume from now on that there is a canonical cut $(A,B)\in \cset'(E'')$ with $|A|,|B|\geq \frac{\hat W}{\log^{5}N}$. Recall that $(A,B)$ is a weakly canonical cut in graph $H_0$, and $|E_{H_0}|\leq \frac{|T|}{\log^6N}\leq \frac{8\hat W}{d'\log^6N}\leq \frac{8}{d'\log N}\cdot \min\set{|A|,|B|}\leq \frac{\phi}{64}\min\set{|A|,|B|}$. Therefore, $(A,B)$ is a weakly canonical cut of sparsity at most $\frac{\phi}{64}$ in $H_0$. Recall also that every edge in $G_0$ has length at most $\frac{d'}{2}=\frac{2}{\phi}$.

For every vertex $v\in V(G_0)$, we denote the weight of $v$ in $G_0$ by $w_0(v)$, and, for every subset $S\subseteq V(G_0)$ of vertices, we denote $w_0(S)=\sum_{v\in S}w_0(v)$. Similarly, for every vertex $v\in V(G)$, we denote the weight of $v$ in $G$ by $w(v)$, and for a subset $S\subseteq V(G)$ of vertices, we denote $w(S)=\sum_{v\in S}w(v)$. Note that, for every vertex $v\in V(G)$, $w_0(v)\leq w(v)$.

Consider now the cut $(A',B')$ in $G_0$, where $A'=A\cap V(G_0)$ and $B'=B\cap V(G_0)$. From \Cref{obs: subdivided cut sparsity}, $(A',B')$ is a weakly well-structured cut in $G_0$, whose sparsity is at most $\frac{\phi}{8}$, with $w_0(A')\geq \frac{|A|}{8}\geq \frac{\hat W}{8\log^5 N}$, and $w_0(B')\geq \frac{|B|}{8}\geq \frac{\hat W}{8\log^5 N}$.

Recall that $d'\leq \frac{2d}{(\log N)^{28}}\leq \frac{n}{256}\leq \frac{\hat W}{128}\leq \frac{\sum_{e\in E(G_0)}\ell(e)}{2}$. Therefore, the length of every edge in $G_0$ is bounded by $\frac{\sum_{e\in E(G_0)}\ell(e)}{2}$. We call now apply the algorithm from \Cref{obs: weakly to strongly well str} to graph $G_0$, parameter $\frac{\phi}{8}$, and the cut $(A',B')$ in $G_0$, to compute a cut $(A'',B'')$ in $G_0$ with $w_0(A'')\geq \frac{w_0(A')}{4}\geq \frac{\hat W}{32\log^5 N}$ and $w_0(B'')\geq \frac{w_0(B')}{4}\geq \frac{\hat W}{32\log^5 N}$, such that $(A'',B'')$ is a strongly well-structured cut in $G_0$, whose sparsity is at most $\phi/4$.
The running time of the algorithm is bounded by $O(|E(G_0)|)\leq O(\hat W^2)$.

Lastly, we define the cut $(A^*,B^*)$ in the current graph $G$ as follows: $A^*=A''\cap V(G)$ and $B^*=B''\cap V(G)$. Since  $(A'',B'')$ is a strongly well-structured cut in $G_0$, and $G\subseteq G_0$, it is easy to verify that $(A^*,B^*)$ is a strongly well-structured cut in $G$. Moreover, since the current phase has not terminated yet, the total weight of all vertices deleted so far from $G$ is bounded by $\frac{\hat W}{\log^{30}N}$, so $w(A^*)\geq w_0(A^*)\geq w_0(A'')-
\frac{\hat W}{\log^{30}N}\geq \frac{w_0(A'')}{2}>\frac{\hat W}{\log^{30}N}$.
Similarly, $w(B^*)\geq w_0(B^*)\geq \frac{w_0(B'')}{2}>\frac{\hat W}{\log^{30}N}$. Since $|E_G(A^*,B^*)|\leq E_{G_0}(A'',B'')$, we conclude that cut $(A^*,B^*)$ has sparsity at most $\phi$. We output the cut $(A^*,B^*)$ as the outcome of the algorithm. From the above discussion, the running time of the algorithm is bounded by $O(\hat W^2)$.

%% file: maintaincluster-expander.tex
\subsection{First Data Structure: Expander $X$ and its Embedding}
\label{subsec: maintaincluster expander}

Throughout the phase, we will maintain an expander $X$, that is defined over the terminals, and is embedded into $H$. We start by providing the algorithm to initialize this data structure, and then provide the algorithm for maintaining it throughout the phase.

\input{expander-initialization}

\input{expander-maintaining}

%% file: expander-initialization.tex
\subsubsection{Initialization}
\label{subsec: expander: initialization}
In this step, we will either compute an expander $X$ over the set $T$ of vertices and embed it into $H$ via short paths that cause low congestion; or we will compute a sparse cut $(A,B)$ in $H$, following which a sufficient number of vertices will be deleted from $H$, so that the current phase will terminate.
We summarize the algorithm for initializing the expander in the following lemma, whose statement uses the constant $\alpha_0$ from \Cref{thm:explicit expander}.

\begin{lemma}\label{lem: phase 1 embed expander}
	There is a randomized algorithm, that either returns ``FAIL'', or computes one of the following:
	
	\begin{itemize}
		\item either  a finalizing cut in $H$; 
		\item or a strongly well-structured $1/d'$-sparse cut $(A,B)$ in $G$ with $w(A),w(B)\geq \frac{\hat W}{(\log N)^{30}}$;
		
		\item or a graph $X^*$ with $V(X^*)=T$, such that $X^*$ has maximum vertex degree at most $24$, and it is an $\frac{\alpha_0}{81}$-expander. In the latter case, the algorithm additionally computes a set $F\subseteq E(X^*)$ of at most $\frac{\alpha_0\cdot |T|}{500}$ edges, and an embedding $\pset^*$ of $X^*\setminus F$ into $H$ via paths that have length at most $O\left (d'\cdot  (\log N)^{16}\right )$ each, that cause congestion $O\left((\log N)^{13}\right )$. 
		\end{itemize}
	
The probability that the algorithm returns ``FAIL'' is at most $\frac{1}{N^7}$, and the running time of the algorithm is bounded by:

\[O\left( c_1\cdot \hat W^{2} \cdot 2^{c_2\sqrt{\log N}}\cdot (\log N)^{16c_2(r-1)+8c_2+3} \right ).\]
\end{lemma}

The remainder of this subsubsection is dedicated to proving \Cref{lem: phase 1 embed expander}.

\subsubsection*{Recap of the Cut-Matching Game}
We use the Cut-Matching Game -- a framework that was initially proposed by \cite{KRV}. The specific variant of the game is due to \cite{KKOV}. Its slight modification (that we also use) was proposed in \cite{detbalanced}, and later adapted to directed graphs in \cite{SCC}; the same modification was used in \cite{CK24}.
We also note that \cite{directed-CMG} extended the original cut-matching game of \cite{KRV} to directed graphs, though we do not use this variation here.
We now define the variant of the game that we use.

The game is played between two players: the Cut Player and the Matching Player. 
The game starts with a directed graph $X$ that contains $\hat n$ vertices  and no edges,  and proceeds in iterations. In the $i$th iteration, the Cut Player chooses a cut $(A_i,B_i)$ in $X$ with $|A_i|,|B_i|\geq \hat n/10$, and $|E_X(A_i,B_i)|\leq \hat n/100$. It then computes two arbitrary equal-cardinality sets $A_i',B_i'$ of vertices with $|A'_i|\geq \hat n/4$, such that either $A_i\subseteq A'_i$, or $B_i\subseteq B'_i$ holds. The matching player then needs to return two perfect matchings: a matching $\mset_i\subseteq A'_i\times B'_i$, and  a matching $\mset'_i\subseteq B'_i\times A'_i$, both of cardinality $|A'_i|$. The edges of $\mset_i\cup \mset'_i$ are added to $X$, and the iteration terminates.
If there is no cut $(A,B)$ in $X$ with $|A|,|B|\geq \hat n/10$, and $|E_X(A,B)|\leq \hat n/100$, then the game must terminate, though we may choose to terminate it earlier.
 The following bound on the number of iterations in this variant of the Cut-Matching Game was shown by \cite{SCC}, extending a similar result of \cite{KKOV} for undirected graphs.

\begin{theorem}[Lemma 7.4 in \cite{SCC}; see also Theorem 7.3 and the following discussion]\label{thm: CMG bound}
	The number of iterations in the above Cut-Matching Game is bounded by $O(\log \hat n)$.
\end{theorem}

From the above theorem, there is an absolute constant $\cCMG>0$, such that, after $\cCMG\log \hat n$ iterations of the above game, for every cut $(A,B)$ in $X$ with $|A|,|B|\geq \hat n/10$, $|E_X(A_i,B_i)|>\hat n/100$ must hold. Notice that the maximum vertex degree in the final graph $X$ is bounded by $2\cCMG\log \hat n$.

\subsubsection*{Implementing the Cut Player}

There are several known algorithms that we could use in order to implement the cut player. For example, \cite{SCC}, generalizing the results of \cite{detbalanced} to directed graphs, provided such an algorithm, whose running time is $\ohat(\hat n)$ (see Theorem 7.3 in \cite{SCC}). We use instead the algorithm of \cite{CK24}, whose analysis is very simple, even though its running time is somewhat higher. The algorithm is summarized in the following theorem, that uses the constant $\alpha_0$ from \Cref{thm:explicit expander}.

\begin{theorem}[Theorem 8.5 in \cite{CK24}]\label{thm: cut player outer}
	There is a deterministic algorithm, that, given an $\hat n$-vertex directed graph $G$ with maximum vertex degree at most $O(\log \hat n)$, returns one of the following:
	
	\begin{itemize}
		\item Either a cut $(A,B)$ in $G$ with $|E_G(A,B)|\leq \hat n/100$ and $|A|,|B|\geq \hat n/10$; or
		
		\item another directed graph $X'$ with $|V(X')|\geq \hat n/4$ and maximum vertex degree at most $18$, such that $X'$ is an $\alpha_0$-expander, together with a subset $F\subseteq E(X')$ of at most $\frac{\alpha_0\cdot \hat n}{1000}$ edges of $X'$, and an embedding $\pset$ of $X'\setminus F$ into $G$ via paths of length at most $O(\log^2 \hat n)$, that cause congestion at most $\tilde \eta\leq O(\log^2 \hat n)$.
	\end{itemize}
	
	The running time of the algorithm is $\otilde(\hat n^2)$.
\end{theorem}

Next, we discuss the implementation of the matching player, which is different from that of \cite{CK24}: to make the algorithm more efficient, we rely on Algorithm $\aset$ for the $r$-restricted \routeandcut problem, whose existence we have assumed.

\subsubsection*{Implementing the Matching Player}
We use the parameter $N'=\max\set{N,\hat W}$, so that $|V(H)|\leq N'$ and $N\leq N'\leq 12N$ holds, and additional parameters
$\hat \Delta=\ceil{\frac{|T|}{4}}$ and $\eta'=(\log N')^{8}$. Throughout the algorithm, we will sometimes compute disjoint subsets $A,B$ of vertices of $H$ of cardinality $\hat \Delta$ each, and we will view $(H,A,B,\hat \Delta,\eta',N')$ as an instance of the $r$-restricted \routeandcut problem. We show that this must be indeed a valid instance of the problem, in the following simple observation  whose proof is deferred to Section \ref{subsec: bad params no event} %and \ref{subsec: valid input for routeandcut} 
of Appendix. %, respectively.

\iffalse
\begin{observation}\label{obs: bad params no event}
$ \eta'\leq \hat \Delta$ and $\frac{\hat W\cdot \eta'}{\hat \Delta}\leq 2^{r\sqrt{\log N'}}$ must hold.
\end{observation}
\fi

\begin{observation}\label{obs: valid input for routeandcut}
For any pair $(A,B)$ of disjoint subsets of $V(H)$ of cardinality $\hat \Delta$ each, $(H,A,B,\hat \Delta,\eta',N')$ is a valid input to the $r$-restricted \routeandcut problem.	
\end{observation}

The matching player will be implemented by applying Algorithm $\aset$ for the $r$-restricted \routeandcut problem to instances $(H,A,B,\hat \Delta,\eta',N')$, where $A,B$ are disjoint subsets of $V(H)$ of cardinality $\hat \Delta$ each. Throughout, we  will use a parameter ${z=\frac{4|T|}{\log^4N}}$.

\subsubsection*{Proof of \Cref{lem: phase 1 embed expander}: the Cut-Matching Game}

We let $\tc$ be a large enough constant, so that, 
in \Cref{thm: cut player outer} the parameter $\tilde \eta\leq \tilde c\log^2 \hat n$. 
We start with the graph $X$, whose vertex set is $T$, and edge set is $\emptyset$, and then run the Cut-Matching game on it. Recall that the number of iterations in the game is bounded by $\cCMG\cdot \log \hat n$. 
This ensures that the graph $X$ that is maintained over the course of the game has maximum vertex degree $O(\log \hn)$.

\subsubsection*{Description of a Single Iteration}

We now describe the execution of the $i$th iteration.
First, we apply the algorithm from \Cref{thm: cut player outer} to the current graph $X$. We say that the current iteration is a \emph{type-1 iteration}, if the algorithm returns 
a cut $(A_i,B_i)$ in $X$ with $|E_X(A_i,B_i)|\leq \hn/100$ and $|A_i|,|B_i|\geq \hn/10$, and otherwise we say that it is a type-2 iteration. Assume first that the current iteration is a type-1 iteration. Then we compute two arbitrary  disjoint subsets $A'_i,B'_i$ of $V(X)$ of cardinality $\ceil{\frac{\hat n}{4}}=\ceil{\frac{|T|}{4}}=\hat \Delta$ each, such that either $A_i\subseteq A'_i$ or $B_i\subseteq B'_i$ holds. We treat $(A'_i,B'_i)$ as the move of the cut player. Recall that the running time of the algorithm from  \Cref{thm: cut player outer} is $\otilde(\hat n^2)\leq \otilde(\hat W^2)$.

Next, we consider two instances of the $r$-restricted \routeandcut problem: instance 
 $(H,A_i',B_i',\hat \Delta,\eta',N')$, and instance  $(H,B_i',A_i',\hat \Delta,\eta',N')$. We apply Algorithm $\aset$ to each of these two instances.
 If the algorithm returns ``FAIL'' on any of these instances, we repeat it up to $\ceil{10\log N}$ times, or until it terminates with a valid output, whichever happens first. If, for any of the two instances of the \routeandcut problem, the algorithm 
 returned ``FAIL'' repeatedly in each of its  $\ceil{10\log N}$ applications, then we terminate the algorithm and return ``FAIL''. Since the probability that an algorithm for the \routeandcut problem returns ``FAIL'' is at most $1/2$, the probability that the current iteration termintates with a ``FAIL'' is bounded by $\frac{2}{N^{10}}$. We assume from now on that the last application of Algorithm $\aset$ to each of the two instances did not terminate with a ``FAIL''. We denote by $\qset'_i$ and $\qset''_i$ the routings that these algorithms return, where $\qset'_i$ is the routing from $A'_i$ to $B'_i$. Recall that the running time of Algorithm $\aset$ on each of these two instances is bounded by:

 \[\begin{split}
 &c_1\cdot \hat W^{2}\cdot 2^{c_2\sqrt{\log N'}}\cdot (\log N')^{16c_2(r-1)+8c_2}
 \\&\quad\quad\quad\quad\quad\quad\leq O\left(c_1\cdot \hat W^{2} \cdot 2^{c_2\sqrt{\log N}}\cdot (\log (12N))^{16c_2(r-1)+8c_2+1}\right ) \\
 &\quad\quad\quad\quad\quad\quad\leq O\left( c_1\cdot \hat W^{2} \cdot 2^{c_2\sqrt{\log N}}\cdot (\log N)^{16c_2(r-1)+8c_2+1} \cdot \left(1+\frac{1}{32c_2r}\right )^{16c_2(r-1)+8c_2+1}\right )\\
  &\quad\quad\quad\quad\quad\quad\leq O\left( c_1\cdot \hat W^{2} \cdot 2^{c_2\sqrt{\log N}}\cdot (\log N)^{16c_2(r-1)+8c_2+1} \cdot \left(1+\frac{1}{32c_2r}\right )^{32c_2r}\right )\\
 &\quad\quad\quad\quad\quad\quad\leq O\left( c_1\cdot \hat W^{2} \cdot 2^{c_2\sqrt{\log N}}\cdot (\log N)^{16c_2(r-1)+8c_2+1} \right )\\
 \end{split} \]

(we have used Inequalities \ref{eq: more logs} and \ref{eq: even more logs}). Since Algorithm $\aset$ may be applied up to $O(\log N)$ times to the resulting instances of the \routeandcut problem, the total time spent in a single iteration by the executions of Algorithm $\aset$ is bounded by:

\[  O\left( c_1\cdot \hat W^{2} \cdot 2^{c_2\sqrt{\log N}}\cdot (\log N)^{16c_2(r-1)+8c_2+2} \right ). \]

We now consider two cases. The first case happens when both $|\qset'_i|,|\qset''_i|\geq \hat \Delta-z$. In this case, we say that the current iteration is of \emph{type (1a)}. Recall that the paths in each of the sets $\qset'_i$, $\qset''_i$ cause congestion at most $4\eta'\cdot \log N'\leq 4(\log N')^{9}\leq 4(\log (12N))^{9}\leq (\log N)^{10}$, since  $\eta'=(\log N')^{8}$ and $N'\leq 12N$. 
Since regular and special edges must alternate on each such path, and since the total number of special edges is at most $\hat W/2$,  $\sum_{P\in \qset'_i}|E(P)|\leq \hat W\cdot  (\log N)^{10}$. Similarly, $\sum_{P\in \qset''_i}|E(P)|\leq \hat W\cdot  (\log N)^{10}$.
 We say that a path $Q\in \qset'_i\cup \qset''_i$ is \emph{long}, if it contains more than $8d'\cdot  (\log N)^{14}$ edges. Clearly, the number of long paths in $\qset'_i$ is bounded by: 
 
\[\frac{\hat W\cdot  (\log N)^{10}}{8d'\cdot  (\log N)^{14}}\leq \frac{\hat W}{2d' (\log N)^{4}}\leq \frac{|T|}{\log^4N}\leq z.\] 

(We have used the fact that $|T|\geq \frac{\hat W}{2d'}$, as Event $\event_1$ did not happen.)
Similarly, at most $z$ paths in $\qset''_i$ are long.
We discard all long paths from both $\qset'_i$ and $\qset''_i$. At the end of this process, every path in $\qset'_i\cup \qset''_i$ has length at most $8d'\cdot  (\log N)^{14}$, and $|\qset'_i|,|\qset''_i|\geq \hat \Delta-2z$.

The set $\qset'_i$ of paths naturally defines a matching $\mset'_i\subseteq A'_i\times B'_i$, where for every path $Q\in \qset_i$ whose endpoints are $a\in A'_i,b\in B'_i$, we add an edge $e=(a,b)$ to the matching. We also let $Q(e)=Q$ be its embedding into $G$. Similarly, the set $\qset''_i$ of paths defines a matching $\mset''_i\subseteq B'_i\times A'_i$, together with an embedding $Q(e)\in \qset''_i$ for each edge $e\in \mset''$. Recall that matching $\mset'_i$ has cardinality at least $\hat \Delta-2z$, while $|A'_i|=|B'_i|=\hat \Delta$. Therefore, by adding a set $F'_i$ of at most $2z$ additional edges (that we refer to as \emph{fake edges}), we can turn matching $\mset'_i$ into a perfect one. Similarly, we add a collection $F''_i$ of at most $2z$ fake edges to $\mset''_i$ in order to turn it into a perfect matching. The edges of $\mset'_i\cup \mset''_i$ are then added to graph $X$, and the current iteration terminates.

Assume now that one of the two sets $\qset'_i,\qset''_i$ computed by Algorithm $\aset$ contains fewer than $\hat \Delta-z$ paths; we assume w.l.o.g. that it is $\qset'_i$. In this case we say that the current iteration is of \emph{type (1b)}.

Let $A''_i\subseteq A'_i,B''_i\subseteq B'_i$ be the sets of vertices that do not serve as endpoints of the paths in $\qset'_i$, so $|A''_i|,|B''_i|\geq z$. Recall that the algorithm must also return a cut $(Z,Z')$ in $H$, such that:

\[ |E_H(Z,Z')|\leq \frac{64\hat \Delta}{\eta'\log^4\hat W}+\frac{256|\qset'_i|}{\eta'}\leq \frac{512\hat \Delta}{\eta'}\leq \frac{256|T|}{(\log N')^8}\leq \frac{\hat W}{d'(\log N)^7},  \]

since $\eta'=(\log N')^{8}$, $\hat \Delta=\ceil{\frac{|T|}{4}}$, $|T|\leq \frac{8\hat W}{d'}$ since Event $\event_1$ did not happen,  and $N'\geq N$.
We are also guaranteed that $A''_i\subseteq Z$ and $B''_i\subseteq Z'$, so $|Z\cap T|,|Z'\cap T|\geq z=\frac{4|T|}{\log^4N}$. 
We use the following simple claim, that either computes a finalizing cut in $H$, or returns a strongly well-structured $1/d'$-sparse cut $(A,B)$ in $G$ with $w(A),w(B)\geq \frac{\hat W}{(\log N)^{30}}$.
The proof is deferred to Section \ref{subsec: proof of finalize the cut} of Appendix.

\begin{claim}\label{claim: finalize the cut}
	There is a deterministic algorithm, that, given a cut $(Z,Z')$ in graph $H$ with $|Z\cap T|,|Z'\cap T|\geq \frac{4|T|}{\log^4N}$ and  $|E_H(Z,Z')|\leq \frac{\hat W}{d'(\log N)^7}$, either  computes a finalizing cut in $H$, or returns a strongly well-structured $1/d'$-sparse cut $(A,B)$ in $G$ with $w(A),w(B)\geq \frac{\hat W}{(\log N)^{30}}$. The running time of the algorithm is $O(\hat W^2)$.
\end{claim}

If the algorithm from \Cref{claim: finalize the cut} returns a finalizing cut in $H$, then we output this cut as the outcome of the algorithm. Otherwise, we output the $1/d'$-sparse cut $(A,B)$ in $G$ with $w(A),w(B)\geq \frac{\hat W}{(\log N)^{30}}$ that the algorithm returns.

Finally, assume that the current iteration $i$ is of type 2. In this case, the algorithm for the Cut Player from \Cref{thm: cut player outer} computed a directed graph 
$X'$ with $|V(X')|\geq \hat n/4$ and maximum vertex degree at most $18$, such that $X'$ is an $\alpha_0$-expander, together with a subset $F'\subseteq E(X')$ of at most $\frac{\alpha_0\cdot n}{1000}$ edges of $X'$, and an embedding $\pset$ of $X'\setminus F'$ into $X$ via paths of length at most $O(\log^2 \hat n)$, that cause congestion at most $\tilde c\log^2 \hat n$.

Denote $\hat F=\bigcup_{i'=1}^{i-1}(F'_i\cup F''_i)$ -- the collection of the fake edges that we have computed so far. Since for all $1\leq i'<i$, $|F'_i\cup F''_i|\leq 4z$ and $i\leq \cCMG\log n$, we get that $|\hat F|\leq 4\cCMG z \log \hat n$.

We also denote $\hat \qset=\bigcup_{i'=1}^{i-1}(\qset'_i\cup \qset''_i)$. Note that $\hat \qset$ defines an embedding of $X\setminus \hat F$ into $H$. All paths in $\hat \qset$ have lengths at most $8d'\cdot  (\log N)^{14}$. Each set $\qset'_i\cup \qset''_i$ causes congestion at most $2(\log N)^{10}$, and so the total congestion caused by the paths in $\hat \qset$ is bounded by $\hat \eta=O\left((\log N)^{11}\right )$. 
We now terminate the execution of the Cut-Matching Game.

As our next step, we will construct a set $F\subseteq E(X')$ of fake edges, and an embedding of the remaining edges of $X'$ into $H$. The latter is done by comining the embedding $\pset$ of $X'\setminus F'$ into $X$ with the embedding $\hat \qset$ of $X\setminus \hat F$ into $H$.

\subsubsection*{Constructing the Set $F$ of Fake Edges for $X'$ and Embedding $X'\setminus F$ into $H$}

We start with $F=F'$. Next, we process every edge $e\in E(X')\setminus F'$ one by one. When an edge $e=(u,v)$ is processed, we consider the path $P'(e)\in \pset$, that serves as an embedding of $e$ into graph $X$. Recall that the length of the path is at most $O(\log^2 \hat n)$. Let $P(e)=(e_1,e_2,\ldots,e_k)$, where $k\leq O(\log^2 \hat n)\leq O(\log^2N)$. 
If path $P(e)$ contains any edges of $\hat F$, then we add $e$ to the set $F$ of fake edges. Otherwise,
for all $1\leq j\leq k$, we consider the path $Q(e_j)\in \hat \qset$, that serves as the embedding of the edge $e_j$ into graph $G$. Recall that the length of $Q$ is at most $8d'\cdot  (\log N)^{14}$. By concatenating the paths $Q(e_1),\ldots,Q(e_k)$, we obtain a path $P'(e)$ connecting $u$ to $v$ in $G$, whose length is at most $O\left (d'\cdot  (\log N)^{16}\right )$.  We turn $P'(e)$ into a simple path in time $\tilde O(d')$, obtaining the final path $\hat P(e)$, that will serve as an embedding of the edge $e$.

Once every edge of $X'\setminus F'$ is processed, we obtain a final set $F$ of fake edges, and, for every remaining edge $e\in E(X')\setminus F$, a simple path $\hat P(e)$ of length $O\left (d'\cdot  (\log N)^{16}\right )$, that serves as the embedding of $e$ in $G$. Recall that $|V(X')|\geq \hat n/4$, and graph $X'$ is an $\alpha_0$-expander with maximum vertex degree at most $18$. 
Next, we bound the congestion caused by the paths in set $\pset^*=\set{\hat P(e)\mid e\in E(X')\setminus F}$ in $H$, and the cardinality of the set $F$ of fake edges.

Consider any edge $e\in E(H)$. Recall that $e$ may participate in at most 	$\hat \eta=O\left((\log N)^{11}\right )$ paths in set $\hat \qset=\set{Q(e')\mid e'\in E(X)\setminus \hat F}$. Every edge $e'\in E(X)\setminus \hat F$, whose corresponding path $Q(e')$ contains the edge $e$, may in turn participate in at most $O(\log^2 N)$ paths of $\pset$. Therefore, overall, every edge $e$ of $G$ participates in at most 	$O\left((\log N)^{13}\right )$ paths of $\pset^*$.

In order to bound the number of the fake edges in $F$,
recall that $|\hat F|\leq 4\cCMG z \log N$. Since the paths in $\pset$ cause congestion at most $\tc\log^2 \hat n=\tc\log^2N$, we get that every edge of $\hat F$ may participate in at most $\tc\log^2 N$ paths of $\pset$. Therefore, $|F\setminus \hat F|\leq (4\cCMG z \log N)\cdot (\tc\log^2 N)\leq 4\tc\cdot \cCMG\cdot z\cdot \log^3N\leq \frac{\alpha_0\cdot |T|}{2000}$,
since $z=\frac{4|T|}{\log^4N}$, and $N$ is sufficiently large.  
Since $|F'|\leq \frac{\alpha_0\cdot |T|}{1000}$, we get that $|F|\leq \frac{3\alpha_0\cdot |T|}{2000}$.

Note that graph $X'$ has all required properties, except that $V(X')$ may not contain all terminals. In our last step, we fix this problem, by connecting the terminals that do not lie in $X'$ to the vertices of $X'$.

\subsubsection*{Extending the Expander to All Terminals}

Let $T_0\subseteq T$ be the set of terminals that lie in $X'$, so $|T_0|\geq \ceil{\frac{|T|}{4}}$. We partition the terminals of $T\setminus T_0$ into three subsets $T_1,T_2,T_3$ arbitrarily, so that every subset contains at most $\ceil{\frac{|T|}{4}}$ terminals. We now perform last three iterations, that we will denote by iterations $(i+1),(i+2)$ and $(i+3)$. In each of these three iterations, we will only invoke the algorithm of the Matching Player, in order to connect the vertices in one of the subsets $T_1,T_2,T_3$ to $T_0$.

We now describe iteration $(i+1)$. The remaining two iterations are performed similarly. We let $A'_{i+1}$ be any subset of $T_0$ containing exactly $\hat \Delta=\ceil{\frac{|T|}{4}}$ vertices. We also let $B'_{i+1}$ be any subset of $\hat \Delta$ terminals with $T_1\subseteq B'_{i+1}$, so that $A'_{i+1}$ and $B'_{i+1}$ are disjoint.
Next, we consider two instances of the $r$-restricted \routeandcut problem: instance 
$(H,A_{i+1}',B_{i+1}',\hat \Delta,\eta',N')$, and instance  $(H,B_{i+1}',A_{i+1}',\hat \Delta,\eta',N')$. We apply Algorithm $\aset$ to each of these two instances.
As before, if either of these algorithms returns ``FAIL'', then we repeat it  at most $\ceil{10\log N}$ times, or until it terminates with a valid input -- whatever happens first. If the algorithm returned ``FAIL'' in all $\ceil{10\log N}$ applications, then we terminate the current algorithm with a ``FAIL''. We assume from now on that in the last application of algorithm $\aset$ to both instances a valid output was computed, and we denote by $\qset'_{i+1}$ and $\qset''_{i+1}$ the routings that these algorithms return, where $\qset'_{i+1}$ is the routing from $A'_{i+1}$ to $B'_{i+1}$.

If either of the sets $\qset'_{i+1},\qset''_{i+1}$ of paths has cardinality less than $\hat \Delta-z$, then we employ the same algorithm as before to   either  compute a finalizing cut in $H$, or to compute a strongly well-structured $1/d'$-sparse cut $(A,B)$ in $G$ with $w(A),w(B)\geq \frac{\hat W}{(\log N)^{30}}$. We then terminate the algorithm and return the resulting cut.
We assume from now on that $|\qset'_{i+1}|,|\qset''_{i+1}|\geq \hat \Delta-z$.
We use the set $\qset'_{i+1}$ of paths as before (except that we ingore those paths that terminate at vertices of $B'_{i+1}\setminus T_1$), in order to compute a matching $M'_{i+1}\subseteq T_0\times T_1$, with $|M'_{i+1}|\geq |T_1|-2z$, together with an embedding of the edges of $M'_{i+1}$ into the graph $H$, via paths of lengths at most $8d'\cdot  (\log N)^{14}$, that cause
congestion at most $O(\log^{10}N)$. We then add a collection $F'_{i+1}$ of at most $2z$ fake edges to $M'_{i+1}$ to turn it into a matching of cardinality $|T_1|$. Similarly, we use the paths in $\qset''_{i+1}$, in order to compute a matching $M''_{i+1}\subseteq T_1\times T_0$, with $|M''_{i+1}|= |T_1|$, together with a subset $F''_{i+1}\subseteq M''_{i+1}$ of at most $2z$ fake edges, and an embedding of the edges of $M''_{i+1}\setminus F''_{i+1}$ into the graph $H$, via paths of lengths at most $8d'\cdot  (\log N)^{14}$, that cause
congestion at most $O(\log^{10}N)$. The remaining two iterations are executed similarly, except that we use sets $T_2$ and $T_3$ of terminals instead of $T_1$.

Consider now the final graph $X^*$, that is obtained from $X'$ by adding the vertices of $T_1\cup T_2\cup T_3$ to it, and the edges of $\mset^*=\bigcup_{j=1}^3(M'_{i+j}\cup M''_{i+j})$. Then $V(X^*)=T$, maximum vertex degree in $X^*$ is at most $18+6=24$. Moreover, from \Cref{obs: expander plus matching}, graph $X^*$ is an $\frac{\alpha_0}{81}$-expander.

Recall that each of the sets $\set{F'_{i+j},F''_{i+j}\mid 1\leq j\leq 3}$ of fake edges has cardinality at most $2z\leq \frac{8|T|}{\log^4N}$. We add the fake edges from all these sets to the set $F$ of fake edges. Since initially $|F|\leq \frac{3\alpha_0\cdot |T|}{2000}$, and since the number of edges that we added to $F$ in this step is bounded by $\frac{100|T|}{\log^4N}$, we get that $|F|\leq \frac{\alpha_0\cdot |T|}{500}$ holds. For every edge $e\in \mset^*\setminus F$, we have computed an embedding path $\tilde P(e)$, that is added to the set $\pset^*$ of paths. The final set $\pset^*$ of paths is an embedding of $X^*\setminus F$ into $H$, where the length of every path is $O\left (d'\cdot  (\log N)^{15}\right )$, and the paths in $\pset^*$ cause congestion $O\left((\log N)^{16}\right )$.

\subsubsection*{Running Time  and Failure Probability Analysis}

The algorithm has $O(\log \hat n)\leq O(\log N)$ iterations. 
The probability that it returns ``FAIL'' in a single iteration is at most $\frac{2}{N^{10}}$, and so the probability that the algorithm ever returns ``FAIL'' is at most $\frac{1}{N^9}$.

the total time spent in a single iteration on the executions of Algorithm $\aset$ is bounded by:

\[  O\left( c_1\cdot \hat W^{2} \cdot 2^{c_2\sqrt{\log N}}\cdot (\log N)^{16c_2(r-1)+8c_2+2} \right ). \]

Additionally, in every iteration we use the algorithm for the Cut Player from \Cref{thm: cut player outer}, whose running time is  $\otilde(\hat W^2)$.

The algorithm from \Cref{claim: finalize the cut}, whose running time is $O(\hat W^2)$, may only be incurred once over the course of the execution of our algorithm.

Additionally, when computing the embedding $\pset^*$ of $X^*\setminus F$ into $G$, we spend time at most $O\left (d'\cdot  (\log N)^{16}\right )$ per edge of $X'$, and, since $|E(X')|\le O(\hat W)$, the time required to compute the embedding is bounded by: $O\left (\hat W \cdot d'\cdot  (\log N)^{15}\right )\leq O\left (\hat W^2\right )$, since $d'\leq  \frac{2d}{(\log N)^{28}}$ and $d\leq n$ from Inequality \ref{eq: bound on d}.
Overall, the running time of the algorithm is bounded by:

\[\begin{split} 
& O\left( c_1\cdot \hat W^{2} \cdot 2^{c_2\sqrt{\log N}}\cdot (\log N)^{16c_2(r-1)+8c_2+3} \right )+\tilde O(\hat W^2)\\
&\quad\quad\quad\quad\quad\quad\quad\quad\quad\leq  O\left( c_1\cdot \hat W^{2} \cdot 2^{c_2\sqrt{\log N}}\cdot (\log N)^{16c_2(r-1)+8c_2+3} \right ). 
\end{split}\]

%% file: expander-maintaining.tex
\subsubsection{Maintaining the Expander}
\label{subsubsec: maintainexpander}

At the beginning of the phase, we execute the algorithm from \Cref{lem: phase 1 embed expander}. If the algorithm terminated with a ``FAIL'', which may only happen with probability at most $\frac{1}{N^7}$, then we terminate the algorithm for the current phase and return ``FAIL''. If the algorithm returned  
 a finalizing cut in $H$, then we apply the algorithm from \Cref{claim: finalizing cut} to this cut. If the algorithm from \Cref{claim: finalizing cut}  estabishes that at least one of the events $\event_2,\event_3$ happened,  we terminate the phase with a ``FAIL''. Otherwise, we obtain a strongly well-structured $\phi$-sparse cut $(A^*,B^*)$ in the current graph $G$, such that $w(A^*),w(B^*)\geq \frac{\hat W}{\log^{30}N}$. In this case, we output the cut $(A^*,B^*)$. Once this cut is processed, at least $\frac{\hat W}{\log^{30}N}$ vertices are deleted from $H$, and the current phase terminates. Finally, if the algorithm from
 \Cref{lem: phase 1 embed expander} returns a strongly well-structured $1/d'$-sparse cut $(A,B)$ in $G$ with $w(A),w(B)\geq \frac{\hat W}{(\log N)^{30}}$, then as before we output this cut, and, once it is processed, at least $\frac{\hat W}{\log^{30}N}$ vertices are deleted from $H$, and the current phase terminates. Therefore, we assume from now on that the algorithm from  \Cref{lem: phase 1 embed expander} computed a graph $X^*$, that we will denote for simplicity by $X$ from now on. Recall that $V(X)=T$, maximum vertex degree in $X$ is at most $24$, and $X$ is an $\frac{\alpha_0}{81}$-expander. 
  Recall that we are also given a set $F\subseteq E(X)$ of at most  $\frac{\alpha_0\cdot |T|}{500}$ fake edges, and an embedding $\pset^*$ of $X\setminus F$ into $H$.
  We let $d''=\Theta\left (d'\cdot  (\log N)^{15}\right )$ and $\eta''=\Theta\left((\log N)^{13}\right )$ be such that every path in $\pset^*$ has length at most $d''$, and the congestion caused by the paths in $\pset^*$ is at most $\eta''$. For convenience, we will denote the set $\pset^*$ of paths by $\qset$ from now on, and, for every edge $\hat e\in E(X)\setminus F$, we denote by $Q(\hat e)\in \qset$ its corresponding embedding path.

 Recall that, since Event $\event_1$ did not happen, we are guaranteed that $\frac{\hat W}{2d'}\leq |T|\leq \frac{8\hat W}{d'}$. As the phase progresses, some vertices and edges may be deleted from $X$ and from $G$. We will view the set $T$ of terminals as a static set, so even if some vertices of $T$ are deleted from $X$ or from $G$, the set $T$ remains unchanged.

Throughout the algorithm, for every edge $e\in E(H)$, we maintain a collection $S'(e)\subseteq E(X)$ of edges $\hat e\in E(X)\setminus F$, such that $e\in Q(\hat e)$. Clearly, the time required to compute these initial sets $S'(e)$ of edges, for all $e\in E(H)$, is subsumed by the time required to compute the embedding of $X\setminus F$ into $H$. Additionally, for every edge $\hat e\in E(X)\setminus F$, we maintain the length $\hat \ell(\hat e)$ of the path $Q(\hat e)$, which is equal to the number of edges on the path, and must initially be bounded by $d''$.

As the phase progresses, graph $H$ undergoes changes, that may require updating the expander $X$ and its embedding. Specifically, we need to ensure that, even as vertices and edges are deleted from $H$, $V(X)\subseteq V(H)$, and, for every edge $\hat e\in E(X)\setminus F$, $Q(\hat e)\subseteq H$ continue to hold. Additionally, as edges in $H$ undergo splitting, some of the embedding paths in $\qset$ may become longer, and we may need to delete the corresponding edges from $X$, in order to ensure that the lengths of all embedding paths are suitably bounded. Lastly, we need to ensure that, throughout the phase, for every pair $x,y\in V(X)$ of vertices, there is a path connecting $x$ to $y$ in $X\setminus F$, whose length is at most $O(\log N)$. Whenever our algorithm discovers that this is no longer the case, it may compute a sparse cut $(A,B)$ in $X$, and delete the vertices of the smaller side of the cut from $X$. We now describe all updates that graph $X$ undergoes over the course of the phase formally. We will ensure that, despite these updates, throughout the phase, $|V(X)|\geq |T|/2$ holds.

%=======================================
%=======================================
%=======================================
%=======================================
%=======================================
%=======================================
%=======================================
%=======================================
\subsubsection*{First Type of Updates: Due to Edge Splitting}
%=======================================
%=======================================
%=======================================
%=======================================
%=======================================
%=======================================
%=======================================
%=======================================

Recall that, as the algorithm progresses, the lengths of some edges $e\in E(G)$ may be doubled. Whenever this happens, we update the corresponding path $P^*(e)$ in $H$, by splitting every special edge $e'\in E(P^*(e))$. The latter is done by replacing the special edge $e'=(x,y)$ with a path $(x,z,z',y)$, where $z,z'$ are new vertices, and $(x,z)$, $(z',y)$ are new special edges.  Notice that once edge $e'$ is split, it no longer lies in graph $H$. For every edge $\hat e\in E(X)\setminus F$ that lies in $S'(e')$, we consider the corresponding embedding path $Q(\hat e)$, that must contain the edge $e'$. We replace $e'$ with the new path $(x,z,z',y)$ on $Q(\hat e)$, and insert $\hat e$ into the sets $S'((x,z))$, $S'((z',y))$, that correspond to the new special edges $(x,z)$ and $(z',y)$. We also update the length $\hat \ell(\hat e)$ of path $Q(\hat e)$, which increases by additive $2$. If, as the result of this procedure, the length of path $Q(\hat e)$ became greater than $8d''$, then edge $\hat e$ is deleted from graph $X$.

We denote by $\hat E_1$ the set of all edges that were deleted from graph $X$ over the course of the phase due to this type of updates. We now show that the number of such edges must be small.

\begin{claim}\label{claim: bound first type edge deletions}
	$|\hat E_1|\leq \frac{\alpha_0\cdot |T|}{5000}$.
\end{claim}
\begin{proof}
	Recall that, over the course of the phase, the algorithm only needs to respond to at most $\Delta'$ queries. Given such a query between a pair $(x,y)$ of vertices, the algorithm must return a path $P$ of length at most $d$ connecting $x$ to $y$. Denote by $\tilde \ell(P)$ the length of the path $P$ at the time when it was returned in response to the query. Then the lengths of some edges on path $P$ are doubled. 
	If $\pset$ is the set of all paths that the algorithm returned in response to queries over the course of the phase, then $\sum_{P\in \pset}\tilde  \ell(P)\leq \Delta'\cdot d$ must hold.

	Let $G_0$ denote the graph $G$ at the beginning of the phase.
	Recall that, if $e$ is any edge of $G_0$, and $\tau,\tau'$ are two consecutive times during the current phase when the length of $e$ is doubled, then $e$ must have participated in at least $\eta$ paths that the algorithm returned during the time interval $(\tau,\tau']$. However, the first time that the length of $e$ is doubled during the current phase may occur even before it participates in $\eta$ such paths. 
	For every edge $e\in E(G_0)$, denote by $\tilde \ell^0(e)$ the length of $e$ at the  beginning of the current phase. If $e\in E(G)$ at the end of the phase, then we denote by $\tilde \ell^*(e)$ the length of $e$ at the end of the phase; otherwise, $\tilde \ell^*(e)$ is the length of $e$ just before it is deleted from $G$.
	
	Let $E'\subseteq E(G_0)$ be the set of all edges $e$ with $\tilde \ell^*(e)>2\tilde \ell^0(e)$, so $E'$ contains all edges of $G_0$ whose lengths were doubled at least twice. 
	Consider now any edge $e\in E'$, and assume that $\tilde \ell^*(e)=2^i$. Let $\tau$ be the time during the current phase when the length of $e$ was doubled from $2^{i-2}$ to $2^{i-1}$, and let $\tau'$ be the time when it was doubled from $2^{i-1}$ to $2^i$. Then $e$ participated in at least $\eta$ paths that were returned by the algorithm in response to queries during the time interval $(\tau,\tau']$. For each such path $P$, it contributed $2^{i-1}=\frac{\tilde \ell^*(e)}2$ to $\tilde \ell(P)$. Therefore, every edge $e\in E'$ contributes at least $\frac{\eta\cdot \tilde \ell^*(e)}{2}$ to $\sum_{P\in \pset}\tilde \ell(P)$. We conclude that $\sum_{e\in E'}\tilde \ell^*(e)\leq \frac{2}{\eta}\sum_{P\in \pset}\tilde \ell(P)\leq \frac{2\Delta'\cdot d}{\eta}$. 
	
	Notice that, for every edge $e\in E(G)$, at any time during the execution of the algorithm, the number of special edges on $P^*(e)$ is equal to the current length of $e$, while the splitting of a speical edge on $P^*(e)$ increases the number of special edges on $P^*(e)$ by $1$. Therefore, if the current length of $e$ is $\ell(e)$, then the edges of $P^*(e)$ underwent $\ell(e)-\tilde \ell^0(e)$ splittings since the beginning of the algorithm.
	
	Consider the graph $H$ at the beginning of the current phase, that we denote by $H_0$. 
	Note that, whenever a special edge $e\in E(H)$ undergoes a splitting, the edge itself is deleted from $H$, and instead two new edges are inserted. Consider now a special edge $e\in E(H)$, and suppose it underwent a splitting at some time. We say that this splitting operation is \emph{unimportant} if $e\in E(H_0)$, and we say that it is \emph{important} otherwise. Note that, if an edge $e\in E(H)$ undergoes an important splitting update, then its parent-edge must lie in $E'$. Therefore, from the above discussion, the total number of important splitting operations in graph $H$ is bounded by $\sum_{e\in E'}\tilde \ell^*(e)$.
	
	Consider now some edge $\hat e\in E(X)\setminus F$, and its embedding path $Q(\hat e)$. Let $\mu(\hat e)$ be the total number of important splitting operations that the edges of path $Q(\hat e)$ underwent over the course of the algorithm. 
	Since the embedding $\qset$ of $X\setminus F$ into $H$ causes congestion at most $\eta''$, we get that $\sum_{e\in X\setminus F}\mu(\hat e)\leq \eta''\cdot \sum_{e\in E'}\tilde \ell^*(e)\leq \frac{2\Delta'\cdot d\cdot \eta''}{\eta}$. On the other hand, if $\hat e\in \hat E_1$, then the length of $Q(\hat e)$ increased from $d''$ to $8d''$ over the course of the algorithm, and so the edges on $Q(\hat e)$ must have undergone at least $6d''$ important splitting operations so $\mu(\hat e)\geq 6d''$. Therefore:
	
	\[ |\hat E'|\leq \frac{2\Delta'\cdot d\cdot \eta''}{6\eta d''}\leq O\left(\frac{\Delta\cdot d}{\eta\cdot d'\cdot (\log N)^{33}}\right )\leq O\left(\frac{n}{d' \cdot(\log N)^{33}}\right ) \leq O\left(\frac{\hat W}{d' \cdot(\log N)^{33}}\right ) <\frac{\alpha_0\cdot |T|}{5000}, \]

	since $d''=\Theta\left (d'\cdot  (\log N)^{16}\right )$, $\Delta'=\frac{\Delta}{(\log N)^{30}}$, and $\eta''=\Theta\left((\log N)^{13}\right )$. We have also used the fact that $\frac{\Delta\cdot d}{\eta}\leq n$ from the problem definition, that $\hat W\geq \frac{n}{2}$, and that $|T|>\frac{8\hat W}{d'}$, since Event $\event_1$ did not happen.
\end{proof}

%=======================================
%=======================================
%=======================================
%=======================================
%=======================================
%=======================================
%=======================================
%=======================================
\subsubsection*{Second Type of Updates: Due to Sparse Cuts in $X$}
%=======================================
%=======================================
%=======================================
%=======================================
%=======================================
%=======================================
%=======================================
%=======================================

Throughout the execution of the phase, we may sometimes find cuts $(A,B)$ in $X\setminus F$, with $|E_{X\setminus F}(A,B)|\leq \frac{\alpha_0}{2^{20}}\cdot \min\set{|A|,|B|}$. Denote $A$ and $B$ by $Z$ and $Z'$, where $Z$ is the smaller of the two sets, breaking ties arbitrarily. We then delete the vertices of $Z$ from graph $X$.
Notice that $|E_{X\setminus F}(A,B)|\leq \frac{\alpha_0}{2^{20}}\cdot|Z|$. We \emph{charge} the vertices of $Z$ for  the edges of $E_{X\setminus F}(A,B)$, where each vertex $v\in Z$ is charged $\frac{\alpha_0}{2^{20}}$ units, so the total charge is at least $|E_{X\setminus F}(A,B)|$. 

Let $\hat E_2$ be the set of edges in the initial graph $X$, that is defined as follows. Initially, $\hat E_2=\emptyset$. Whenever our algorithm computes a cut 
$(A,B)$ in $X\setminus F$, with $|E_{X\setminus F}(A,B)|\leq \frac{\alpha_0}{2^{20}}\cdot \min\set{|A|,|B|}$, we add the edges of $E_{X\setminus F}(A,B)$ to set $\hat E_2$. From the charging scheme that we just described, it is easy to verify that, at the end of the algorithm, $|\hat E_2|\leq \frac{\alpha_0\cdot |T|}{2^{20}}$ holds.

%======================================================
%======================================================
%======================================================
%======================================================
%======================================================
%======================================================
%======================================================
%======================================================
%======================================================
%======================================================
%======================================================
\subsubsection*{Third Type of Updates: Due to Sparse Cuts in $G$}

Recall that our algorithm may, at any time, produce a strongly well-structured $\phi$-sparse cut $(A,B)$ in graph $G$. If $w(A)\leq W(B)$, then the vertices of $A$ are deleted from $G$, and otherwise the vertices of $B$ are deleted from $G$. Additionally,  the endpoints of the edges of $E_G(A,B)$ are also deleted from  $G$, and graph $H$ is also updated with corresponding edge- and vertex-deletions. This may require updating graph $X$, for example, if some of the embedding paths in $\qset$ use vertices or edges that were deleted from $H$. We emphasize that edges and vertices may only be deleted from $G$ via this procedure, and the only other updates that graph $G$ undergoes is the doubling of the lengths of its edges.

Let $(A_1,B_1),\ldots,(A_q,B_q)$ denote the sequence of all strongly well-structured $\phi$-sparse cuts in $G$ that our algorithm produced during the current phase, indexed in the order in  which the algorithm produced them. If the last sparse cut that the algorithm produced caused the total number of vertices that were deleted from $H$ over the course of the phase to reach $\frac{\hat W}{\log^{30}N}$ (after which the current phase is immediately terminated), then we do not include that cut in this sequence. For all $1\leq i\leq q$, denote by $\tau_i$ the time when cut $(A_i,B_i)$ was produced, by $E_i=E_G(A_i,B_i)$ at time $\tau_i$, and by $J_i\subseteq A_i,J'_i\subseteq B_i$ the vertices that serve as endpoints of the edges in $E_i$ (with respect to $G\attime[\tau_i]$ -- graph $G$ at time $\tau_i$). 
If $w(A_i)\geq w(B_i)$, then we denote $Z_i=A_i\setminus J_i$, and otherwise we denote $Z_i=B_i\setminus J'_i$. We also  let $\overline Z_i=V(G\attime[\tau_i])\setminus Z_i$ contain the remaining vertices of the current graph $G\attime[\tau_i]$. For convenience, if $Z_i=A_i\setminus J_i$, then we say that cut $(A_i,B_i)$ is of \emph{type 1}, and otherwise we say that it is of \emph{type 2}.

We denote by $G_0$ the graph $G$ at the beginning of the phase, and, for all $1\leq i\leq q$, we denote by $G_i$ the graph $G$ that is obtained after cut $(A_i,B_i)$ has been processed. Therefore, for all $1\leq i\leq q$, graph $G_i$ is obtained from graph $G_{i-1}$ by deleting the vertices of $\overline Z_i$ from it, or, equivalently, $G_{i}=G_{i-1}[Z_i]$.

We denote by $H_0$ the graph $H$ at the beginning of the current phase. For all  $1\leq i\leq q$, we denote by $H'_i$ the graph $H$ just before the cut $(A_i,B_i)$ is processed, and we denote by $H_i$ the graph $H$ immediately after cut $(A_i,B_i)$ is processed.
Therefore, graph $H_i$ is obtained from $H'_{i}$ by deleting from it all vertices and regular edges that were deleted from $G_{i-1}$; additionally, for every special edge $e$ that was deleted from $G_{i-1}$, all vertices and edges on the corresponding path $P^*(e)$ are deleted from $H'_{i}$. We denote $V(H_i)$ by $Z^*_i$, and we denote by $\overline{Z}^*_i=V(H'_{i})\setminus V(H_i)$ the set of all vertices that were deleted from $H'_i$ when cut $(A_i,B_i)$ of $G$ was processed.

We note that, in addition to the updates due to processing the cuts $(A_j,B_j)$, the only other kind of updates that graph $H$ undergoes is the splitting of some of its special edges, that happens whenever the length of the corresponding parent-edge is doubled. Therefore, for all $1\leq i\leq q$, graph $H'_i$ can be obtained from $H_{i-1}$ by executing a sequence of special edge splitting operations.

Consider some index $1\leq i\leq q$, and consider an edge $e=(x,y)\in E_i$. Recall that this edge is represented by a path $P^*(e)$ in graph $H$. Additionally, every regular edge $e'$ that is incident to $x$ or $y$ in $G$ is present in graph $H$. We let $\Gamma(e)\subseteq E(H'_{i})$ contain all regular edges that are incident to $x$ or $y$ in $H'_i$, and all edges of $P^*(e)$.
We then set $\Gamma_i=\bigcup_{e\in E_i}\Gamma(e)$. Notice that, if $(A_i,B_i)$ is a type-1 cut (so $Z_i=A_i\setminus J_i$), then graph $H'_{i}\setminus \Gamma_i$ contains no edge connecting a vertex of $Z^*_i$ to a vertex of $\overline Z^*_i$. Similarly, if $(A_i,B_i)$ is a type-2 cut (so $Z_i=B_i\setminus J'_i$), the graph $H'_{i}\setminus \Gamma_i$ contains no edge connecting a vertex of $\notZ^*_i$ to a vertex of $Z^*_i$. Therefore, it may be convenient to think of the update to graph $H'_{i}$ due to the cut $(A_i,B_i)$ in $G$ as follows: first, we delete the edges of $\Gamma_i$ from $H'_{i}$. Then we obtain a sparsity-$0$ cut in the resulting graph (with the cut being $(Z^*_i,\notZ^*_i)$ if $(A_i,B_i)$ is a type-1 cut, or $(\notZ^*_i,Z^*_i)$ otherwise). After that we delete the vertices of one of the sides of the cut (vertices of $\notZ^*_i$) from the resulting graph, obtaining graph $H_i$.

Before we describe the updates to the expander $X$, recall that $X$ undergoes other types of updates, in which edges and vertices may be deleted from it. But for the sake of analyzing the updates of type 3, it will be helpful for us to ignore the other types of updates that expander $X$ undergoes. Therefore, we denote the initial graph $X$ by $X_0$, and, for $1\leq i\leq q$, we denote by $X_i$ the graph that is obtained from $X_{i-1}$ once the cut $(A_i,B_i)$ is processed. Except for updating the graph $X$ while processing the cuts $(A_i,B_i)$, we currently do not allow any other updates to $X$. We will show how combine all types of updates to $X$ together later.

Fix again an index $1\leq i\leq q$, and consider an edge $e=(x,y)\in E_i$. Recall that we have defined a collection $\Gamma(e)$ of edges of $H'_{i}$, that includes all edges on the path $P^*(e)$, and all regular edges that are incident to $x$ or to $y$ in $H'_i$. We define the corresponding subset $\hat \Gamma(e)$ of edges of $X_{i-1}\setminus F$: $\hat \Gamma(e)=\bigcup_{e'\in \Gamma(e)}S'(e')$. In other words, an edge $\hat e\in E(X_{i-1})\setminus F$ belongs to the set $\hat \Gamma(e)$ if and only if its corresponding embedding path $Q(\hat e)$ contains at least one edge of $\Gamma(e)$. We need the following simple observation.

\begin{observation}\label{obs: few embedding edges affected}
	For every edge $e=(x,y)\in E_i$, $|\hat \Gamma(e)|\leq 2\eta''$. 
\end{observation}

\begin{proof}
	Let $T(e)$ be the set of all terminals that lie on path $P^*(e)$.
	Since we have assumed that Event $\event_4$ did not happen, $|T(e)|\leq 16\log N$.
	
	We classify the edges in set $\hat \Gamma(e)$ into three different types, and bound the number of edges of each type separately. Consider any edge $\hat e\in \hat \Gamma(e)$.
	
	We say that  $\hat e$ is of type 1, if at least one of its endpoints lies on $P^*(e)$. On this case, at least one endpoint of $\hat e$ belongs to set $T(e)$. Since $|T(e)|\leq 16\log N$, and since the degree of every vertex in $X$ is at most $24$, we get that the total number of type-1 edges in $\hat \Gamma(e)$ is bounded by $O(\log N)\leq \eta''$. %We can compute the set of all type-1 edges in $\hat \Gamma(e)$ in time $O(\ell(e))\leq O(d')$, by considering the path $P^*(e)$ and computing the set $T(e)$ of terminals.
	
	We say that edge $\hat e$ is of type 2, if it is not of type 1, and, additionally, its embedding path $Q(\hat e)$ contains at least one edge of $P^*(e)$. Since every inner vertex on $P^*(e)$ has in- and out-degree $1$, and the endpoints of $\hat e$ do not lie on $P^*(e)$, it must be the case that $P^*(e)\subseteq Q(\hat e)$. Since the embedding of $X\setminus F$ into $H$ causes congestion at most $\eta''$, we get that the total number of type-2 edges in $\hat \Gamma$ is bounded by $\eta''$. 
	
	If edge $\hat e$ is not of the first two types, then we say that it is of type 3. We claim that no type-3 edges may exist. Indeed, assume for contradiction that some edge $\hat e\in \hat \Gamma(e)$ is of type $3$. Then path $Q(\hat e)$ must contain some regular edge $e'$ of $H'_i$ that is incident to $x$ or $y$; assume w.l.o.g. that $e'$ is incident to $x$. However, $x$ is not an endpoint of $\hat e$, and so it is not an endpoint of $Q(\hat e)$. Since the only edge leaving $x$ in $G$ is $e$, path $Q(\hat e)$ must contain the first edge on $P^*(e)$. But then $\hat e$ must belong to one of the first two types, a contradiction. 
\end{proof}

Fix again an index $1\leq i\leq q$. We now define the set $\hat \Gamma_i$ of edges of $X_{i-1}\setminus F$, as follows: $\hat \Gamma_i=\bigcup_{e\in E_i}\hat \Gamma(e)$. 
Note that, equivalently, $\hat \Gamma_i$ contains all edges $\hat e$ of $X_{i-1}\setminus F$, whose embedding path $Q(\hat e)$ contains at least one edge of $\Gamma_i$. Notice that, from \Cref{obs: few embedding edges affected}, $|\hat \Gamma_i|\leq 2|E_i|\cdot \eta''$.
We denote $\hat E_3=\bigcup_{i=1}^q\hat \Gamma_i$.
We need the following simple observation.

\begin{observation}\label{obs: few deleted edges}
$\sum_{i=1}^q|E_i|\leq \frac{\phi\cdot \hat W}{\log^{30}N}$, and $|\hat E_3|\leq \frac{\alpha_0\cdot |T|}{2^{20}}$.
%\frac{2\phi\cdot \eta''\cdot \hat W}{\log^{30}N}\leq 
\end{observation}
\begin{proof}
Fix an index $1\leq i\leq q$. Since cut $(A_i,B_i)$ is $\phi$-sparse in $G_{i-1}$, $|E_i|\leq \phi\cdot \min\set{w(A_i),w(B_i)}$. Moreover, once cut $(A_i,B_i)$ is processed, the total weight of all vertices in $G$ decreases by at least $\min\set{w(A_i),w(B_i)}\geq \frac{|E_i|}{\phi}$, and so at least $\frac{|E_i|}{\phi}$ vertices are deleted from $H'_i$ to obtain $H_i$. 
%In other words, $|V(H_i)|\leq |V(H_{i-1})|-\frac{|E_i|}{\phi}$.
Since we have assumed that, by the time that cut $(A_q,B_q)$ has been processed, fewer than $\frac{\hat W}{\log^{30}N}$ vertices were deleted from $H$, we get that 
${\sum_{i=1}^q|E_i|\leq \frac{\phi\cdot \hat W}{\log^{30}N}}$.
	
From \Cref{obs: few embedding edges affected}, for all $1\leq i\leq q$, $|\hat \Gamma_i|\leq 2|E_i|\cdot \eta''$. Therefore:

\[|\hat E_3|=\sum_{i=1}^q|\hat \Gamma_i|\leq 2\eta''\cdot \sum_{i=1}^q|E_i|\leq \frac{2\phi\cdot \eta''\cdot \hat W}{\log^{30}N}\leq O\left (\frac{\hat W}{d'\cdot \log^{17}N}\right )\leq \frac{\alpha_0\cdot \hat W}{2^{21}\cdot d'}\leq \frac{\alpha_0\cdot |T|}{2^{20}},\]

since $\eta''=\Theta(\log^{13}N)$, $\phi=\frac{4}{d'}$, and $N$ is large enough.
\end{proof}

%Recall that, following the deletion of the vertices of $Z'_i$ from $G$, we also update graph $H$, so that all vertices and regular edges that were deleted from $G$ are also deleted from $H$, and, for every special edge $e$ that is deleted from $G$, the vertices and the edges of the corresponding path $P^*(e)$ are deleted from $H$. Equivalently, graph $H$ can be obtained from $G[Z_i]$, buy replacing every special edge $e\in G[Z_i]$ with the corresponding path $P^*(e)$ that contans $\ell(e)$ special edges. For convenience, we denote by $Z_i^*$ the set of vertice of this updated graph $H$, and by $\overline Z_i^*$ the set of vertices that were deleted from $H$. Let $H_{i-1}$ denote graph $H$ before the current update, and let $H_i$ denote the same graph after the current update. Notice that, if $|A_i|\geq |B_i|$ (so $Z_i=A_i\setminus J_i$), then graph $H_{i-1}\setminus \Gamma_i$ contains  no edge connecting a vertex of $Z^*_i$ to a vertex of $\overline Z_i^*$. Similarly, if $|B_i|>|A_i|$ (so $Z_i=B_i\setminus J'_i$), then $H_{i-1}\setminus \Gamma_i$ contains no edge connecting a vertex of $\overline Z_i^*$ to  a vertex of $Z^*_i$.

We are now ready to describe the type-3 updates to the expander $X$. Fix an index $1\leq i\leq q$. In order to update expander $X_{i-1}$ following the cut $(A_i,B_i)$ in graph $G_{i-1}$ that our algorithm produced, we start by deleting the edges of $\hat \Gamma_i$ from $X_{i-1}$. 
For each such edge $\hat e\in \hat \Gamma_i$, we consider every edge $e\in Q(\hat e)$, and update the set $S'(e)$ of edges by deleting $\hat e$ from it. 
Next, we denote $T_i=V(X_{i-1})\cap Z^*_i$, and $\overline T_i=V(X_{i-1})\cap \notZ_i^*$. If $\notT_i\neq\emptyset$, then we delete the vertices of $\notT_i$ from $X_{i-1}$, obtaining the final graph $X_i$. We now discuss this latter update in more detail, presenting a slightly different view of this update, that will be helpful for us later.

Assume first that $(A_i,B_i)$ in $G_{i-1}$ is a type-1 cut, so $Z_i=A_i\setminus J_i$. As observed already, in this case, graph $H'_{i}\setminus \Gamma_i$ contains no edge connecting a vertex of $Z^*_i$ to a vertex of $\notZ^*_i$. It is then easy to verify that graph $X_{i-1}\setminus \hat \Gamma_i$ contains no edge connecting a vertex of $T_i$ to a vertex of $\overline T_i$. Indeed, if such an edge $\hat e$ existed, then its corresponding path $Q(\hat e)$ would connect a vertex of $Z^*_i$ to a vertex of $\notZ^*_i$ in $H_i'$, and so it must contain an edge of $\Gamma_i$. But then $\hat e\in \hat \Gamma_i$ holds, a contradiction. Therefore, if $(A_i,B_i)$ is a type-1 cut, and $T_i,\overline T_i\neq \emptyset$, then $(T_i,\overline T_i)$ is a cut of sparsity $0$ in $X_{i-1}\setminus \hat \Gamma_i$.

Next, assume that $(A_i,B_i)$ is a type-2 cut. As observed already, in this case, graph $H'_i\setminus \Gamma_i$ contains no edge connecting a vertex of $\notZ^*_i$ to a vertex of $Z^*_i$. From the same arguments as above, graph $X_{i-1}\setminus \hat \Gamma_i$ contains no edge connecting a vertex of $\notT_i$ to a vertex of $\overline T_i$. Therefore, if $(A_i,B_i)$ is a type-2 cut, and $T_i,\overline T_i\neq \emptyset$, then $(\notT_i,T_i)$ is a cut of sparsity $0$ in $X_{i-1}\setminus \hat \Gamma_i$.

To conclude, we can view the update to graph $X_{i-1}$ as follows: first, we delete all edges of $\hat \Gamma_i$ from $X_{i-1}$. Then we either delete all vertices from $X_{i-1}$ (if $T_i=\emptyset$), or delete no vertices from $X_{i-1}$ (if $\notT_i=\emptyset$), or we obtain a sparsity-$0$ cut in the resulting graph $X_{i-1}\setminus\hat \Gamma$, and we delete all vertices on one side of this cut from $X_{i-1}$.
This view of the update procedure will be helpful for us later.
It is easy to verify that, at the end of this update, the collection $\set{Q(\hat e)\in \qset\mid \hat e\in E(X_i)\setminus F}$ defines a valid embedding of graph $X_i\setminus F$ into graph $H_i$. Indeed, we have ensured that $T_i\subseteq V(H_i)$. Additionally, for every edge $\hat e\in E(X_i)\setminus F$, both endpoints of $\hat e$ lie in set $Z^*_i$, and, since path $Q(\hat e)$ is disjoint from the edges of $\Gamma_i$, it may not use vertices of $\notZ^*_i$, and must therefore be contained in $H_i$.

Lastly, we would like to claim that $|V(X_q)|$ remains sufficiently large compared to $|T|=|V(X_0)|$, but unfortunately this may not be the case. Instead we will show that, if $1\leq i\leq q$ is the smallest index, for which $|X_i|\leq |T|\cdot \left(1-\frac{8}{\log^4N}\right )$ holds, then we can compute a finalizing cut in graph $G_{i-1}$. 
%either correctly establish that at least one of the events $\event_2,\event_3$ has happened; or we can compute a $\phi$-sparse cut $(A^*,B^*)$ in $G_i$, such that, once cut $(A^*,B^*)$ is processed in $G_i$, and graph $H_i$ is updated with the resulting deletions of edges and vertices, the total number of vertices deleted from $H$ since the beginning of the phase becomes at least $\frac{\hat W}{(\log N)^{30}}$. In the former case, we will terminate the algorithm with ``FAIL'', while in the latter case, we will process cut $(A^*,B^*)$ in $G_i$, following which a sufficient number of vertices will be deleted from $H$ so that the first phase will terminate. 
In our algorithm, we will monitor the number of vertices in the current graph $X_i$, and, if it ever falls below $|T|\cdot \left(1-\frac{8}{\log^4N}\right )$, we will compute a finalizing cut, and apply the algorithm from \Cref{claim: finalizing cut} to it, after which the current phase will be terminated, so in particular no further cuts $(A_j,B_j)$ will be produced by the algorithm during the current phase. Therefore, essentially we only need to consider the case where $|V(X_{q-1})|\geq |T|\cdot \left(1-\frac{8}{\log^4N}\right )$ but $|V(X_q)|< |T|\cdot \left(1-\frac{8}{\log^4N}\right )$. The following claim summarizes the algorithm that will allow us to complete the current phase in such a case. For a vertex $v\in V(G_0)$, we denote by $\hat w(v)$ its weight at the beginning of current phase, and for a subset $S$ of vertices of $G_0$, we denote by $\hat w(S)=\sum_{v\in S}\hat w(v)$.

\begin{claim}\label{claim: lots of vertices in X after second type of updates}
	Suppose that $|V(X_{q-1})|\geq |T|\cdot \left(1-\frac{8}{\log^4N}\right )$ but $|V(X_q)|< |T|\cdot \left(1-\frac{8}{\log^4N}\right )$. Then there is a deterministic algorithm, that, given graphs $G_0,X_0$, and the sequence $(A_1,B_1),\ldots,(A_q,B_q)$ of cuts in $G$, and graph $H_{q-1}$, computes a finalizing cut $(\tA,\tB)$ in $H_{q-1}$. The running time of the algorithm is $O(\hat W^2)$.
\end{claim}
\begin{proof}
	For convenience, in this proof we denote graph $G_{q-1}$ by $G'$ and graph $H_{q-1}$ by $H'$.
	We compute the finalizing cut $(\tA,\tB)$ in $H'$ as follows. First, for every vertex $v\in V(G')$, if $v\in A_q$, then we add $v$ to $\tA$, and otherwise we add $v$ to $\tB$. Next, we consider every special edge $e=(x,y)$ in $G'$. If $x,y\in A_q$, then we add all inner vertices on path $P^*(e)$ to $\tA$. If $x,y\in B_q$, then we add all inner vertices on path $P^*(e)$ to $\tB$.
	
	Next, we consider two cases. The first case happens if $(A_q,B_q)$ is a type-1 cut, so $V(G_q)=A_q\setminus J_q$. For each remaining special edge (that must have first endpoint in $A_q$ and last endpoint in $B_q$), we add all inner vertices of $P^*(e)$ to $\tB$. 
	
	Note that, for every special edge $e$ of $G'$, the first and the last edges on path $P^*(e)$ are special edges, so the resulting cut $(\tA,\tB)$ in $H'$ is weakly well-structured. Clearly:

	\[|E_{H'}(\tA,\tB)|=|E_{G'}(A_q,B_q)|\leq \phi\cdot \min\set{w(A_q),w(B_q)}\leq \frac{4}{d'}\cdot \frac{\hat W}{\log^{30}N}\leq \frac{8|T|}{\log^{30}N},\]
	
	since $\phi=\frac{4}{d'}$. 
	We have also used the fact that, since the current phase has not terminated yet, $w(A_q),w(B_q)\leq \frac{\hat W}{\log^{30}N}$, and that, since Event $\event_1$ did not happen, $|T|\geq\frac{\hat W}{2d'}$.
	
	Notice that, from our construction, $\tA=V(H_q)\cup J_q$. All vertices of $\tB$ are deleted from $H'$ when the cut $(A_q,B_q)$ is processed, and, since the current phase did not terminate yet, $|\tB|\leq \frac{\hat W}{\log^{30}N}$. Since $\tA=V(H_q)\cup J_q$, we get that $|\tA\cap T|\leq |T_q|+|J_q|$. Recall that we have assumed that $|T_q|=|V(X_q)|< |T|\cdot \left(1-\frac{8}{\log^4N}\right )$. Additionally, $|J_q|\leq |E_{G'}(A_q,B_q)|\leq \frac{8|T|}{\log^{30}N}$. Altogether, we conclude that $|\tA\cap T|\leq |T|\cdot \left(1-\frac{7}{\log^4N}\right )$, and $|\tA|\geq |V(H')|-\frac{\hat W}{32}$. Therefore, cut $(\tA,\tB)$ is a finalizing cut in graph $H'$.
	
	The case where $(A_q,B_q)$ is a type-2 cut is treated similarly, except that now, for every special edge $e$ of $G'$ whose endpoints belong to different sets in $\set{A_q,B_q}$, we place all inner vertices of $P^*(e)$ in $\tA$.
	Using the same reasoning as before, all edges of $E_{H'}(\tA,\tB)$ are special edges, and $|E_{H'}(\tA,\tB)|\leq \frac{8|T|}{\log^{30}N}$ holds. 
	 Our construction ensures that $\tB=V(H_q)\cup J'_q$. Using the same argument as above, we conclude that $|\tB\cap T|\leq |T_q|+|J'_q|\leq |T_q|+|E_{H'}(\tA,\tB)\leq |T|\cdot \left(1-\frac{7}{\log^4N}\right )$. Similarly, $|\tB|\geq |V(H')|-\frac{\hat W}{32}$, and so cut $(\tA,\tB)$ is finalizing.
\end{proof}

Whenever our algorithm computes a strongly well-structured $\phi$-sparse cut $(A_i,B_i)$ in the current graph $G$, we will first compute the corresponding sets $E_i$, $\Gamma_i$, and $\hat \Gamma_i$ of edges, together with the set $T_i$ of terminals, before applying any changes to graphs $G$ and $H$. If $|T_i|\geq |T|\cdot \left(1-\frac{8}{\log^4N}\right )$, then we proceed with updating the graphs $G,H$ and $X$, as described above. But otherwise, we use the algorithm from 
\Cref{claim: lots of vertices in X after second type of updates} in order to compute a finalizing cut $(\tA,\tB)$ in the current graph $H$, and then employ the algorithm from \Cref{claim: finalizing cut}, that either correctly establishes that at least one of the events $\event_2,\event_3$ has happened (in which case we terminate the phase with a ``FAIL''), or  computes a strongly well-structured $\phi$-sparse cut $(A^*,B^*)$ in the current graph $G$, such that $w(A^*),w(B^*)\geq \frac{\hat W}{\log^{30}N}$. In the latter case, we do not process the cut $(A_i,B_i)$, but instead output the cut $(A^*,B^*)$, and update both graphs $G$ and $H$ accordingly, following which at least $\frac{\hat W}{\log^{30}N}$ vertices must be deleted from $H$, and the current phase terminates.

\subsubsection*{Combining All Three Types of Updates}

We now complete our description of the algorithm for maintaining the expander $X$ and the embedding of $X\setminus F$ into $H$ throughout the phase. We will also maintain another graph $X'$. Initially, $X'=X$. Whenever a type-3 update $(A_i,B_i)$ occurs, we update the graph $X'$ as described above, but we will ignore type-1 and type-2 updates when maintaining $X'$. We will sometimes refer to $X'$ as the ``shadow expander'', while we refer to $X$ as the ``expander'', even though, as the algorithm progresses, it is possible that both graphs stop being expanders.

In order to maintain the graph $X$, whenever a type-1 or a type-2 update arrives, we process it exactly as describe above. Assume now that the $i$th type-3 update $(A_i,B_i)$ arrives.
In this case, we use graph $X'$ in order to compute the set $\hat \Gamma_i$ of edges, and the partition $(T_i,\notT_i)$ of the vertices of $X'$, as described above, and we update graph $X'$ accordingly. If, as the result, $|V(X')|$ falls below 
$|T|\cdot \left(1-\frac{8}{\log^4N}\right )$, then we invoke the algorithm from \Cref{claim: lots of vertices in X after second type of updates}, that computes a finalizing cut $(\tilde A,\tilde B)$ in the current graph $H$. We then apply the algorithm from  \Cref{claim: finalizing cut} to the resulting cut. If the algorithm
correctly estabishes that at least one of the events $\event_2,\event_3$ has happened,  then we terminate our algorithm and return ``FAIL''. Otherwise,
we obtain a strongly well-structured $\phi$-sparse cut $(A^*,B^*)$ in the current graph $G$, such that $w(A^*),w(B^*)\geq \frac{\hat W}{\log^{30}N}$.
We output this cut, and update the graphs $G$ and $H$ accordingly, following which at least $\frac{\hat W}{\log^{30}N}$ vertices are deleted from $H$ and the current phase terminates.
%The total running time of the algorithms from  \Cref{claim: lots of vertices in X after second type of updates} and \Cref{claim: finalizing cut} is $O(\hat W^2)$.

%If this algorithm establishes that at least one of the events $\event_2$ or $\event_3$ has happened, then we terminate our algorithm and return ``FAIL''. Otherwise, we obtain  a well-structured $\frac{\phi}2$-sparse cut $(A,B)$ in graph $G_0$ with $w(A),w(B)\geq \frac{2\hat W}{(\log N)^{30}}$. Notice that the current graph $G$ is identical to graph $G_i$, that is obtained from the initial graph $G_0$ at the beginning of the phase, by only processing type-3 updates. The only difference is that the weights of some vertices in $G$ may be greater than those in $G_i$, since, in addition to the changes incurred while processing type-3 updates, the lengths of some edges of $G$ may have grown. Therefore, we can exploit the cut $(A,B)$ in $G_0$, as described above, in order to compute a $\phi$-sparse well-structured graph $(A_{q+1},B_{q+1})$ in the current graph $G$, with $w(A),W(B)\geq \frac{\hat W}{\log^{30}N}$. Once graphs $G$ and $H$ are updated with this cut, at least $\frac{\hat W}{\log^{30}N}$ vertices are deleted from $H$ and the current phase terminates.

Assume now that $|V(X')|\geq |T|\cdot \left(1-\frac{8}{\log^4N}\right )$ continues to hold after the current update. We then update graph $X$ as follows: we delete the edges of $\hat \Gamma_i$, and the vertices of $\notT_i$ from it. 
This completes the description of the algorithm for maintaining graph $X$. It is easy to verify that, at all times, $X\subseteq X'$ holds. Therefore, when the $i$th type-3 update is processed, immediately after the edges of $\hat \Gamma_i$ are deleted from $X$, one of the cuts $(T_i\cap V(X),\notT_i\cap V(X))$, $(\notT_i\cap V(X),T_i\cap V(X))$ must be a cut of sparsity $0$ in the resulting graph $X$.

Let $E^*=F\cup \hat E_1\cup \hat E_2\cup \hat E_3$. Recall that $|F|\leq \frac{\alpha_0\cdot |T|}{500}$, and we have shown that 
$|\hat E_1|\leq \frac{\alpha_0\cdot |T|}{5000}$ (in \Cref{claim: bound first type edge deletions}), $|\hat E_2|\leq \frac{\alpha_0\cdot |T|}{2^{20}}$, and $|\hat E_3|\leq \frac{\alpha_0\cdot |T|}{2^{20}}$ (in \Cref{obs: few deleted edges}). 
Altogether, we get that $|E^*|\leq \frac{\alpha_0\cdot |T|}{400}$ holds throughout the phase.
We now prove that, throughout the phase, $|V(X)|\geq |T|/2$ holds.

\begin{claim}\label{claim: expander remains large}
	Throughout the phase,  $|V(X)|\geq |T|/2$.
\end{claim}
\begin{proof}
	Assume otherwise. Consider the first time $\tau$ at which $|V(X)|<|T|/2$ holds. For convenience, we denote by $X_0$ the original graph $X$, which is an $\alpha_0$-expander, with $|T|$ vertices and we denote by $\tilde X$ the graph $X\setminus E^*$. 
	
	We construct the set $\yset$ of subsets of vertices of $X_0$, as follows. Initially, $\yset=\emptyset$. Whenever a type-2 update arrives, with a cut $(A,B)$ of the current graph $X$, if $|A|\leq |B|$, then we add $A$ to $\yset$ (recall that in this case, the vertices of $A$ are deleted from $X$); otherwise, we add $B$ to $\yset$. Clearly the cardinality of the set that we just added to $\yset$ is bounded by $|T|/2$. Moreover, since the edges of $E_X(A,B)$ belong to $E^*$, cut $(A,B)$ has sparsity $0$ in $\tilde X$.
	
	Assume now that a type-3 update arrives, and let $T'$ denote the set of vertices deleted from $X$ in the current iteration. We then add $T'$ to $\yset$. Since we did not terminate the current phase after the current iteration, $|V(X')|\geq |T|\cdot \left(1-\frac{8}{\log^4N}\right )$ must currently hold, and so $|T'|\leq \frac{8|T|}{\log^4N}\leq \frac{|T|}{2}$ must hold. Moreover, as we have shown, one of the cuts $(T',V(X)\setminus T')$ or $(V(X)\setminus T',T')$ is a sparsity-$0$ cut in graph $\tilde X$.
	
	Let $\yset=(Y_1,\ldots,Y_q)$ be the final collection of vertex subsets obtained after all updates that occured until time $\tau$ have been processed, where the sets are indexed in the order in wich they were added to $\yset$. We also add the set $Y_{q+1}$ that contains all vertices that lie in $X$ at time $\tau$ to $\yset$. From our assumption, $|X|\leq |T|/2$, and so for all $1\leq q'\leq q+1$, $|Y_{q'}|\leq |T|/2$.
	
	For all $1\leq j\leq q$, let $Y'_j=Y_{j+1}\cup\cdots\cup Y_q$. At the time when $Y_j$ was added to $\yset$, it defined a cut of sparsity $0$ in graph $\tilde X$. In other words, either $E_{\tilde X}(Y_j,Y'_j)=\emptyset$, or
	$E_{\tilde X}(Y_j,Y'_j)=\emptyset$. Let $\alpha$ be such that $\max_{Y_j\in \yset}\set{|Y_j|}=\alpha\cdot |T|$. From our discussion so far, $\alpha\leq \half$. From \Cref{obs: from many sparse to one balanced}, there is a cut $(A,B)$ in $\tilde X$ with $|A|,|B|\geq |T|/4$, and $E_{\tilde X}(A,B)=\emptyset$.

	Observe that $|E_{X_0}(A,B)|\leq |E^*|\leq \frac{\alpha_0\cdot |T|}{400}$. Since $|A|,|B|\geq \frac{|T|}{4}$, we obtain a cut $(A,B)$ in $X_0$, whose sparsity is below $\alpha_0$, contradicting the fact that $X_0$ is an $\alpha_0$-expander.
\end{proof}

%% file: connect-to-terminals.tex
\subsection{Second Data Structure: Connecting to the Vertices of $X$}
\label{subsec: connect to terminals}

At the beginning of the phase, we initialize the algorithm for the \connecttocenters problem from \Cref{thm: connecttocenters} on the current graph $G$, with $\beta=V(X)$, parameters $d'$ and $r$ remaining unchanged, and parameter $N$ replaced by $N'=8N$.
We slightly modify the lengths of the special edges: for each such edge $e\in E^{\spec}$, we define the new length $\ell'(e)=\max\set{\ell(e),2\ell^0(e)}$. We define, for every vertex $v\in V(G)$, a weight $w'(v)$ to be the length $\ell'(e)$ of the unique special edge $e$ incident to $v$, and we let $W'=\sum_{v\in V(G)}w'(v)$. Clearly, $W'\leq 2\hat W$. As the lengths of special edges in the original graph $G$ grow, we always set, for every edge $e\in E(G)$, the length $\ell'(e)$ that is used by the algorithm for the \connecttocenters problem to $\ell'(e)=\max\set{\ell(e),2\ell^0(e)}$, and we also define the weights $w'(v)$ of vertices $v\in V(G)$ with respect to these edge lengths.  Clearly, for all $e\in E^{\spec}$, $\ell(e)\leq \ell'(e)\leq 2\ell(e)$ and for all $v\in V(G)$, $w(v)\leq w'(v)\leq 2w(v)$ always hold. We use $W'$ to denote the sum of the initial weights $w'(v)$ for $v\in V(G)$, so it does not change over the course of the algorithm.
Finally, the graphs $G$ and the corresponding subdivided graph $H'$ that serve as input to the \connecttocenters problem must be given in the adjacency-list representation. We can compute both adjacency-list representations in time $O(|V(H')|^2)\leq O((W')^2)\leq O((W^0)^2)$.
We prove that we obtain a valid input to the \connecttocenters problem in the following simple claim, whose proof deferred to Section \ref{subsec: valid inst of connecttocenters} of Appendix.

\begin{claim}\label{claim: valid input to conencttocenters}
	$(G,\set{\ell'(e)}_{e\in E(G)},H', N',d',r,\beta)$ is a valid input to the \connecttocenters problem.
\end{claim}

As the algorithm progresses, we will ensure that $\beta=V(X)$ always holds. Since, from \Cref{claim: expander remains large},  $|V(X)|\geq |T|/2$ holds throghout the phase, this will ensure that at most half the vertices may be deleted from the initial set $\beta$ over the course of the phase. From \Cref{obs: final length of edges},
$\sum_{e\in E^{\spec}}\ell^*(e) \leq  2\sum_{e\in E^{\spec}}\ell^0(e)+2n$. Therefore, 
 the total increase in the lengths $\ell'(e)$ of edges $e\in E(G)$ over the course of the entire algorithm is at most $2n\leq 8|V(G)|$, as required.
 
Consider the graph $H'$, that is the subdivided graph corresponding to $G$ and the lengths $\ell'(e)$ of its edges, and recall that $H$ is the subdivided graph corresponding to $G$ and the lengths $\ell(e)$ of its edges. Recall that, for every edge $e\in E(G)$, $\ell(e)\leq \ell'(e)\leq 2\ell(e)$. Since the lengths of all special edges are integral powers of $2$,  for every special edge $e\in E(G)$, either $\ell'(e)=\ell(e)$ or $\ell'(e)=2\ell(e)$ holds. Therefore, graph $H'$ can be obtained from graph $H$ as follows: for every edge $e\in E^{\spec}$ with $\ell'(e)=2\ell(e)$, we subdivide every special edge $\hat e\in E(H)$ that lies on the path $P^*(e)$ representing $e$ once, obtaining a path $\hat P(\hat e)$ with $2$ special edges, representing edge $\hat e$ in $H'$. For all edges $\hat e\in E(H')\cap E(H)$, we let $\hat P(\hat e)=(\hat e)$.

Recall that the algorithm for the \connecttocenters problem may, at any time, produce Outcome \ref{outcome: finalizing cut}, namely, a weakly well-structured cut $(A',B')$ in graph $H'$, with $|A'|\geq \frac{511|V(H')|}{512}$,  $|A'\cap V(X)|<\frac{2|V(X)|}{\log^3N'}$, and $|E_{H'}(A',B')|\leq  \frac{|V(H')|}{1024d'\cdot \log^6N'}$. 
Consider the cut $(A,B)$ in $H$ with $A=A'\cap V(H)$ and $B=B'\cap V(H)$. We claim that $(A,B)$ is a finalizing cut in $H$, in the following simple observation.

\begin{observation}\label{obs: finalizing cut from cut in H'}
	Cut $(A,B)$ is a finalizing cut in $H$.
\end{observation}
\begin{proof}
 Since, at the beginning of the phase, $|V(H')|=W'\leq 2\hat W$, and over the course of the phase, the total increase in the lengths $\ell'(e)$ of edges $e\in E(G)$, and hence in $|V(H')|$, is at most $4n\leq 8\hat W$, we get that $|V(H')|\leq 10\hat W$ holds	throughout the phase.

Consider any edge $e=(x,y)\in E_H(A,B)$. First, we claim that $e$ is a special edge. Indeed, otherwise, $e\in E_{H'}(A',B')$ must hold, contradicting the fact that $(A',B')$ is a weakly well-structured cut. If $e\in E(H')$, then $e\in E_{H'}(A',B')$ must hold. Otherwise, at least one of the two special edges that lie on the path $\hat P(e)$ representing $e$ in $H'$ belong to $E_{H}(A,B)$. Therefore:

\[|E_H(A,B)|\leq |E_{H'}(A',B')|\leq  \frac{|V(H')|}{1024d'\cdot \log^6N'}\leq \frac{10\hat W}{1024d'\cdot\log^6N}\leq \frac{|T|}{8\log^6N},\]

since $|T|\geq \frac{\hat W}{2d'}$, as Event $\event_1$ did not happen.
Recall that
 $|A'|\geq \frac{511|V(H')|}{512}$, and so: 
 
 \[|B|\leq |B'|\leq \frac{|V(H')|}{512}\leq \frac{10\hat W}{512}\leq \frac{\hat W}{32}.\]
 
 Therefore, $|A|\geq |V(H)|-\frac{\hat W}{32}$ must hold. Lastly, 
 $|A\cap V(X)|\leq |A'\cap V(X)|<\frac{2|V(X)|}{\log^3N'}<\frac{|T|}{2}$.
We conclude that $(A,B)$ is a finalizing cut in $H$.
\end{proof}

For brevity, whenever  the algorithm for the \connecttocenters on $G$  produces Outcome \ref{outcome: finalizing cut}, we say that it produces a finalizing cut $(A,B)$ in $H$.

Recall that the algorithm for the \connecttocenters problem uses a parameter $\rho=\frac{2^{12}\cdot W'\cdot \log^3N'}{|T|}\leq O\left(\frac{\hat W\log^3N}{|T|}\right )$, and that the total increase in the lengths of all edges is denoted by $\Lambda$; in our case, $\Lambda\leq O(n)$ holds. Recall also that $\beta'$ denotes the set $V(X)$ of vertices at the end of the algorithm. The running time of the algorithm from \Cref{thm: connecttocenters} for this instance of the \connecttocenters problem is then bounded by:

\[\begin{split}
&	O\left(c_1\cdot \hat W\cdot (\hat W+|T|
\cdot \rho +\Lambda)\cdot  2^{c_2\sqrt{\log N'}}\cdot (\log N')^{16c_2(r-1)+8c_2+26}\right )\\
&\quad\quad\quad\quad\quad\quad \leq O\left(c_1\cdot (W^0)^2\cdot  2^{c_2\sqrt{\log N'}}\cdot (\log N')^{16c_2(r-1)+8c_2+29}\right )\\
&\quad\quad\quad\quad\quad\quad \leq O\left(c_1\cdot (W^0)^2\cdot  2^{c_2\sqrt{\log N}}\cdot (\log N)^{16c_2(r-1)+8c_2+30} \cdot \left(1+\frac{1}{32c_2r}\right )^{16c_2(r-1)+8c_2+30}\right )\\
&\quad\quad\quad\quad\quad\quad \leq O\left(c_1\cdot (W^0)^2\cdot  2^{c_2\sqrt{\log N}}\cdot (\log N)^{16c_2(r-1)+8c_2+30}\right ).
\end{split}
\]

(we have used Inequalities \ref{eq: more logs} and \ref{eq: even more logs}.)

We denote the data structure for the \connecttocenters problem that we described above by $\DSout$, and we denote the set $\qset=\bigcup_{i=1}^{\lambda-1}\qset_i$ of paths that it maintains by $\qset^{\out}$.
If a vertex $v$ lies in layer $U_{\lambda}$, then we say that $v$ is \emph{disconnected} in $\DSout$. For each vertex $v\in \bigcup_{i=1}^{\lambda-1}U_i$, the unique path of $\qset_i$ that originates at $v$ is denoted by $K(v)$.

Next, we consider the graph $\revG$, obtained from $G$ by reversing the direction of its edges, and its corresponding subidivided graph $\revH$. It is easy to verify that $\revG$ is a perfect well-structured graph. We define an instance of the \connecttocenters problem on graph $\revG$, with the edge lengths $\set{\ell'(e)}_{e\in E(\revG)}$, set $\beta$ of vertices, and parameters $N',d'$ and $r$ defined exactly as before, and then apply the algorithm for the \connecttocenters problem from \Cref{thm: connecttocenters} to this instance. We denote the data structure maintained by this algorithm by $\DSin$. If a vertex $v$ lies in layer $U_{\lambda}$, then we say that $v$ is \emph{disconnected} in $\DSin$.
In order to avoid confusion, we denote the layers of this data structure by $U'_0,\ldots,U'_{\lambda}$. 
Recall that, for all $1\leq i<\lambda$, the algorithm maintains a collection $\qset_i=\set{K(v)\mid v\in U'_i}$ of paths in graph $\revH$, where for all $v\in U_i$, path $K(v)$ connects $v$ to some vertex of $\bigcup_{i'<i}U'_{i'}$. By reversing the direction of edges in this path, we obtain a path $K'(v)$ in $H$, that connects some vertex of $\bigcup_{i'<i}U'_{i'}$ to $v$. We denote by $\qset'_i=\set{K'(v)\mid v\in U'_i}$, and by $\qset^{\inn}=\bigcup_{i=1}^{\lambda-1}\qset'_i$. To avoid confusion, for a vertex $v\not \in U_{\lambda}$, we use $K(v)$ only to denote the path maintained by $\DSout$, that is, $K(v)$ is the unique path originating  from $v$ in $\qset^{\out}$, and for $v\not\in U'_{\lambda}$, we use $K'(v)$ to denote the unique path terminating at $v$ in $\qset^{\inn}$.

If the algorithm for the \connecttocenters problem in graph $\revG$ produces Outcome \ref{outcome: finalizing cut}, then, using the same reasoning as in the proof of \Cref{obs: finalizing cut from cut in H'}, we obtain a finalizing cut $(A,B)$ in graph $\revH$. It is then immediate to verify that $(B,A)$ is a finalizing cut in graph $H$. For brevity, whenever the algorithm for the \connecttocenters problem in graph $\revG$ produces Outcome \ref{outcome: finalizing cut}, we say that it produces a finalizing cut in $H$.

Assume now that the algorithm for the \connecttocenters problem in graph $\revG$ produces Outcome \ref{outcome: sparse cut}; that is, it computes
a strongly well-structured $\frac{1}{d'}$-sparse cut $(A,B)$ in $\revG$. It is then easy to see that $(B,A)$ is a strongly well-structured $\frac{1}{d'}$-sparse cut in $G$. Therefore, if 
the algorithm for the \connecttocenters problem in graph $\revG$ produces Outcome \ref{outcome: sparse cut}, we will say that it produces a strongly well-structured $\frac{1}{d'}$-sparse cut in $G$.

From our discussion, the total running time of the algorithms that maintain data structures $\DSin$ and $\DSout$, including the time that is needed in order to construct the corresponding two  instances of the \connecttocenters problem, and including the time that is needed in order to compute a finalizing cut in $H$, if any of the two algorithms produce  Outcome \ref{outcome: finalizing cut}, is bounded by:

\[O\left(c_1\cdot (W^0)^2\cdot  2^{c_2\sqrt{\log N}}\cdot (\log N)^{16c_2(r-1)+8c_2+30}\right ).\]

%% file: cleanup.tex
\subsection{Additional Data Structures}

Recall that our algorithm maintains a graph $X$ (that is initially an expander), a subset $F$ of edges of $X$ (the fake edges), and an embedding $\qset$ of $X\setminus F$ into $H$. Additionally, it maintains a \emph{shadow expander} $X'$, which is initially identical to $X$, but, as the time progresses, we may ignore some of the updates to graph $X$ and only apply a subset of such updates to $X'$. We also maintain a subset $F'$ of edges of $X'$, which is initially identical to $F$, and an embedding $\qset'$ of $X'\setminus F'$ into $H$, which is initially identical to $\qset$. 
Throughout the algorithm, $X'\subseteq X$  and $F'\subseteq F$ hold. Additionally, for every edge $e\in E(X')\setminus F$, its embedding path $Q'(e)\in \qset'$ is identical to the embedding path $Q(e)\in \qset$.
As the time progresses, vertices and edges may be deleted from $X$ and from $X'$. Whenever an edge $e\in F$ is deleted from $X$, we delete it from $F$ as well, and similarly, whenever an edge $e'\in F'$ is deleted from $X'$, we delete it from $F'$ as well. For every edge $e\in E(X)$, its embedding path $Q(e)\in \qset$ remains unchanged throughout the algorithm, except that, whenever some special edge on path $Q(e)$ undergoes splitting, we modify the path accordingly. Embedding $\qset'$ of $X'\setminus F'$ into $H$ is maintained similarly.

For every edge $e\in E(X)\setminus F$, we maintain the length $\hat \ell(e)$ of the embedding path $Q(e)$. We initialize the value $\hat \ell(e)$ at the beginning of the phase, once the graph $X$ and the embedding $\qset$ of $X\setminus F$ into $H$ are computed. After that, whenever a special edge of $Q(e)$ undergoes splitting, we increase $\hat\ell(e)$ by (additive) 2. The time that is required in order to maintain the values $\hat \ell(e)$ for edges $e\in E(X)\setminus F$ is clearly subsumed by the time required to compute and to maintain the collection $\qset$ of paths. In our description of the remainder of the algorithm below, we will ignore the task of maintaining these values.

For every edge $e\in E(H)$, we also maintain a list $S(e)$ of all edges $\hat e\in E(X)\setminus F$, such that $e\in Q(\hat e)$. The task of maintaining these lists is quite straightforward; we will ignore it when we describe the remainder of the algorithm below, but we will show how to extend the algorithm in order to maintain these lists efficiently later.

Throughout, we use a distance parameter $\hat d=\frac{2^{30}\log \hat W}{\alpha_0}$.
Throughout the phase, we maintain an arbitrary vertex  $v\in V(X)$,
and two shortest-paths trees $\hat \tau'$ and $\hat \tau''$ rooted at $v$ in graph $X\setminus F$, where the edges in $\hat \tau'$ are directed away from $v$, and the edges in $\hat \tau''$ are directed towards $v$. In other words, for every vertex $a\in V(X)$, the unique $s$-$a$ path in $\hat \tau'$ is a shortest $s$-$a$ path in $X\setminus F$, and the unique $a$-$s$ path in $\hat \tau''$ is a shortest $a$-$s$ path in $X\setminus F$. Note that both trees can be computed in time $O(|T|)$, since $|V(X)|\leq |T|$, and maximum vertex degree in $X$ is $O(1)$. Whenever an edge or a vertex are deleted from $X$, we recompute both trees from scratch, selecting a new root vertex $v\in V(X)$ if needed. Since $|E(X)|\le O(|T|)$, the trees need to be recomputed at most $O(|T|)$ times, and so the total time spent over the course of the phase on computing the two trees is bounded by $O(|T|^2)\leq O((W^0)^2)$.

\subsection{The Remainder of the Algorithm}
\label{subsection: cleanup}

In this subsection we describe the remainder of the algorithm for a single phase of the algorithm for the \maintaincluster problem.
At the beginning of the phase, we construct the subdivided graph $H$ corresponding to the initial graph $G$, and then construct the set $T$ of terminals, as described in Section \ref{subsec: subsample}. We check whether either of the bad events $\event_1$ or $\event_4$ happened, and if so, we terminate the phase and return ``FAIL''. We assume from now on that neither of these events happened. The running time of the algorithm so far is $O(|E(H)|)\leq O((W^0)^2)$.

We execute the algorithm from \Cref{lem: phase 1 embed expander} in order to construct the expander $X$. If the algorithm terminates with a ``FAIL'', which may only happen with probability at most $\frac{1}{N^7}$, then we terminate the algorithm for the current phase and return ``FAIL''. If the algorithm returns 
a finalizing cut in $H$, then we apply the algorithm from \Cref{claim: finalizing cut} to this cut. If the algorithm from \Cref{claim: finalizing cut}  estabishes that at least one of the events $\event_2,\event_3$ happened,  we terminate the phase with a ``FAIL''. Otherwise, we obtain a strongly well-structured $\phi$-sparse cut $(A^*,B^*)$ in the current graph $G$, such that $w(A^*),w(B^*)\geq \frac{\hat W}{\log^{30}N}$. In this case, we output the cut $(A^*,B^*)$. Once this cut is processed, at least $\frac{\hat W}{\log^{30}N}$ vertices are deleted from $H$, and the current phase terminates. Finally, if the algorithm from
\Cref{lem: phase 1 embed expander} returns a strongly well-structured $1/d'$-sparse cut $(A,B)$ in $G$ with $w(A),w(B)\geq \frac{\hat W}{(\log N)^{30}}$, then as before we output this cut, and, once it is processed, at least $\frac{\hat W}{\log^{30}N}$ vertices are deleted from $H$, and the current phase terminates. Therefore, we assume from now on that the algorithm from  \Cref{lem: phase 1 embed expander} computed a graph $X$ with $V(X)=T$, whose maximum vertex degree is at most $24$, such that $X$ is an $\frac{\alpha_0}{81}$-expander. 
Recall that we are also given a set $F\subseteq E(X)$ of at most  $\frac{\alpha_0\cdot |T|}{500}$ fake edges, and an embedding $\qset$ of $X\setminus F$ into $H$. Every path in $\pset^*$ has length at most $d''$, and the congestion caused by the paths in $\pset^*$ is at most $\eta''$.
The running time of the algorithm from \Cref{lem: phase 1 embed expander} is bounded by:

\[O\left( c_1\cdot \hat W^{2} \cdot 2^{c_2\sqrt{\log N}}\cdot (\log N)^{16c_2(r-1)+8c_2+3} \right )\leq O\left( c_1\cdot (W^0)^{2} \cdot 2^{c_2\sqrt{\log N}}\cdot (\log N)^{16c_2(r-1)+8c_2+3} \right ),\]

and the running time of the algorithm from \Cref{claim: finalizing cut} is $O(\hat W^2)\leq O((W^0)^2)$.

We also initialize data structures $\DSin$ and $\DSout$. Recall that the total time for maintaining both data structures is bounded by:

\[O\left(c_1\cdot (W^0)^2\cdot  2^{c_2\sqrt{\log N}}\cdot (\log N)^{16c_2(r-1)+8c_2+30}\right ).\]

Recall that the algorithm for the \connecttocenters problem, that is used to maintain data structures $\DSin$ and $\DSout$ may return ``FAIL'' with probability at most $\frac{1}{N^7}$. If any of these two instantiations of the algorithm terminate with a ``FAIL'', we terminate the algorithm for the current phase, and return ``FAIL''.

The running time of the algorithm so far, including the initialization of all data structures, and maintaining data structures $\DSin$ and $\DSout$ throughout the phase, is bounded by:

\[O\left(c_1\cdot (W^0)^2\cdot  2^{c_2\sqrt{\log N}}\cdot (\log N)^{16c_2(r-1)+8c_2+30}\right ).\]

Our algorithm employs a cleanup procedure, that we describe below, once at the beginning of the phase, and once after each query is processed. The purpose of the procedure is to ensure that the following conditions hold:

\begin{properties}{C}
%	\item For every edge $e\in E(G)$, $\ell(e)\le \frac{d'}{2}$; \label{cond: short edges}
%	\item The length of every path in $\qset^{\inn}\cup \qset^{\out}$ is bounded by  $\frac{d}{8\lambda}$;\label{cond: short in and out paths}
		\item The length of every path in $\qset$ is bounded by  $8d''$; \label{cond: short paths in embedding}
	\item The depth of each of the trees $\hat \tau',\hat\tau''$ is at most $\hat d$; and \label{cond: expander small diam}
	\item $U_{\lambda}=\emptyset$ and $U'_{\lambda}=\emptyset$. \label{cond: everyone connected}
\end{properties}

(here, $U_{\lambda}$ and $U'_{\lambda}$ are the sets of disconnected vertices from data structures $\DSout$ and $\DSin$, respectively).
Before we describe the cleanup procedure that ensures all these conditions, we describe our algorithm for responding to queries.

\subsubsection{Responding to Queries}

We assume that, at the time when a query $(x,y)$ arrives, Conditions \ref{cond: short paths in embedding}--\ref{cond: everyone connected} hold. We start by computing a path $P_1$ in graph $H$, connecting $x$ to some vertex $x'\in V(X)$, as follows. Since $U_{\lambda}=\emptyset$ currently holds, there is some layer $0\leq i<\lambda$ in data structure $\DSout$, with $x\in U_i$. If $i=0$, then we let the path $P_1$ contain a single vertex -- vertex $x$ itself. Otherwise, we use the unique path $K(x)\in \qset_i$ to connect $x$ to some vertex $x_1\in U^{<i}$. Let $U_{i_1}$ be the layer in which $x_1$ lies, so $i_1<i$. We say that $K(x)$ is the first \emph{subpath} of $P_1$. 
If $x_1\in V(X)$, then we have completed the construction of the path $P_1$. Otherwise, using a similar process, we connect $x_1$ to some vertex $x_2\in U^{<i_2}$, via the path $K(x_1)$, that becomes the second subpath of $P_1$ and so on. After at most $\lambda-1$ iterations, we are then guaranteed to reach a vertex $x'\in V(X)$, thereby obtaining a path $P_1$ connecting $x$ to $x'$. Since the length of every path in $\qset^{\out}$ is at most $16d'\log^{10}N'$ (from the relaxed condition \ref{prop: routing paths}), while $d'\leq \frac{2d}{(\log N)^{28}}$ and $\lambda\leq O(\log N)$, we get that the length of $P_1$ is bounded by:

\[16d'\lambda\log^{10}N'\leq O\left(\frac{d}{(\log N)^{17}}\right ).  \]

Using a similar algorithm, together with data structure $\DSin$, we compute a path $P_2$ in graph $H$, of length at most $O\left(\frac{d}{(\log N)^{17}}\right )$, connecting some vertex $y'\in V(X)$ to $y$. The time required to compute the paths $P_1$ and $P_2$ is bounded by $O(|E(P_1)|+|E(P_2)|)$.

Next, we use the tree $\hat \tau''$ to compute a path $P'_3$ in $X\setminus F$, connecting $x'$ to the root $v$ of the tree $\hat \tau''$, and we use the tree $\hat \tau'$ to compute a path $P'_4$ in $X\setminus F$ connecting $v$ to $y'$. By concatenating both paths, we obtain a path $P'$ in $X\setminus F$ that connects $x'$ to $y'$, whose length is at most $2\hat d$.

Denote $P'=(e_1,\ldots,e_z)$. We  compute a path $P_3$ connecting $x'$ to $y'$ in graph $H$ by concatenating the corresponding embedding paths $Q(e_1),\ldots,Q(e_z)$ in $\qset$. Recall that, from Condition \ref{cond: short paths in embedding}, the length of each such path is bounded by $8 d''$.
Since $\hat d=\frac{2^{30}\log \hat W}{\alpha_0}\leq O(\log N)$ and
$d''=\Theta\left (d'\cdot  (\log N)^{16}\right )\leq O\left (\frac{d}{(\log N)^{12}}\right )$, as $d'\leq \frac{2d}{(\log N)^{28}}$, we get that  the length of the resulting path $P_3$ is bounded by:

\[2\hat d\cdot 8d''\leq O\left (\frac{d}{(\log N)^{11}}\right ).\]
 
 The time required to compute the path $P_3$ is bounded by $|E(P_3)|$.
 By concatenating the paths $P_1,P_3$ and $P_2$, we obtain a path $P'$ connecting $x$ to $y$ in $H$, whose length is bounded by:
 
\[ O\left(\frac{d}{(\log N)^{17}}\right )+O\left (\frac{d}{(\log N)^{11}}\right )<d, \]
 
since $N$ is sufficiently large. 
The total time spent on computing path $P'$ is $O(|E(P')|)\leq O(d)$.
We then turn $P'$ into a simple $x$-$y$ path of length at most $d$ in graph $H$, in time $O(d)$. Finally, we convert $P'$ into a simple $x$-$y$ path in graph $G$ of length at most $d$ in a straightforward way, in time $O(d)$, and return the path $P$.

The time to respond to a single query is bounded by $O(d)\leq O(n)\leq O(W^0)$, and the time to respond to all queries in a single phase is bounded by $O(W^0\cdot \Delta')\leq O\left(\frac{W^0\cdot \Delta}{(\log N)^{30}}\right )$ (we have used the fact that $\Delta'=\frac{\Delta}{\log^{30}N}$, and Inequality \ref{eq: bound on d}).

%Let $P$ be the path in $G$ that corresponds to the path $P'$
%We return path $P$ in response to the query. Notice that the time spent on computing the path $P$ is bounded by $O(\ell(P))$.
%We then return the path $P$. The total time the algorithm has spent on computing the path $P$ is bounded by $O(|E(P)|)+O\left(\frac{W^0\cdot (\log N)^{28}}{d}\right )$.

\iffalse
Let $P_1,\ldots,P_q$ be all paths that the algorithm returned in response to queries in the current phase. Recall that $q\leq \Delta'\leq \frac{\Delta}{(\log N)^{30}}$.  Then the total time that the algorithm spends on processing all queries is bounded by:

\[\begin{split}
O\left(\sum_{i=1}^q|E(P_i)|\right )+ O\left(\frac{W^0\cdot \Delta'\cdot (\log N)^{28}}{d}\right )&\leq O\left(\sum_{i=1}^q|E(P_i)|\right )+ O\left(\frac{W^0\cdot \Delta}{d\cdot \log^2 N}\right )\\
\end{split}
\] 
 
\mynote{if the trees in the expander are maintained for us we don't need to worry}
$d\leq \frac{n\cdot \eta}{\Delta}$

 $\frac{\Delta\cdot d}{\eta}\leq 2n$
 
% We transform $P'$ into a simple path $P$ in graph $G$ in time $O(d)$, and return path $P$ in response to the query. 
%So far, the time that is required in order to process a single query is bounded by $O(d)+O(N)$.

\fi

\subsubsection{Post-Query Updates}
Recall that, after responding to a query, the algorithm receives a subset $E'\subseteq E(P)$ of special edges on the path $P$ that was returned in response to the query, whose lengths need to be doubled. We can then compute a subset $E''\subseteq E'$ of edges, whose lengths $\ell'(e)$ that are used by the algorithms that maintain data structures $\DSin$ and $\DSout$ need to be doubled, and supply the set $E''$ of edges to these two algorithms, in time $O(|E'|)$.
Next, we consider the edges in $E'$ one by one. For each edge $e\in E'$, we consider every special edge $e'$ on the path $P^*(e)$ representing $e$ in $H$. We split each such edge $e'=(x,y)$, by replacing it with the path $(x,x',y',y)$ in $H$. If any of the endpoints $x,y$ of $e'$ are terminal vertices that lie in $T$, this operation does not affect them. 
Notice that this operation inserts two new special edges into graph $H$: edge $e'_1=(x,x')$ and edge $e'_2=(y,y')$. We initialize the lists $S(e'_1)=S(e'_2)=\emptyset$.

Next, we consider every edge $\hat e\in S(e')$, one by one. For each such edge $\hat e$, we update the embedding path $Q(\hat e)$ by splitting the edge $e'$, and also update the length $\hat \ell(e)$ of the path by increasing it by (additive) $2$. We also insert $\hat e$ into $S(e_1')$ and $S(e_2')$. If the length $\hat \ell(e)$ of $e$ became greater than $8d''$, then we flag this edge, to be taken care of by the Cleanup procedure.

\iffalse
This concludes the algorithm for processing queries. We now bound the total time that is required in order to process all queries during a single phase. Since the number of queries that we are required to process in a single phase is bounded by $\Delta'=\frac{\Delta}{(\log N)^{30}}$, from the above discussion, the time required in order to compute a path in response to each query in the current phase is bounded by:

\[O\left (\Delta'\cdot \left(\frac{W^0\cdot (\log N)^{28}}{d}+d\right )\right )\leq O\left(\frac{\Delta\cdot W^0}{d\log^2N}+\frac{d\cdot \Delta}{(\log N)^{30}}\right )\leq O\left((W^0)^2\right ),\]

since, from the problem definition, $\frac{\Delta\cdot d}{\eta}\leq 2n\leq O(W^0)$
\fi

In order to bound the time required for post-query updates, let $E^*$ be the set of all special edges that ever belonged to graph $H$ over the course of the phase. Since, from \Cref{obs: final length of edges}, $\sum_{e\in E^{\spec}}\ell^*(e)\leq 2\sum_{e\in E^{\spec}}\ell^0(e)+4n\leq O(W^0)$, and since, for each edge $e\in E(G)$, the number of its child edges that ever belonged to $H$ over the course of the phase is asymptotically bounded $\ell^*(e)$, we get that $|E^*|\leq O(W^0)$.

Recall that the congestion caused by the embedding $\qset$ of $X\setminus F$ into $H$ is bounded by $O((\log N)^{13})$. Therefore, for every edge in $E^*$ we spend $O((\log N)^{13})$ time to process it when it is first inserted into $H$. The total time required for all post-query updates is then bounded by:

\[O( W^0\cdot (\log N)^{13}).  \]

The total running time of the algorithm so far, including responding to queries, but excluding the Cleanup Procedure, is bounded by:

\[O\left(c_1\cdot (W^0)^2\cdot  2^{c_2\sqrt{\log N}}\cdot (\log N)^{16c_2(r-1)+8c_2+30}\right )+O\left(\frac{W^0\cdot \Delta}{(\log N)^{30}}\right ).\]

\subsubsection{The Cleanup Procedure}

 The cleanup procedure is executed once at the beginning of the phase, and then after every query, once all edge-splittings that follow the query are processed. The purpose of the procedure is to ensure that Properties \ref{cond: short paths in embedding}--\ref{cond: everyone connected} are satisfied. As part of the procedure, we will sometimes compute $\phi$-sparse strongly well-structured cuts in graph $G$. Whenever this happens, we will call to another procedure, called \processcutt, that will process the resulting cut. We view Procedure \processcutt as part of the cleanup procedure, and we describe it in the following subsection. We note that, if the cleanup procedure ever identifies a $\phi$-sparse strongly well-structured cut $(A.B)$ in $G$, our algorithm does not immediately output the cut, but rather calls to Procedure \processcutt. The procedure may choose to either output the cut $(A,B)$, after which the vertices of the side of the cut with smaller weight and the endpoints of the edges in $E_G(A,B)$ are deleted from $G$, and graphs $H$, $X$ and $X'$ are updated accordingly; or it may instead compute another finalizing cut $(A',B')$ in $H$, and then call to \Cref{claim: finalizing cut} to process it. In the latter case the algorithm does not output the initial cut $(A,B)$, but after cut $(A',B')$ is processed the phase will terminate.
We now turn to describe the cleanup procedure.

% We then perform a number of iterations. At the beginning of every iteration, we assume that we are given some vertex $v\in V(X)$, and two shortest-paths trees $\hat \tau'$ and $\hat \tau''$ in $X\setminus F$ that are both rooted at $v$, with the edges of $\hat \tau'$ directed away from $v$, and the edges of $\hat \tau'$ directed towards $v$. 

The procedure consists of iterations, which are executed as long as at least one of the conditions  \ref{cond: short paths in embedding}--\ref{cond: everyone connected}  is violated.

Assume first that Condition \ref{cond: everyone connected} is violated. In this case we say that the current iteration is of type 1.
Assume w.l.o.g. that $U_{\lambda}\neq \emptyset$; the other case is symmetric. We execute a single iteration of the algorithm for the \connecttocenters problem that maintains data structure $\DSout$. We now consider three cases. If the algorithm produced  Outcome 
\ref{outcome: LCDS}, that is, it obtained a valid \LCDS with $U_{\lambda}=\emptyset$, then the current iteration terminates.
If the algorithm produced a finalizing cut $(A',B')$ in the current graph $H$, then the cut processed like decribed in Section \ref{subsec: connect to terminals}, after which the current phase terminates. Recall that the running time for processing the cut $(A',B')$ is included in our analysis of the running time of the algorithm that maintains data structure $\DSout$. Lastly, if the algorithm for the \connecttocenters problem produces a  a strongly well-structured $\frac{1}{d'}$-sparse cut $(A,B)$ in $G$, then we call to Procedure \processcutt. The procedure will either produce a finalizing cut in $H$, that will subsequently terminate the current phase; or it will output the cut $(A,B)$, following which, if $w(A)\leq w(B)$, then the vertices of $A$, and the endpoints of the edges in $E_G(A,B)$ are deleted from $G$, while otherwise, the vertices of $B$, and the endpoints of the edges in $E_G(A,B)$ are deleted from $G$, and the current iteration terminates.
The total running time of all type-1 iterations of the Cleanup procedure executed over the course of the phase, excluding the running time of Procedure \processcutt, is already included in the running times of the algorithms that maintain data structures $\DSin$ and $\DSout$.

Assume now that Condition \ref{cond: short paths in embedding} is violated. In this case we say that the current iteration is of type 2. Let $\hat e\in E(X)$ be the edge with $\hat \ell(\hat e)>8d''$; recall that all such edges were flagged by the algorithm that performed post-query updates. We then simply delete the edge $\hat e$ from graph $X$. 
We also recompute the shortest path threes $\hat \tau'$ and $\hat \tau''$ in $X\setminus F$ that are rooted at the same vertex $v$, in time $O(|E(X)|)\leq O(|T|)\leq O(W^0)$.
The total running time of type-2 iterations over the course of the whole phase, excluding the time needed for recomputing the trees $\hat \tau'$ and $\hat \tau''$, is asymptotically bounded by the running time of the algorithm that performs post-query updates after every query.
Since at least one edge is deleted from $X$ in each type-2 iteration, and $|E(X)|\leq O(|T|)\leq O(W^0)$, the total time spent on recomputing the trees $\hat \tau'$ and $\hat \tau''$ in all type-2 iterations of the phase is bounded by $O((W^0)^2)$.

Lastly, assume that Condition \ref{cond: expander small diam} is violated. In this case we say that the current iteration is of type 3. Then we can compute, in time $O(1)$, two vertices 
a pair $a,b\in V(X)$ with $\dist_{X\setminus F}(a,b)>\hat d$.
Recall that  $\hat d=\frac{2^{30}\log \hat W}{\alpha_0}$, and the maximum vertex degree in $X$ is bounded by $D=24$. We can now use the algorithm from \Cref{obs: ball growing} to compute a cut $(A,B)$ in $X\setminus F$ of sparsity at most $\phi'=\frac{32D\log \hat W}{\hat d}\leq \frac{\alpha_0}{2^{20}}$. Let $Z\in \set{A,B}$ be the set of vertices of smaller cardinality, breaking ties arbitrarily.
For every vertex $x\in Z$, we delete $x$ from $V(X)$. 
We supply the list of all vertices that are deleted from $X$ to the algorithms for the \connecttocenters problem that are used to maintain data structures $\DSin$ and $\DSout$.
Recall that the running time of the algorithm from \Cref{obs: ball growing} is bounded by $O(D\cdot |V(X)|)\leq O(\hat W)\leq O(W^0)$. 
We also select an arbitrary vertex $v\in V(X)$, and compute  shortest-path trees $\hat \tau'$ and $\hat \tau''$ in $X\setminus F$ rooted at $v$. Note that both trees can be computed in time $O(|E(X)|)\leq O(\hat W)\leq O(W^0)$.
Every time a type-3 iteration occurs, at least one vertex is deleted from $X$. Therefore, the total time that the algorithm spends on type-3 iterations over the course of the entire phase is bounded by $ O((W^0)^2)$.

This completes the algorithm for the Cleanup procedure. From our discussion, the total running time of the algorithm for the \maintaincluster problem so far, excluding Procedure \processcutt, is bounded by:

\[O\left(c_1\cdot (W^0)^2\cdot  2^{c_2\sqrt{\log N}}\cdot (\log N)^{16c_2(r-1)+8c_2+30}\right )+O\left(\frac{W^0\cdot \Delta}{(\log N)^{30}}\right ).\]
 
Next, we describe the algorithm for Procedure \processcutt.

\subsubsection{Procedure \processcutt}

The input to Procedure \processcutt is a strongly well-structured $\phi$-sparse cut $(A,B)$ in graph $G$. Recall that, if our algorithm chooses to output the cut $(A,B)$, then the vertices of the side of the cut that has smaller weight will be deleted from $G$, together with the endpoints of the edges in $E_G(A,B)$. After that, graph $H$ will be updated with a corresponding sequence of vertex- and edge-deletions, and expander graph $X$, together with the shadow expander $X'$ will also be updated, as described in Section \ref{subsubsec: maintainexpander}. If we apply all these updates directly, we may discover that too many vertices have been deleted from graph $X'$. In such a case, we could use the algorithm from \Cref{claim: lots of vertices in X after second type of updates}, in order to compute a finalizing cut $(\tA,\tB)$ in graph $H$. However, $(\tA,\tB)$ is a valid finalizing cut in the original graph $H$, to which the updates following the cut $(A,B)$ have not been applied yet. Therefore, our algorithm will, instead, simulate the updates to graphs $G$, $H$ and $X'$ in order to find out whether the number of vertices that are deleted from $X'$ following these updates is too high. If so, it will apply the algorithm from \Cref{claim: lots of vertices in X after second type of updates} to compute a finalizing cut $(\tA,\tB)$ in graph $H$, followed by the algorithm from \Cref{claim: finalizing cut}, that either correctly establishes  that at least one of the events $\event_2,\event_3$ has happened (in which the current phase terminates with ``FAIL''),  or computes a strongly well-structured $\phi$-sparse cut $(A^*,B^*)$ in the current graph $G$, such that $w(A^*),w(B^*)\geq \frac{\hat W}{\log^{30}N}$, which our algorithm then outputs; in this case, our algorithm does not output the cut $(A,B)$. The vertices of the side of the cut $(A^*,B^*)$ that has smaller weight are then deleted from $G$, together with the endpoints of the edges of $E_G(A^*,B^*)$, and the current phase terminates. If our algorithm establishes that $|X'|$ remains sufficiently high even after we update the graphs $G$, $H$ and $X'$ with the cut $(A,B)$, then it will simply output the cut $(A,B)$, and update the graphs $G,H,X'$ and $X$ accordingly.

We now provide a formal description of Procedure \processcutt. We assume that our algorithm maintains the graph $G_0$, which is the graph $G$ at the beginning of the current phase, and the sequence of strongly well-structured $\phi$-sparse cuts $(A_1,B_1),\ldots,(A_q,B_q)$ that it has output so far over the course of the phase. We also assume that it maintains the invariant that $|V(X')|\geq |T|\cdot \left(1-\frac{8}{\log^4N}\right )$. This invariant clearly holds at the beginning of the phase, when $V(X')=T$.

We now assume that Procedure \processcutt is called with a strongly well-structured $\phi$-sparse cut $(A,B)$ in the current graph $G$. We assume that so far the algorithm  returned $q\geq 0$ well-structured $\phi$-sparse cuts, that we denote by $(A_1,B_1),\ldots,(A_q,B_q)$, in the order in which the algorithm returned them.
We also denote the cut $(A,B)$ by $(A_{q+1},B_{q+1})$, and let $E_{q+1}=E_G(A_{q+1},B_{q+1})$.
Let
$J_{q+1}\subseteq A_{q+1},J'_{q+1}\subseteq B_{q+1}$ the vertices that serve as endpoints of the edges in $E_{q+1}$. As before, if $w(A_{q+1})\geq w(B_{q+1})$, then we denote  $Z_{q+1}=A_{q+1}\setminus J_{q+1}$, and otherwise we denote $Z_{q+1}=B_{q+1}\setminus J'_{q+1}$. We also let $\overline Z_{q+1}=V(G)\setminus Z_{q+1}$ contain the remaining vertices of the current graph $G$. As before, if $Z_{q+1}=A_{q+1}\setminus J_{q+1}$, then we say that cut $(A_{q+1},B_{q+1})$ is of \emph{type 1}, and otherwise we say that it is of \emph{type 2}.
We also denote by $G_q$ and $H_q$ the graphs $G$ and $H$, respectively, at the beginning of Procedure \processcutt, and we denote $z_{q+1}=\min\set{w(A),w(B)}$.

For every edge $e=(x,y)\in E_{q+1}$, we compute the set $\Gamma(e)$ of edges in $H_q$. Recall that the set $\Gamma(e)$ consists of all vertices on the path $P^*(e)$ that represents $e$ in graph $H_q$, as well as all regular edges insident to vertices $x$ and $y$ in $H_q$. It is easy to verify that $|\Gamma(e)|\leq |V(H_q)|\leq O(\hat W)$. Next, we compute the set $\Gamma_{q+1}=\bigcup_{e\in E_{q+1}}\Gamma(e)$ of edges in $\hat H$. 
The sets $\Gamma(e)$ of edges for all $e\in E_{q+1}$, together with the set $\Gamma_{q+1}$ of edges can be computed in time $O(|E_{q+1}|\cdot \hat W)\leq O(\phi z_{q+1}\cdot \hat W)$.

We also compute the partition of the vertices of $H_q$ into sets $Z^*_{q+1}$ and $\notZ^*_{q+1}$, as follows. Set $\notZ^*_{q+1}$ contains all vertices that lie in set $\notZ_{q+1}$. Additionally, for every special edge $e$ of $G$ with both endpoints in $\notZ_{q+1}$, set $\notZ^*_{q+1}$ contains all vertices on path $P^*(e)$. Similarly, set $Z^*_{q+1}$ contains all vertices that lie in set $Z_{q+1}$, and additionally, for every special edge $e$ of $G$ with both endpoints in $Z_{q+1}$, set $Z^*_{q+1}$ contains all vertices on path $P^*(e)$. It is easy to verify that $(Z^*_{q+1},\notZ^*_{q+1})$ is a partition of $V(H_{q+1})$, since cut $(A_{q+1},B_{q+1})$ is strongly well-structured in $G$.  The time required to compute the partition  $(Z^*_{q+1},\notZ^*_{q+1})$  of $V(H_{q+1})$ is bounded by $O(|V(H_{q+1})|)\leq O(\hat W)$.

Next, we compute the set $\hat \Gamma_{q+1}=\bigcup_{e\in E_{q+1}}\hat \Gamma(e)$ of edges in graph $X'\setminus F$, where for an edge $e\in E_{q+1}$, $\hat \Gamma(e)=\bigcup_{e'\in \Gamma(e)}S(e')$. In other words, an edge $\hat e\in E(X')\setminus F$ belongs to the set $\hat \Gamma(e)$ if and only if its corresponding embedding path $Q(\hat e)$ contains at least one edge of $\Gamma(e)$, or, equivalently, edge $\hat e$ belongs to set $S(e')$ of edges for some edge $e'\in \Gamma(e)$. From \Cref{obs: few embedding edges affected}, for every edge $e\in E_{q+1}$,  $|\hat \Gamma(e)|\leq 2\eta''$.
Therefore, $|\hat \Gamma_{q+1}|\leq 2|E_{q+1}|\cdot \eta''\leq 2\phi \cdot z_{q+1}\cdot \eta''\leq O(\phi\cdot z_{q+1}\cdot (\log N)^{13})$. 
%For every edge $e\in E_{q+1}$, the corresponding set $\hat \Gamma(e)$ of edges can be computed in time $O(d'+\eta'')$, as follows. 
Moreover, given an edge $e\in E_{q+1}$, set $\hat \Gamma(e)$ can be computed in time $O(\hat W\cdot \eta'')$. Indeed, recall that, from the proof of \Cref{obs: few embedding edges affected}, an edge $\hat e\in E(X)\setminus F$ may lie in $\hat \Gamma(e)$ in one of the following two cases: (i) an endpoint of $\hat e$ lies in $T\cap P^*(e)$ -- all such edges $\hat e$ can be identified in time $O(|E(P^*(e))|)\leq O(\hat W)$; or (ii) if $e'$ is the first edge on $P^*(e)$ and $\hat e\in S(e')$ -- since $|S(e')|\leq \eta''$, all such edges can be identified in time $O(\eta'')$. Therefore, the time required to compute the set $\hat \Gamma_{q+1}$ of vertices is bounded by $O(|E_{q+1}|\cdot \hat W\cdot \eta'')\leq O(\phi\cdot z_{q+1}\cdot \hat W\cdot \eta'')\leq  O(\phi\cdot z_{q+1}\cdot \hat W\cdot (\log N)^{13})$.

We also compute a partition $(T_{q+1},\notT_{q+1})$ of the vertices of $X'$, where $T_{q+1}=V(X')\cap Z^*_{q+1}$ and $\notT_{q+1}=V(X')\cap \notZ^*_{q+1}$.
Notice that the running time of the procedure so far is bounded by:

\[O\left(\phi z_{q+1}\cdot (\log N)^{13}\cdot \hat W\right )+O(\hat W)\leq 
O\left(\frac{\hat W \cdot z_{q+1}\cdot (\log N)^{13}}{d'}\right )+\leq O(\hat W \cdot z_{q+1} ),\]

since $d'\geq \frac{d}{(\log N)^{28}}$, and $d\geq (\log N)^{64}$ from the definition of the problem.

Next, we consider two cases. The first case happens if $|T_{q+1}|\geq |T|\cdot \left(1-\frac{8}{\log^4N}\right )$. 
In this case we output the cut $(A_{q+1},B_{q+1})$, and then update our data structures accordingly: we delete the vertices of $\notZ_{q+1}$ from $G$, and we delete the vertices of $\notZ^*_{q+1}$ from $H$. We also delete the edges of $\hat \Gamma_{q+1}$ and the vertices of $\notT_{q+1}$ from $X$ and from $X'$.
Finally, we update the algorithms that maintain data structures $\DSin$ and $\DSout$ with the vertices that were just deleted from $X$.
 In this case, the running time of the procedure remains at most $O(\hat W \cdot z_{q+1} )$.
Notice that, following this procedure, at least $z_{q+1}$ vertices are deleted from $H$. We \emph{charge} the running time of the procedure to the vertices that are deleted from $H$, where the charge to each vertex is $O(\hat W)$.

We now consider the case where  $|T_{q+1}|< |T|\cdot \left(1-\frac{8}{\log^4N}\right )$. In this case, we apply the algorithm from \Cref{claim: lots of vertices in X after second type of updates} to graph $G_0$ and the sequence $(A_1,B_1),\ldots,(A_{q+1},B_{q+1})$ of cuts in $G$. Recall that the running time of the algorithm is $O(\hat W^2)$, and it returns a  finalizing cut $(\tA,\tB)$ in the current graph $H$.
Next, we apply the algorithm from \Cref{claim: finalizing cut}
to graph $G_0$, the sequence $(A_1,B_1),\ldots,(A_{q},B_{q})$ of cuts in $G$, and the finalizing cut  $(\tA,\tB)$ in the current graph $H$. Recall that the running time of the algorithm from  \Cref{claim: finalizing cut} is $O(\hat W^2)$. If the algorithm from  \Cref{claim: finalizing cut} estabishes that at least one of the events $\event_2,\event_3$ has happened, then we terminate the current phase and return ``FAIL''. Otherwise, we obtain a strongly well-structured $\phi$-sparse cut $(A^*,B^*)$ in the current graph $G$, such that $w(A^*),w(B^*)\geq \frac{\hat W}{\log^{30}N}$. 
Let $Z^*\in\set{A^*,B^*}$ be the side of smaller weight; if $w(A^*)=w(B^*)$, then we let $Z^*=B^*$. We delete the vertices of $Z^*$ from $G$, together with the endpoints of the edges in $E_G(A^*,B^*)$. Every vertex $v$ that was deleted from $G$ is also deleted from $H$, and for every special edge $e$ that was deleted from $G$, all vertices and edges on path $P^*(e)$ are deleted from $H$. As the result, at least $\frac{\hat W}{\log^{30}N}$
vertices are deleted from $H$ and the current phase terminates.

Overall, the last execution of Procedure \processcutt in a phase takes time $O(\hat W^2)$. If $q^*$ is the total number of times that the procedure is executed during the phase, then all but the last execution of the procedure take time at most:

\[\sum_{q'=1}^{q^*-1}O(z_{q'}\cdot \hat W))\leq O((W^0)^2).\]

(Here, we used the fact that, from \Cref{obs: final length of edges}, the total number of vertices that ever belonged to $H$ is bounded by $O(W^0)$.)
The total running time of all executions of Procedure \processcutt over the course of a phase is then bounded by $O(\hat W^2)$, and the total running time of a single phase of the algorithm for the \maintaincluster problem remains at most:

\[O\left(c_1\cdot (W^0)^2\cdot  2^{c_2\sqrt{\log N}}\cdot (\log N)^{16c_2(r-1)+8c_2+30}\right )+O\left(\frac{W^0\cdot \Delta}{(\log N)^{30}}\right ).\]

\subsubsection{Maintaining the Additional Data Structures and the Final Accounting}

\paragraph{Additional Data Structures.}
It remains to show an algorithm that maintains, for every edge $e\in E(H)$, the set $S(e)\subseteq E(X)\setminus F$ of edges. It is straightforward to augment our algorithm to maintain this data structure without increasing its asymptotic running time.

Once the algorithm from \Cref{lem: phase 1 embed expander} computes 
the expander $X$ and the embedding $\qset$ of $X\setminus F$ into $H$, we initialize the sets $S(e)$ for all edges $e\in E(H)$ in a straightforward way: we start by setting $S(e)=\emptyset$ for all $e\in E(H)$. We then consider every edge $\hat e\in E(X)\setminus F$ one by one. When edge $\hat e$ is considered, we process every edge $e\in E(Q(\hat e))$, and add $\hat e$ to the corresponding set $S(e)$ of edges. Whenever some edge $\hat e\in E(X)\setminus F$ is deleted from $X$, we similarly process every edge $e\in E(Q(\hat e))$, and we delete $\hat e$ from set $S(e)$. Lastly, when an edge $e\in E(H)$ undergoes splitting, and new special edges $e'$ and $e''$ are inserted into graph $H$ as the result of this split, for every edge $\hat e\in S(e)$, we update the path $Q(\hat e)$ with the corresponding edge split, and we insert $\hat e$ into the lists $S(e')$ and $S(e'')$. The updates to the lists $\set{S(e)}_{e\in E(H)}$ due to edge splitting were already taken into account in the running time analysis of the algorithm that processes post-query updates. The additional time that is required in order to maintain the lists  $\set{S(e)}_{e\in E(H)}$ is asymptotically bounded by the time required to initialize these lists at the beginning of the phase, which, in turn, is asymptotically bounded by the running time of the algorithm from \Cref{lem: phase 1 embed expander}.

\paragraph{Failure probability and running time analysis.}
Recall that our algorithm may return ``FAIL'' in the following cases: (i) Event $\event$ has happened -- the probability for that is bounded by $1/16$; or (ii) the algorithm from  \Cref{lem: phase 1 embed expander} for constructing the expander $X$ returns ``FAIL'': the probability for that is bounded by $\frac{1}{N^7}$; or (iii) the algorithm for the \connecttocenters problem, that is used in order to maintain data structures $\DSin$ and $\DSout$ returns ``FAIL'': this can happen with probability at most  $\frac{1}{N^7}$ for each of the two instantiations of the algorithm. Therefore, overall, the probability that the algorithm for a single phase terminates with a ``FAIL'' is bounded by $\frac{1}{2}$. From our analysis in \Cref{subsec: phases}, the number of phases in the algorithm is  bounded by $O\left((\log N)^{33}\right )$. Since the running time of every phase is bounded by:

\[O\left(c_1\cdot (W^0)^2\cdot  2^{c_2\sqrt{\log N}}\cdot (\log N)^{16c_2(r-1)+8c_2+30}\right )+O\left(\frac{W^0\cdot \Delta}{(\log N)^{30}}\right ),\]

using the fact that $c_2$ is a large enough constant and $N$ is sufficiently large, we get that the total running time of the algorithm is bounded by:

\[\begin{split}
&O\left(c_1\cdot (W^0)^2\cdot  2^{c_2\sqrt{\log N}}\cdot (\log N)^{16c_2(r-1)+8c_2+63}\right )+O\left(W^0\cdot \Delta\cdot \log^3 N\right )\\
&\quad\quad\quad\quad\quad\quad\quad\quad \leq c_1\cdot  (W^0)^2\cdot 2^{c_2\sqrt{\log N}}\cdot (\log N)^{16c_2r}+c_1\cdot W^0\cdot \Delta\cdot \log^4 N.
\end{split}\]

This completes the proof of \Cref{thm: from routeandcut to maintaincluster}.

%% file: appx-expander-tools.tex
\section{Proofs Omitted from \Cref{sec: expander tools}}
\label{sec: appx expander tools}

\subsection{Proof of \Cref{obs: subdivided cut sparsity}}
\label{subsec: proof of subdivided cut sparsity}

	Consider first any edge $e=(u,v)\in E(G)$ with $u\in A$ and $v\in B$. We claim that $e$ must be a special edge. Indeed, assume otherwise, that is, $e$ is a regular edge. Then $u\in A'$ and $v\in B'$ must hold. Moreover, edge $e=(u,v)$ lies in graph $G^+$, so $e\in E_{G^+}(A',B')$. Since cut $(A',B')$ in $G^+$ is weakly well-structured, this is impossible, a contradiction. We conclude that every edge $e\in E_G(A,B)$ must be a special edge, and so cut $(A,B)$ is a weakly well-structured cut in $G$.
	
Next, we prove that $w(A)\geq\frac{|A'|}{8}$; the proof that $w(B)\geq  \frac{|B'|}{8}$ is symmetric.
We partition the vertices of $A'$ into three subsets. The first subset is $A'_1=A$. If vertex $v\in A'$ does not lie in $A$, then there  is some special edge $e=(x,y)\in E(G)$, such that $v$ is an inner vertex on path $P(e)$. If any of the endpoints $x,y$ of $e$ lie in $A$, then we add vertex $v$ to the second set $A'_2$, and otherwise it is added to the third set, $A'_3$. 
	
	Clearly, $w(A)\geq |A|\geq |A_1'|$, and it is easy to verify that $w(A)\geq \frac{|A_2'|} 2$ holds as well, since the weight of every vertex $x$ in $G$ is equal to the length of the unique special edge $e$ that is incident to $x$, which, in turn, is equal to half the number of vertices on the corresponding path $P(e)$.
	Therefore,  $w(A)\geq \frac{|A'_1|+|A'_2|}{4}$ must hold. Next, we show that $|A'_3|\leq \frac{|A'|}2$.
	
	Indeed, let $\tilde E$ be the collection of all special edges $e=(x,y)$ in $G$ with $x,y\in B$, such that at least one inner vertex on the corresponding path $P(e)$ belongs to $A'$. Clearly, for each such edge $e$, at least one edge $e'\in P(e)$ must belong to $E_{G^+}(A',B')$. Therefore, $|\tilde E|\leq E_{G^+}(A',B')\leq \phi\cdot |A'|$. Since we have assumed that graph $G$ does not contain edges of length greater than $\frac{1}{4\phi}$, we get that, for every edge $e\in \tilde E$, the number of vertices lying on the corresponding path $P(e)$ is at most $\frac{1}{2\phi}$. Therefore, $|A'_3|\leq \frac{1}{2\phi}\cdot |\tilde E|\leq \frac{|A'|}2$.  We conclude that $|A_3'|\leq \frac{|A'|}{2}$, and so $w(A)\geq \frac{|A'_1|+|A'_2|}{4}\geq \frac{|A'|}{8}$. Using a similar reasoning, we get that $w(B)\geq \frac{|B'|}8$.
	
Recall that $\Phi_{G^+}(A',B')=\frac{|E_{G^+}(A',B')|}{\min\set{w(A'),w(B')}}=\frac{|E_{G^+}(A',B')|}{\min\set{|A'|,|B'|}}$, since the weight of every vertex in $G^+$ is $1$. Consider now any edge $e=(u,v)\in E_G(A,B)$. Since $(A,B)$ is a well-structured cut in $G$, $e$ must be a special edge, and moreover, $u\in A'$, $v\in B'$ must hold. Therefore, at least one edge on the path $P(e)$ must lie in $E_{G^+}(A',B')$. We conclude that $|E_G(A,B)|\leq |E_{G^+}(A',B')|$.
 Altogether, we get that:

	\[ \Phi_G(A,B)=\frac{|E_G(A,B)|}{\min\set{w(A),w(B)}}\leq \frac{8|E_{G^+}(A',B')|}{\min\set{|A'|,|B'|}}=8\Phi_{G^+}(A',B'). \]

\subsection{Proof of \Cref{obs: weakly to strongly well str}}
\label{subsec: proof of weakly to strongly well str}

	We start with $(A',B')=(A,B)$, and then consider every vertex $x\in A$	one by one. If there is a special edge $e=(y,x)\in E_G(B',A')$, then if $w(A')\leq w(B')$ currently holds, we move $y$ from $B'$ to $A'$, and otherwise we move $x$ from $A'$ to $B'$. Notice that, since graph $G$ is well-structured, if a regular edge $(a,b)$ lies in $G$, then $a\in L$ and $b\in R$ must hold. Additionally, since $(y,x)$ is a special edge, $y\in R$ and $x\in L$ must hold. Therefore, the only edge entering $x$ is the special edge $(y,x)$, and all edges leaving $x$ are regular edges. By moving $x$ from $A'$ to $B'$ we therefore cannot introduce any new edges to $E_G(A',B')$, or any new special edges to $E_G(B',A')$. Similarly, the only edge leaving $y$ is the special edge $(y,x)$, and all edges entering $y$ are regular edges. Therefore, by moving $y$ from $B'$ to $A'$, we cannot introduce any new edges to $E_G(A',B')$, or any new special edges to $E_G(B',A')$.
	
	Once every vertex of $A$ is processed, we obtain the final cut $(A',B')$, that is strongly well-structured. Note that $|E_G(A',B')|\leq |E_G(A,B)|\leq \frac{\phi}2\cdot\min\set{w(A),w(B)}$. Let $z=\min\set{w(A),w(B)}$, and let recall that $w(V(G))=2\sum_{e\in E}\ell(e)=2W$. Notice that $z\leq W$ must hold.
	
	If no vertices were ever removed from $A'$, then $w(A')=w(A)$ holds. Otherwise, consider the last time when a vertex was removed from $A'$. We denote the removed vertex by $v$. Then, just before $v$ was moved from $A'$ to $B'$, $w(A')\geq w(B')$ held. Therefore, at the end of the algorithm, $w(A')\geq W-w(v)\geq \frac{W}{2}$, since the weight $w(v)$ of vertex $v$ is the length of the unique special edge incident to $v$, which must be bounded by $\frac{W}{2}$. We conclude that, at the end of the algorithm, $w(A')\geq \min\set{w(A),\frac{W}{2}}\geq \frac{w(A)}{4}$ holds. Using similar reasoning, $w(B')\geq \min\set{w(B),\frac{W}{2}}\geq \frac{w(B)}{4}$. 
	Moreover, $\min\set{w(A'),w(B')}\geq \min\set{w(A),w(B),\frac{W}{2}}\geq \half \min\set{w(A),w(B)}$.
	Altogether, $|E_G(A',B')|\leq \frac{\phi}2\cdot\min\set{w(A),w(B)}\leq  \phi\cdot\min\set{w(A'),w(B')}$.
	
The running time of the above algorithm is $O(\vol(A))$.	If $\vol(B)< \vol(A)$, then we can execute the same procedure by examining the vertices of $B$ instead of $A$, in time $O(\vol(B))$.

%% file: appx-routeandcut-from-maintaincluster.tex
\section{Proofs Omitted from \Cref{sec: SSSP alg}}
\label{sec: appx: proofs for routeandcut}

\subsection{Proof of \Cref{obs: valid instance}}
\label{subsec: appx: proof valid instance}
%\begin{proof}
From Invariant \ref{inv: perfect well structured}, graph $H[X^0]$ is a perfect well-structured graph. It is immediate to verify that the lengths $\ell(e)$ of edges $e\in E(H)$ define a proper assignment of lengths to the edges of $H[X^0]$. Recall that we have denoted by $W^0(X)$ the total weight of all vertices in $H[X^0]$ when $X$ is added to $\xset$, and, from \Cref{obs: bound edge lengths}, we are guaranteed that $W^0(X)\leq N'$, where $N'=\max\set{N,8n\cdot \log^5n}$ is the parameter that we have defined. 
From the definition of the \stSP problem, we also get that $1\leq \eta\leq \Delta$ and $1\leq \Delta \leq n-|B^0|\leq N\leq N'$. 
From our assumption that $N$ is greater than a sufficiently large constant, with $N>2^{2^{c_2}}$, the same holds for $N'$.

Since $X$ is a non-leaf cluster,$|X^0| > \frac{\Delta\cdot (\log N)^{128}}{\eta}$ must hold. Recall that $d_X=\min\set{\frac{|X^0|\cdot \eta}{\Delta},d}$. From Inequality \ref{eq: two bounds on d}, $d\geq 2^{2\sqrt{\log N}}$. Therefore, if $d_X=d$, then $d_X\geq 2^{2\sqrt{\log N}}\geq 
2^{\sqrt{\log N'}}\geq (\log N')^{64}$ (from Inequality \ref{eq: bound on N'}, and since $N'\geq N$ is large enough). Otherwise:

\[d_X= \frac{|X^0|\cdot \eta}{\Delta}\geq  (\log N)^{128}\geq (\log N')^{64},\] 

since $\log N'\leq 2\log N\leq \log^2N$
from Inequality \ref{eq: bound on log n'}. In either case, $d_X\geq (\log N')^{64}$ holds.

Lastly, since $d_X=\min\set{\frac{|X^0|\cdot \eta}{\Delta},d}$, we get that $d_X\leq \frac{|X^0|\cdot \eta}{\Delta}$, and so  $\frac{\Delta\cdot d_X}{\eta}\leq |X^0|$ holds. This establishes that $(H[X^0],\Delta,\eta,d_X,N')$ is a valid input instance for the \maintaincluster problem. 

In order to show that this instance is $r$-restricted, we need to show that $r\leq \ceil {\sqrt{\log N'}}$ and that $d_X\leq 2^{r\cdot\sqrt{\log N'}}$ holds. 
The former follows immediately since, from the definition of the $(r+1)$-restricted \stSP problem,  $r+1\leq \ceil{\sqrt{\log N}}\leq \ceil{\sqrt{\log N'}}$. The latter also immediately follows since $d_X\leq d$, and since $d\leq 2^{r\cdot\sqrt{\log N}}\leq 2^{r\cdot\sqrt{\log N'}}$ from Inequality \ref{eq: two bounds on d}. 
%\end{proof}
%==============================
%================================

\subsection{Proof of \Cref{obs: both graphs are perfect well-structured}}
\label{subsec: appx: proof of both graphs are perfect well-str}
%\begin{proof}
From Invariant \ref{inv: perfect well structured}, graph $H[X]$ is a perfect well-structured graph, so every vertex of $X$ is incident to exactly one special edge in $H[X]$. Clearly, both $H[Y']$ and $H[Z']$ are well-structured graphs. It remains to show that each of them is a strongly well-structured graph.

We start by proving that $H[Y']$ is a strongly well-structured graph. It is enough to show that, for every vertex $v\in Y'$, if $e$ is the unique special edge that is incident to $v$ in $G$, then the other endpoint $u$ of $e$ also lies in $Y'$. Indeed, assume otherwise. Since $H[X]$ is a perfect well-structured graph, $u\in X$ must hold. It is impossible that $u\in J$, since in this case $u\in R$, and the unique special edge incident to $u$ lies in the cut $E_H(Y,Z)$. Therefore, $u\in Z$ must hold. Since $v\not\in J$, edge $e$ must lie in $E_H(Z,Y)$. But from the definition of a strongly well-structured cut, $e$ must be a regular edge, a contradiction. 

In order to prove that $H[Z']$ is a strongly well-structured graph, it is enough to show that, for every vertex $v\in Z'$, if $e$ is the unique special edge that is incident to $v$ in $G$, then the other endpoint $u$ of $e$ also lies in $Z'$. Assume otherwise. As before, since $H[X]$ is a perfect well-structured graph, $u\in X$ must hold. It is impossible that $u\in J'$, since in this case $u\in L$, and the unique special edge incident to $u$ lies in the cut $E_H(Y,Z)$. Therefore, $u\in Y$ must hold. Since $v\not\in J'$, edge $e$ must lie in $E_H(Z,Y)$. But from the definition of a strongly well-structured cut, $e$ must be a regular edge, a contradiction.
%\end{proof}
%=========================

\subsection{Proof of \Cref{obs: new right to left}}
\label{subsec: appx: proof of obs new right to left}
%\begin{proof} 
Let $e=(u,v)$ be a new right-to-left edge that was added to the \ATO as the result of the current transformation. Assume first that $v\in J'$. Since $J'\subseteq L$, edge $e$ must be the unique special edge entering $v$, and, from the definition of the set $J'$ of vertices, $e\in E_{H[X]}(Y,Z)=E_H(Y,Z)$ must hold. Therefore, $u\in J$ must hold. %We will show that the edges of $E_{G[X]}(Y,Z)$, that connect vertices of $J$ to vertices of $J'$, are the only right-to-left edges that were added during the current update.
Assume next that $v\in Z'$. Since edge $e$ is a right-to-left edge, it must be the case that $u\in Y'\cup J=Y$. However, there can be no edges connecting $Y$ to $Z'$ in $G[X]$, from the definition of the sets $J'$ and $Z'$ of vertices. Assume now that $v\in Y'$. Since edge $e$ is a right-to-left edge, $u\in J$ must hold. However, as observed already, $J\subseteq R$ must hold, and for every vertex $u\in J$, the only edge that leaves $u$ is the unique special edge incident to $u$, whose other endpoint is in $J'$. Therefore, it is impossible that $u\in J$ and $v\in Y'$ holds for an edge $(u,v)$. Lastly, we need to consider the case where $v\in J$. Since $(u,v)$ is a right-to-left edge, it must be the case that $u\in J$. But since $J\subseteq R$, it is impossible that an edge of $G$ has both endpoints in $J$. 
%\end{proof}
%===========================
%===========================
%============================

\subsection{Proof of \Cref{obs: valid spec instance}}
\label{subsec: appx: proof valid spec instance}
%\begin{proof}
As observed already, graph $H[U^0]$ is a perfect well-structured graph. It is immediate to verify that the initial lengths $\ell(e)$ of edges $e\in E(H)$ define a proper assignment of lengths to the edges of $H[U^0]$. Moreover, initially, the length of every special edge of $G$ is set to $1$. 
Recall that $|U^0|>2(n-|B^0|)\cdot (\log N)^{128}$, and that $|U^0\setminus \beta|\leq n-|B^0|$ must hold. Therefore, $|\beta|\geq \frac{99|U^0|}{100}$ holds, as required.
From our assumption that $N$ is greater than a sufficiently large constant, so is $N'$.

Recall that we have denoted by $W^0(U)$ the total initial weight of all vertices in $H[U^0]$. It is immediate to verify that the initial weight of every vertex in $H[U^0]$ is $1$, and so $W^0(U)\leq n\leq N\leq N'$. 
From the definition of the \stSP problem, we also get that $1\leq \eta\leq \Delta$ and $1\leq \Delta \leq n-|B^0|\leq \frac{|\beta|}{(\log N)^{128}}$ (the latter inequality follows since we have established that $|\beta|\geq  \frac{99|U^0|}{100}$ and $|U^0|\geq 2(n-|B^0|)\cdot (\log N)^{128}$).

Recall that $d_U= \max\set{(\log N')^{64},\left (|U^0\setminus B^0|+\Delta\right )\frac{\eta}{4\Delta\cdot 2^{\sqrt{\log N}}}}$, so $d_U\geq (\log N')^{64}$ holds.
%
%
%from Inequalities \ref{eq: two bounds on d} and \ref{eq: bound on N'}, $d\geq %2^{2\sqrt{\log N}}\geq 
%2^{\sqrt{\log N'}}$.
Lastly, we need to establish that $\frac{\Delta\cdot d_U}{\eta}\leq |U^0\setminus B^0|+\Delta\cdot (\log N')^{64}$.
Indeed, since \newline $d_U= \max\set{(\log N')^{64},\left (|U^0\setminus B^0|+\Delta\right )\frac{\eta}{2\Delta\cdot 2^{\sqrt{\log N}}}}$, we get that:

\[
%\begin{split}
\frac{\Delta\cdot d_U}{\eta}=\max\set{\frac{\Delta\cdot (\log N')^{64}}{\eta},\frac{|U^0\setminus B^0|}{4\cdot 2^{\sqrt{\log N}}}+\frac{\Delta}{4\cdot 2^{\sqrt{\log N}}}} \leq |U^0\setminus B^0|+\Delta\cdot(\log N')^{64}.
%\frac{|U^0|}{4\cdot 2^{5\sqrt{\log N}}}\leq |U^0\setminus \beta|=|U^0|-|\beta|.
%\end{split}
\]

In order to show that this instance is $r$-restricted, we need to show that $r\leq \ceil {\sqrt{\log N'}}$ and that $d_U\leq 2^{r\cdot\sqrt{\log N'}}$ holds. 
The former follows immediately since, from the definition of the $(r+1)$-restricted \stSP problem,  $r+1\leq \ceil{\sqrt{\log N}}\leq \ceil{\sqrt{\log N'}}$. The latter also immediately follows since:

\[
\begin{split}
d_U&=\max\set{(\log N')^{64},\left (|U^0\setminus B^0|+\Delta\right )\frac{\eta}{4\Delta\cdot 2^{\sqrt{\log N}}}}\\
&\leq \max\set{(\log N')^{64},\frac{(n-|B^0|)\cdot\eta}{2\Delta\cdot 2^{\sqrt{\log N}}}}\\
&\leq \max\set{(\log N')^{64},2d}\\
&\le 2^{r\sqrt{\log N'}},
\end{split}
\]

since $d= \frac{(n-|B^0|)\cdot\eta}{4\Delta\cdot 2^{\sqrt{\log N}}}$ and
 $d\leq 2^{r\cdot\sqrt{\log N}-1}\leq 2^{r\sqrt{\log N'}-1}$ from Inequality \ref{eq: two bounds on d}. 
%\end{proof}
%==============================
%================================

\subsection{Proof of \Cref{obs: bound final weights in U}}
\label{subsec: appx: proof of bound on final weights in U}

Recall that, at the beginning of the algorithm, for every edge $e\in E^{\spec}\setminus E^*$, $\ell(e)=1$, and for every vertex $v\in \tilde U$, $w(v)=1$ holds.  
For every edge $e\in (E^{\spec}\setminus E^*)\cap H[U^0]$, if both endpoints of $e$ lie in $U$ at the end of the algorithm, then we set $\tilde \ell(e)$ to be the length of $e$ at the end of the algorithm, and otherwise we let $\tilde \ell(e)$ be the length of $e$ just before its endpoints are deleted from $U$. 
Let $E'\subseteq (E^{\spec}\setminus E^*)\cap H[U^0]$ be the set of all edges $e\in (E^{\spec}\setminus E^*)\cap H[U^0]$ with $\tilde \ell(e)>1$. 
Then $\sum_{v\in \tilde U}\tilde w(v)\leq |\tilde U|+2\sum_{e\in E'}\tilde \ell(e)$.

Let $q$ be the total number of queries to algorithm $\aset'(U)$, so that $q\leq \Delta$, and let $\qset'=\set{P_1,\ldots,P_q}$ be the paths returned by Algorithm $\aset'(U)$ in response to the queries, where the paths are indexed in the order in which they were returned. For all $1\leq i\le q$, let $C_i$ denote the length of path $P_i$ when it was added to $\qset$. Clearly, $C_i\leq d_U$, and $\sum_{i=1}^qC_i\leq q\cdot d_U\leq \Delta\cdot d_U$.

Consider now some edge $e\in E'$, and recall that $\ell(e)$ was doubled at least once over the course of Algorithm $\aset'(U)$. Let $\tau$ be the last time when the length of $e$ was doubled. Then, prior to time $\tau$, there were at least $\eta$ iterations $i$, during which the path $P_i$ that was added to $\qset'$ contained $e$, and the length of $e$ during the corresponding iteration was $\frac{\tilde \ell(e)}{2}$. In other words, edge $e$ contributes at least $\frac{\tilde \ell(e)\eta}{2}$ to $\sum_{i=1}^qC_i$. Altogether, we get that $\sum_{i=1}^qC_i\geq \sum_{e\in E'}\frac{\tilde \ell(e)\eta}{2}$. Therefore:

\[\sum_{e\in E'}\tilde \ell(e)\leq \frac{2\sum_{i=1}^qC_i}{\eta}\leq \frac{2\Delta \cdot d_U}{\eta}\leq |U^0\setminus B^0|+2\Delta\cdot (\log N')^{64},\]

since
$d_U= \max\set{(\log N')^{64},\left (|U^0\setminus B^0|+\Delta\right )\frac{\eta}{4\Delta\cdot 2^{\sqrt{\log N}}}}$.

Altogether, we get that:

\[\sum_{v\in \tilde U}\tilde w(v)\leq |\tilde U|+2\sum_{e\in E'}\tilde \ell(e)\leq 4|U^0\setminus B^0|+8\Delta\cdot (\log N')^{64},\]

since the set $\tilde U$ contains all vertices of $U^0\setminus B^0$, and additionally all vertices of $B^0\cap U^0$ that were deleted from $B^0$ over the course of the algorithm, whose number is bounded by $\Delta$.

%===================================================
%===================================================
%===================================================

%===================================================
%===================================================
%===================================================
\subsection{Proof of \Cref{claim: close descendant}}
\label{subsec: proof of close descendant}

%===================================================
%===================================================
%===================================================
Consider any cluster $X\in \tilde X'_2$. Observe first that at most one child-cluster of $X$ may be its close descendant. Indeed, assume otherwise, and let $X_1,X_2$ be two child-clusters of $X$ that are its close descendants, and assume that $X_1$ was added to $\xset$ before $X_2$. Then $|X^0_1|,|X^0_2|>|X^0|/2$ must hold, and so $X^0_1\cap X^0_2\neq \emptyset$. But this is impossible since, after cluster $X_1$ was added to $\xset$, all vertices of $X_1^0$ were deleted from $X$. Recall that, if $X_1$ is a child-cluster of $X$, and it is a close descendant of $X$, then, when $X_1$ was created, $\sum_{v\in X_1}w(v)\leq \sum_{v\in X}w(v)/2$ held, following which the vertices of $X_1^0$ were deleted from $X$.
We start with the following claim that bounds the number of close descendants of a cluster $X$, whose initial weights are not much lower than that of $X$.

\begin{claim}\label{claim: short subsequence}
	Consider a sequence $(X_1,X_2,\ldots,X_z)$ of distinct clusters in $\tilde \xset_2'$, such that for all $1\leq i<z$, $X_{i+1}$ is a child cluster of $X_i$, clusters $X_2,\ldots,X_z$ are close descendants of $X_1$, and for all $2\leq i\leq z$, $W^0(X_i)\geq W^0(X_1)/2$. Then $z\leq 16$ must hold.
\end{claim}
\begin{proof}
	For all $1\leq i\leq z$, let $\tau_i$ be the time when cluster $X_i$ was added to $\xset$, and for all $1<i\leq r$, let $X'_{i-1}$ be the set of all vertices that lied in $X_{i-1}$ just before time $\tau_i$, excluding the vertices of $X_i$. 
	For convenience, we also denote by $X'_z=X_z$. For all $1\leq i<z$, we denote by $W_i$ be the total weight of all vertices in set $X'_i$ at time $\tau_{i+1}$, and we denote by $W_z=W^0(X_z)$. Note that the sets $X'_1,\ldots,X'_z$ of vertices are mutually disjoint.
	
	Consider now an index $1<i\leq z$. Since $|X_i^0|\geq |X_1^0|/2$ must hold, we get that $|X_i^0|\geq |X_{i-1}^0|/2$, and so, from our discussion, at time $\tau_i$, the total weight of the vertices of $X_{i}$ was less than or equal to the total weight of the vertices of $X'_{i-1}$. In other words, for all $1< i\leq z$, $W_{i-1}\geq W^0(X_{i})\geq W^0(X_1)/2$. 
	From our definitions, $W_z=W^0(X_z)\geq W^0(X_1)/2$ as well. Therefore:
	
	\[\sum_{i=1}^zW_i\geq \frac{W^0(X_1)\cdot z}{2}.\]
	%Since $W^0(X_{i+1}
	%	 )\ge W^0(X_1)$, we get that $W_i\geq W^0(X_1)/2$. 

	In the following observation,  we show that $\sum_{i=1}^{z}W_i\leq 8W^0(X_1)$ must hold, and so $z\leq 16$, as required. The following observation then completes the proof of \Cref{claim: short subsequence}.
	
	\begin{observation}\label{obs: small partial total weight}
		$\sum_{i=1}^{z}W_i\leq 8W^0(X_1)$.
	\end{observation}
	\begin{proof}
		Recall that, from our definition, the sets $\set{X'_1,\ldots,X'_{z}}$ of vertices are mutually disjoint. For convenience, denote $\tau_{z+1}=\tau_z$, and let $Y=\bigcup_{i=1}^zX'_i$, so $Y\subseteq X^0_1$. Recall that, for all $1\leq i< z$, $W_i$ denotes the weight of the vertices in $X_i'$ at time $\tau_{i+1}$. In order to prove the observation, we will observe how the weights of the vertices of $Y$ evolve during the time interval $(\tau_1,\tau_z]$, except that, for all $1\leq i\leq z$, we  will ``freeze'' the weights of the vertices in $X'_i$ at time $\tau_{i+1}$, so they are not allowed to grow anymore.
		
		Specifically, for every vertex $v\in Y$, we denote by $\tw(v)$ the weight of $v$ at time $\tau_1$. For all $1\leq i\leq z$, for every vertex $v\in X_i$, we denote by $\tw'(v)$ the weight of $v$ at time $\tau_{i+1}$. Clearly, 
		$\sum_{v\in Y}\tw(v)\leq W^0(X_1)$, and $\sum_{i=1}^{z}W_i=\sum_{v\in Y}\tw'(v)$.
		It now remains to bound $\sum_{v\in Y}(\tw'(v)-\tw(v))$. 
		Let $\tilde E'$ be the set of all special edges that lie in $E(H[X^0_1])$.
		For every special edge $e\in \tilde E'$, we denote by $\tilde \ell(e)$ the length of $e$ at time $\tau_1$, and recall that $W^0(X_1)=2\sum_{e\in \tilde E'}\tilde \ell(e)$.

		Recall that for each cluster $X\in \tilde \xset_2'$, we have defined a parameter $d_{X}=\min\set{\frac{|X^0|\cdot \eta}{\Delta},d}$. Since, for all $1<i\leq z$, cluster $X_i$ is a descendant of $X_1$, we get that $|X^0_i|\leq |X^0|$, and so $d_{X_i}\leq d_{X_1}$.

		Consider now the following process. Initially, every vertex $v\in Y$ is assigned weight $\tw(v)$, and every special edge $e\in \tilde E'$ is assigned length $\tilde \ell(e)$. Then we obeserve the iterations of the algorithm that occured during the time interval $[\tau_1,\tau_z)$ in the order in which they occur. Let $j$ be any such iteration. Let $1\leq i\leq z$ be the integer for which iteration $j$ occurs during the time interval $[\tau_{i},\tau_{i+1})$. We then consider the cluster $X_i$, and denote by $P_j$ the path that Algorithm $\aset(X_i)$ returned in response to a \shortpath query during iteration $j$. Denote by $C_j$ the length of $P_j$ at the time when it was returned, and recall that $C_j\leq d_{X_i}\leq d_{X_1}$ must hold. We also denote by $\tilde E_j\subseteq E(P_j)$ the set of special edges of $H[X_i]$, whose lengths were doubled as the result of this iteration. For every vertex $v\in X_i\cap Y$, that serves as an endpoint of an edge of $\tilde E_j$, we double the weight $\tw(v)$. For every edge $e\in \tilde E_j$, we also double its length $\tilde \ell(e)$.
		
		It is easy to verify that, at the end of this process, the final weight of every vertex $v\in Y$ is precisely $\tw'(v)$. 
		For every special edge $e\in \tilde E'$, let $\tilde \ell'(e)$ be the length of $e$ at the end of this process. Then $\sum_{v\in Y}\tw'(v)\leq 2\sum_{e\in \tilde E'}\tilde \ell'(e)$.
		Let $r$ be the total number of iterations of the algorithm that occur during the time interval $[\tau_1,\tau_z)$. Consider now any special edge $e\in \tilde E'$. From the definition of our algorithm, if $\tau_1\leq \tau<\tau'\leq \tau_z$ are any two
		times at which the length of $e$ was doubled, then edge $e$ participated in at least $\eta$ of the paths of $\set{P_1,\ldots,P_r}$ that were returned by the above process during the time interval $[\tau,\tau')$.

		Recall that we have established that $\sum_{e\in \tilde E'}\tilde \ell(e)=\frac{W^0(X_1)}{2}$. Let $\tilde E''\subseteq \tilde E'$ be the set of all edges $e$ for which $\tilde \ell'(e)>2\tilde \ell(e)$. Consider now any such edge $e\in \tilde E''$. Then the length of $e$ must have been doubled at least twice during the process that we described above. Let $\tau<\tau'$ be the last two times  during the process when the length of $e$ was doubled. Then,  at least $\eta$ of the paths 
		of $\set{P_1,\ldots,P_r}$ that were returned during the time interval  $[\tau,\tau')$ contained $e$, and for each such path $P_j$, edge $e$ contributed at least $\tilde \ell'(e)/2$ to the length of the path.
		Therefore, every edge $e\in \tilde E''$ contributes at least $\frac{\tilde \ell'(e)\cdot \eta}{2}$ to $\sum_{j=1}^rC_j$, and so $\sum_{j=1}^rC_j\geq \sum_{e\in \tilde E''}\frac{\tilde \ell'(e)\cdot \eta}{2}$. On the other hand:

		\[\sum_{j=1}^rC_j\leq r\cdot d_{X_1}\leq \Delta\cdot d_{X_1}\leq \Delta\cdot \min\set{\frac{|X_1^0|\cdot \eta}{\Delta},d} \leq |X_1^0|\cdot \eta\leq W^0(X_1)\cdot \eta.
		\]
		
		Altogether, we get that 	$\sum_{e\in \tilde E''}\tilde \ell'(e)\leq \frac{2}{\eta}\cdot \sum_{j=1}^rC_j\leq 2W^0(X_1)$, and so:
		
		\[\sum_{i=1}^zW_i=\sum_{v\in Y}\tw'(v)\leq 2\sum_{e\in \tilde E'}\tilde \ell'(e)\leq 4\sum_{e\in \tilde E'}\tilde \ell(e)+2\sum_{e\in \tilde E''}\tilde \ell'(e)\leq 8W^0(X_1).
		\]
	\end{proof} 
\end{proof}

Recall that, for every cluster $X\in \tilde \xset'_2$, $\frac{\Delta\cdot (\log N)^{128}}{\eta}\leq W^0(X)\leq 8n\log^5n$ holds (see the definition of leaf clusters  (\Cref{def: leaf})  and Inequality \ref{eq: final weight of a cluster}). We now prove that, for all $\floor{\log\left(\frac{\Delta\cdot(\log N)^{128}}{\eta}\right )}<j\leq 4\ceil{\log n}$, if $X\in \tilde \xset'_2$ is a cluster with $2^j\leq W^0(X)<2^{j+1}$, then $X$ has at most $16j$ close descendants.

The proof is by induction on $j$. The base case is when $j=\floor{\log\left(\frac{\Delta\cdot (\log N)^{128}}{\eta}\right )}$. Recall that, if $W^0(X')<\frac{\Delta\cdot (\log N)^{128}}{\eta}$, then $X'$ is a leaf cluster. Therefore,  if $X\in \tilde \xset'_2$ is a cluster with $2^j\leq W^0(X)<2^{j+1}$, and $X'\in \tilde \xset_2$ is its close descendant, then $W^0(X')\geq \frac{\Delta\cdot (\log N)^{128}}{\eta}\geq \frac{W^0(X)} 2$ must hold. 
For a cluster $X\in \tilde \xset_2'$ with $2^{j}\leq W^0(X)<2^{j+1}$, 
let $\dset(X)\subseteq \tilde X'_2$ be the collection of the close descendants of $X$, and denote $|\dset(X)|=z$. Then we can define an ordering $(X=X_1,X_2,\ldots,X_z)$ of the clusters in $\dset(X)$, such that, for all $1\leq i<z$, $X_{i+1}$ is a child-cluster of $X_i$. From \Cref{claim: short subsequence}, we are then guaranteed that $z\leq 16$.

For the induction step, assume that the claim holds for some integer $j\geq \floor{\log\left(\frac{\Delta\cdot (\log N)^{128}}{\eta}\right )}$, and consider a cluster $X\in \tilde \xset'_2$ with $2^{j+1}\leq W^0(X)<2^{j+2}$. Let $\dset(X)\subseteq \tilde X'_2$ be the set of all close descendants of $X$, and denote $|\dset(X)|=z$. As before, we can define an ordering $(X=X_1,X_2,\ldots,X_z)$ of the clusters in $\dset(X)$, such that, for all $1\leq i<r$, $X_{i+1}$ is a child-cluster of $X_i$. Let $i'> 1$ be the smallest index with $W^0(X_{i'})<W^0(X)/2$. Then for all $1\leq i< i'$, $W^0(X_i)\geq W^0(X)/2$, and so, from \Cref{claim: short subsequence}, $i'\leq 17$. Moreover, clusters $X_{i'+1},\ldots,X_z$ are close descendants of cluster $X_{i'}$, and so, from the induction hypothesis, $z-i'+1\leq 16j$. Altogether, we get that $z= (z-i'+1)+i'-1\leq 16(j+1)$. Since, for every cluster $X\in \tilde \xset_2'$, $W^0(X)\leq 8n\log^5n<2^{4\ceil{\log n}}$, we get that each such cluster has at most $64\ceil{\log n}$ close descendants.
%\end{proof}

%% file: appx-proofs-for-connectivity.tex
\section{Proofs Omitted from \Cref{sec: layered connectivity DS}}

\subsection{Proof of \Cref{claim: routeandcut instance}}
\label{subsec: proof of valid routeandcut}

Observe first that, from our definition, graph $H'$ is a well-structured graph.

Recall that $N\geq \frac{W}{32}$, while $|V(H')|\leq |V(H)|+\beta\cdot M\leq 2W\log^3N$.  Since $N'=64N\log^3N$, $N'\geq |V(H')|$ holds. %The second parameter that we use is $\Delta'=k$, that replaces $\Delta$, and the third parameter is $\highlightf{ \eta'=2^{10}\cdot d'\cdot \log^6N}$, that replaces $\eta$.  
Since we have assumed that  $k\geq 2^{20}\cdot d'\cdot \log^6N$, we get that $1\le \eta'\leq \Delta'$, as required.  Since  $U_{\lambda}\subseteq A'$, while $B'$ contains all vertices of $U_0$ and their copies, it is easy to verify that $|A'|,|B'|\geq k=\Delta'$.  Therefore, for any pair $(A',B')$ of disjoint subsets of $V(H')$ with $U_{\lambda}\subseteq A$, and with all vertices of $U_0$ and their copies lying in $B'$,  $(H',A',B',\Delta', \eta',N')$ defines a valid instance of the \routeandcut problem.
We now verify that this instance is $r$-restricted.

In order to do so, it is enough to show that $1\leq r\leq \ceil{\sqrt{\log N'}}$ and $\frac{(|V(H')|-|B'|)\cdot \eta'}{\Delta'}\leq 2^{r\cdot \sqrt{\log N'}}$ holds. Since $1\leq r\leq \sqrt{\log N}$, we get that  $1\leq r\leq \ceil{\sqrt{\log N'}}$ also holds.

Recall that $|V(H')\setminus B'|\leq 4\lambda k$. Therefore:

\[\frac{(|V(H')|-|B'|)\cdot \eta'}{\Delta'}\leq \frac{4\lambda k\cdot  2^{20}\cdot d'\cdot \log^6N}{k}\leq  d'\cdot \log^8N\leq 2^{r\sqrt{\log N}}\leq 2^{r\cdot \sqrt{\log N'}},
\]

since $\lambda=\ceil{\log(64N)}+1$ and $d'\leq\frac{2^{r\sqrt{\log N}}}{\log^8N}$ from the definition of Procedure \reconnect.

\subsection{Proof of \Cref{obs: lots of terminals on one side}}
\label{subsec: proof of lots of terminals on one side}

We first prove that $|\tilde Y|\leq O(k\log N)$. Recall that set $\tilde Y$ may contain all vertices of $H'$, except for the vertices of $B'_j$. It is easy to verify that $|B'_j|\geq |B_j|-k$, since $|\rset'_j|<k$. Recall that we have shown that $|V(H')\setminus B_j|\leq O(\lambda\cdot k)\leq O(k\log N)$. Therefore, $|\tilde Y|\leq |V(H')\setminus B_j|+|B_j\setminus B'_j|\leq O(k\log N)$.

In order to prove the remaining two assertions, we partition the set $\beta$ of vertices into three subsets. The first set, $\beta_1$, contains all vertices $x\in \beta$ that lie in $\tilde Y'$. The second set, $\beta_2$, contains all vertices $x\in \beta$ that lie in $\tilde Y$, such that every copy of $x$ lies in $\tilde Y$ as well. The third set, $\beta_3$, contains all remaining vertices of $\beta$.

Notice that, for every vertex $x\in \beta_3$, some edge that is incident to one of the copies of $x$ must lie in $E_{H'}(\tilde Y,\tilde Y')$, and so $|\beta_3|\leq  |E_{H'}(\tilde Y,\tilde Y')|\leq \frac{k}{d'\log^6N}\leq \frac{|\beta|}{\log^4N}$,
since $|\beta|\geq \frac{W}{d'\cdot \log^2N}\geq\frac{k}{d'\cdot\log^2N}$.

Notice also that, for every vertex $x\in \beta_2$, all $M$ copies of $x$ (including $x$ itself) must serve as endpoints of the paths in $\rset'_j$, so $|\beta_2|\leq \frac{|\rset'_j|}{M}\leq \frac{k\cdot |\beta|}{W\log^3N}\leq \min\set{\frac{k}{\log^3N},\frac{|\beta|}{\log^3N}}$, since $M=\frac{W\cdot \log^3N}{|\beta|}$ and $k,|\beta|\leq |V(H)|\leq W$. 

Overall, we get that: 

\[|\tilde Y'\cap \beta|=|\beta_1|\geq |\beta|-|\beta_2|-|\beta_3|\geq |\beta|\cdot\left(1-\frac{1}{\log^4N}-\frac{1}{\log^3N}\right )\geq |\beta|\cdot\left(1-\frac{2}{\log^3N}\right ).\]

Additionally,

\[ |\tilde Y'\cap \beta|\geq |\beta|-|\beta_2|-|\beta_3|\geq |\beta|-\frac{k}{\log^3N}-\frac{k}{d'\log^6N}\geq |\beta|-\frac{2k}{\log^3N}. \]

Lastly, since $|A'_j|\geq \frac{k}{64\lambda}\geq \frac{k}{2^9\log N}$, as $\lambda=\ceil{\log(64N)}+1\leq 8\log N$, while $|\beta_2|\leq  \frac{k}{\log^3N}$ 
and $|\beta_3|\leq \frac{k}{d'\log^6N}$, we get that $|A'_j|\geq (|\beta_2|+|\beta_3|)\cdot \log N= |\tilde Y\cap \beta|\cdot \log N$.

\subsection{Proof of \Cref{claim: bound lambda and deleted from beta}}
\label{subsec: proof of bound on lambda}
%\begin{proof}
Let $\beta'_1$ be the set of vertices that were deleted from $\beta$ directly, and let $\beta'_2=\beta'\setminus\beta'_1$. Since at most one vertex may be deleted from $\beta$ directly in every iteration, $|\beta'_1|\leq \Delta$ must hold.

Let $E'$ denote the set of all special edges whose lengths have increased over the course of the algorithm. 
Consider any such edge $e=(u,v)\in E'$, and let $\tau$ be the first time when the length of $e$ was increased. Then at time $\tau$, neither of the vertices $u,v$ may lie in $\beta$. Therefore, either both of these vertices lied in $V(G^0)\setminus\beta^0$; or at least one of these vertices lies in $\beta'_1$ (since, once a vertex is added to $\beta_2'$, it is deleted from $G$). We conclude that $|E'|\leq |V(G^0)|-|\beta^0\setminus \beta'_1|\leq n-|\beta^0|+\Delta$.
For every edge $e\in E'$, let $\ell^0(e)=1$ be its initial length.
Clearly, $\sum_{e\in E'}\ell^0(e)\leq n-|\beta^0|+|\beta'_1|\leq n-|\beta^0|+\Delta$.

Denote by $\set{P_1,\ldots,P_q}$ the set of all paths that the algorithm returned in response to the queries. For all $1\leq i\leq q$, let $C_i$ be the length of path $P_i$ at the time it was returned, so $C_i\leq d$, and $\sum_{i=1}^qC_i\leq d\cdot q\leq d\cdot \Delta$. 

Let $\hat E\subseteq E'$ be the set of edges whose length increased more than once. For every edge $e\in \hat E$, let $\ell^*(e)$ be the final length of the edge $e$, after its last increase, and let $\tau'(e)<\tau''(e)$ be the times of the last two increases in the length of $e$. Then during the time interval $(\tau'(e),\tau''(e)]$, at least $\eta$ paths that were returned in response to the queries contained $e$, and edge $e$ contributed at least $\frac{\ell^*(e)}{2}$ to the length $C_i$ of each such path $P_i$. In other words, every edge $e\in \hat E$ contributes at least $\frac{\eta\cdot \ell^*(e)}{2}$ to $\sum_{i=1}^qC_i$. Therefore:

\[ \sum_{i=1}^qC_i\geq \frac{\eta}{2}\sum_{e\in \hat E}\ell^*(e),\]

and:

\[ \sum_{e\in \hat E}\ell^*(e)\leq \frac{2}{\eta}\cdot \sum_{i=1}^qC_i\leq \frac{2d\Delta}{\eta}.\]

Altogether, we get that:

\[
\begin{split}
\Lambda&\leq \sum_{e\in E'}2\ell^0(e)+\sum_{e\in \hat E}\ell^*(e)\\
&\leq 2n-2|\beta^0|+2\Delta+\frac{2d\Delta}{\eta}\\
%&\leq  2n-2|\beta^0|+2\Delta+n-|\beta^0|+2\Delta\cdot (\log N)^{64}\\
&\leq 4n-4|\beta^0|+3\Delta\cdot (\log N)^{64}.
\end{split} \]

(we have used the fact that, from the problem definition, $\frac{\Delta\cdot d}{\eta}\leq |V(G)|-|\beta^0|+\Delta\cdot (\log N)^{64}$).	Since, from the problem definition, $\Delta\leq \frac{|\beta^0|}{(\log N)^{128}}$, we get that $\Lambda\leq 4n-4|\beta^0|+3\Delta\cdot (\log N)^{64}\leq 4n$.

Next, we bound $|\beta'_2|$. Denote $\hat \beta=\beta^0\setminus \beta'_1$, and recall that $|\hat \beta|\geq |\beta^0|-\Delta$. For every vertex $v\in V(G^0)$, let $w^0(v)=1$ denote the initial weight of $v$, and let $w^*(v)$ denote the weight of $v$ at the end of the algorithm (if $v\in V(G)$ holds), or just before $v$ is deleted from $G$ (otherwise). We also denote by $\tilde W=\sum_{v\in V(G^0)\setminus \hat \beta}w^*(v)$.

Notice that, at the beginning of the algorithm:

\[\sum_{v\in V(G^0)\setminus \hat \beta}w^0(v)=n-|\hat \beta|\leq n-|\beta^0|+\Delta.\]

Since the total increase in all edge lengths is $\Lambda\leq 4n-4|\beta^0|+5\Delta\cdot (\log N)^{64}$, we get that:

\[
\begin{split}
\tilde W\leq \sum_{v\in V(G^0)\setminus \hat \beta}w^0(v)+2\Lambda\leq 9n-9|\beta^0|+11\Delta \cdot (\log N)^{64}.
\end{split}
\]

Consider now any strongly well-structured cut $(X,Y)$ in $G$ of sparsity $\Phi_G(X,Y)\leq \frac{(\log N)^{64}}{d}$, such that $w(X)\geq 1.8|\beta\cap X|$ holds, that the algorithm produces. Let $J'\subseteq Y$ denote the set of vertices that serve as endpoints of the edges of $E_G(X,Y)$, and recall that $|J'|\leq |E_G(X,Y)|\leq \frac{w(X)\cdot (\log N)^{64}}{d}\leq w(X)$, since $d\geq (\log N)^{64}$ from the problem definition. Note that $\hat \beta\cap X\subseteq \beta\cap X$, and so $w(X)\geq  1.8|\beta\cap X|\ge 1.8|\hat \beta\cap X|$ holds. Since the weights of all vertices in $\hat \beta$ are $1$, we get that:

\[w(X\setminus \hat \beta)=w(X)-|X\cap \hat \beta|\geq w(X)\left(1-\frac{1}{1.8}\right )\geq 0.44 w(X).\]

The number of vertices that were indirectly deleted due to the cut is bounded by:

\[|X\cap \hat \beta|+|J'|\leq 2w(X)\leq 5w(X\setminus\hat \beta)\leq 5\sum_{v\in X\setminus\hat \beta}w^*(v). \]

We assign to every vertex $v\in X\setminus \hat \beta$ a \emph{charge} of $5w^*(v)$ units, so that the total charge to all vertices of $X\setminus\hat \beta$ is greater than the number of vertices that are deleted from $\beta$ in the current iteration. Notice that the vertices of $X\setminus \hat \beta$ are deleted from $G$, and so they will never be charged again. We conclude that:

\[|\beta'_2|\leq 5\hat W\leq 45n-45|\beta^0|+55\Delta  \cdot (\log N)^{64}. \]

Altogether, we get that $|\beta'|\leq |\beta'_1|+|\beta'_2|\leq  45n-45|\beta^0|+56\Delta  \cdot (\log N)^{64}$.

%\end{proof}

\subsection{Proof of \Cref{cor: small weight}}
\label{subsec: proof of small weight}

First, observe that: 

\[\sum_{v\in V(G)\setminus \beta}w(v)\leq |V(G)|-|\beta|+2\Lambda\leq |V(G)|-|\beta|+8(n-|\beta^0|)+10\Delta(\log N)^{64},\]

from \Cref{claim: bound lambda and deleted from beta}.
Notice that, if  $v\in V(G)\setminus \beta$, then either $v\in V(G^0)\setminus \beta^0$; or $v$ belonged to $\beta$, but was directly removed from $\beta$. Since the number of vertices that are directly removed from $\beta$ is bounded by $\Delta$, and since, from the problem definition, $\Delta\leq \frac{|\beta^0|}{(\log N)^{128}}$, we get that:

\[\sum_{v\in V(G)\setminus \beta}w(v)\leq 9(n-|\beta^0|)+11\Delta(\log N)^{64}\leq 9(n-|\beta^0|)+\frac{|\beta^0|}{1000}. \]

 Lastly, since $|\beta^0|\geq \frac{99n}{100}$ from the problem definition, we get that $(n-|\beta^0|)\leq \frac{n}{100}\leq \frac{|\beta^0|}{99}$.
Using the fact that $|\beta|\geq \frac{|\beta^0|}{2}$ holds throughout the algorithm, we now get that:

\[\sum_{v\in V(G)\setminus \beta}w(v)\leq \frac{9|\beta^0|}{99}+\frac{|\beta^0|}{1000}\leq 0.1|\beta^0|\leq 0.2|\beta|.\]

\subsection{Proof of \Cref{obs: smaller side}}
\label{subsec: proof of smaller side}

First, from \Cref{cor: small weight}, $w(A\setminus \beta)\leq \sum_{v\in V(G)\setminus \beta}w(v)\leq 0.2|\beta|$.

On the other hand, $w(A\setminus \beta)= w(A)-w(A\cap \beta)=w(A)-|A\cap \beta|\geq w(A)\left(1-\frac{1}{1.8}\right )\geq 0.44 w(A)$. Therefore, $w(A)\leq 2.27w(A\setminus \beta)\leq 0.46|\beta|$. In particular, $|\beta\cap A|\leq w(A)\leq 0.46|\beta|$. but then:

\[w(B)\geq |B\cap \beta|\geq |\beta|-|A\cap \beta|\geq 0.54|\beta|>w(A).\]

%=================================
%=================================
%=================================
%=================================
%=================================

%% file: appx-finalizing-cut.tex
\section{Proofs Omitted from \Cref{sec: alg for maintaincluster from routeandcut}}

\subsection{Proof of \Cref{obs: final length of edges}}
\label{subsec: proof of final length of edges}
%\begin{proof}
Let $\pset$ be the set of all paths that the algorithm returned in response to queries. 
For every path $P\in \pset$, we denote by $\ell'(P)$ the length of the path $P$ at the time when it was returned, so $\ell'(P)\leq d$. Since $|\pset|\leq \Delta$, we get that $\sum_{P\in \pset}\ell'(P)\leq d\cdot \Delta$.

Let $E'\subseteq E^{\spec}$ be the set of all edges $e$ with $\ell^*(e)>2\ell^0(e)$. Consider any such edge $e\in E'$, and denote $\ell^*(e)=2^i$. Let $\tau$ be the first time when the length of $e$ was set to $2^{i-1}$, and let $\tau'$ be the time when the length of $e$ was set to $2^i$. Then, in time interval $(\tau,\tau']$, at least $\eta$ paths were added to $\pset$ that contained $e$, and, for each such path $P$, edge $e$ contributed at least $2^{i-1}=\frac{\ell^*(e)} 2$ to $\ell'(P)$.
Therefore:

\[\sum_{P\in \pset}\ell'(P)\geq \sum_{e\in E'}\frac{\ell^*(e)\cdot \eta}{2}. \]

Since $\sum_{P\in \pset}\ell'(P)\leq d\cdot \Delta$, we get that $\sum_{e\in E'}\ell^*(e)\leq \frac{2d\Delta}{\eta}$.
Altogether, we get that:

\[ \sum_{e\in E^{\spec}}\ell^*(e)\leq 2\sum_{e\in E^{\spec}}\ell^0(e)+\sum_{e\in E'}\ell^*(e)\leq 2\sum_{e\in E^{\spec}}\ell^0(e)+\frac{2d\Delta}{\eta} \leq 2\sum_{e\in E^{\spec}}\ell^0(e)+2n, \]

since $\frac{\Delta d}{\eta}\leq n$ from the problem definition. 
%\end{proof}
%=======================

\subsection{Proof of \Cref{claim: compute intermediate cut}}
\label{sec: proof of intermediate finalizing cut}

Recall that $(A_1,B_1),\ldots,(A_q,B_q)$ is the sequence of all stronlgy well-structured $\phi$-sparse cuts that our algorithm produced during the current phase so far, indexed in the order in  which the algorithm produced them. For all $1\leq i\leq q$, denote by $\tau_i$ the time when cut $(A_i,B_i)$ was computed, by $E_i=E_G(A_i,B_i)$, and by $J_i\subseteq A_i,J'_i\subseteq B_i$ the vertices that serve as endpoints of the edges in $E_i$ (with respect to $G\attime[\tau_i]$ -- graph $G$ at time $\tau_i$). If, at time $\tau_i$, $w(A_i)> w(B_i)$, then we denote $Z_i=A_i\setminus J_i$ and $\overline Z_i=B_i\setminus J'_i$; otherwise, we denote $Z_i=B_i\setminus J_i'$ and $\overline Z_i=A_i\setminus J_i$. % We also denote by $Z_i$ the larger of the sets $A_i\setminus J_i,B_i\setminus J'_i$, and we let $\overline Z_i=V(G\attime[\tau_i])\setminus Z_i$. For convenience, 
In the former case, we say that cut $(A_i,B_i)$ is of \emph{type 1}, and in the latter case it is of \emph{type 2}.

We denote by $G_0$ the graph $G$ at the beginning of the phase, and, for all $1\leq i\leq q$, we denote by $G_i$ the graph $G$ that is obtained after cut $(A_i,B_i)$ has been processed. Therefore, for all $1\leq i\leq q$, graph $G_i$ is obtained from graph $G_{i-1}$ by deleting the vertices of $\overline Z_i\cup J_i\cup J'_i$ from it, or, equivalently, $G_{i}=G_{i-1}[Z_i]$. The current graph $G$ is then denoted by $G_q$, and we denote the current graph $H$ by $H_q$. We also let $E'=\bigcup_{j=1}^qE_j$.
%
%For all $1\leq i\leq q$, we let $E_j=E_G(A_j,B_j)$, where we consider the graph $G$ at the time when cut $(A_j,B_j)$ was returned by the algorithm. We denote the current graph $G$ by $G_q$, and 

Since the total number of vertices deleted so far from $H$ during the current phase is less than $\frac{\hat W}{\log^{30}N}$, and since, for all $1\leq j\leq q$, $|E_j|\leq \phi\cdot \min\set{w(A_j),w(B_j)}$, we conclude that: 

\begin{equation}\label{eq: bounding E'}
|E'|\leq\sum_{j=1}^q|E_j|\leq \frac{\phi\cdot \hat W}{\log^{30}N}=\frac{4\hat W}{d'\cdot \log^{30}N}\leq \frac{8|T|}{\log^{30}N},
\end{equation} 

since $|T|\geq \frac{\hat W}{2d'}$, as Event $\event_1$ did not happen.

Our first step is to compute a partition of the vertices of $G_0$ into three subsets, $S_1,S_2,S_3$, such that $S_2=V(G_q)$, and the partition has several other useful properties, summarized in the following claim.

\begin{claim}\label{claim: compute the sequence}
	There is a deterministic algorithm, that computes a partition $(S_1,S_2,S_3)$ of the vertices of $G_0$, such that the following hold:
	
	\begin{itemize}
		\item $S_2=V(G_q)$;
		\item if $e=(x,y)$ is an edge of $G_0$ connecting a vertex $x\in S_j$ to a vertex $y\in S_i$ with $1\leq j<i\leq 3$, then $e$ is a special edge, and moreover, $e\in E'$ holds; and
		
		\item if $e=(x,y)$ is an edge of $G_0$ connecting a vertex $x\in S_j$ to a vertex $y\in S_i$ with $1\leq i<j\leq 3$, then $e$ is a regular edge.
	\end{itemize}
	The running time of the algorithm is $O((\hat W)^2)$.
\end{claim}

\begin{proof}
	We start by constructing a sequence $\sset=(Y_1,\ldots,Y_{q+1})$ of subsets of vertices of $G_0$, over the course of $q$ iterations. 
	Specifically, for all $1\leq i\leq q$, in iteration $i$ we will construct a sequence $\sset_i$
	of $i+1$ subsets of vertices of $G_0$. We will also identify a single set $S_i^*\in \sset_i$ as the \emph{distingiushed} vertex set of $\sset_i$, and we will define a subset $J_i^*\subseteq S_i^*$ of vertices of $S_i^*$ called \emph{boundary} vertices. We require that the following properties hold.
	
	\begin{properties}{R}
		\item The sets of vertices in $\sset_i$ define a partition of $V(G_0)$; \label{prop 1: partition}
		\item The set $J^*_i$ of boundary vertices contains all vertices of $S^*_i$ that serve as endpoints of the edges of $\bigcup_{i'=1}^{i}E_{i'}$ and only them;\label{prop: boundary vertices incident to cut edges}
		\item $V(G_i)=S_i^*\setminus J_i^*$;\label{prop: distinguished is the graph}
		\item For every boundary vertex $v\in J_i^*$, every edge incident to $v$ in $G_0[S^*_i]$ is a regular edge;  \label{prop: boundary incident to regular}
			\item If $S,S'\in \sset_i$, with $S$ appearing before $S'$ in the sequence, and $e$ is an edge of $G_0$ connecting a vertex of $S'$ to a vertex  of $S$, then $e$ is a regular edge of $G_0$; and \label{prop: forward edges}
			
		\item If $S,S'\in \sset_i$, with $S$ appearing before $S'$ in the sequence, and $e$ is an edge of $G_0$ connecting a vertex of $S$ to a vertex  of $S'$, then $e\in \bigcup_{i'=1}^iE_{i'}$. \label{prop: backward edges}

	\end{properties}

	In the first iteration, we consider the cut $(A_1,B_1)$, and we let the initial sequence be $\sset_1=(A_1,B_1)$. Notice that every edge connecting a vertex of $A_1$ to a vertex of $B_1$ is an edge of $E_1$, and, since cut $(A_1,B_1)$ is strongly well-structured, every edge connecting a vertex of $B_1$ to a vertex of $A_1$ is a regular edge.
	This immediately establishes Properties \ref{prop 1: partition}, \ref{prop: forward edges} and \ref{prop: backward edges} for $\sset_1$.
	
	If $(A_1,B_1)$ is a type-1 cut, that is, $w(A_1)>w(B_1)$, then we designate $A_1$ as the distinguished set $S^*_1$ of vertices for $\sset_1$, and we let $J_1$ be the corresponding set $J^*_1$ of boundary vertices. 
	Alternatively, if $(A_1,B_1)$ is a type-2 cut, then we designate $B_1$ as the distinguished set $S^*_1$ of vertices for $\sset_1$, and we let $J'_1$ be the corresponding set $J^*_1$ of boundary vertices. 
	It is immediate to verify that, with this definition, Properties \ref{prop: boundary vertices incident to cut edges}--\ref{prop: boundary incident to regular} hold.

	Assume now that for some $1\leq i<q$, we have computed the sequence $\sset_i=(S_1,S_2,\ldots,S_{i+1})$, for which Properties \ref{prop 1: partition} -- \ref{prop: backward edges} hold. Consider the distinguished set $S^*_i$ of $\sset_i$ that we denote for simplicity by $S^*$, and its corresponding set $J^*$ of boundary vertices that we denote by $J^*$. Recall that, from Requirement \ref{prop: distinguished is the graph}, $V(G_i)=S^*\setminus J^*$. Consider now the cut $(A_{i+1},B_{i+1})$ in $G_i$. We extend this cut to a partition $(A',B')$ of $S^*$ as follows. Initially, we set $(A',B')=(A_{i+1},B_{i+1})$. Next, we consider the boundary vertices $v\in J^*$ one by one. Recall that, from Property \ref{prop: boundary incident to regular}, only regular edges may be incident to $v$ in $G_0[S^*]$. Therefore, if $v\in L$, then it has no incoming edges in $G_0[S^*]$; in this case we add it to $B'$. Similarly, if $v\in R$, then it has no outgoing edges in $G_0[S^*]$; in this case we add it to $A'$. Once all vertices of $J^*$ are processed, we obtain the partition $(A',B')$ of $S^*$. Notice that $E_{G_{0}}(A',B')=E_{i+1}$, and all edges of $E_{G_0}(B',A')$ are regular edges. In order to obtain sequence $\sset_{i+1}$, we replace the set $S^*$ in $\sset_i$ with the sets $A',B'$, in this order. It is immediate to verify that Properties \ref{prop 1: partition}, \ref{prop: forward edges} and \ref{prop: backward edges}  continue to hold for $\sset_{i+1}$.
	
	If cut $(A_{i+1},B_{i+1})$ is a type-1 cut, that is, $w(A_{i+1})>w(B_{i+1})$, then we designate $A'$ as the distinguished set $S^*_{i+1}$ of vertices for $\sset_{i+1}$, and we let $J_{i+1}\cup (A'\cap J^*)$ be the corresponding set $J^*_{i+1}$ of boundary vertices. 
	Alternatively, if $(A_{i+1},B_{i+1})$ is a type-2 cut, then we designate $B'$ as the distinguished set $S^*_{i+1}$ of vertices for $\sset_{i+1}$, and we let $J'_{i+1}\cup (B'\cap J^*)$ be the corresponding set $J^*_{i+1}$ of boundary vertices. 
	It is immediate to verify that, with this definition, Properties \ref{prop: boundary vertices incident to cut edges}--\ref{prop: boundary incident to regular} continue to hold.
	The algorithm consists of $q\leq O(\hat W)$ iterations, and each iteration can be executed in time $O(\hat W)$. Therefore, the running time of the algorithm is $O(\hat W^2)$.

	We are now ready to define the partition $(S_1,S_2,S_3)$ of the vertices of $G_0$. Consider the final sequence $\yset_q=(Y_1,\ldots,Y_{q+1})$, and assume that its distinguished set is $Y_j$. We denote by $J^*$ the set of the boundary vertices of $Y_j$. Initially, we set $S_1=\bigcup_{j'=1}^{j-1}Y_{j'}$, $S_2=Y_j$, and $S_3=\bigcup_{j'=j+1}^{q}Y_{j'}$. Clearly, $(S_1,S_2,S_3)$ is a partition of $V(G_0)$, and, from properties \ref{prop: forward edges} and \ref{prop: backward edges}, for every pair $1\leq i<i'\leq 3$ of indices, for every edge $e=(x,y)$ with $x\in S_i,y\in S_{i'}$, $e$ must be a special edge that lies in $E'$, while for every edge $e'=(x',y')$ with  $x'\in S_{i'},y'\in S_{i}$, edge $e'$ is a regular edge. Additionally, from Property \ref{prop: distinguished is the graph}, $S_2=V(G_q)\cup J^*$. In our last step we will move every vertex of $J^*$ to either $S_1$ or $S_3$, in order to ensure that $S_2=V(G_q)$ holds.
	
	For convenience, if $e=(x,y)$ is an edge of $G_0$, with $x\in S_i$ and $y\in S_{i'}$, then we say that $e$ is a \emph{left-to-right} edge if $i<i'$, and that it is a \emph{right-to-left} edge if $i'<i$. Notice that the statement of the claim requires that all left-to-right edges are special edges that lie in $E'$, and all right-to-left edges are regular edges.
	
	We consider every vertex $v\in J^*$ one by one. 
	Recall that, from Property \ref{prop: boundary incident to regular}, all edges incident to $v$ in $G_0[Y_j]$ are regular edges. Assume first that $v\in L$. Let $e=(u,v)$ be the unique special edge incident to $v$. Then $u\in S_1$ must hold, and $e$ is the only edge entering $v$. We move $v$ from $S_2$ to $S_3$. Notice that, if $e'$ is a  regular edge incident to $v$, then $e'=(v,u')$ must hold for some vertex $u'$, and so this move does not create any new left-to-right edges, while it may only introduce new right-to-left edges that are regualr edges.
	
	Assume now that $v\in R$. Let $e=(v,u)$ be the unique special edge incident to $v$. Then we are guaranteed that $u\in S_3$ must hold. Edge $e$ is the only edge that leaves $v$ in $G_0$, and all other edges incident to $v$ are incoming edges, which must also be regular edges. We move vertex $v$ from $S_2$ to $S_1$. As before, this move does not introduce any new left-to-right edges, and the only new right-to-left edges that it may introduce are regular edges that are incident to $v$.
	
	Once every vertex of $J^*$ is processed we obtain the partition $(S_1,S_2,S_3)$ with the required properties. It is easy to verify that the running time of the algorithm is $O((\hat W)^2)$.
\end{proof}

Note that our updates to the graph $G$ after each cut $(A_i,B_i)$ ensure that $G$ remains a perfect well-structured graph throughout the algorithm, so in particular, for every vertex $v\in S_2=V(G_q)$, if $e$ is the unique special edge that is incident to $v$ in $G_0$, then its other endpoint $u$ also lies in $S_2$. The partition $(S_1,S_2,S_3)$ of the vertices of $G_0$ naturally defines a partition $(\tilde S_1,\tilde S_2,\tilde S_3)$ of the vertices of $H_0$, as follows. First, for every vertex $v\in V(G_0)$, if $v\in S_i$, for $1\leq i\leq 3$, then we add $v$ to $\tilde S_i$ as well. Next, we consider the special edges of $E(G_0)$. Consider any such special edge $e=(x,y)$, and assume that $x\in S_i$ and $y\in S_{i'}$, for indices $1\leq i,i'\leq 3$. If $i=i'$, then we add all vertices lying on the path $P^*(e)$ to set $\tilde S_i$. Otherwise, from our construction $i<i'$ must hold, and moreover, from the above discussion, $i=1$ and $i'=3$ must hold. We then add all inner vertices of $P^*(e)$ to set $\tilde S_1$. Consider the resulting partition $(\tilde S_1,\tilde S_2,\tilde S_3)$ of the vertices of $V(H_0)$. As before, if $e=(x,y)$ is any edge of $H_0$ with $x\in \tilde S_i$ and $y\in \tilde S_{i'}$, then we say that $e$ is a \emph{left-to-right} edge if $i<i'$ and we say that it is a \emph{right-to-left} edge if $i'<i$. Our construction insures that, if $e'$ is a left-to-right edge of $H_0$ with respect to the partition $(\tilde S_1,\tilde S_2,\tilde S_3)$, then its parent-edge $e$ is a left-to-right edge of $G_0$ with respect to the partition $(S_1,S_2,S_3)$. Moreover, if $e$ is a left-to-right edge of $G_0$ with respect to the partition $(S_1,S_2,S_3)$, then exactly one edge of $P^*(e)$ is a left-to-right edge of $H_0$ with respect to  the partition $(\tilde S_1,\tilde S_2,\tilde S_3)$. We conclude that the total number of left-to-right edges in $H_0$ with respect to the partition $(\tilde S_1,\tilde S_2,\tilde S_3)$
is bounded by $|E'|\leq  \frac{8|T|}{\log^{30}N}$.
Notice that, since the phase has not terminated so far, the total number of vertices that were deleted from $H_0$ since the beginning of the phase is at most $\frac{\hat W}{\log^{30}N}$, and so in particular, $|\tilde S_1|+|\tilde S_3|\leq \frac{\hat W}{\log^{30}N}$.

Next, we consider two cases. The first case happens if $|\tilde S_2\cap T|\leq \left(1-\frac{4}{\log^4N}\right )\cdot |T|$. In this case, either $|\tilde S_1\cap T|\geq \frac{2|T|}{\log^4N}$, or $|\tilde S_3\cap T|\geq \frac{2|T|}{\log^4N}$.
Assume first that the former holds. In this case, we return the cut $(\tilde S,\tilde S')$, where $\tilde S=\tilde S_1$ and $\tilde S'=\tilde S_2\cup \tilde S_3$. From our discussion, $|E_{H_0}(\tilde S,\tilde S')|\leq \frac{8|T|}{\log^{30}N}\leq \frac{|T|}{\log^6N}$. Since $|\tilde S_1|\leq  \frac{\hat W}{\log^{30}N}$ and $|\tilde S_1\cap T|\geq \frac{2}{\log^4N}$, we get that $\tilde S'$ contains all but at most $\frac{\hat W}{\log^{30}N}\leq \frac{\hat W}{8}$ vertices of $H_0$, and at most $\left(1-\frac{2}{\log^4N}\right )|T|$ vertices of $T$.

Similarly, if $|\tilde S_3\cap T|\geq \frac{2|T|}{\log^4N}$, then we return the cut $(\tilde S,\tilde S')$, where $\tilde S=\tilde S_1\cup \tilde S_2$ and $\tilde S'=\tilde S_3$. Using the same reasoning as before, it is easy to see that $\tilde S$  contains all but at most $\frac{\hat W}{\log^{30}N}\leq \frac{\hat W}{8}$ vertices of $H_0$, and at most $\left(1-\frac{2}{\log^4N}\right )|T|$ vertices of $T$.

It now remains to consider the second case, where $|\tilde S_2\cap T|> \left(1-\frac{4}{\log^4N}\right )\cdot |T|$. Intuitively, we would now like to use the finalizing cut $(\tilde A,\tilde B)$ in graph $H_q$, in order to compute the cut $(\tilde S,\tilde S')$. For example, we could start with $\tilde S=\tilde S_1$ and $\tilde S'=\tilde S_3$, and then add vertices of $\tilde A$ to $\tilde S$, and vertices of $\tilde B$ to $\tilde S'$. The technical issue with this approach is that $\tilde S_2$ is a subset of vertices of graph $H_0$. Specifically, $\tilde S_2$ is the set of vertices of the subdivided graph of $G_0[S_2]$. Intuitively, we can think of graph $H_q$ as being obtained from $H_0[\tilde S_2]$ by iteratively subdividing some of the special edges. Therefore, not all vertices of $H_q$ lie in $\tilde S_2$, and so it is not immediately clear how to use the cut $(\tilde A,\tilde B)$ in order to obtain the cut $(\tilde S,\tilde S')$ in graph $H_0$ with the desired properties. In order to overcome this technical difficulty, we will slightly modify the cut $(\tilde A,\tilde B)$ in $H_q$. The goal of this modification is to ensure that, for every special edge $e\in E(G_q)$, all inner vertices of $P^*(e)$ lie on the same side of the cut, and, if both endpoints of $e$ lie on the same side of the cut, then so do all vertices of $P^*(e)$. At the same time we need to ensure that we preserve some of the useful properties of the cut $(\tilde A,\tilde B)$: namely, that $|E_{H_q}(\tilde A,\tilde B)|$ is small, and that one of the sides of the cut contains many vertices of $H_q$ but relatively few vertices of $T$.

We provide the algorithm for modifying the cut $(\tilde A,\tilde B)$ in the next observation. For simplicity, in the observation and its proof we denote the graph $H_q$ by $H$ and $G_q$ by $G$.

\begin{observation}\label{claim: transform the cut}
	There is a deterministic algorithm, that, given a finalizing cut $(\tA,\tB)$ in graph $H$, computes a cut $(\tA',\tB')$ in $H$, such that the following hold:
	
	\begin{itemize}
		\item for each special edge $e=(x,y)$ of $G$, either all inner vertices on $P^*(e)$ lie in $\tA'$, or all inner vertices of $P^*(e)$ lie in $\tB'$. Moreover, if both endpoints of $e$ lie in $\tA'$, then so do all vertices of $P^*(e)$, and if both endpoitns of $e$ lie in $\tB'$, then so do all vertices of $P^*(e)$:
			
		\item all edges in $E_H(\tA',\tB')$ are special edges;
		\item $|E_H(\tA',\tB')|\leq \frac{|T|}{8\log^6N}$; and
		
		\item one of the sides of the cut $\tilde Z\in \set{\tA',\tB'}$ contains at least $|V(H)|-\frac{\hat W}{16}$ vertices of $H$, and at most $|T|\cdot \left(1-\frac{6}{\log^4N}\right )$ vertices of $T$.
	\end{itemize}
	
	The running time of the algorithm is $O(\hat W^2)$.
\end{observation}

We complete the proof of \Cref{claim: transform the cut} below, after we complete the proof of \Cref{claim: compute intermediate cut} using it.
Consider some special edge $e\in E(G_q)$. For simplicity, we denote the subdivided path $P^*(e)$ in the corresponding graph $H_q$ by $P'(e)$, while we use $P^*(e)$ to denote the subdivided path corresponding to $e$ in $H_0$. Therefore, $P'(e)$ is obtained from $P^*(e)$ by iteratively subdividing some of the special edges.

We construct the cut $(\tilde S,\tilde S')$ in graph $H_0$ as follows. We start with $\tilde S=\tilde S_1$ and $\tilde S'=\tilde S_3$. Next, we consider all vertices $v\in S_2=V(G_q)$. For each such vertex, if $v\in \tA'$, then we place $v$ in $\tilde S$, and otherwise, we place it in $\tilde S'$. Lastly, we consider every special edge $e\in E(G_q)$. If all inner vertices of $P'(e)$ lie in $\tA'$, then we place all inner vertices of $P^*(e)$ in $\tilde S$; otherwise, we place all inner vertices of $P^*(e)$ in $\tilde S'$. 

We now verify that the resulting cut $(\tilde S,\tilde S)$ has all required properties. Consider first any edge $e\in E_{H_0}(\tilde S,\tilde S')$. If this edge is incident to a vertex of $\tilde S_1$ or to a vertex of $\tilde S_2$, then we are guaranteed that $e\in E'$ holds. Otherwise, there is a parent-edge $e'=(x,y)\in G_0$ of $e$, with $x\in \tilde A'$, $y\in \tilde B'$. In this case, path $P'(e')$ contributed exactly one special edge to the set $E_{H_q}(\tA',\tB')$ in $H_q$, and path $P^*(e')$ contributes exactly one special edge to the set $E_{H_0}(\tilde S,\tilde S')$ -- the edge $e$. Therefore, every edge in $E_{H_0}(\tilde S,\tilde S')$ is a special edge, and:

\[|E_{H_0}(\tilde S,\tilde S')|\leq |E'|+|E_H(\tA',\tB')|\leq  \frac{8|T|}{\log^{30}N}+ \frac{|T|}{8\log^6N}\leq \frac{|T|}{4\log^6N}.  \]

Recall that one of the sides of the cut $\tilde Z\in \set{\tA',\tB'}$ contains all but at most  $\frac{\hat W}{16}$ vertices of $H_q$, and at most $|T|\cdot \left(1-\frac{6}{\log^4N}\right )$ vertices of $T$.

Assume first that $\tilde Z=\tA'$. Recall that we have established that $|\tilde S_3|\leq \frac{\hat W}{\log^{30}N}$. Since $|\tB'|\leq \frac{\hat W}{16}$, we get that $|\tilde S'|\leq |\tilde S_3|+|\tB'|\leq \frac{\hat W}{\log^{30}N}+\frac{\hat W}{16}\leq \frac{\hat W}{8}$. Additionally, since we have assumed that  $|\tilde S_2\cap T|> \left(1-\frac{4}{\log^4N}\right )\cdot |T|$, we get that $|\tilde S_1\cap T|\leq \frac{4|T|}{\log^4N}$, and so $|\tilde S\cap T|\leq |\tilde S_1\cap T|+|\tA'\cap T|\leq \left(1-\frac{2}{\log^4N}\right )\cdot |T|$.

Similarly, if $\tilde Z=\tB'$, then $|\tilde S|\leq |\tilde S_1|+|\tA'|\leq \frac{\hat W}{\log^{30}N}+\frac{\hat W}{16}\leq \frac{\hat W}{8}$. Since 
we have assumed that  $|\tilde S_2\cap T|> \left(1-\frac{4}{\log^4N}\right )\cdot |T|$, we get that $|\tilde S_3\cap T|\leq \frac{4|T|}{\log^4N}$, and so $|\tilde S'\cap T|\leq |\tilde S_3\cap T|+|\tB'\cap T|\leq \left(1-\frac{2}{\log^4N}\right )\cdot |T|$.

We now analyze the running time of the algorithm. The running times of the algorithms from \Cref{claim: compute the sequence} and from \Cref{claim: transform the cut} are bounded by $O(\hat W^2)$. It is easy to verify that the running time of the remainder of the algorithm is bounded by $O(\hat W^2)$ as well, so that total running time of the algorithm is bounded by $O(\hat W^2)$.

In order to complete the proof of \Cref{claim: compute intermediate cut}, it is now enough to prove \Cref{claim: transform the cut}, which we do next. In the proof of the observation we refer to graph $H_q$ as $H$, and to graph $G_q$ as $G$.

\begin{proofof}{\Cref{claim: transform the cut}}
	We partition the set of all special edges of $E(G)$ into four subsets. The first set, $E^A$, contains all special edges $(x,y)$ with $x,y\in \tA$. The second set, $E^B$, contains all special edges $(x,y)$ with $x,y\in \tB$. The third set, $E^{AB}$ contains all special edges $(x,y)$ with $x\in \tA$ and $y\in \tB$, and the fourth set, $E^{BA}$, contains all special edges $(x,y)$ with $x\in \tB$ and $y\in \tA$. We say that a a special edge $e\in E(G)$ is \emph{bad}, if some edge of $P^*(e)$ lies in $E_H(\tA,\tB)$. Clearly, the number of bad special edges is bounded by $|E_H(\tA,\tB)|\leq \frac{|T|}{8\log^6N}$. 	 
	Recall that, since cut $(\tA,\tB)$ is finalizing, one of the sides of the cut $Z\in \set{\tA,\tB}$ contains at least $|V(H)|-\frac{\hat W}{32}$ vertices of $H$, and at most $|T|\cdot \left(1-\frac{7}{\log^4N}\right )$ vertices of $T$.
	
	We start with $(\tA',\tB')=(\tA,\tB)$, and then modify the cut $(\tA',\tB')$, in two steps. In the first step, we consider every bad edge $e\in E(G)$ in turn. If $e\in E^A\cup E^{AB}\cup E^{BA}$, then we move all inner vertices of $P^*(e)$ to set $\tA'$. Otherwise, $e\in E^B$ must hold, and we move all inner vertices of $P^*(e)$ to $\tB'$. 
	
Consider the resulting cut $(\tA',\tB')$, that is obtained after every bad special edge of $G$ has been processed.
	It is immediate to verify that $|E_H(\tA',\tB')|\leq |E_H(\tA,\tB)|\leq \frac{|T|}{8\log^6N}$ continues to hold. Indeed, every bad edge $e\in E(G)$ contributes at least one edge of $P^*(e)$ to $E_H(\tA,\tB)$, and it contributes at most one edge of $P^*(e)$ to $E_H(\tA',\tB')$. 
It is immediate to verify that every edge in $E_H(\tA',\tB')$ is a special edge, since, for every special edge $e$ of $G$, the first and the last edge on path $P^*(e)$ are special edges of $H$.
	Moreover, every edge in $E_H(\tA',\tB')$ is a special edge that is a child-edge of some bad edge in $E^{AB}$.

Recall that the number of bad edges in $G$ is bounded by $|E_H(\tA,\tB)|\leq \frac{|T|}{8\log^6N}$.
Consider the set $Z\in \set{\tA,\tB}$, such that contains at least $|V(H)|-\frac{\hat W}{32}$ vertices of $H$, and at most $|T|\cdot \left(1-\frac{7}{\log^4N}\right )$ vertices of $T$. Since Event $\event_4$ did not happen, for every bad edge $e\in E(G)$, the correpsonding path $P^*(e)$ contains at most $16\log N$ vertices of $T$. 
Therefore, the total number of vertices of $T$ that our procedure could have added to $Z$ during the first step is bounded by $ \frac{|T|}{8\log^6N}\cdot 16\log N\leq \frac{2|T|}{\log^5N}$. Therefore, after the transformation, $Z$ contains at most $|T|\cdot \left(1-\frac{6}{\log^4N}\right )$ vertices of $T$. 
Additionally, since the length of every edge $e$ in $G$ is bounded by $d'/2$, for every special edge $e\in E(G)$, the number of vertices on path $P^*(e)$ is bounded by $d'$. Therefore, the total number of vertices that our algorithm may have removed from $Z$ is bounded by $d'\cdot \frac{|T|}{8\log^6N} \leq \frac{\hat W}{\log^6N}$ (since $|T|\leq \frac{8\hat W}{d'}$, as Event $\event_1$ did not happen). Therefore, after the first step, $Z$ still contains at least $|V(H)|-\frac{3\hat W}{64}$ vertices. To summarize, after the first step, $Z$ contains at least $|V(H)|-\frac{3\hat W}{64}$ vertices, and at most $|T|\cdot \left(1-\frac{6}{\log^4N}\right )$ vertices of $T$.

Consider any special edge $e=(x,y)$ of $G$. If $x,y\in \tA'$, and any vertex of $P^*(e)$ originally lied in $\tB'$, then $e$ must have been a bad edge, as some edge on $P^*(e)$ must have lied in $E_H(\tA,\tB)$. Our algorithm then moved all inner vertices on $P^*(e)$ to $\tA'$. Simiarly, if $x,y\in \tB'$, then our algorithm ensures that all inner vertices on $P^*(e)$ lie in $\tB'$. If $x\in \tA$ and $y\in \tB$ held, then $e$ must be a bad edge, since some edge on $P^*(e)$ must lie in $E_H(\tA,\tB)$. Our algorithm moved all inner vertices of $P^*(e)$ to $\tA'$. Lastly, assume that $x\in \tB'$ and $y\in \tA'$ holds. If $e$ is a bad edge, then we have moved all inner vertices of $P^*(e)$ to $\tA'$. It now only remains to take care of the subset $E^*\subseteq E^{BA}$ of special edges, that contains every edge $e\in E^{BA}$ that is not bad. We take care of these edges in the second step.

Intuitively, in the second step, we consider each edge $e\in E^*$ one by one, and for each such edge, we move all inner vertices of $P^*(e)$ to one side of the cut. However, we need to be careful so that we do not remove too many vertices from the set $Z$, and that we do not add too many vertices of $T$ to $Z$.

We define a subset $E^*_1\subseteq E^*$ of edges to contain every edge $e\in E^*$, such that the path $P^*(e)$ contains at least one vertex of $T$. If $|E^*_1|>\ceil{\frac{6|T|}{\log^4N}}$, then we discard edges arbitrarily from $E^*_1$, until $|E^*_1|=\ceil{\frac{6|T|}{\log^4N}}$ holds.

We now consider two cases. The first case happens when $Z=\tA$.
In this case, for every edge $e\in E^*_1$, we move all inner vertices on path $P^*(e)$ to $\tB'$, and for every edge $e\in E^*\setminus E^*_1$, we move all inner vertices on path $P^*(e)$ to $\tA'$. Notice that, if $|E^*_1|=\frac{6|T|}{\log^4N}$, then this process guarantees that $|\tB'\cap T|\geq  \frac{6|T|}{\log^4N}$ holds at the end of the second step, and so $|\tA'\cap T|\leq  |T|\cdot \left(1-\frac{6}{\log^4N}\right )$ holds. If $|E^*_1|<\frac{6|T|}{\log^4N}$, then we did not add any new terminals to set $\tA'$ in this step, and so $|\tA'\cap T|\leq  |T|\cdot \left(1-\frac{6}{\log^4N}\right )$  continues to hold. Since, for every edge $e\in E^*_1$, the number of inner vertices on path $P^*(e)$ is bounded by $d'$, the number of vertices that we have deleted from $\tA'$ in this step is bounded by:

\[d'\cdot |E^*_1|\leq d'\cdot \ceil{\frac{6|T|}{\log^4N}}\leq d'\cdot \frac{100\hat W}{d'\log^4N}\leq \frac{\hat W}{\log^3N}. \]

(We have used the fact that $|T|\leq \frac{8\hat W}{d'}$, since Event $\event_1$ did not happen).
Therefore, after the second step, $\tA'$ contains at least $|V(H)|-\frac{3\hat W}{64}-\frac{\hat W}{\log^3N}\geq |V(H)|-\frac{\hat W}{16}$ vertices.	
	
If $Z=\tB$, then 	for every edge $e\in E^*_1$, we move all inner vertices on path $P^*(e)$ to $\tA'$, and for every edge $e\in E^*\setminus E^*_1$, we move all inner vertices on path $P^*(e)$ to $\tB'$. Using the same arguments as before, we conclude that, after the second step, $|\tB'\cap T|\leq  |T|\cdot \left(1-\frac{6}{\log^4N}\right )$  and $|\tB'|\geq  |V(H)|-\frac{\hat W}{8}$ hold.
	
It is easy to verify that the time required to compute the modified cut $(\tA',\tB')$ is bounded by $O(|E(H)|)\leq O(\hat W^2)$. 
\end{proofof}

\subsection{Proof of \Cref{obs: valid input for routeandcut}}
\label{subsec: bad params no event}

%==================================
%==================================
%==================================
%==================================
%==================================

It is enough to prove that inequalities  $r\leq \ceil{\sqrt{\log N'}}$, $1\leq \eta'\leq \hat \Delta$, and $\frac{(|V(H)|-|B|)\cdot \eta'}{\hat \Delta}\leq 2^{r\sqrt{\log N'}}$ all hold.

%First, it is immediate to verify that $1\leq \hat \Delta\leq N'$, 
%since $\hat \Delta\leq |T|\leq |V(H)|=\hat W$, while $N'\geq \hat W$. Additionally, 
First, since $r\leq \ceil{\sqrt{\log N}}$, and $N\leq N'$, it is immediate to verify that $r\leq \ceil{\sqrt{\log N'}}$.

Second, since Event $\event_1$ did not happen, $\hat \Delta\geq \frac{|T|}{4}\geq \frac{\hat W}{8d'}$.
Since $\eta'=(\log N')^{8}$, we then get that:

\[
\begin{split}
\hat \Delta\geq\frac{\hat W}{8d'}\geq \frac{n}{16d'}\geq \frac{n\cdot (\log N)^{28}}{32d}\geq \frac{(\log N)^{28}}{64}\geq \eta'.
\end{split}
\] 

(we have used the fact that $\hat W\geq|V(G)|\geq \frac{n}{2}$; $d'\leq \frac{2d}{(\log N)^{28}}$;  $d\leq n$ from Inequaltiy \ref{eq: bound on d}), and $\eta'=(\log N')^8$. Therefore, $1\leq \eta'\leq \hat \Delta$ holds.
%	\[|T|\geq \frac{\hat W}{2d'}\ge \frac{n}{4d'}\geq 4(\log n)^{6\hat c}\geq  4(\log N')^{\hat c+7}=4\eta'\] 

% (we have used the fact that $d\leq n$, $d'=\frac{d}{\poly\log n}\leq \frac{n}{\poly\log n}$, and, from Inequality \ref{eq: comparing logs}, $\log n\geq \sqrt{\log N}$, while $N'\leq 8N$).
%	Therefore, if Event $\event$ does not happen, $\hat \Delta\geq \frac{|T|}{4}\geq \eta'$ holds, and so  $1\le \eta'\leq \hat \Delta$, as required.
Lastly:

\[\frac{(|V(H)|-|B|)\cdot \eta'}{\hat \Delta}\leq \frac{\hat W\cdot \eta'}{\hat \Delta}\leq 8d' (\log N')^{8}\leq  d\leq  2^{r\cdot \sqrt{\log N}}\leq  2^{r\cdot \sqrt{\log N'}}. \]

%==================================
%==================================
%==================================
%==================================
%==================================
%\subsection{Proof of \Cref{cor: valid input for routeandcut}}
%\label{subsec: valid input for routeandcut}

%==================================
%==================================
%==================================
%==================================
%==================================

\subsection{Proof of \Cref{claim: finalize the cut}}
\label{subsec: proof of finalize the cut}

We start by turning the cut $(Z,Z')$ into a weakly well-structured one, using the following straightforward algorithm. Consider every vertex $v\in Z$ one by one. If there is a regular edge $(v,u)\in E_H(Z,Z')$, then we move $v$ to $Z'$. Notice that, if a regular edge $(v,u)$ exists, then $v\in L$ must hold, so moving $v$ to $Z'$ may introduce at most one new edge into the set $E_H(Z,Z')$ -- the special edge that is incident to $v$. At the same time, at least one regular edge is deleted from $E_H(Z,Z')$, so $|E_H(Z,Z')|$ does not increase. At the end of this procedure, cut $(Z,Z')$ becomes weakly well-structured, and $|E_H(Z,Z')|\leq \frac{\hat W}{d'(\log N)^7}$ continues to hold. Since  $|T|\geq \frac{\hat W}{2d'}$, as Event $\event_1$ did not happen, we get that $|E_H(Z,Z')|\leq \frac{2|T|}{\log^7N}$. The number of vertices that were deleted from $Z$ is bounded by the original number of edges in $E_H(Z,Z')$, which is bounded by $\frac{\hat W}{d'(\log N)^7}\leq \frac{2|T|}{\log^7N}\leq \frac{|T|}{8\log^6N}$. Therefore, at the end of this procedure, $|Z\cap T|,|Z'\cap T|\geq \frac{7|T|}{8\log^4N}$ holds. Notice that the running time of the algorithm so far is bounded by $O(|E(H)|)\leq O(\hat W^2)$.
If either of the sets $Z$ or $Z'$ contains at least $|V(H)|-\frac{\hat W}{32}=\frac{31\hat W}{32}$ vertices, then cut $(Z,Z')$ is a finalizing cut (see \Cref{def: finalizing cut}). We then terminate the algorithm and return this cut.

We assume from now on that $|Z|,|Z'|\geq \frac{\hat W}{32}$. Since $|E_H(Z,Z')|\leq \frac{\hat W}{d'(\log N)^7}\leq \frac{32}{d'\cdot (\log N)^7}\cdot \min\set{|Z|,|Z'|}$, we get that the sparsity of the cut $(Z,Z')$ is at most $\frac{32}{d'\cdot (\log N)^7}$.

Next, we turn the cut $(Z,Z')$ into a 
cut $(A',B')$ in $G$, in a natural way: we let $A'=Z\cap V(G)$, and $B'=Z'\cap V(G)$.
Recall that the length of every edge in $G$ is at most $\frac{d'}{2}$. From \Cref{obs: subdivided cut sparsity}, $(A',B')$ is a weakly well-structured cut in $G$ of sparsity at most $\frac{256}{d'\log^7N}$, with $w(A')\geq \frac{Z}{8}\geq \frac{\hat W}{256}$ and $w(B')\geq \frac{|Z'|}8\geq \frac{\hat W}{256}$. The time required to compute the cut $(A',B')$ in $G$ given the cut $(Z,Z')$ in $H$ is bounded by $O(|V(H)|)\leq O( \hat W)$.

Recall that for every special edge $e\in E(G)$, $\ell(e)\leq \frac{d'}2\leq \frac{n}{16}$, while $\sum_{e\in E(G)}\ell(e)\geq  \frac{n}{4}$. 
We can now apply the algorithm from \Cref{obs: weakly to strongly well str} to cut $(A',B')$, to compute a strongly well-structured cut $(A,B)$ in $G$ with $w(A)\geq \frac{w(A')}{4}\geq \frac{\hat W}{1024}\geq \frac{\hat W}{\log^{30}N}$ and $w(B)\geq \frac{w(B')}{4}\geq \frac{\hat W}{\log^{30}N}$, whose sparsity is at most $\frac{1024}{d'\log^7N}\leq \frac{1}{d'}$. 
The running time of the algorithm from  \Cref{obs: weakly to strongly well str} is $O(|E(H)|)\leq O(\hat W^2)$. We then return the cut $(A,B)$.

From the above discussion, the total running time of the algorithm is $O(\hat W)^2$.

\subsection{Proof of \Cref{claim: valid input to conencttocenters}}
\label{subsec: valid inst of connecttocenters}
Clearly, $G$ is a perfect well-structured graph, and the  assignment $\ell'(e)$ of lengths to its edges $e\in E(G)$ is a proper assignment. The current weight of all vertices in $G$ with respect to the edge lengths $\ell'(\cdot)$ is $W'\leq 2\hat W\leq 24W^0$, while $N\ge W^0$ holds from the definition of the \maintaincluster problem. Therefore, $N'=8N\geq 8W^0\geq 8n\geq 8|V(G)|$, and $N'=8N\geq 8W^0\geq  \frac{W'}{8}$ holds, so $N'\geq \max\set{|V(G)|,\frac{ W'}{8}}$, as required.
Since $N$ is large enough from the definition of the \maintaincluster problem, so is $N'$.
Since Event $\event_1$ did not happen, $|\beta|=|V(X)|=|T|\geq \frac{\hat W}{2d'}\geq \frac{W'}{256d'}$, as required. 
Since $1\leq r\leq \ceil{\sqrt{\log N}}$ holds from the definition of $r$-restricted \maintaincluster problem, and $N'\geq N$, we get that $1\leq r\leq \ceil{\sqrt{\log N'}}$.
Lastly, we need to verify that $d'\leq\min\set{\frac{2^{r\sqrt{\log N'}}}{\log^8N'},\frac{|V(G)|}{8}}$ holds.
Indeed, recall that $d'\leq \frac{2d}{(\log N)^{28}}$. Moreover, from Inequality 
\ref{eq: bound on d},  $d\leq n$, and from the  definition of the $r$-restricted \maintaincluster problem, 
$d\leq 2^{r\cdot \sqrt{\log N}}$. Therefore: 

\[\begin{split}
d'&\leq \frac{2d}{(\log N)^{28}}\\
&\leq \min\set{\frac{2n}{(\log N)^{28}},\frac{2\cdot  2^{r\cdot \sqrt{\log N}}}{(\log N)^{28}}}\\
&\leq \min\set{\frac{2^{30}\cdot n}{(\log N')^{28}},\frac{2^{29}\cdot  2^{r\cdot \sqrt{\log N'}}}{(\log N')^{28}}}\\
&\leq \min\set{\frac{|V(G)|}{8},\frac{2^{r\sqrt{\log N'}}}{\log^8N'}}.
\end{split}\]

(We have used the fact that, from Inequality \ref{eq: more logs}, $\log N'\leq 2\log N$).

%% file: even-faster-classical-matching.bbl
\newcommand{\etalchar}[1]{$^{#1}$}
\def\cprime{$'$} \def\cprime{$'$}
\begin{thebibliography}{vdBCK{\etalchar{+}}23}

\bibitem[AAP93]{AAP93}
Baruch Awerbuch, Yossi Azar, and Serge~A. Plotkin.
\newblock Throughput-competitive on-line routing.
\newblock In {\em 34th Annual Symposium on Foundations of Computer Science,
  Palo Alto, California, USA, 3-5 November 1993}, pages 32--40, 1993.

\bibitem[ABF23]{DBLP:conf/stoc/AbboudBF23}
Amir Abboud, Karl Bringmann, and Nick Fischer.
\newblock Stronger 3-sum lower bounds for approximate distance oracles via
  additive combinatorics.
\newblock In Barna Saha and Rocco~A. Servedio, editors, {\em Proceedings of the
  55th Annual {ACM} Symposium on Theory of Computing, {STOC} 2023, Orlando, FL,
  USA, June 20-23, 2023}, pages 391--404. {ACM}, 2023.

\bibitem[ABKZ22]{abboud2022hardness}
Amir Abboud, Karl Bringmann, Seri Khoury, and Or~Zamir.
\newblock Hardness of approximation in p via short cycle removal: Cycle
  detection, distance oracles, and beyond.
\newblock {\em arXiv preprint arXiv:2204.10465}, 2022.

\bibitem[AMV20]{AMV20}
Kyriakos Axiotis, Aleksander Madry, and Adrian Vladu.
\newblock Circulation control for faster minimum cost flow in unit-capacity
  graphs.
\newblock In {\em 61st {IEEE} Annual Symposium on Foundations of Computer
  Science, {FOCS} 2020, Durham, NC, USA, November 16-19, 2020}, pages 93--104,
  2020.

\bibitem[Ber17]{Bernstein}
Aaron Bernstein.
\newblock Deterministic partially dynamic single source shortest paths in
  weighted graphs.
\newblock In {\em LIPIcs-Leibniz International Proceedings in Informatics},
  volume~80. Schloss Dagstuhl-Leibniz-Center for Computer Science, 2017.

\bibitem[BGS20]{SCC}
Aaron Bernstein, Maximilian~Probst Gutenberg, and Thatchaphol Saranurak.
\newblock Deterministic decremental reachability, scc, and shortest paths via
  directed expanders and congestion balancing.
\newblock In {\em 2020 IEEE 61st Annual Symposium on Foundations of Computer
  Science (FOCS)}, pages 1123--1134. IEEE, 2020.

\bibitem[BGWN20]{AlmostDAG2}
Aaron Bernstein, Maximilian~Probst Gutenberg, and Christian Wulff-Nilsen.
\newblock Near-optimal decremental sssp in dense weighted digraphs.
\newblock In {\em 2020 IEEE 61st Annual Symposium on Foundations of Computer
  Science (FOCS)}, pages 1112--1122. IEEE, 2020.

\bibitem[CGL{\etalchar{+}}20]{detbalanced}
Julia Chuzhoy, Yu~Gao, Jason Li, Danupon Nanongkai, Richard Peng, and
  Thatchaphol Saranurak.
\newblock A deterministic algorithm for balanced cut with applications to
  dynamic connectivity, flows, and beyond.
\newblock In {\em 2020 IEEE 61st Annual Symposium on Foundations of Computer
  Science (FOCS)}, pages 1158--1167. IEEE, 2020.

\bibitem[CK24]{CK24}
Julia Chuzhoy and Sanjeev Khanna.
\newblock A faster combinatorial algorithm for maximum bipartite matching.
\newblock In {\em Proceedings of the Thirty-Fifth Annual {ACM-SIAM} Symposium
  on Discrete Algorithms (SODA)}, 2024.

\bibitem[CKL{\etalchar{+}}22]{ChenKLPGS22}
Li~Chen, Rasmus Kyng, Yang~P. Liu, Richard Peng, Maximilian~Probst Gutenberg,
  and Sushant Sachdeva.
\newblock Maximum flow and minimum-cost flow in almost-linear time.
\newblock In {\em 63rd {IEEE} Annual Symposium on Foundations of Computer
  Science, {FOCS} 2022, Denver, CO, USA, October 31 - November 3, 2022}, pages
  612--623, 2022.

\bibitem[CMSV17]{CMSV17}
Michael~B. Cohen, Aleksander Madry, Piotr Sankowski, and Adrian Vladu.
\newblock Negative-weight shortest paths and unit capacity minimum cost flow in
  {\~{o}} (\emph{m}\({}^{\mbox{10/7}}\) log \emph{W}) time (extended abstract).
\newblock In {\em Proceedings of the Twenty-Eighth Annual {ACM-SIAM} Symposium
  on Discrete Algorithms, {SODA} 2017, Barcelona, Spain, Hotel Porta Fira,
  January 16-19}, pages 752--771, 2017.

\bibitem[DHZ00]{DorHZ00}
Dorit Dor, Shay Halperin, and Uri Zwick.
\newblock All-pairs almost shortest paths.
\newblock {\em {SIAM} J. Comput.}, 29(5):1740--1759, 2000.

\bibitem[Din06]{Dinitz}
Yefim Dinitz.
\newblock Dinitz' algorithm: The original version and {Even's} version.
\newblock In {\em Theoretical computer science}, pages 218--240. Springer,
  2006.

\bibitem[DP09]{dubhashi2009concentration}
Devdatt~P Dubhashi and Alessandro Panconesi.
\newblock {\em Concentration of measure for the analysis of randomized
  algorithms}.
\newblock Cambridge University Press, 2009.

\bibitem[DS08]{DS08}
Samuel~I. Daitch and Daniel~A. Spielman.
\newblock Faster approximate lossy generalized flow via interior point
  algorithms.
\newblock In Cynthia Dwork, editor, {\em Proceedings of the 40th Annual {ACM}
  Symposium on Theory of Computing, Victoria, British Columbia, Canada, May
  17-20, 2008}, pages 451--460. {ACM}, 2008.

\bibitem[ES81]{EvenS}
Shimon Even and Yossi Shiloach.
\newblock An on-line edge-deletion problem.
\newblock {\em Journal of the ACM (JACM)}, 28(1):1--4, 1981.

\bibitem[FF56]{FF56}
L.~R. {Ford, Jr.} and D.~R. Fulkerson.
\newblock Maximal flow through a network.
\newblock {\em Canadian Journal of Mathematics}, 8:399--404, 1956.

\bibitem[Fle00]{Fleischer00}
Lisa Fleischer.
\newblock Approximating fractional multicommodity flow independent of the
  number of commodities.
\newblock {\em {SIAM} J. Discrete Math.}, 13(4):505--520, 2000.

\bibitem[Gab17]{Gabow17}
Harold~N. Gabow.
\newblock The weighted matching approach to maximum cardinality matching.
\newblock {\em Fundam. Informaticae}, 154(1-4):109--130, 2017.

\bibitem[GG81]{GabberG81}
Ofer Gabber and Zvi Galil.
\newblock Explicit constructions of linear-sized superconcentrators.
\newblock {\em J. Comput. Syst. Sci.}, 22(3):407--420, 1981.
\newblock announced at FOCS'79.

\bibitem[GK98]{GK98}
Naveen Garg and Jochen K{\"{o}}nemann.
\newblock Faster and simpler algorithms for multicommodity flow and other
  fractional packing problems.
\newblock In {\em 39th Annual Symposium on Foundations of Computer Science,
  {FOCS} '98, November 8-11, 1998, Palo Alto, California, {USA}}, pages
  300--309, 1998.

\bibitem[GK04]{GoldbergK04}
Andrew~V. Goldberg and Alexander~V. Karzanov.
\newblock Maximum skew-symmetric flows and matchings.
\newblock {\em Math. Program.}, 100(3):537--568, 2004.

\bibitem[GKK10]{GKK10}
Ashish Goel, Michael Kapralov, and Sanjeev Khanna.
\newblock Perfect matchings in o(\emph{n} log \emph{n}) time in regular
  bipartite graphs.
\newblock In {\em Proceedings of the 42nd {ACM} Symposium on Theory of
  Computing, {STOC} 2010, Cambridge, Massachusetts, USA, 5-8 June 2010}, pages
  39--46, 2010.

\bibitem[GWN20]{gutenberg2020decremental}
Maximilian~Probst Gutenberg and Christian Wulff-Nilsen.
\newblock Decremental sssp in weighted digraphs: Faster and against an adaptive
  adversary.
\newblock In {\em Proceedings of the Fourteenth Annual ACM-SIAM Symposium on
  Discrete Algorithms}, pages 2542--2561. SIAM, 2020.

\bibitem[HK73]{HK73}
John~E. Hopcroft and Richard~M. Karp.
\newblock An n\({}^{\mbox{5/2}}\) algorithm for maximum matchings in bipartite
  graphs.
\newblock {\em {SIAM} J. Comput.}, 2(4):225--231, 1973.

\bibitem[HK95]{HenzingerKing}
Monika~Rauch Henzinger and Valerie King.
\newblock Fully dynamic biconnectivity and transitive closure.
\newblock In {\em Foundations of Computer Science, 1995. Proceedings., 36th
  Annual Symposium on}, pages 664--672. IEEE, 1995.

\bibitem[HKNS15]{HenzingerKNS15}
Monika Henzinger, Sebastian Krinninger, Danupon Nanongkai, and Thatchaphol
  Saranurak.
\newblock Unifying and strengthening hardness for dynamic problems via the
  online matrix-vector multiplication conjecture.
\newblock In {\em Proceedings of the forty-seventh annual ACM symposium on
  Theory of computing}, pages 21--30, 2015.

\bibitem[IM81]{IM81}
Oscar~H. Ibarra and Shlomo Moran.
\newblock Deterministic and probabilistic algorithms for maximum bipartite
  matching via fast matrix multiplication.
\newblock {\em Inf. Process. Lett.}, 13(1):12--15, 1981.

\bibitem[Kin99]{ES-tree-directed}
Valerie King.
\newblock Fully dynamic algorithms for maintaining all-pairs shortest paths and
  transitive closure in digraphs.
\newblock In {\em 40th Annual Symposium on Foundations of Computer Science
  (Cat. No. 99CB37039)}, pages 81--89. IEEE, 1999.

\bibitem[KKOV07]{KKOV}
Rohit Khandekar, Subhash Khot, Lorenzo Orecchia, and Nisheeth~K Vishnoi.
\newblock On a cut-matching game for the sparsest cut problem.
\newblock {\em Univ. California, Berkeley, CA, USA, Tech. Rep.
  UCB/EECS-2007-177}, 6(7):12, 2007.

\bibitem[KRV09]{KRV}
Rohit Khandekar, Satish Rao, and Umesh Vazirani.
\newblock Graph partitioning using single commodity flows.
\newblock {\em Journal of the ACM (JACM)}, 56(4):19, 2009.

\bibitem[Lou10]{directed-CMG}
Anand Louis.
\newblock Cut-matching games on directed graphs.
\newblock {\em CoRR}, abs/1010.1047, 2010.

\bibitem[LRS13]{LRS13}
Yin~Tat Lee, Satish Rao, and Nikhil Srivastava.
\newblock A new approach to computing maximum flows using electrical flows.
\newblock In {\em Symposium on Theory of Computing Conference, STOC'13, Palo
  Alto, CA, USA, June 1-4, 2013}, pages 755--764, 2013.

\bibitem[LS19]{LS19}
Yin~Tat Lee and Aaron Sidford.
\newblock Solving linear programs with sqrt(rank) linear system solves.
\newblock {\em CoRR}, abs/1910.08033, 2019.

\bibitem[LS20]{LS20_stoc}
Yang~P. Liu and Aaron Sidford.
\newblock Faster energy maximization for faster maximum flow.
\newblock In {\em Proccedings of the 52nd Annual {ACM} {SIGACT} Symposium on
  Theory of Computing, {STOC} 2020, Chicago, IL, USA, June 22-26, 2020}, pages
  803--814, 2020.

\bibitem[Mad10]{madry2010faster}
Aleksander Madry.
\newblock Faster approximation schemes for fractional multicommodity flow
  problems via dynamic graph algorithms.
\newblock In {\em Proceedings of the forty-second ACM symposium on Theory of
  computing}, pages 121--130, 2010.

\bibitem[Mad13]{Madry13}
Aleksander Madry.
\newblock Navigating central path with electrical flows: From flows to
  matchings, and back.
\newblock In {\em 54th Annual {IEEE} Symposium on Foundations of Computer
  Science, {FOCS} 2013, 26-29 October, 2013, Berkeley, CA, {USA}}, pages
  253--262, 2013.

\bibitem[Mad16]{Madry16}
Aleksander Madry.
\newblock Computing maximum flow with augmenting electrical flows.
\newblock In {\em {IEEE} 57th Annual Symposium on Foundations of Computer
  Science, {FOCS} 2016, 9-11 October 2016, Hyatt Regency, New Brunswick, New
  Jersey, {USA}}, pages 593--602, 2016.

\bibitem[Mar73]{Margulis}
G.~A. Margulis.
\newblock Explicit construction of concentrators.
\newblock {\em Problemy Peredafi Iqfiwmacii}, 9(4):71--80, 1973.
\newblock (English translation in Problems Inform. Transmission (1975)).

\bibitem[MS04]{MS04}
Marcin Mucha and Piotr Sankowski.
\newblock Maximum matchings via gaussian elimination.
\newblock In {\em 45th Symposium on Foundations of Computer Science {(FOCS}
  2004), 17-19 October 2004, Rome, Italy, Proceedings}, pages 248--255, 2004.

\bibitem[MV80]{MicaliV80}
Silvio Micali and Vijay~V. Vazirani.
\newblock An o(sqrt({\(\vert\)}v{\(\vert\)}) {\(\vert\)}e{\(\vert\)}) algorithm
  for finding maximum matching in general graphs.
\newblock In {\em 21st Annual Symposium on Foundations of Computer Science,
  Syracuse, New York, USA, 13-15 October 1980}, pages 17--27. {IEEE} Computer
  Society, 1980.

\bibitem[RZ11]{RodittyZ11}
Liam Roditty and Uri Zwick.
\newblock On dynamic shortest paths problems.
\newblock {\em Algorithmica}, 61(2):389--401, 2011.

\bibitem[Sch03]{Sch03}
A.~Schrijver.
\newblock {\em Combinatorial Optimization: Polyhedra and Efficiency}.
\newblock Springer, 2003.

\bibitem[ST04]{ST04}
Daniel~A. Spielman and Shang{-}Hua Teng.
\newblock Nearly-linear time algorithms for graph partitioning, graph
  sparsification, and solving linear systems.
\newblock In {\em Proceedings of the 36th Annual {ACM} Symposium on Theory of
  Computing, Chicago, IL, USA, June 13-16, 2004}, pages 81--90, 2004.

\bibitem[Vaz94]{Vazirani94}
Vijay~V. Vazirani.
\newblock A theory of alternating paths and blossoms for proving correctness of
  the \emph{O}(sqrt\{\emph{v e}\}) general graph maximum matching algorithm.
\newblock {\em Comb.}, 14(1):71--109, 1994.

\bibitem[vdBCK{\etalchar{+}}23]{brand2023deterministic}
Jan van~den Brand, Li~Chen, Rasmus Kyng, Yang~P. Liu, Richard Peng,
  Maximilian~Probst Gutenberg, Sushant Sachdeva, and Aaron Sidford.
\newblock A deterministic almost-linear time algorithm for minimum-cost flow.
\newblock In {\em 64th {IEEE} Annual Symposium on Foundations of Computer
  Science (FOCS), (to appear)}, 2023.

\bibitem[vdBLN{\etalchar{+}}20]{BrandLNPSS0W20}
Jan van~den Brand, Yin~Tat Lee, Danupon Nanongkai, Richard Peng, Thatchaphol
  Saranurak, Aaron Sidford, Zhao Song, and Di~Wang.
\newblock Bipartite matching in nearly-linear time on moderately dense graphs.
\newblock In {\em 61st {IEEE} Annual Symposium on Foundations of Computer
  Science, {FOCS} 2020, Durham, NC, USA, November 16-19, 2020}, pages 919--930,
  2020.

\end{thebibliography}
